\documentclass[12pt,twoside]{mitthesis}
\usepackage{lgrind}

\usepackage{cite,cmap,mathtools,amssymb,braket,mathrsfs,tikz,subfig,xspace,xcolor,float,color,siunitx,comment,afterpage,multirow,array,hhline,framed,colortbl}
\usepackage[T1]{fontenc}
\usepackage[utf8]{inputenc}
\usepackage[countmax]{subfloat}
\usepackage[hidelinks]{hyperref}

\usepackage{amsthm}
\newtheorem*{EFPMultiCorre}{Multigraph/EFP Correspondence}
\newtheorem*{conceptdef}{Quark and Gluon Jet Definition (Conceptual)~\cite{Gras:2017jty}}
\newtheorem*{opdef}{Quark and Gluon Jet Definition (Operational)}
\newtheorem{lemma}{Lemma}
\newtheorem{conjecture}{Conjecture}

\newtheorem{definition}{Definition}
\newtheorem*{ircsafety}{Infrared and Collinear Safety}

\providecommand{\href}[2]{#2}
\newcommand{\pd}[2]{\frac{\partial #1}{\partial #2}}
\DeclareMathOperator{\tr}{tr}
\DeclareMathOperator*{\argmin}{arg\,min}

\DeclareRobustCommand{\Sec}[1]{Sec.~\ref{#1}}
\DeclareRobustCommand{\Secs}[2]{Secs.~\ref{#1} and \ref{#2}}

\DeclareRobustCommand{\App}[1]{App.~\ref{#1}}
\DeclareRobustCommand{\Tab}[1]{Table~\ref{#1}}
\DeclareRobustCommand{\Tabs}[2]{Tables~\ref{#1} and \ref{#2}}
\DeclareRobustCommand{\Fig}[1]{Fig.~\ref{#1}}
\DeclareRobustCommand{\Figs}[2]{Figs.~\ref{#1} and \ref{#2}}

\DeclareRobustCommand{\Eq}[1]{Eq.~(\ref{#1})}
\DeclareRobustCommand{\Eqs}[2]{Eqs.~(\ref{#1}) and (\ref{#2})}
\DeclareRobustCommand{\Eqss}[3]{Eqs.~(\ref{#1}), (\ref{#2}), and (\ref{#3})}
\DeclareRobustCommand{\Ref}[1]{Ref.~\cite{#1}}
\DeclareRobustCommand{\Refs}[1]{Refs.~\cite{#1}}
\newcommand{\pythia}{\textsc{Pythia}\xspace}

\newcommand{\B}{\text{EFP}\xspace}
\newcommand{\Bs}{\text{EFPs}\xspace}
\newcommand{\EMD}{\text{EMD}\xspace}
\newcommand{\E}{\mathcal{E}}
\renewcommand{\O}{\mathcal{O}}

\definecolor{table_red}{rgb}{0.501960,0.000000,0.008131}
\definecolor{table_yellow}{rgb}{0.501959,0.501930,0.014756}
\definecolor{table_green}{rgb}{0.065930,0.501799,0.006832}
\definecolor{table_teal}{rgb}{0.064121,0.501820,0.501973}
\definecolor{table_red_bg}{RGB}{249,243,242}
\definecolor{table_yellow_bg}{RGB}{249,249,243}
\definecolor{table_green_bg}{RGB}{245,249,244}
\definecolor{table_teal_bg}{RGB}{244,249,249}

\begin{document}

\title{Energy Flow in Particle Collisions}

\author{Eric Mario Metodiev}

 \prevdegrees{A.B., Physics and Mathematics, Harvard University (2016) \\
                    A.M., Physics, Harvard University (2016)}

\department{Department of Physics}

\degree{Doctor of Philosophy}

\degreemonth{September}
\degreeyear{2020}
\thesisdate{July 30, 2020}

\supervisor{Jesse Thaler}{Associate Professor of Physics}

\chairman{Nergis Mavalvala}{Associate Department Head, Physics}

\maketitle

\cleardoublepage

\setcounter{savepage}{\thepage}

\begin{abstractpage}
In this thesis, I introduce a new bottom-up approach to quantum field theory and collider physics, beginning from the observable energy flow: the energy distribution produced by particle collisions.
First, I establish a metric space for collision events by comparing their energy flows.
I unify many ideas spanning multiple decades, such as observables and jets, as simple geometric objects in this new space.
Second, I develop a basis of observables by systematically expanding in particle energies and angles, encompassing many existing observables and uncovering new analytic structures.
I highlight how the traditional criteria for theoretical calculability emerge as consistency conditions, due to the redundancy of describing an event using particles rather than its energy flow.
Finally, I propose a definition of particle type, or flavor, which makes use of only observable information.
This definition requires refining the notion of flavor from a per-event label to a statistical category, and I showcase its direct experimental applicability at colliders.
Throughout, I synthesize concepts from particle physics with ideas from statistics and computer science to expand the theoretical understanding of particle interactions and enhance the experimental capabilities of collider data analysis techniques.
\end{abstractpage}

\cleardoublepage

\section*{Acknowledgments}

I am deeply grateful for the many great people I have had the chance to interact with during my time at MIT.
I have had the privilege of learning a huge amount from the members of the phenomenology group in the CTP, especially during the Friday journal clubs.
I would also like to thank my officemates: Adam Bene Watts, Patrick Fitzpatrick, Linghang Kong, and Mehdi Soleimanifar.
I absorbed a great deal about quantum information and dark matter from them by osmosis, and they made the office a terrific place to work.
One of the major highlights of my time at MIT (and before) was collaborating with Patrick Komiske, with whom I had an amazing collaboration that produced many of the results presented in this thesis.
I learned the most about particle physics while studying for the qualifying exam, and he, Anthony Grebe, and Tej Kanwar made the experience a very memorable one.

I am lucky to have many wonderful colleagues and collaborators beyond MIT.
From the world of experimental particle physics, I would particularly like to thank Matt LeBlanc, Ben Nachman, and Jennifer Roloff.
I would also like to acknowledge Andrew Larkoski for great discussions and fruitful collaborations.
I was fortunate to spend a year at Harvard as a Visiting Fellow in the Center for the Fundamental Laws of Nature, where the phenomenology group at Harvard was exceptionally welcoming.
I have learned a tremendous amount from Matt Schwartz, who has been a fantastic pedagogue and mentor to me since before I began at MIT.
I also need to thank many colleagues at Harvard for great discussions and collaborations, particularly: Anders Andreassen, Cari Cesarotti, Chris Frye, and Hofie Hannesdottir.

I am eternally grateful to my wife, Miruna Oprescu for far more than I can list here.
I am also thankful for my parents and grandparents, who originally inspired me to go down the path of scientific research by their words and their actions.

Lastly, I would like to thank my advisor, Jesse Thaler, from whom I learned an unbelievable amount about particle physics and far beyond.
He has deeply influenced how I do research, the way I give talks, my style of writing papers, and much more.
I am forever grateful to Jesse, and I could not imagine a more excellent time at MIT.

\pagestyle{plain}
\tableofcontents
\newpage
\listoffigures
\newpage
\listoftables

\chapter{Introduction}

\section{Particle Collisions and their Energy Flow}

Particles and their interactions give rise to the richness of the universe around us.
The standard model of particle physics consists of all of the presently known particles and forces.
Quantum field theory is the theoretical framework for understanding and predicting the interactions between these particles.
Among them, electrons and photons are described by the theory of electromagnetism: quantum electrodynamics, while quarks and gluons are governed by the theory of the strong force: quantum chromodynamics.
To search for new particles and forces beyond the standard model, particles are collided at high energies to provide an experimental probe of physics at increasingly small distances.
These particle collisions, or ``events,'' produce complex mosaics of particles, encoding the high energy interactions that took place.
By better theoretically understanding and experimentally analyzing the patterns of particles collisions, we can gain new insights into the fundamental interactions of nature.

An essential perspective for understanding particle interactions and collisions is to consider the energy flow, or angular distribution of energy, that they produce.
For an event consisting of $M$ particles with four-momenta $\{p_i^\mu\}_{i=1}^M$, the energy flow is:
\begin{equation}
\mathcal E(\hat n) = \sum_{i=1}^M E_i \,\delta(\hat n - \hat n_i),
\end{equation}
where $E_i$ is the energy and $\hat n_i = \vec p_i/E_i$ is the velocity of particle $i$.
Here, I consider specifically theories of massless particles, with the massive case discussed in \App{sec:mass}.
Fundamentally, the energy flow can be written in terms of the energy-momentum tensor $T_{\mu\nu}$ of the underlying quantum field theory as~\cite{Sveshnikov:1995vi,Cherzor:1997ak,Korchemsky:1997sy,Bauer:2008dt}:
\begin{equation}
\mathcal E(\hat n) = \lim_{r\to\infty} r^2 \int_0^\infty dt\,\hat n^i T_{0i} (t,r\hat n),
\end{equation}
which is precisely the amount of energy that flows in direction $\hat n$.

The energy flow is deeply related to the observable information in an event, namely those quantities that can be theoretically calculated and experimentally measured.
In massless theories, only the energy flow of an event is observable, which excludes all information about the number and types of particles that comprise the event.
In massive theories, this feature manifests as the robustness of the energy flow to low energy effects, such as nonperturbative or detector effects.
This behavior is typically phrased as ``infrared and collinear safety'': insensitivity to low energy and collinear modifications of the event.
From this foundation, the energy flow has been used as the framework for a variety of ideas and developments~\cite{Tkachov:1995kk,Korchemsky:1999kt,Belitsky:2001ij,Berger:2003iw,Lee:2006nr,Hofman:2008ar,Mateu:2012nk,Belitsky:2013xxa,Chen:2019bpb}, ranging from understanding nonperturbative corrections to calculating observables without scattering amplitudes.

In this thesis, I take a new bottom-up approach to quantum field theory and collider physics, with the unifying theme of starting from the observable information: the energy flow of the event.
First, I establish a space of events by introducing a metric between their energy flows: the ``work'' required to rearrange one event into another.
I show that a host of classic collider observables and concepts, such as jets, emerge as simple geometric objects in this new space.
Second, I develop a basis of observables by performing a systematic expansion in particle energies and angles, with many existing observables directly encompassed in the basis and numerous new analytic structures appearing.
Infrared and collinear safety emerges as a consistency condition from the redundancy of describing events using particles instead of their energy flow.
Finally, I define how different types, or flavors, of particles can be defined using only the observed energy flow.
This provides an operational and data-driven definition that is applicable both theoretically and experimentally at colliders today.
I approach these questions through the lens of quantum chromodynamics relevant for physics at the Large Hadron Collider, though the general considerations and results apply to quantum field theories more broadly.
This thesis is largely based on work done in collaboration with Patrick Komiske and Jesse Thaler in \Refs{Komiske:2020qhg,Komiske:2017aww,Komiske:2018vkc}.
In all of these cases, I synthesize physical concepts, such as the energy flow and factorization, with ideas from statistics and computer science, such as optimal transport and topic modeling, to enable new theoretical developments in particle physics and expand the experimental capabilities of collider data analysis.


\section{The Space of Events and its Geometry}

In Chapter 2, I  establish that many fundamental concepts and techniques in quantum field theory and collider physics can be naturally understood and unified through a simple new geometric language, based on work with my collaborators in \Ref{Komiske:2020qhg}.
The idea is to equip the space of events with a metric from which other geometric objects can be rigorously defined.
The analysis is based on the energy mover's distance, a metric which operates purely at the level of the observable energy flow, and allows for a clarified definition of infrared and collinear safety and related concepts.
A number of well-known collider observables can be exactly cast as the minimum distance between an event and various manifolds in this space.
Jet definitions, such as exclusive cone and sequential recombination algorithms, can be directly derived by finding the closest few-particle approximation to the event.
Several area- and constituent-based pileup mitigation strategies are naturally expressed in this formalism as well.
Finally, I lift the reasoning to develop a precise distance between theories, which are treated as collections of events weighted by cross sections.
In all of these various cases, a better understanding of existing methods in our geometric language suggests interesting new ideas and generalizations.

Formulating a metric requires asking a key question: When are two events similar?
Despite the simplicity and generality of this question, there had been no established notion of the distance between two events.
To address this question, with my collaborators in \Ref{Komiske:2019fks}, I developed a metric for the space of collider events based on the earth mover's distance: the ``work'' required to rearrange the radiation pattern of one event into another.
This new distance, called the energy mover's distance, can be formualted as an optimal transport problem between two energy flows $\E$ and $\E'$ as:
\begin{equation}
\label{eq:emdintro}
\EMD_{\beta,R} (\mathcal E, \mathcal E') = \min_{\{f_{ij}\ge0\}} \sum_{i=1}^M\sum_{j=1}^{M'} f_{ij} \left( \frac{\theta_{ij}}{R} \right)^\beta + \left|\sum_{i=1}^M E_i - \sum_{j=1}^{M'}E_j'\right|,
\end{equation}
\begin{equation}
\label{eq:emdconstraintsintro}
\sum_{i=1}^M f_{ij} \le E_j', \quad\quad \sum_{j=1}^{M'} f_{ij} \le E_i, \quad\quad \sum_{i=1}^M\sum_{j=1}^{M'} f_{ij} = \min\left(\sum_{i=1}^M E_i,\sum_{j=1}^{M'}E_j'\right),
\end{equation}
where $\theta_{ij}^2 = -(n_i^\mu - n_j^\mu)^2$ is an angular ground metric with $n^\mu = p^\mu/E$, $R>0$ is a parameter controlling the tradeoff between transporting energy and destroying it, and $\beta > 0$ is an angular weighting.
I exposed interesting connections between this metric and the structure of infrared- and collinear-safe observables, providing a novel technique to quantify event modifications due to hadronization, pileup, and detector effects.
I showcased how this metrization unlocks powerful new tools for analyzing and visualizing collider data without relying upon a choice of observables.
More broadly, this framework paves the way for data-driven collider phenomenology without specialized observables or machine learning models.

Beyond solely developing these event geometry ideas theoretically, I explored this new space using jets in public collider data from the CMS experiment with my collaborators in \Ref{Komiske:2019jim}.
Starting from 2.3~fb$^{-1}$ of 7 TeV proton-proton collisions collected at the Large Hadron Collider in 2011, we isolated a sample of 1,690,984 central jets with transverse momentum above 375 GeV.
To validate the performance of the CMS detector in reconstructing the energy flow of jets, we compared the CMS Open Data to corresponding simulated data samples for a variety of jet kinematic and substructure observables.
Even without detector unfolding, we find very good agreement for track-based observables after mitigating the impact of pileup.
I performed a range of novel analyses, using the energy mover's distance to measure the pairwise difference between jet energy flows, such as the ones shown in \Fig{fig:cmsod}.
The metric allowed us to quantify the impact of detector effects, visualize the metric space of jets, extract their fractal dimension, and identify the most and least typical jet configurations.
To facilitate future studies with CMS Open Data, we made our datasets and analysis code available, amounting to around two gigabytes of distilled data and one hundred gigabytes of simulation files.
With this analysis using CMS Open Data, I brought event geometry ideas from their original theoretical development all the way to their first explorations in real collider data.

\begin{figure}
\centering
\includegraphics[scale=0.925]{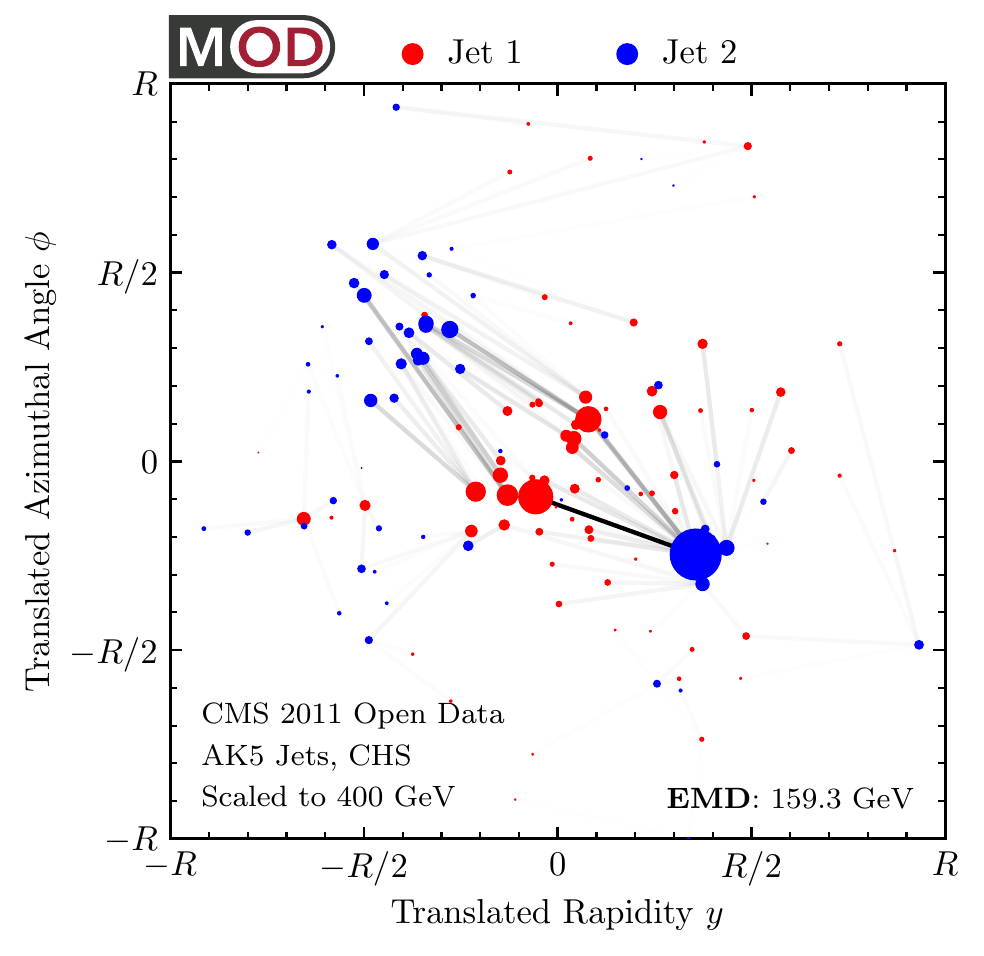}
\caption{\label{fig:cmsod}The energy mover's distance between two jets from CMS Open Data~\cite{Komiske:2019jim}.
Particles are shown at their position in the rapidity-azimuth plane with sizes corresponding to their transverse momenta.
The optimal transport plan to rearrange one jet into the other is shown by the intensity of the lines connecting pairs of points.}
\end{figure}

%

\section{Observables in Massless Quantum Field Theories}

In quantum physics, an observable $\mathcal O$ is any quantity that takes a well-defined value on some basis of states, with $\mathcal O \ket X =\mathcal O(X) \ket X$.
Any such observable can be theoretically calculated and experimentally measured using the usual framework of quantum physics.
The situation becomes more complicated, however, for massless quantum field theories.
With a basis of multiparticle states $\ket{\{p_1^\mu, \cdots, p_M^\mu\}}$, it is no longer true that quantities defined on this basis can always be theoretically calculated or experimentally observed.
Infrared divergences, or the tendencies to emit low energy or collinear particles, cause infinite results for most such quantities in perturbation theory.
Even simple and intuitive quantities, such as the number of particles in the state, do not yield finite results.
Importantly, these divergences are not a problem with the theory and instead highlight that most such quantities are not physically observable in these theories.
This issue signals that a reorganization of our understanding of observables for massless quantum field theories is potentially required, related to the vanishing of the $S$-matrix in theories with massless particles.

Since massless quantum field theories are ubiquitously used in particle physics, such as for precision collider calculations, it is imperative to understand the structure of observables in these theories.
The key concept for identifying observable quantities is their infrared and collinear safety, or robustness to low energy emissions and collinear splittings.
Conceptually, using the event energy flow $\mathcal E$, I will argue that observables $\mathcal O$ are essentially those quantities that can be written as:
\begin{equation}
\mathcal O \ket{\mathcal E} = \mathcal O(\mathcal E) \ket{\mathcal E},
\end{equation}
where $\ket{\mathcal E}$ is any quantum state with well-defined energy flow $\mathcal E$, such as a multiparticle state.
Infrared- and collinear-safe observables are theoretically calculable with finite and physical results, while unsafe observables will yield infinite results.
In massive theories, namely those with a mass gap or massive states in the spectrum, these unsafe observables might be calculable and observable.
However, they may be dominated by low energy or nonperturbative effects rather than the high energy structure of the event.
Typically, infrared and collinear safety is checked observable by observable without recourse to a more fundamental underlying structure or basis of observables.
Understanding observables in quantum field theories at the same level as in ordinary quantum physics requires a detailed exploration of their infrared and collinear safety.

In Chapter 3, I introduce a complete linear basis of infrared- and collinear-safe observables: the energy flow polynomials, based on work done with my collaborators in \Ref{Komiske:2017aww}.
I focus specifically on jets and their substructure, though the analysis holds more broadly at event-level for quantum field theories with massless particles.
Energy flow polynomials are multiparticle energy correlators with specific angular structures that are a direct consequence of infrared and collinear safety.
I establish a powerful graph theoretic representation of these polynomials which allows for their efficient organization and computation.
The energy flow polynomial corresponding to a multigraph $G$ with $N$ vertices is:
\begin{equation}
\text{EFP}_G(\mathcal E) = \sum_{i_1=1}^M \cdots \sum_{i_N=1}^M E_{i_1}\cdots E_{i_N} \prod_{(k,\ell)\in G} \theta_{i_k i_\ell}^\beta,
\end{equation}
where $\theta_{ij}^2 = -(n_i^\mu - n_j^\mu)^{2}$ is an angular distance measure with $n^\mu = p^\mu/E$, $\beta$ is an angular weighting, and $(k,\ell)$ are the pairs of vertices connected by edges in $G$.
Many common collider observables are shown to be exact linear combinations of energy flow polynomials.
I demonstrate the linear spanning nature of the energy flow basis by performing direct regression for several common observables.

The energy flow polynomials form an overcomplete basis, with linear relations among them emerging in various contexts.
It is crucial to analyze these relations to understand the independent energy-momentum tensor correlations that are probed by these observables.
More generally, the energy flow polynomials are one of a broad class of objects, called ``multiparticle correlators'', which are frequently encountered in particle physics and beyond.
By translating multiparticle correlators into the language of graph theory, new insights into their structure can be gained.
In \Ref{Komiske:2019asc} with my collaborators, I highlighted the power of this graph-theoretic approach by ``cutting open'' the vertices and edges of the graphs, allowing for the systematic classification of linear relations among multiparticle correlators and the development of faster methods for their computation.
The naive computational complexity of an $N$-point correlator among $M$ particles is $\mathcal O(M^N)$, but when the pairwise distances between particles can be cast as an inner product, I showed that all such correlators can be computed in linear $\mathcal O(M)$ runtime.
By introducing novel tensorial objects called energy flow moments:
\begin{equation}
\mathcal I^{\mu_1\cdots\mu_v} =2^{v/2} \sum_{i=1}^M E_i n_i^{\mu_1}\cdots n_i^{\mu_v},
\end{equation}
I achieved a fast implementation of collider observables widely used at the Large Hadron Collider to identify boosted hadronic resonances.
As another application, I computed the number of leafless multigraphs with $d$ edges up to $d=16$ (15,641,159), conjecturing that this is the same as the number of independent kinematic polynomials of degree $d$, previously known only to $d=8$ (279) in a string theory context~\cite{Boels:2013jua}.

A better understanding of the structure of observables in massless quantum field theories also enabled new developments in machine learning for particle physics.
A key question for machine learning approaches in particle physics is how to best represent and learn from collider events.
As an event is intrinsically a variable-length, unordered set of particles, I built upon recent machine learning efforts to learn directly from sets of features or ``point clouds''.
In \Ref{Komiske:2018cqr} with my collaborators, by adapting and specializing the Deep Sets framework~\cite{DBLP:conf/nips/ZaheerKRPSS17} to particle physics, I introduced energy flow networks, architectures which respect infrared and collinear safety by construction.
I also developed particle flow networks, which allow for general energy dependence and the inclusion of additional particle-level information such as charge and flavor.
The energy flow networks feature a per-particle internal (latent) representation, and summing over all particles yields an overall event-level latent representation:
\begin{equation}
\text{EFN}(\mathcal E) = F\left(  \sum_{i=1}^M E_i \Phi(n_i^\mu) \right),
\end{equation}
for two functions $F$ and $\Phi$, which are learned.
I showed how this latent space decomposition unifies existing event representations based on detector images and radiation moments.
To demonstrate the power and simplicity of this set-based approach, I applied these networks to the collider task of discriminating quark jets from gluon jets, finding similar or improved performance compared to existing methods.
These architectures also achieved state-of-the-art performance for boosted top quark identification, which I showcased in \Ref{Kasieczka:2019dbj} as part of a community comparison.
I also showed how the learned event representation can be directly visualized, providing insight into the inner workings of the model.
In this way, a better understanding of observables in quantum field theory can translate into more efficient processing and analyzing of events for a wide variety of tasks at the Large Hadron Collider.

\section{Particle Flavors and Factorization}

What exactly is a particle?
Quarks and gluons, for instance, are never directly observed in experiments. Instead, they manifest as jets: collimated sprays of particles.
In fact, electrons in massless quantum electrodynamics also share this same issue.
While ``quark'' and ``gluon'' jets are often treated as separate, well-defined objects in both theoretical and experimental contexts, no precise, practical, and cross section-level definition of jet flavor had existed.
A crucial question is then whether particle flavors such as ``electron'', ``photon'', ``quark'', and ``gluon''  are purely theoretical constructs or whether they can be defined using observable quantities.

Remarkably, particle flavors can indeed be defined from observables.
A key concept for understanding this and related questions is the factorization of observables.
Factorization describes how physics at different energy scales affects the observed radiation pattern in a structured and hierarchical way.
Intuitively, the observed event emerges from the outgoing high energy particles followed by the formation of their substructure (e.g.~jets), up to subleading corrections.
This principle of factorization has also been shown to apply at the level of the full energy flow~\cite{Bauer:2008jx}.
Translating this idea into a particle or jet flavor definition requires rethinking the standard notion of per-event flavors and sharpening the idea of flavor as a statistical category.

In Chapter 4, I develop and advocate for a definition of particle flavor that solely uses observable information and is built on factorization, based on work with my collaborators in \Ref{Komiske:2018vkc}.
I focus on providing a fully data-driven and operational definition of quark and gluon jets that is readily applicable at colliders, though the conclusions apply more generally for particle flavor in massless quantum field theories.
Rather than specifying a per-jet flavor label, we aggregately define quark and gluon jets at the distribution level in terms of measured hadronic cross sections.
Intuitively, quark and gluon jet ``topics'' emerge as the two maximally separable categories within two jet samples in data.
From two mixed samples $M_1$ and $M_2$, taking $M_1$ to be more quark-enriched, the operational quark and gluon distributions are:
\begin{equation}
p_q(\mathcal O) = \frac{p_{M_1}(\mathcal O) - \kappa_{12}\, p_{M_2}(\mathcal O)}{1 - \kappa_{12}}, \quad\quad p_g(\mathcal O) = \frac{p_{M_2}(\mathcal O) - \kappa_{21}\, p_{M_1}(\mathcal O)}{1 - \kappa_{21}},
\end{equation}
in terms of the two reducibility factors $\kappa_{12}=\min_\mathcal O \frac{p_{M_1}(\mathcal O)}{p_{M_2}(\mathcal O)}$ and $\kappa_{21}=\min_\mathcal O \frac{p_{M_2}(\mathcal O)}{p_{M_1}(\mathcal O)}$.
Benefiting from my previous work on data-driven classifiers and topic modeling for jets, I show that the practical tools needed to implement this definition already exist for experimental applications.
As an informative example, I demonstrate the power of this operational definition using $Z$+jet and dijet samples, illustrating that pure quark and gluon distributions and fractions can be successfully extracted in a fully well-defined manner.
In fact, this new flavor definition has already been experimentally applied by the ATLAS experiment~\cite{Aad:2019onw}, giving rise to the data-driven quark and gluon jet distributions (topics) shown in \Fig{fig:atlasfig}.
\begin{figure}[t]
\centering
\includegraphics[width=0.625\columnwidth]{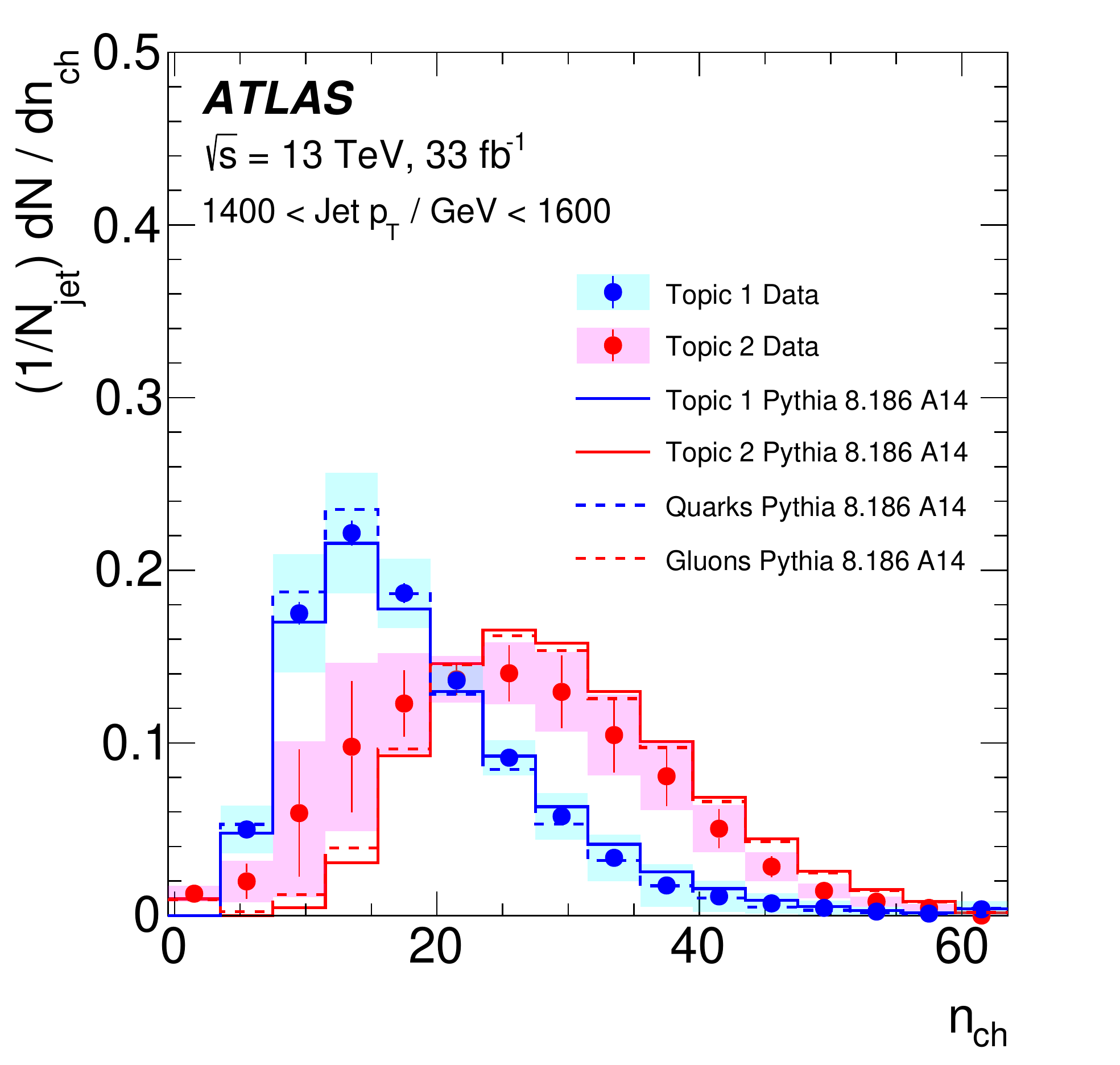}
\caption{\label{fig:atlasfig} The quark and gluon jet topic distributions extracted by the ATLAS experiment with the number of charged particles $n_\text{ch}$ using high energy jets~\cite{Aad:2019onw}.
There is agreement between the fully data-driven flavor definition and the unphysical simulation-based quark and gluon labels from \textsc{Pythia}.}
\end{figure}

More broadly, understanding jets initiated by quarks and gluons is of fundamental importance in collider physics.
Efficient and robust techniques for quark versus gluon jet classification have consequences for new physics searches, precision studies of the strong coupling constant, parton distribution function extractions, and many other applications.
Numerous machine learning analyses have attacked the problem, demonstrating that good performance can be obtained but generally not providing an understanding for what properties of the jets are responsible for that separation power.
In \Ref{Larkoski:2019nwj} with my collaborator, I provided an extensive and detailed analysis of quark versus gluon classification from first-principles theoretical calculations.
Working in the strongly-ordered soft and collinear limits, I calculated probability distributions for fixed $N$-body kinematics within jets with up through three resolved emissions.
This enables explicit calculation of quantities central to machine learning such as the likelihood ratio, the area under the receiver operating characteristic curve, and reducibility factors within a well-defined approximation scheme.
Further, I related the existence of a consistent power counting procedure for classification to ideas for operational flavor definitions, and I used this relationship to construct a power counting for quark versus gluon classification as an expansion in $e^{C_F-C_A}\ll 1$, the exponential of the fundamental and adjoint Casimirs.
These calculations provide insight into the classification performance of particle multiplicity and show how observables sensitive to all emissions in a jet are optimal.
I compared the predictions to the performance of individual observables and neural networks with parton shower event generators, validating that the predictions describe the features identified by machine learning techniques.

Understanding particle flavors using the mathematics of topic modeling builds upon work done in \Ref{Metodiev:2018ftz} with my collaborator, where I first introduced a framework to identify underlying classes of particle-types from collider data.
Because of a close mathematical relationship between factorized distributions of observables and emergent themes in sets of documents, I applied ideas from topic modeling to extract topics from data with minimal or no input from simulation or theory.
As a proof of concept, I applied the topic modeling framework to determine separate quark and gluon jet distributions for constituent multiplicity.
I also determined separate quark and gluon rapidity spectra from a mixed $Z$+jet sample.
While the topics are defined directly from hadron-level multi-differential cross sections, one can also predict topics from first-principles theoretical calculations, with implications for how to define quark and gluon jets and other particle flavors beyond leading-logarithmic accuracy.
These investigations suggest that factorized topic modeling approaches will be useful for extracting underlying particle distributions and fractions in a wide range of contexts at the Large Hadron Collider.

 \begin{figure}[t]
\centering
\includegraphics[width=0.625\columnwidth]{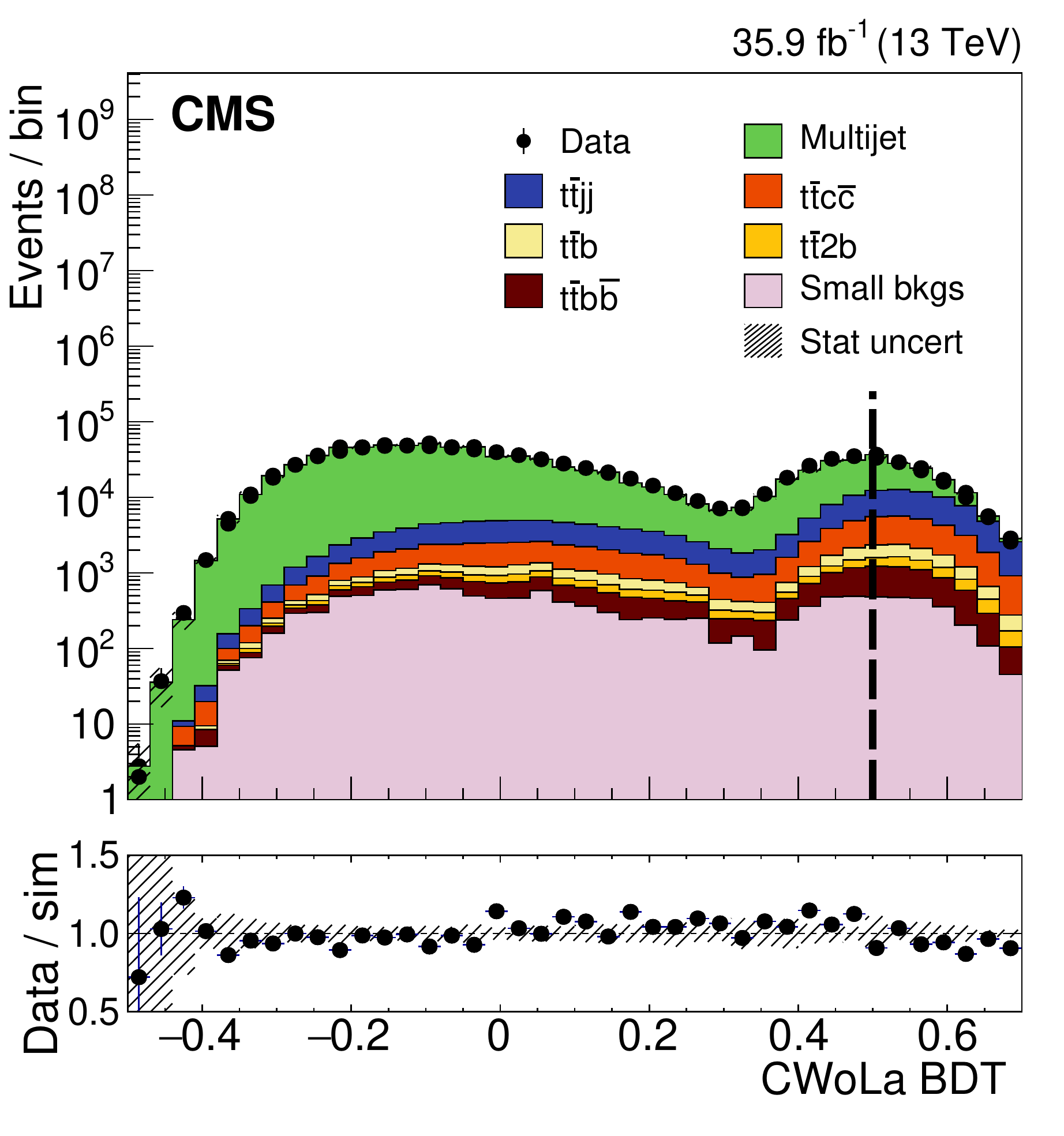}
\caption{\label{fig:cmsfig}The output of a boosted decision tree (BDT) from the CMS experiment, trained directly on data using the classification without labels (CWoLa) framework to identify various $t\bar t$+jets events from the multijet background~\cite{Sirunyan:2019jud}.}
\end{figure}

The insight of viewing collider data in terms of factorized distributions also gave rise to developments in data-driven machine learning at colliders.
Using these ideas, together with my collaborators in \Ref{Metodiev:2017vrx}, I introduced a new ``classification without labels'' paradigm to train machine learning classifiers directly on real collider data.
Modern machine learning techniques can be used to construct powerful models for difficult collider physics problems.
In many applications, however, these models are trained on imperfect simulations due to a lack of truth-level information in the data, which risks the model learning artifacts of the simulation.
In this new framework, a classifier is trained to distinguish statistical mixtures of classes, which are common at colliders due to factorization.
Crucially, neither individual labels nor class proportions are required, yet I proved that the optimal classifier in this paradigm is also the optimal classifier in the traditional fully-supervised case where all label information is available.
After demonstrating the power of this method in an analytical toy example, I considered a realistic benchmark for collider physics: distinguishing quark- versus gluon-initiated jets using mixed quark and gluon training samples.
Further, with my collaborators in \Ref{Komiske:2018oaa}, I demonstrated that complex, high-dimensional classifiers can also be trained on impure mixtures using these weak supervision techniques.
The quark versus gluon jet performance of the weakly supervised methods is comparable with what can be achieved using pure samples, which I originally studied in depth with my collaborators in \Ref{Komiske:2016rsd}.
This work opens the door to a new regime whereby complex models are trained directly on data, sidestepping simulations and providing direct access to probe the underlying physics.

More generally, this framework for classification without labels can be applied to any classification problem where labels or class proportions are unknown or simulations are unreliable, but statistical mixtures of the classes are available.
This technique has already been experimentally applied by the CMS experiment for a  measurement of the $t\bar t b \bar b$ cross section~\cite{Sirunyan:2019jud}, yielding the fully data-driven trained model shown in \Fig{fig:cmsfig}.
The paradigm has also been used as the foundation for model-agnostic new physics search strategies~\cite{Collins:2018epr,Collins:2019jip}.
These ideas have been recently applied by the ATLAS experiment to search for resonant new physics in dijets in a model-independent way for the first time~\cite{Aad:2020cws}.

\section{Experimental Aspects}

The theoretical ideas presented in this thesis are deeply influenced by and benefit from close contact with the realities of particle experiments.
It was by starting from the experimentally robust and observable quantities, namely cross sections of infrared- and collinear-safe observables, that many of the previous theoretical developments were enabled.
Beyond the research described above, I have also done work on experimentally-focused aspects of collider physics: pileup mitigation and unfolding.

Pileup involves the contamination of the energy distribution arising from the primary collision of interest by radiation from additional soft collisions.
With my collaborators in \Ref{Komiske:2017ubm}, I developed a new technique for removing this pileup contamination using machine learning and convolutional neural networks.
The network takes as input the energy distribution of charged leading vertex particles, charged pileup particles, and all neutral particles and outputs the energy distribution of particles coming from leading vertex alone.
The algorithm performs remarkably well at eliminating pileup distortion on a wide range of simple and complex jet observables.
I tested the robustness of the algorithm in a number of ways and discussed how the network can be trained directly on data.

Collider data must be corrected for detector effects (``unfolded'') to be compared with many theoretical calculations and measurements from other experiments.
Unfolding is traditionally done for individual, binned observables without including all information relevant for characterizing the detector response.
Together with my collaborators in \Ref{Andreassen:2019cjw}, I introduced an unfolding method that iteratively reweights a simulated dataset, using machine learning to capitalize on all of the  available information.
The approach is unbinned, works for arbitrarily high-dimensional data, and naturally incorporates information from the full phase space.
I illustrated this technique on a realistic jet substructure example from the Large Hadron Collider and compared it to standard binned unfolding methods.
This new paradigm enables the simultaneous measurement of all observables, including those not yet invented at the time of the analysis.

\chapter{The Hidden Geometry of Particle Collisions}

\section{Introduction}
\label{sec:intro}

Unification of ideas in physics has been an important way of achieving elegance, clarity, and simplicity, which in turn helps inspire meaningful new developments.
In this chapter, we use the energy mover's distance (EMD) between collider events~\cite{Komiske:2019fks} to provide a natural geometric language that unifies many important concepts and techniques in quantum field theory and collider physics from the past five decades.
Furthermore, we introduce and discuss several new ideas inspired by this geometric approach to studying the space of events.

Throughout this chapter, we refer to an event and its \emph{energy flow} interchangeably.
The energy flow, or distribution of energy, is the kinematic information that is experimentally observable and perturbatively well-defined in quantum field theories with massless particles~\cite{Tkachov:1995kk}.
As it relates to collider physics, the energy flow has been extensively studied~\cite{Tkachov:1995kk,Sveshnikov:1995vi,Korchemsky:1997sy,Basham:1978zq,Cherzor:1997ak,Tkachov:1999py,Korchemsky:1999kt,Belitsky:2001ij,Berger:2002jt,Berger:2003iw,Lee:2006nr,Bauer:2008dt,Hofman:2008ar,Mateu:2012nk,Belitsky:2013xxa,Komiske:2017aww,Komiske:2018cqr,Komiske:2019asc,Chen:2019bpb}, and this chapter builds on many of these previous concepts.
For an event consisting of $M$ particles with positive energies $E_i$ and angular directions $\hat n_i$, the energy flow is:
\begin{equation}
\label{eq:energyflow}
\E(\hat n) = \sum_{i=1}^ME_i\,\delta(\hat n - \hat n_i).
\end{equation}
Note that the energy flow is insensitive to charge and flavor information.
Particles are taken to be massless in the body of this chapter, with $n_i^\mu = (1, \hat{n}_i)^\mu = p^\mu_i/E_i$, and the case of massive particles is discussed in \App{sec:mass}.
In a hadron collider context, particle transverse momenta $p_{T,i}$ are typically used in place of particle energies, but we focus on energies in this chapter to minimize extraneous notation.

The EMD was introduced in \Ref{Komiske:2019fks} as a metric between events.
It is based on the well-known earth mover's distance~\cite{DBLP:journals/pami/PelegWR89,Rubner:1998:MDA:938978.939133,Rubner:2000:EMD:365875.365881,DBLP:conf/eccv/PeleW08,DBLP:conf/gsi/PeleT13}, also known as the Wasserstein metric~\cite{wasserstein1969markov,dobrushin1970prescribing}.
Intuitively, the EMD between two events is the amount of ``work'' required to rearrange one event to the other.
Its value can be obtained by solving the following optimal transport problem between energy flows $\E$ and $\E'$:
\begin{equation}
\label{eq:emd}
\EMD_{\beta,R} (\mathcal E, \mathcal E') = \min_{\{f_{ij}\ge0\}} \sum_{i=1}^M\sum_{j=1}^{M'} f_{ij} \left( \frac{\theta_{ij}}{R} \right)^\beta + \left|\sum_{i=1}^M E_i - \sum_{j=1}^{M'}E_j'\right|,
\end{equation}
\begin{equation}
\label{eq:emdconstraints}
\sum_{i=1}^M f_{ij} \le E_j', \quad\quad \sum_{j=1}^{M'} f_{ij} \le E_i, \quad\quad \sum_{i=1}^M\sum_{j=1}^{M'} f_{ij} = \min\left(\sum_{i=1}^M E_i,\sum_{j=1}^{M'}E_j'\right),
\end{equation}
where $\theta_{ij}$ is a pairwise distance between particles known as the \emph{ground metric}, $R>0$ is a parameter controlling the tradeoff between transporting energy and destroying it, and $\beta > 0$ is an angular weighting exponent.%
\footnote{\label{footnote:pWasser}
Strictly speaking, for the case of $\beta>1$, one must raise the first term in \Eq{eq:emd} to the $1/\beta$ power for the EMD to be a proper metric satisfying the triangle inequality, in which case it is known as a $p$-Wasserstein metric with $p = \beta$.
Additionally, $2R$ should be larger than or equal to the maximum distance in the ground space for the EMD to satisfy the triangle inequality.
When written without subscripts, $\EMD(\E,\E')$ refers to the case of $\beta=1$ and a sufficiently large $R$ to ensure that we have a proper metric.
Even if the EMD is not a proper metric, though, it is still a valid optimal transport problem for any positive values of $\beta$ and $R$.
}
For the angular metric between two massless particles, we focus on the case of
\begin{equation}
\label{eq:theta_def}
\theta_{ij} = \sqrt{2n_i^\mu n_{j\mu}} = \sqrt{2 (1 - \hat{n}_i \cdot \hat{n}_j)},
\end{equation}
which reduces to their opening angle in the nearby limit.%
\footnote{Many modifications to this EMD definition are possible, including alternative angular distances such as strict opening angle or rapidity-azimuth distance as well as alternative notions of energy such as transverse momentum. In addition, energies can be normalized by dividing by their total scalar sum so that energy flows become proper probability distributions. If desired, the EMD in the center-of-mass frame can be phrased in a manifestly Lorentz-invariant way by replacing the particle energies $E_i$ with $p_i^\mu P_\mu/\sqrt{P_\mu P^\mu}$, where $P_\mu$ is the total event four-momentum.}
The first term in \Eq{eq:emd} quantifies the difference in radiation patterns while the second term, which vanishes in the case of normalized energy flows, allows for the comparison of events with different total energies.
The constraints in \Eq{eq:emdconstraints} specify that the amount of energy moved to or from a particle cannot exceed its initial energy, and that as much energy must be moved as possible.

\begin{figure}[t]
\centering
\includegraphics[scale=0.8]{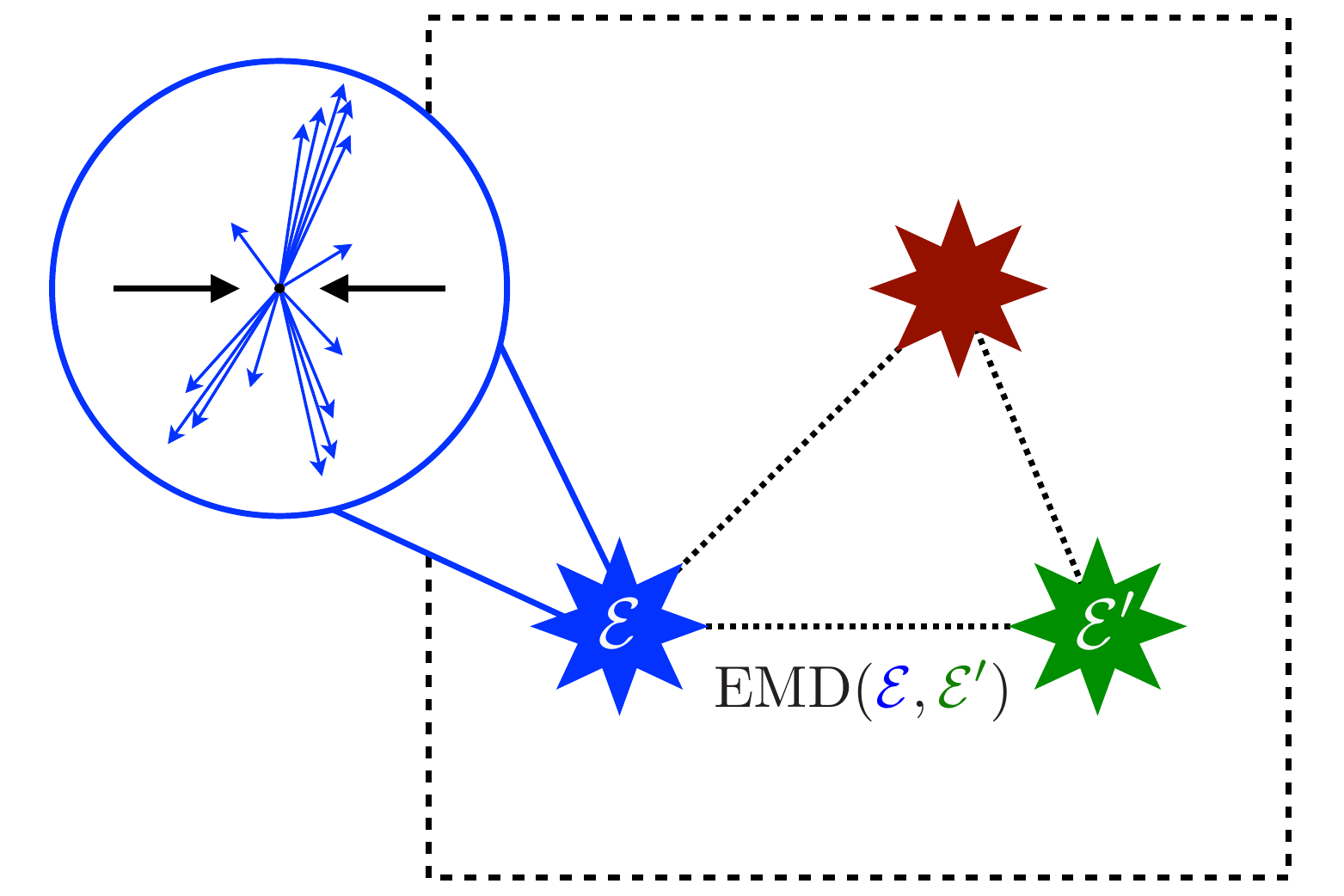}
\caption{\label{fig:event_space} An illustration of the space of events.
Each point in the space is a collider event consisting of the particles produced in a collision, as indicated by the blue event.
The distance between events is quantified by the EMD, giving rise to a metric space.
Geometry in this abstract space of events provides a natural language to understand many ideas and developments in quantum field theory and collider physics.}
\end{figure}

The EMD has previously been used to bound modifications to infrared- and collinear-safe (IRC-safe) observables, distinguish different types of jets, and enable visualizations of the space of events~\cite{Komiske:2019fks}.
It has also been used to explore the space of jets and quantify detector effects with CMS Open Data from the Large Hadron Collider (LHC)~\cite{Komiske:2019jim}.
Alternative pairwise event distances were considered in \Ref{Mullin:2019mmh} in the context of new physics searches.
Here, we demonstrate that the EMD can be used to clarify numerous concepts throughout quantum field theory and collider physics using a unified language of event space geometry.
In addition to demonstrating how concepts such as IRC safety, observables, jet finding, and pileup subtraction are related, we will develop new ideas and techniques in each of these areas, which we describe below.

Equipping collider events with a metric allows us to explore interesting geometric and topological ideas in the space of events.
\Fig{fig:event_space} illustrates the space of events with the EMD as a metric.
One key construction for relating these concepts is the notion of a manifold in the space of events, which will allow us to define the distance between an event and a manifold, as well as the point of closest approach on a manifold.
Since fixed-order perturbation theory works with a definite number of particles, an important type of manifold will be the idealized massless $N$-particle manifold:
\begin{equation}
\label{eq:npmanifold}
\mathcal P_N = \left\{\left.\sum_{i=1}^N E_i\, \delta(\hat n - \hat n_i)\,\, \right| \,\, E_i\ge0 \right\},
\end{equation}
which, intuitively, is the set of all possible events with $N$ massless particles.
Note that $\mathcal P_N\supset\mathcal P_{N-1}\supset\cdots\mathcal P_2\supset \mathcal P_1\supset\mathcal P_0$ via soft and collinear limits, so that the idealized $N$-particle manifold contains each manifold of smaller particle multiplicity.

\begin{table}[t]
\centering
\begin{tabular}{|c|c|c|c|}
\hline\hline
\bf Sec. & \bf Concept &  \bf Equation & \bf Illust.  \\
\hline\hline
& &  &  \multirow{5}{*}{\raisebox{-4em}{\includegraphics[scale=0.19]{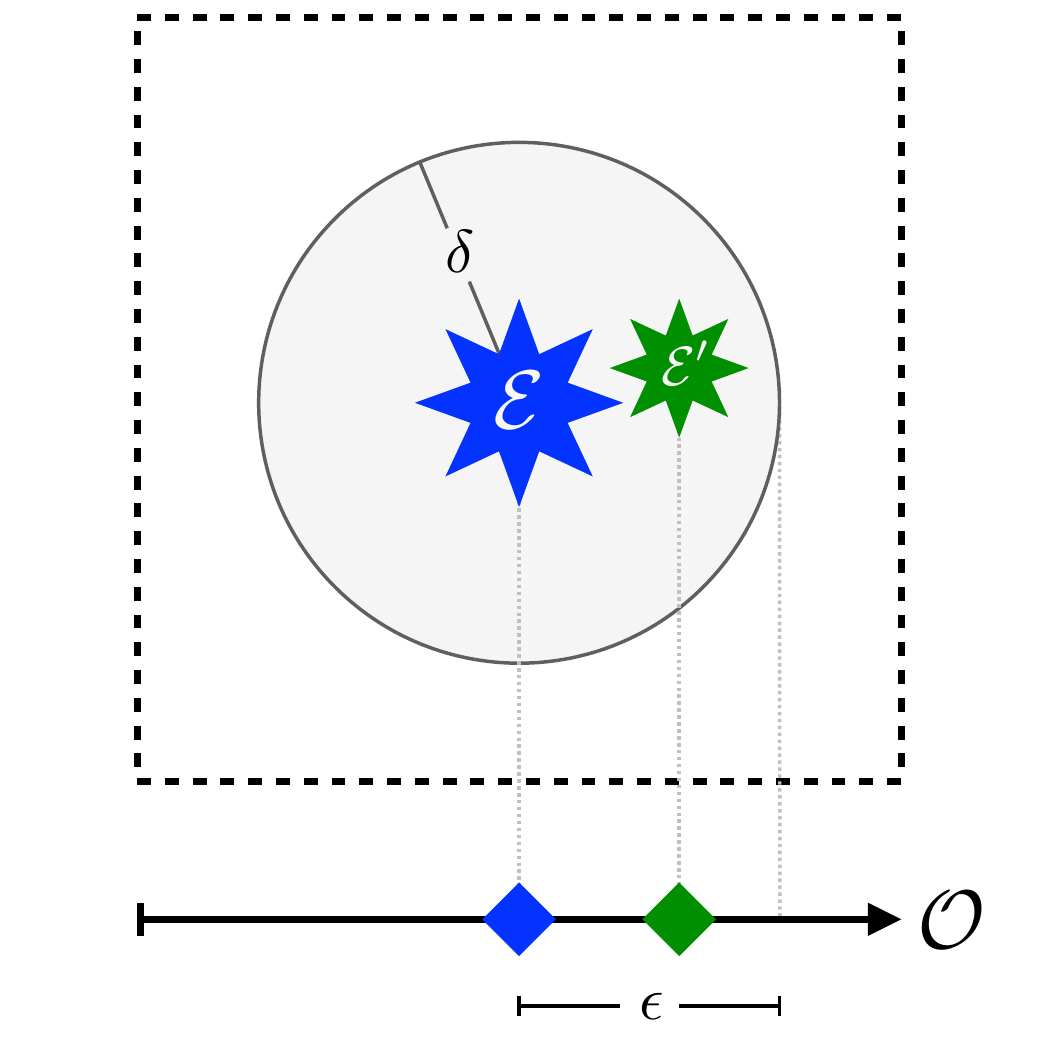}}} \\
 \ref{sec:safety} & {\bf Infrared and} & $\displaystyle\text{EMD}(\E,\E')<\delta  \implies$ & \\ 
 & {\bf Collinear Safety} & $|\O(\E) - \O(\E')| < \epsilon$ &  \\
& \cite{Kinoshita:1962ur,Lee:1964is,Sterman:1977wj,Sterman:1978bi,Sterman:1978bj,Sterman:1979uw,sterman1995handbook,Weinberg:1995mt,Ellis:1991qj,Banfi:2004yd,Larkoski:2013paa,Larkoski:2014wba,Larkoski:2015lea} &  & \\
& &  &  \\ \hline
& &  & \multirow{5}{*}{\raisebox{-4em}{\includegraphics[scale=0.19]{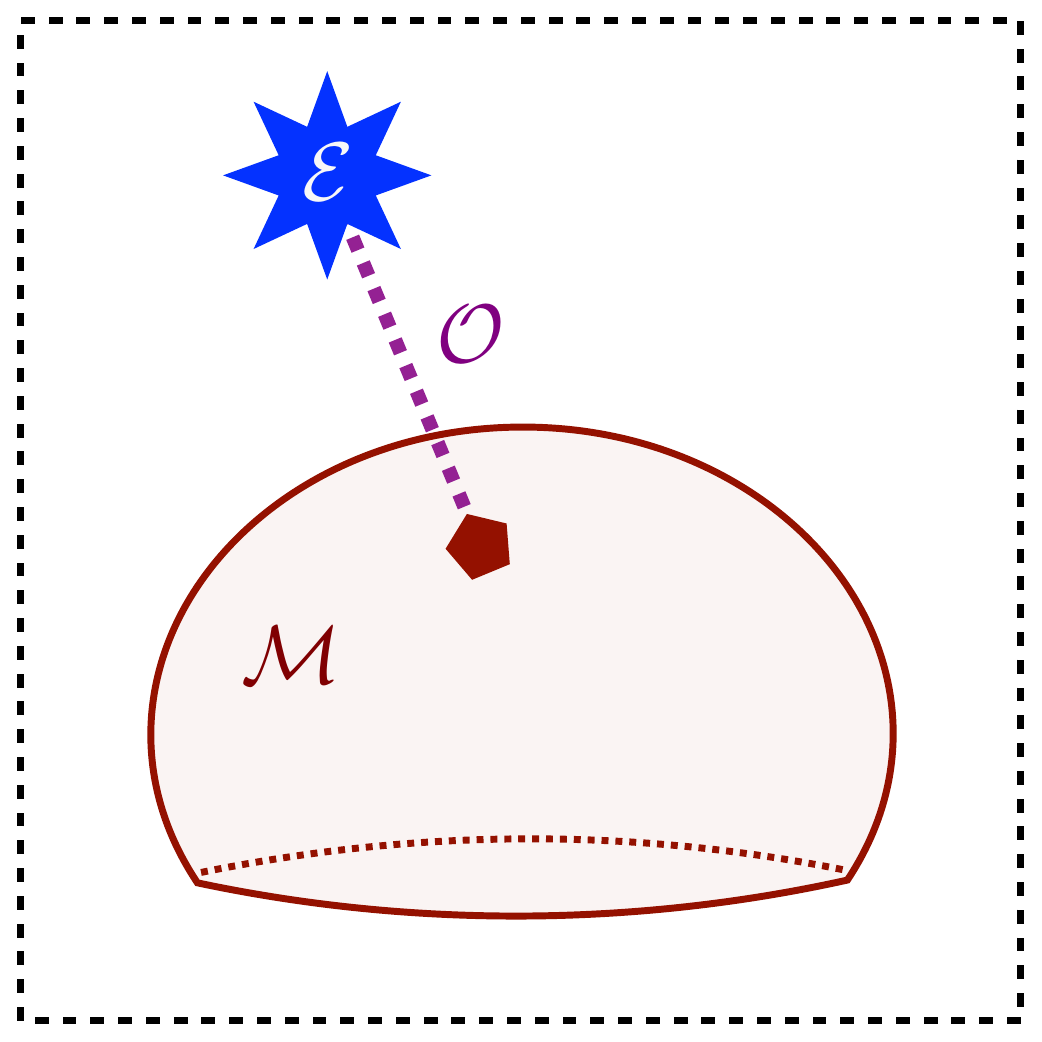}}} \\ 
 \ref{sec:observables} &{\bf Observables} &  $\displaystyle\O(\E) = \min_{\E'\in\mathcal M}\EMD(\E,\E')$ &   \\
\ref{sec:eventobservables} & Event Shapes~\cite{Brandt:1964sa,Farhi:1977sg,Georgi:1977sf,Larkoski:2014uqa,Stewart:2010tn,isotropytemp} &  &    \\
 \ref{sec:jetobservables} & Jet Shapes~\cite{Ellis:2010rwa,Thaler:2010tr,Thaler:2011gf}  &  &    \\
& &  &  \\ \hline
& &  & \multirow{5}{*}{\raisebox{-4em}{\includegraphics[scale=0.19]{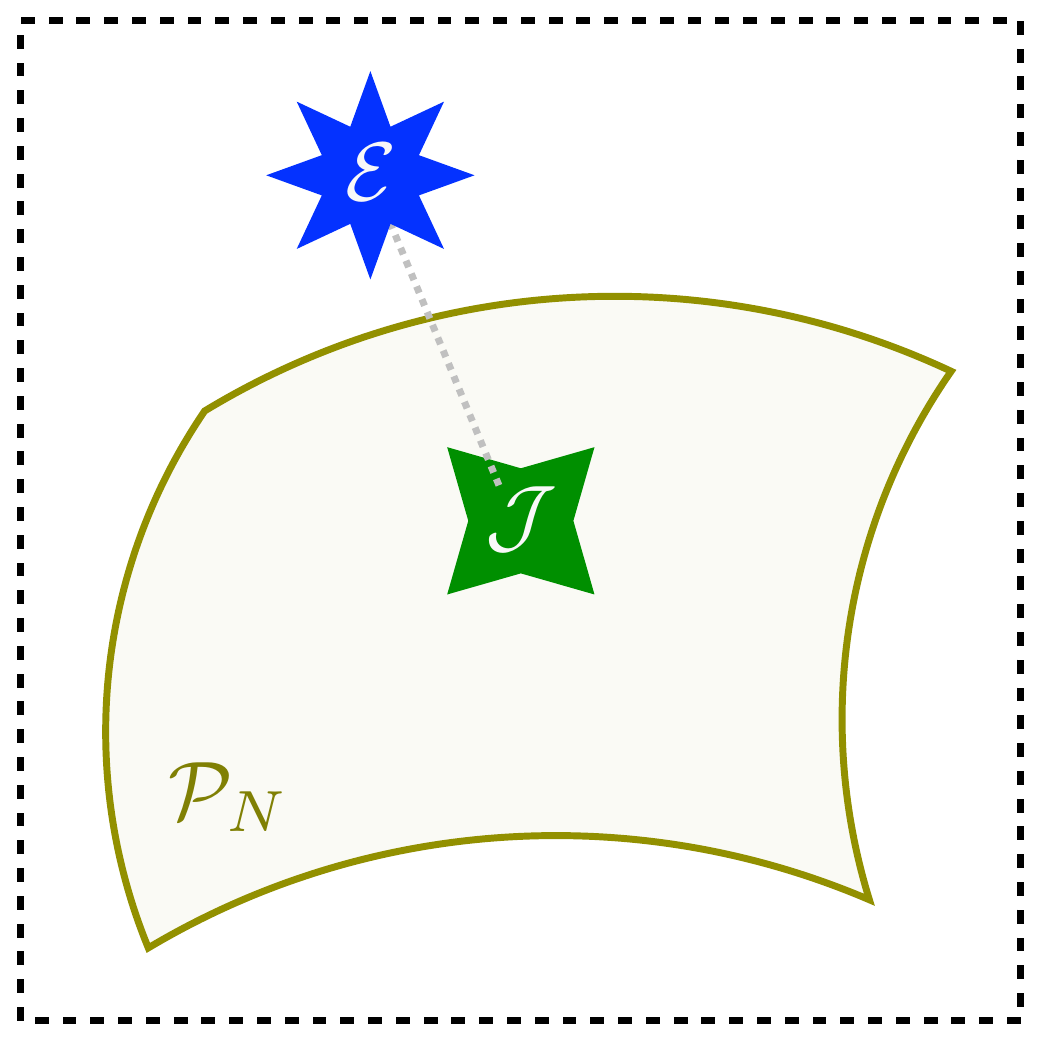}}} \\ 
 \ref{sec:jets} &  {\bf Jets} & $\displaystyle\mathcal J(\E)=\argmin_{\mathcal J\in\mathcal P_N}\,\EMD(\E,\mathcal J)$ &   \\
 \ref{subsec:xcone} &  Cone Finding~\cite{Stewart:2015waa,Thaler:2015xaa}  & & \\
 \ref{sec:seqrec} & Seq. Rec.~\cite{Catani:1993hr,Ellis:1993tq,Bertolini:2013iqa,Salambroadening} &  & \\
& &  & \\ \hline
& &  & \multirow{5}{*}{\raisebox{-4em}{\includegraphics[scale=0.19]{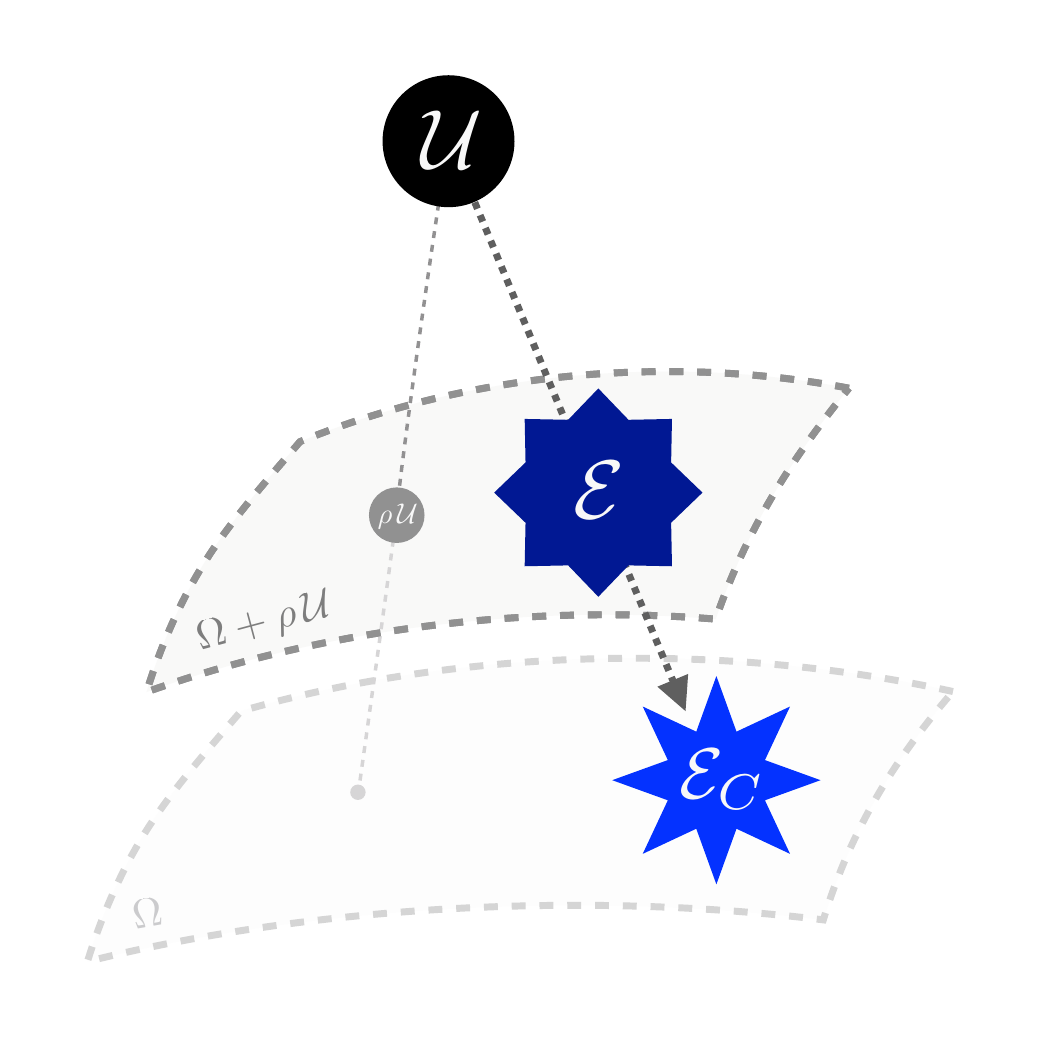}}} \\ 
\ref{sec:pileup} & {\bf Pileup} &  $\mathcal E_C(\mathcal E, \rho)  = \displaystyle\argmin_{\mathcal E' \in \Omega}\,\EMD(\mathcal E,\mathcal E' + \rho\,\mathcal U)$ & \\
 &\cite{Cacciari:2007fd,Cacciari:2008gn,Soyez:2012hv,Berta:2014eza,Bertolini:2014bba,Soyez:2018opl,Berta:2019hnj}  &  & \\
& &  & \\
& &  & \\ \hline
& & & \multirow{3}{*}{\raisebox{-5em}{\includegraphics[scale=0.19]{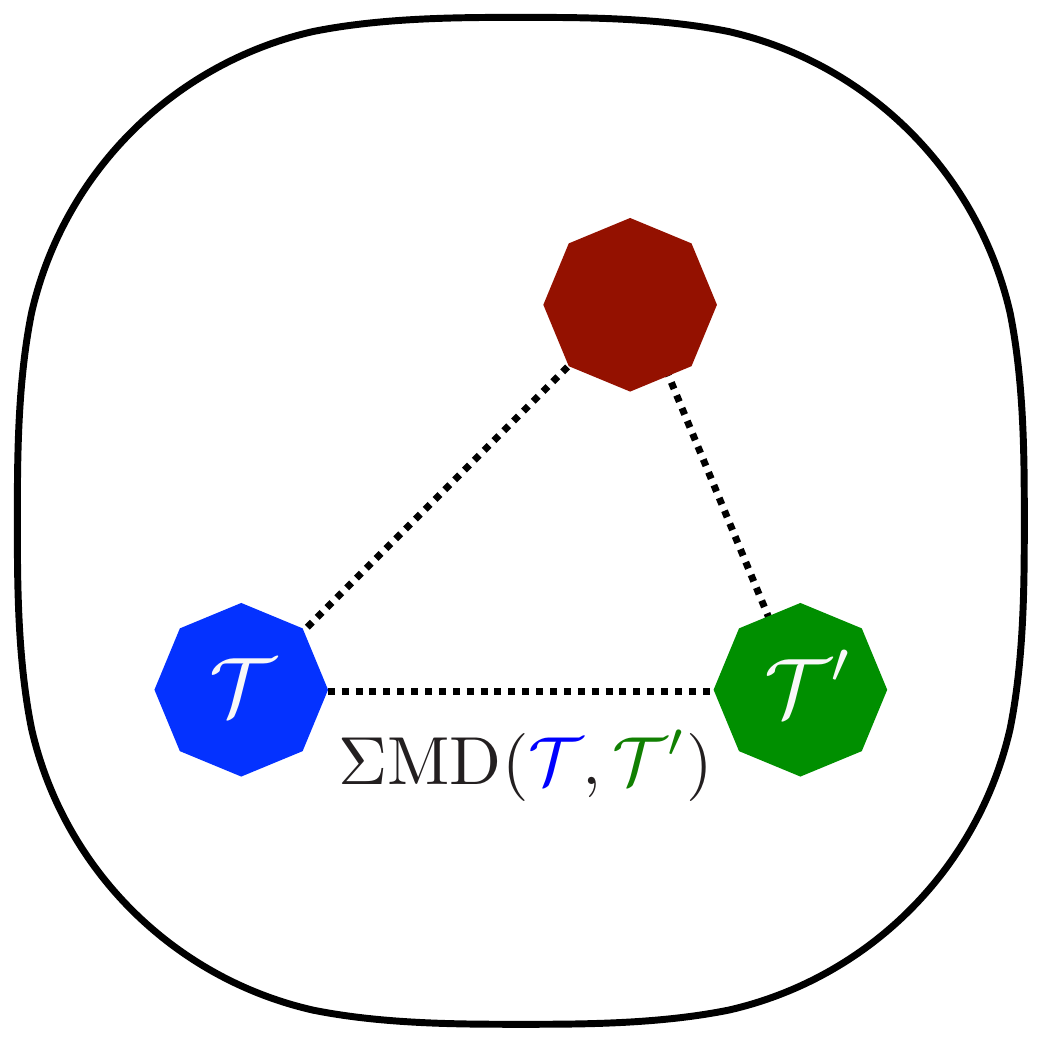}}}  \\
\ref{sec:theory} & {\bf Theory Space} & $\displaystyle\mathcal{T}(\mathcal E) = \sum_{i = 1}^N \sigma_i \, \delta(\mathcal E - \mathcal E_i)$ & \\
& & & \\
& & & \\
\hline\hline
\end{tabular}
\caption{\label{tab:outlinefigs}
Concepts from quantum field theory and collider physics, unified in this chapter as geometric and topological constructions in the space of events.
In \Sec{sec:safety}, IRC safety is identified as continuity in this space.
In \Sec{sec:observables}, many classic collider observables are shown to be the shortest distance between the event and a manifold of events.
In \Sec{sec:jets}, popular jet algorithms are derived by projecting the event onto manifolds of $N$-particle events.
In \Sec{sec:pileup}, common pileup mitigation strategies are cast as transporting away uniform radiation.
In \Sec{sec:theory}, a space of theories is developed using a distance between event distributions.
}
\end{table}
\afterpage{\clearpage}

The key concepts unified in this chapter are outlined and summarized in \Tab{tab:outlinefigs}.
In \Sec{sec:safety}, we discuss observables as functions defined on event radiation patterns and IRC safety as smoothness in the space of energy flows.
Colloquially, the label ``IRC safe'' indicates that an observable should be well-defined and calculable in perturbation theory~\cite{Kinoshita:1962ur,Lee:1964is} due to its robustness to long-distance effects (e.g.~hadronization in the case of QCD).
This ``perturbatively accessible'' IRC safety is traditionally connected to the observable being ``insensitive'' to the addition of low energy particles or collinear splittings of particles~\cite{Sterman:1977wj,Sterman:1978bi,Sterman:1978bj,Sterman:1979uw,sterman1995handbook,Weinberg:1995mt,Ellis:1991qj,Banfi:2004yd}.
Here, we refine the definition of IRC safety and clarify when discontinuities in an observable spoil its perturbative calculability.
Critical to our formulation is the notion of continuity with respect to the metric topology provided by the EMD:
\begin{definition}\label{def:emdcontinuity}
An observable $\O$ is $\EMD$ continuous at an event $\E$ if, for any $\epsilon>0$, there exists a $\delta>0$ such that for all events $\E'$:
\begin{equation}\label{eq:emdcontinuity}
\text{\emph{EMD}}(\E,\E')<\delta \quad \implies\quad |\O(\E) - \O(\E')| < \epsilon.
\end{equation}
\end{definition}
\noindent We argue that IRC safety is $\EMD$ continuity everywhere except a negligible set of events, where a negligible set is one that contains no EMD balls of non-zero radius.
Using the EMD provides a definition of IRC safety that does not refer to particles directly, which circumvents many pathologies of previous definitions.
We argue that observables that are calculable in fixed-order perturbation are exactly those that satisfy a slightly stronger continuity condition known as H\"older continuity~\cite{Ortega2000,Gilbarg2001}, which restricts the types of divergences that can appear in the distribution of an observable~\cite{Sterman:1979uw,Banfi:2004yd}.
Fascinatingly, this framework naturally accommodates Sudakov-safe observables~\cite{Larkoski:2013paa,Larkoski:2014wba,Larkoski:2015lea} as those that are IRC safe but fail to satisfy EMD H\"older continuity on a non-negligible subset of some $\mathcal P_N$ (where a non-negligible subset of $\mathcal P_N$ is one that has measure in $\mathcal P_N$).
This suggests, in agreement with \Ref{Larkoski:2015lea}, that Sudakov safe observables are indeed perturbatively calculable once properly regulated.

In \Sec{sec:observables}, we highlight that many well-known collider observables can be viewed as the distance of closest approach between an event and a manifold of events.
Many of the observables we consider can be exactly cast as:
\begin{equation}
\label{eq:obsdefemd}
\O(\E) = \min_{\E'\in\mathcal M}\EMD_{\beta,R}(\E,\E'),
\end{equation}
for particular choices of the manifold $\mathcal M$ and parameters $\beta$ and $R$.
Observables that have the form of \Eq{eq:obsdefemd} include thrust~\cite{Brandt:1964sa,Farhi:1977sg}, spherocity~\cite{Georgi:1977sf}, (recoil-free) broadening~\cite{Larkoski:2014uqa}, and $N$-jettiness~\cite{Stewart:2010tn}.
Particularly interesting is the event isotropy, recently proposed in \Ref{isotropytemp}, which was inspired by EMD geometry and is directly based on optimal transport.
This geometric framework also includes jet substructure observables such as jet angularities~\cite{Ellis:2010rwa} and  $N$-subjettiness~\cite{Thaler:2010tr,Thaler:2011gf}.

In \Sec{sec:jets}, we demonstrate how jet finding can be phrased in our geometric language.
Intuitively, a jet algorithm ``approximates'' an $M$-particle event with $N<M$ objects called jets.
To phrase this geometrically, we are interested in the point of closest approach in $\mathcal P_N$ to our event, allowing us to define jets as:
\begin{equation}
\label{eq:jetdefemd}
\mathcal J(\E)=\argmin_{\mathcal J\in\mathcal P_N}\,\EMD_{\beta,R}(\E,\mathcal J),
\end{equation}
where $\mathcal J$ is the collection of $N$ jets corresponding to the event $\E$.
Many common jet finding algorithms can be derived in full detail from \Eq{eq:jetdefemd}.
For instance, we show that jets defined by \Eq{eq:jetdefemd} are precisely those found by XCone~\cite{Stewart:2015waa,Thaler:2015xaa}, where $\beta$ is the angular weighting exponent and $R$ is the jet radius.
Also, several popular sequential clustering algorithms and recombination schemes, such as $k_T$ clustering~\cite{Catani:1993hr,Ellis:1993tq} with winner-take-all recombination~\cite{Bertolini:2013iqa,Larkoski:2014uqa,Salambroadening}, can be exactly obtained by iterating \Eq{eq:jetdefemd} with $N = M-1$ for various $\beta$.
It is satisfying that a rich diversity of jet algorithms can be concisely encoded using event geometry, and we find that several new schemes not previously appearing in the literature naturally emerge.

In \Sec{sec:pileup}, we connect several pileup mitigation strategies to optimal transport through the EMD.
There is a long-established relationship between pileup subtraction and geometric concepts~\cite{Cacciari:2007fd,Cacciari:2008gn,Soyez:2012hv,Berta:2014eza,Bertolini:2014bba,Soyez:2018opl,Berta:2019hnj}.
Since pileup is reasonably modeled as uniform contamination in rapidity and azimuth~\cite{Soyez:2018opl}, we phrase pileup subtraction as removing a uniform distribution of radiation from the event using optimal transport.
Intuitively, pileup mitigation finds the event that, when combined with an amount $\rho$ of uniform radiation $\mathcal U$, is closest to the given event:
\begin{equation}\label{eq:pileupemdintro}
\mathcal E_C(\mathcal E, \rho) = \argmin_{\mathcal E' \in \Omega}\,\EMD_\beta(\mathcal E,\mathcal E' + \rho\,\mathcal U),
\end{equation}
yielding the pileup-corrected event $\mathcal E_C$.
Here, $\Omega$ refers to the space of all possible energy flows and $\EMD_\beta$ compares events of equal energy, as described at the beginning of \Sec{sec:observables}.
We demonstrate that Voronoi area subtraction~\cite{Cacciari:2007fd,Cacciari:2008gn} and constituent subtraction~\cite{Berta:2014eza} can be phrased exactly as \Eq{eq:pileupemdintro} in the small-pileup limit.
Generalizing this to the large-pileup limit, we develop two new pileup subtraction schemes, Apollonius subtraction and iterated Voronoi subtraction, and discuss their prospects and potential advantages.

In \Sec{sec:theory}, we introduce a distance between theories: the cross section mover's distance (stylized as $\Sigma$MD, using the typical greek letter for cross section).
Here, a ``theory'' $\mathcal T$ is taken to be a distribution over (or collection of) events $\{\E_i\}$ weighted by cross sections $\{\sigma_i\}$:
\begin{equation}
\label{eq:T_defintro}
\mathcal{T}(\mathcal E) = \sum_{i = 1}^N \sigma_i \, \delta(\mathcal E - \mathcal E_i).
\end{equation}
The $\Sigma$MD is formulated as an optimal transport problem with EMD as the ground metric and cross sections as the weights.
The similarity of the constructions of EMD and $\Sigma$MD are highlighted in \Tab{tab:emdsmdcomp}.
Interestingly, we connect $\Sigma$MD to a recently proposed technique for probing jet modifications due to the quark-gluon plasma by comparing similar sets of events between proton-proton and heavy-ion collisions~\cite{Brewer:2018dfs}.
We also demonstrate that representative events can be identified by clustering using the $\Sigma$MD, analogously to how particles are clustered into jets.
The $\Sigma$MD provides the foundation for a rigorous formulation of ``theory space'', quantifying how different two theories are based on all of their physically observable quantities simultaneously.

\begin{table}[t]
\centering
\begin{tabular}{rcc}
\hline\hline
& Energy Mover's Distance & Cross Section Mover's Distance  \\ \hline \hline
Symbol & \text{EMD} & \text{$\Sigma$MD} \\
Description & Distance between events & Distance between theories \\ \hline
Weight & Particle energies $E_i$ & Event cross sections $\sigma_i$ \\
Ground Metric & Particle distances $\theta_{ij}$ & Event distances $\text{EMD}(\mathcal E_i, \mathcal E_j)$ \\
\hline\hline
\end{tabular}
\caption{\label{tab:emdsmdcomp}
Comparing the constructions of EMD and $\Sigma$MD as optimal transport problems.
Events are treated as energy-weighted angular distributions, whereas theories are treated as cross section-weighted event distributions.
This connection allows us to bootstrap the EMD as a ground metric for the $\Sigma$MD to develop a rigorous notion of theory space.
}
\end{table}

Our conclusions are presented in \Sec{sec:conc}, where we also highlight the interesting and unique interplay between machine learning and the natural sciences in this story.
\clearpage

\section{Infrared and collinear safety: Smoothness in the space of events}
\label{sec:safety}

\begin{figure}[t]
\centering
\includegraphics[width=\columnwidth]{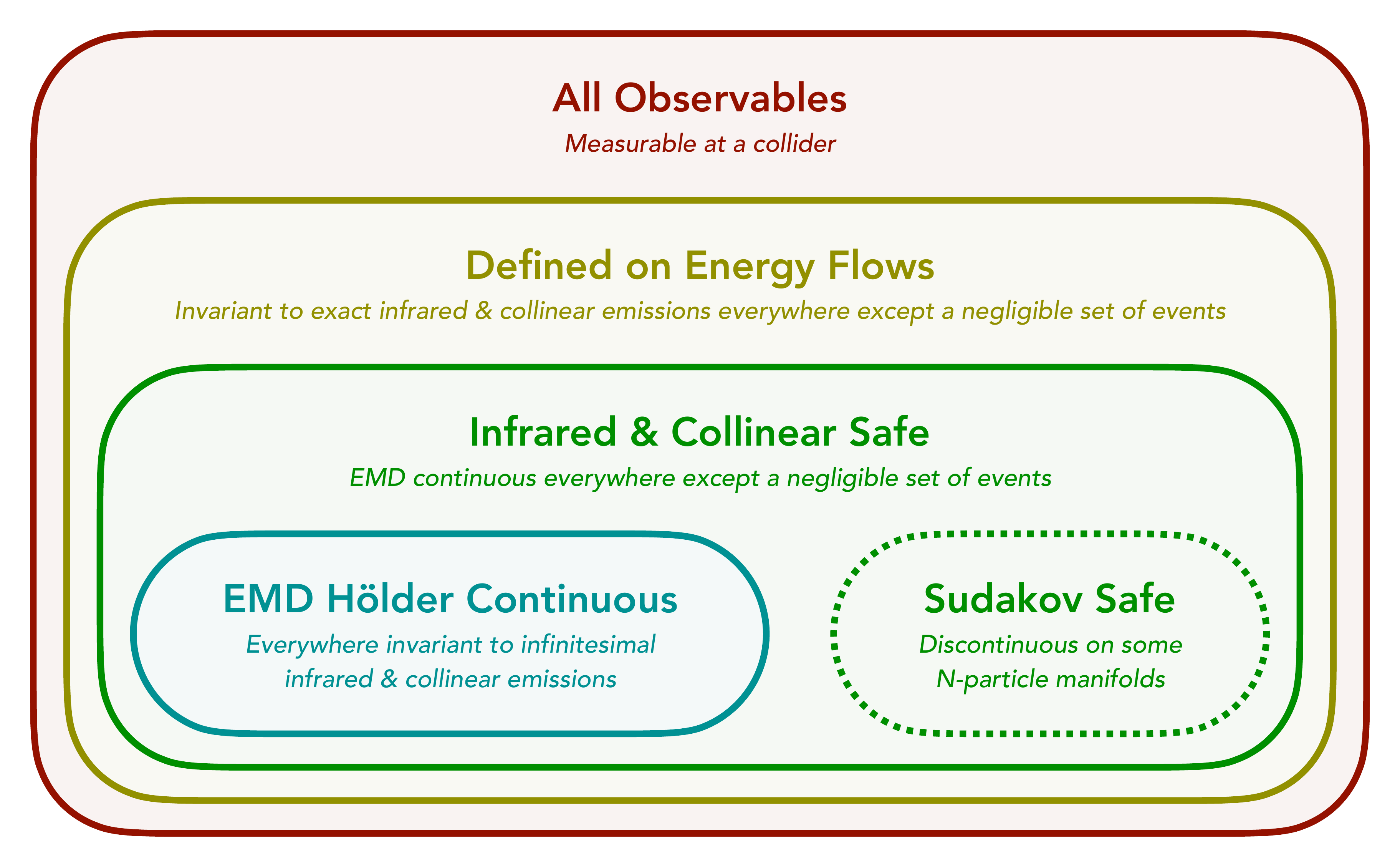}
\caption{\label{fig:obsset}
An illustration of the set of observables partitioned according to various IRC-invariance properties.
Examples of observables in each category are listed in \Tab{tab:exampleobs}.
}
\end{figure}

IRC safety is a central notion in collider physics because it indicates when an observable is robust to long distance effects and hence can be described in perturbation theory~\cite{Kinoshita:1962ur,Lee:1964is} using a combination of fixed-order calculations and resummation.
This insensitivity is frequently connected to the invariance of an observable under certain modifications of the event, namely soft and collinear splittings~\cite{Sterman:1977wj,Sterman:1978bi,Sterman:1978bj,sterman1995handbook,Weinberg:1995mt,Ellis:1991qj,Banfi:2004yd}.

In this section, we review some of the common mathematical statements of this invariance that have appeared in the literature, with the goal of clarifying and categorizing their implications.
We arrive at a simple, unified description of IRC safety and related concepts (including Sudakov safety) as statements about continuity in the space of energy flows.
In \Fig{fig:obsset}, we show the breakdown of observables into broad classes according to our categorization.
A few common examples of each category are given in \Tab{tab:exampleobs}.

\begin{table}[p]
\begin{flushright}

\begin{tabular}{p{0.56\textwidth} p{0.38\textwidth}}
\hline\hline
\cellcolor{table_red_bg} \textbf{\textcolor{table_red}{All Observables}} & \cellcolor{table_red_bg} Comments \\
 \hline
\cellcolor{table_red_bg} Multiplicity $\left(\sum_i 1\right)$ & \cellcolor{table_red_bg}  IR unsafe and C unsafe\\
\cellcolor{table_red_bg} Momentum Dispersion~\cite{CMS-PAS-JME-13-002} $\left(\sum_i E_i^2 \right)$ & \cellcolor{table_red_bg}  IR safe but C unsafe\\ 
\cellcolor{table_red_bg} Sphericity Tensor~\cite{Bjorken:1969wi} $\left(\sum_i p_i^\mu p_i^\nu \right)$ & \cellcolor{table_red_bg}  IR safe but C unsafe\\
\cellcolor{table_red_bg} Number of Non-Zero Calorimeter Deposits & \cellcolor{table_red_bg}  C safe but IR unsafe \\
  \hline \hline
\end{tabular}

\begin{tabular}{p{0.51\textwidth} p{0.38\textwidth}}
 \multicolumn{2}{l}{\cellcolor{table_yellow_bg} \textbf{\textcolor{table_yellow}{Defined on Energy Flows}}} \\
\hline
\cellcolor{table_yellow_bg} Pseudo-Multiplicity ($\min \{N\, |\, \mathcal T_N = 0\}$)& \cellcolor{table_yellow_bg} Robust to exact IR or C emissions\\
\hline  \hline
\end{tabular}

\begin{tabular}{p{0.46\textwidth} p{0.38\textwidth}}
 \multicolumn{2}{l}{\cellcolor{table_green_bg} \textbf{\textcolor{table_green}{Infrared \& Collinear Safe}}} \\
 \hline
\cellcolor{table_green_bg} Jet Energy $\left(\sum_i E_i \right)$ &\cellcolor{table_green_bg}  Disc.\ at jet boundary\\
\cellcolor{table_green_bg} Heavy Jet Mass~\cite{Clavelli:1981yh} &\cellcolor{table_green_bg}  Disc.\ at hemisphere boundary \\
\cellcolor{table_green_bg} Soft-Dropped Jet Mass~\cite{Dasgupta:2013ihk,Larkoski:2014wba} & \cellcolor{table_green_bg} Disc.\ at grooming threshold\\
\cellcolor{table_green_bg} Calorimeter Activity~\cite{Pumplin:1991kc} ($N_{95}$) & \cellcolor{table_green_bg} Disc.\ at cell boundary\\
  \hline
  \hline
   \multicolumn{2}{l}{\cellcolor{table_green_bg} \textit{\textcolor{table_green}{Sudakov Safe}}} \\
  \hline
\cellcolor{table_green_bg} Groomed Momentum Fraction~\cite{Larkoski:2015lea} ($z_g$) & \cellcolor{table_green_bg} Disc.\ on $1$-particle manifold\\
\cellcolor{table_green_bg} Jet Angularity Ratios~\cite{Larkoski:2013paa} & \cellcolor{table_green_bg} Disc.\ on 1-particle manifold \\
\cellcolor{table_green_bg} $N$-subjettiness Ratios~\cite{Thaler:2010tr,Thaler:2011gf} ($\tau_{N+1} / \tau_{N}$) & \cellcolor{table_green_bg} Disc.\ on $N$-particle manifold\\
\cellcolor{table_green_bg} $V$ parameter~\cite{Banfi:2004yd} (\Eq{eq:Vobs}) & \cellcolor{table_green_bg} H-disc.\ on 3-particle manifold \\
 \hline \hline
\end{tabular}

\begin{tabular}{p{0.405\textwidth} p{0.385\textwidth}}
 \multicolumn{2}{l}{\cellcolor{table_teal_bg} \textbf{\textcolor{table_teal}{EMD H\"{o}lder Continuous Everywhere}}} \\
\hline
\cellcolor{table_teal_bg} Thrust~\cite{Brandt:1964sa,Farhi:1977sg} & \cellcolor{table_teal_bg} \\
\cellcolor{table_teal_bg} Spherocity~\cite{Georgi:1977sf} & \cellcolor{table_teal_bg}  \\ 
\cellcolor{table_teal_bg} Angularities~\cite{Berger:2003iw} & \cellcolor{table_teal_bg}  \\ 
\cellcolor{table_teal_bg} $N$-jettiness~\cite{Stewart:2010tn} $\left(\mathcal T_N\right)$ & \cellcolor{table_teal_bg}  \\
\cellcolor{table_teal_bg} $C$ parameter~\cite{Parisi:1978eg,Donoghue:1979vi,Ellis:1980wv,Catani:1997xc} & Resumming beneficial at $C = \frac34$ \cellcolor{table_teal_bg} \\
\cellcolor{table_teal_bg} Linear Sphericity~\cite{Donoghue:1979vi} $\left(\sum_i E_i n_i^\mu n_i^\nu \right)$ & \cellcolor{table_teal_bg}  \\
\cellcolor{table_teal_bg} Energy Correlators~\cite{Banfi:2004yd,Larkoski:2013eya,Larkoski:2014gra,Moult:2016cvt} & \cellcolor{table_teal_bg}  \\
\cellcolor{table_teal_bg} Energy Flow Polynomials~\cite{Komiske:2017aww,Komiske:2019asc} & \cellcolor{table_teal_bg}  \\
\hline\hline
\end{tabular}
\end{flushright}
\caption{
Examples of well-known collider observables, along with their classification according to \Fig{fig:obsset}.
The observables satisfy the conditions of all bold-faced categories above them in the table.
Note that via our classification, Sudakov safe observables are IRC safe, since the discontinuities appear on $N$-particle manifolds which are negligible sets in the full space.
}
\label{tab:exampleobs}
\end{table}

\subsection{Review of infrared and collinear invariance}
\label{sec:invariance}

The most straightforward statement of IRC invariance is that an observable $\mathcal O$ is unchanged under the addition of an exactly zero energy particle or an exactly collinear splitting~\cite{sterman1995handbook}:
\begin{align}
\label{eq:exactirsafety}\text{Exact Infrared Invariance:}&\quad \mathcal O(p_1^\mu, \ldots,p_M^\mu) =  \mathcal O(0 p_0^\mu, p_1^\mu,  \ldots, p_M^\mu),\\
\label{eq:exactcsafety}\text{Exact Collinear Invariance:}&\quad \mathcal O(p_1^\mu,  \ldots, p_M^\mu) = \mathcal O(\lambda p_1^\mu, (1-\lambda) p_1^\mu,\ldots, p_M^\mu),
\end{align}
for any soft momentum $p_0^\mu$ and collinear splitting fraction $\lambda\in[0,1]$.
These conditions correctly rule out some observables from having a perturbative description, such as the number of particles in an event, which change by a finite amount under any splitting.
Exact IRC invariance, however, is not sufficiently restrictive to guarantee perturbative calculability of an observable.
For instance, the number of calorimeter cells with non-zero energy is safe according to \Eqs{eq:exactirsafety}{eq:exactcsafety}, though it is highly sensitive to arbitrarily low-energy effects~\cite{Pumplin:1991kc}.
Similarly, the pseudo-multiplicity, which we define as the smallest $N$ that yields zero $N$-jettiness (see \Sec{subsec:nsubjettiness} below), is unchanged by exact infrared and collinear emissions,%
\footnote{We thank Andrew Larkoski for discussions related to this point.}
but is highly sensitive to any emissions at finite energy or angle.

Another common statement of IRC invariance refines the concept by invoking the limit as particles become soft or collinear~\cite{Sterman:1978bi,Sterman:1978bj,Weinberg:1995mt,Banfi:2004yd}:
\begin{align}
\label{eq:nearirsafety}\text{Near Infrared Invariance:}&\quad \mathcal O(p_1^\mu, \ldots,p_M^\mu) =  \lim_{\epsilon\to0}\mathcal O(\epsilon p_0^\mu, p_1^\mu,  \ldots, p_M^\mu),\\
\label{eq:nearcsafety}\text{Near Collinear Invariance:}&\quad \mathcal O(p_1^\mu, \ldots, p_M^\mu) = \lim_{p_0^\mu\to p_1^\mu}\mathcal O(\lambda p_{0}^\mu, (1-\lambda) p_1^\mu, \ldots, p_M^\mu).
\end{align}
One issue with this definition is that many reasonable observables that have hard boundaries in phase space are excluded, such as jet kinematics due to sensitivity to particles on a jet boundary.
Hybrid definitions mixing exact and near IRC invariance also appear in the literature but they suffer from the same pathologies.
Another issue is that \Eqs{eq:nearirsafety}{eq:nearcsafety} (and also \Eqs{eq:exactirsafety}{eq:exactcsafety}) do not guarantee insensitivity to multiple soft or collinear splittings.

Several of these issues were previously identified in \Ref{Banfi:2004yd}, which utilized a limit-based statement of IRC invariance, recognized the importance of allowing for multiple soft and collinear emissions, and allowed for exceptions on sets of measure zero.
Despite noting that a rigorous mathematical definition of IRC safety would be desirable, \Ref{Banfi:2004yd} concluded that formulating one without pathologies was challenging and that a satisfactory definition had not yet been obtained.
Here, we explore how the geometric picture provided by the EMD yields a natural and elegant way to phrase IRC safety and to control these various subtleties.
This builds on the notion of ``$C$-continuity'' advocated for in \Refs{Tkachov:1995kk,Tkachov:1999py}, which argue that the perturbative calculability of $C$-continuous observables can be seen by relating the energy flow to the energy-momentum tensor of the underlying quantum field theory.

\subsection{Infrared and collinear safety in the space of events}
\label{sec:ircsafety}

The EMD provides a natural language for understanding IRC-safe observables as continuous functions on the space of events.
To make this precise, we first must understand which observables are well-defined functions of the energy flow.

We can show that observables that are defined on \emph{all} energy flows are precisely those which have exact IRC invariance according to \Eqs{eq:exactirsafety}{eq:exactcsafety}.
First, an observable is well defined on the space of energy flows if its value is the same on events that are zero EMD apart. 
The following lemma establishes the remaining connection to exact IRC invariance.
\begin{lemma}
Two events are zero EMD apart if and only if they differ by zero energy emissions or exactly collinear splittings.
\end{lemma}
\begin{proof}
Adding a zero energy particle or a collinear splitting to an event manifestly does zero energy moving, proving the forward direction.
To prove the reverse direction, suppose that two events are zero EMD apart and take their energy flows to be:
\begin{equation}
\mathcal E(\hat n) = \sum_{i=1}^M E_i\,\delta(\hat n - \hat n_i),\quad\quad \mathcal E'(\hat n) = \sum_{j=1}^{M'} E'_j \,\delta(\hat n - \hat n'_j).
\end{equation}
Since the EMD is a proper metric between energy flows, the identity of indiscernibles says that $\EMD(\E(\hat n),\E'(\hat n))=0$ implies $\E(\hat n)=\E'(\hat n)$.
For any direction $\hat n$ with at least one particle, either the sums of energies in that direction are equal between the two events or the particle has zero energy.
In the first case, the events differ by exactly collinear splittings in that direction, and in the second case they differ by zero energy particles.
\end{proof}
By this lemma we see that exact IRC invariance ensures that we can write $\mathcal O(\mathcal E)$ rather than $\mathcal O(p_1^\mu, \cdots, p_M^\mu)$ for an observable.
As discussed in \Sec{sec:invariance}, exact IRC invariance is insufficient to guarantee IRC safety and we must formulate a stronger condition phrased in the geometric language of the space of events.

We propose that IRC safety is achieved by requiring an observable to be EMD continuous, in the sense of Definition~\ref{def:emdcontinuity}, except possibly on a negligible set of events.
We define a negligible set to be one that contains no EMD ball.
The (open) \EMD ball $B_r(\E)$ around an event $\mathcal E$ is defined as all events within an EMD of $r>0$:
\begin{equation}
B_r(\E)=\left\{\E'\in\Omega\,\Big|\,\EMD(\E,\E')<r\right\},
\end{equation}
where $\Omega$ is the space of all energy flows.
Implicit in the above requirement is that an observable must be well defined on energy flows.
Concretely, we state IRC safety as the following:
\begin{framed}
\begin{ircsafety}
An observable is IRC safe if it is EMD continuous for all energy flows, except potentially on a negligible set of events.
\end{ircsafety}
\end{framed}
\noindent This new formulation of IRC safety has many aspects of existing ideas of safety discussed in \Sec{sec:ircsafety} wrapped into a concise and rigorous statement.
It makes mathematically precise the intuitive notion that small perturbations in the energy flow of the event give rise to small perturbations in the observable.
This notion of EMD continuity for IRC safe observables is illustrated in \Fig{fig:space_ircsafe}.
The exception for negligible sets allows observables to be discontinuous in a way that affords them the opportunity to depend sharply on phase space but does not spoil their calculability.
Calculability is a statement about integrability, and removing a negligible set of points from an integral cannot change its value.

\begin{figure}[t]
\centering
\includegraphics[scale=0.75]{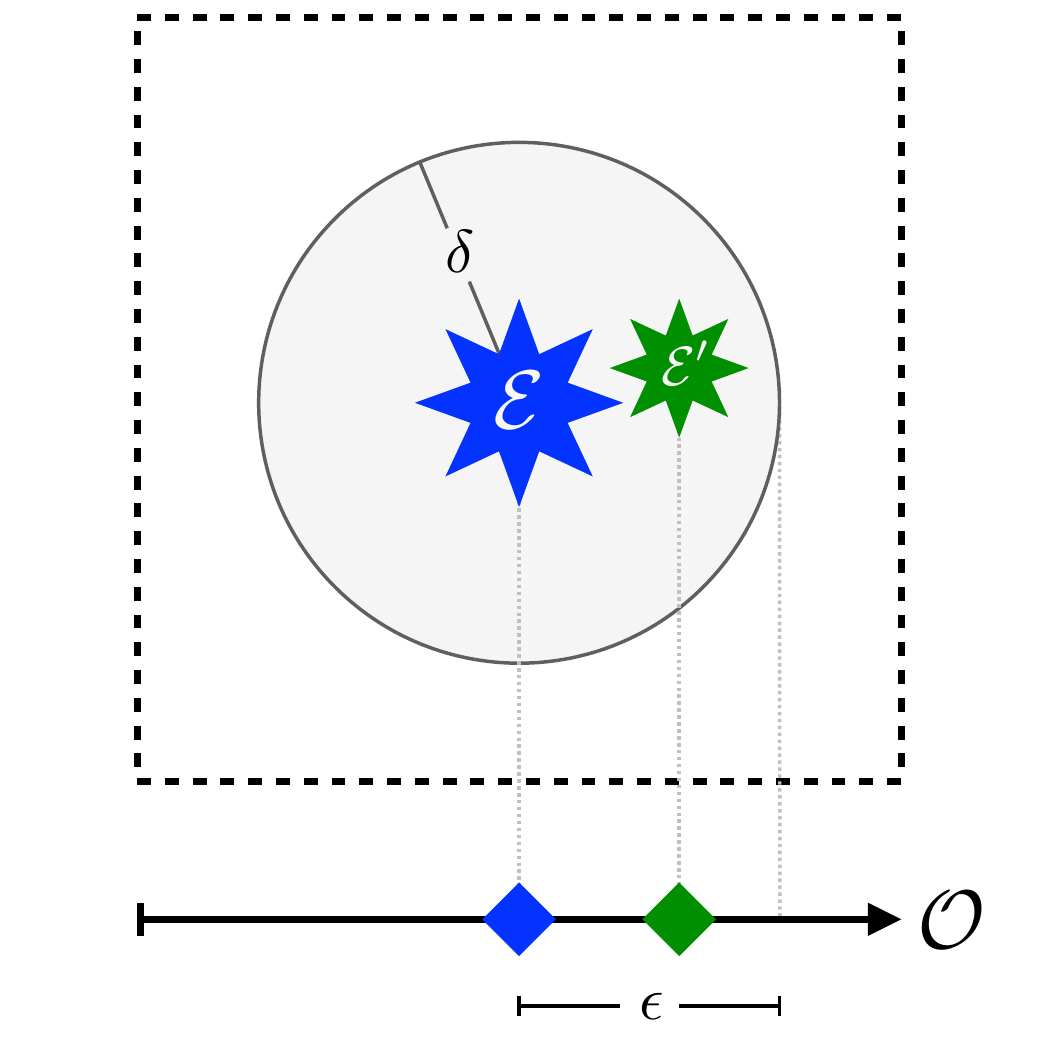}
\caption{\label{fig:space_ircsafe} An illustration of IRC safety of an observable as continuity in the space of events.
As formulated in \Eq{eq:emdcontinuity}, small perturbations to the event, as measured by EMD, yield small changes in the observable value.}
\end{figure}

To get some familiarity with this definition, consider additive IRC-safe observables, which are ubiquitous structures~\cite{Komiske:2019asc} that take the form $\mathcal O(\mathcal E) = \sum_{i=1}^M E_i f(\hat n_i)$ for an angular function $f$.
One can prove that they are Lipschitz continuous in the space of events assuming $f$ is Lipschitz continuous~\cite{Komiske:2019fks}, and therefore they naturally satisfy continuity according to the EMD.
As a generalization of additive observables, energy flow networks~\cite{Komiske:2018cqr} are a machine learning architecture that can approximate any IRC-safe observable through an additive IRC-safe latent space.
As long as the activation functions are continuous almost everywhere, then the final energy flow network output will be IRC safe.

There are also observables that fail the criteria of \Eqs{eq:nearirsafety}{eq:nearcsafety} for small sets of events but are safe according to our definition and are indeed calculable.
The energy of a jet is a simple example where emissions on the jet boundary result in discontinuous behavior of the observable, but this discontinuity is integrable in fixed-order perturbation theory.
A more complicated example is the invariant mass after soft drop grooming~\cite{Larkoski:2014wba,Dasgupta:2013ihk}: for events on the threshold of having an emission dropped, tiny perturbations can give rise to discontinuously large changes in the observable.
This issue, however, only occurs on a negligible set, satisfying our definition of safety and avoiding serious analytic pathologies~\cite{Frye:2016okc,Frye:2016aiz,Marzani:2017mva,Marzani:2017kqd}.
Piecewise continuity does, however, complicate analyzing the nonperturbative corrections~\cite{Hoang:2019ceu} and detector response~\cite{ATL-PHYS-PUB-2019-027,Aad:2019vyi} of soft-dropped jet mass.

Our definition also includes observables that would sometimes not be called IRC safe since they do not have a well defined Taylor expansion in the small parameter of the theory (e.g.\ $\alpha_s$ for QCD).
These observables are nevertheless perturbatively calculable, though methods beyond fixed-order perturbation theory may be required.
The next subsections are devoted to exploring which IRC-safe observables are calculable in fixed-order perturbation theory and which require additional techniques.

\subsection{Calculability in fixed-order perturbation theory}
\label{sec:fopt}

IRC safety has long been connected with the notion of calculability order-by-order in perturbative quantum field theory.
However, IRC safety according to our Definition~\ref{def:emdcontinuity} includes observables that are not calculable in fixed-order perturbation theory, which we explore further in the next subsection.
Here, building off the work in \Refs{Sterman:1979uw,Banfi:2004yd}, we formulate the stronger notion of EMD H\"older continuity~\cite{Ortega2000,Gilbarg2001} and argue that it is the appropriate condition to guarantee order-by-order perturbative control:
\begin{definition}\label{def:emdholdercontinuity}
An observable $\mathcal O$ is EMD H\"older continuous with exponent $\alpha\in(0,1]$ at an event $\E$ if there exists $K>0$ such that for all $\E'$ in some neighborhood of $\E$:
\begin{equation}\label{eq:Kalpha}
|\mathcal O(\E)-\mathcal O(\E')|\le K\,\text{\emph{EMD}}(\E,\E')^\alpha.
\end{equation}
\end{definition}
\noindent Note that the case of $\alpha=1$ corresponds to Lipschitz continuity at $\E$, and in general we have containment such that H\"older continuity with exponent $\alpha$ implies H\"older continuity with exponent $\beta$ if $\beta\le\alpha$.
EMD H\"older continuity effectively specifies that the $\delta$ in Definition~\ref{def:emdcontinuity} is no smaller than $\epsilon$ to some power (times a constant) for all points in a neighborhood of $\E$, and thus it is a stronger requirement than plain EMD continuity.

To connect to fixed-order perturbation theory, we state the following conjecture:
\begin{framed}
\begin{conjecture}\label{conj:fopt}
An observable is calculable order-by-order in perturbation theory if it is EMD H\"older continuous on all but a negligible set of events in each $N$-particle manifold. 
\end{conjecture}
\end{framed}
\noindent This relation phrases the ideas of \Ref{Sterman:1979uw} and ``Version 2'' of the IRC safety definition of \Ref{Banfi:2004yd} in our geometric language via the EMD.
While these criteria were originally formulated for the calculability of moments of an observable, they appear to also extend to the calculability of distributions of observables~\cite{Sterman:2006uk}.

It is possible to demonstrate a precise equivalence between our Conjecture~\ref{conj:fopt} and the following criteria of \Ref{Sterman:1979uw} regarding when the average value of an observable $\O$ is calculable in fixed-order perturbation theory:
\begin{equation}
\label{eq:stermanenergycond}
\lim_{|\vec p_i|\to0}\frac{\O(\vec p_1,\ldots,\vec p_i,\ldots,\vec p_M)-\O(\vec p_1,\ldots,\vec p_{i-1},\vec p_{i+1},\ldots,\vec p_M)}{|\vec p_i|^a}=0,
\end{equation}
\begin{equation}
\label{eq:stermanthetacond}
\lim_{\theta_{ij}\to0}\frac{\O(\vec p_1,\ldots,\vec p_i,\ldots,\vec p_j,\ldots)-\O(\vec p_1,\ldots,\vec p_i+\vec p_j,\ldots,\vec p_{j-1},\vec p_{j+1},\ldots)}{\theta_{ij}^b}=0,
\end{equation}
where the powers $a$ and $b$ are positive and the choices of $i$ and $j$ are arbitrary.
Here, \Eq{eq:stermanenergycond} is a statement of H\"older continuity in the energy of particle $i$, which implies ordinary soft safety.
Similarly, \Eq{eq:stermanthetacond} is a statement of H\"older continuity in the angular distance between particles $i$ and $j$, which implies ordinary collinear safety.
In these soft and collinear limits, $\EMD(\mathcal E, \mathcal E') \propto E_i$ and $\EMD(\mathcal E, \mathcal E') \propto \theta_{ij}$ respectively, and so \Eqs{eq:stermanenergycond}{eq:stermanthetacond} can be phrased compactly as:
\begin{equation}
\lim_{\mathcal E' \to \mathcal E} \frac{\mathcal O(\mathcal E) - \mathcal O(\mathcal E')}{\EMD(\mathcal E, \mathcal E')^c} = 0.
\end{equation}
for some positive exponent $c$.
This is equivalent to the H{\"{o}}lder continuity of the observable $\mathcal O$ at $\E$ with some exponent $\alpha \ge c$, connecting the formulation of \Ref{Sterman:1979uw} to our conjecture.

Our Conjecture~\ref{conj:fopt} also nicely connects to ``Version 2'' of the IRC safety definition in \Ref{Banfi:2004yd}, which we restate here with a suggestive relabeling of the original notation.
The criteria for fixed-order calculability of an observable in \Ref{Banfi:2004yd} are as follows:
\begin{quote}
\Ref{Banfi:2004yd}: Given almost any fixed set of particles and any value $n$, then for any $\epsilon>0$, however small, there should exist a $\delta>0$ such that producing $n$ extra soft or collinear emissions, each emission being at a distance of no more than $\delta$ from the nearest particle, then the value of the observable does not change by more than $\epsilon$. Furthermore, there should exist a positive power $c$ such that for small $\epsilon$, $\delta^c$ can always be taken greater than $\epsilon$.
\end{quote}
By equipping the space of events with these topological and geometric structures via EMD, our language provides a natural language to sharply mathematically formulate this discussion.
The first sentence can be encoded as EMD continuity of the observable on all but a negligible set of events.
The power relation between the $\epsilon$ and $\delta$ parameters is precisely captured by EMD H\"{o}lder continuity with some exponent $\alpha > c$, connecting to our Conjecture~\ref{conj:fopt}.

A variety of observables are considered in \Ref{Banfi:2004yd} at the boundary of perturbative calculability, which helpfully illustrate the various requirements in their definition.\footnote{We thank Gavin Salam for discussions related to this point.}
An observable that is useful to consider is:
\begin{equation}\label{eq:Vobs}
V(\mathcal E) = \mathcal T_2(\mathcal E) \left(1 + \frac{1}{\ln E(\mathcal E)/{\mathcal T_3(\mathcal E)}}\right),
\end{equation}
where $\mathcal T_N$ are $N$-jettiness observables~\cite{Stewart:2010tn} discussed further in \Sec{sec:njettiness}, and $E$ is the total energy of the event.
We will refer to this observable as the ``$V$ parameter''.
The double logarithmic structure of $\mathcal T_3$ spoils the integrability of $V$ at fixed order due to its behavior as $\mathcal T_3$ goes to zero~\cite{Banfi:2004yd}, which occurs on the three-particle manifold $\mathcal P_3$.
Nonetheless, this observable can be calculated using techniques beyond fixed-order perturbation theory, such as the Sudakov safety approach discussed in the next section.

The relation between our formalism and fixed-order perturbative calculability is phrased as a conjecture since additional subtleties or nuances about this type of calculability may emerge with future research.
Nonetheless, it is very satisfying that our geometric language provides an efficient encapsulation and unification of the existing formulations of \Refs{Sterman:1979uw,Banfi:2004yd}.
In future work, it would be interesting to find a geometric phrasing of recursive IRC safety~\cite{Banfi:2004yd}, which is a more restrictive condition than EMD H\"older continuity and relevant for understanding factorization and resummation.
It would also be interesting to find a geometric phrasing of unsafe observables that can be nevertheless be computed with the help of nonperturbative fragmentation functions (see \Ref{Elder:2017bkd} for a broad class of such observables).
We hope that further refinements and developments will benefit from and be enabled by the rigorous geometric and topological constructions we have introduced for the space of events via the EMD.

\subsection{A refined understanding of Sudakov safety}

Sudakov-safe observables~\cite{Larkoski:2013paa,Larkoski:2014wba,Larkoski:2015lea} are an interesting class of observables that are not typically considered IRC safe because divergences may appear order by order in perturbation theory; this issue was originally pointed out in \Ref{Soyez:2012hv}.
Nevertheless, the distribution for a Sudakov-safe observable $\mathcal{O}_s$ can be computed perturbatively by calculating its conditional distribution with an IRC-safe companion observable $\mathcal{O}_c$, resumming the $\mathcal{O}_c$ distribution, and then marginalizing over $\mathcal{O}_c$ to obtain a finite answer~\cite{Larkoski:2015lea}:
\begin{equation}
\label{eq:sudakov_safe_strategy}
p(\mathcal{O}_s) = \int \text{d}\mathcal{O}_c \, p(\mathcal{O}_s | \mathcal{O}_c) \, p(\mathcal{O}_c).
\end{equation}
The conditional probability $p(\mathcal{O}_s | \mathcal{O}_c)$ can either be computed in fixed-order perturbation theory or it can be further resummed to obtain a more accurate prediction for $p(\mathcal{O}_s)$.

Here, we interpret Sudakov-safe observables as observables that \emph{are} IRC safe according to our definition but may be EMD (H\"{o}lder) discontinuous on sets with non-zero measure when restricted to some idealized massless $N$-particle manifold $\mathcal P_N$, defined in \Eq{eq:npmanifold}.
The relevant manifolds are the $N$-particle manifolds since these contain the infrared singular regions of massless gauge theories, namely configurations that differ by soft and collinear splittings.
The IRC safety of an observable according to our definition guarantees that any potentially problematic energy flows are infinitesimally close to energy flows for which the observable is well defined.
The strategy in \Eq{eq:sudakov_safe_strategy} also enables the computation of observables such as the $V$ parameter in \Eq{eq:Vobs}, which are EMD continuous everywhere but exhibit H\"{o}lder discontinuities on sets with non-zero measure in $\mathcal P_N$ and are therefore incalculable with fixed-order perturbation theory alone.

It is instructive to make a connection to practical methods of computing Sudakov-safe observables.
In a quantum field theory of massless particles, the cross section to produce events with exactly $N$ particles is zero (i.e.~the naive $S$-matrix is zero), and such theories ultimately yield smooth predictions in the space of events.
Hence, divergences that appear in the calculation of such an observable in a fixed-order expansion can be regulated by a joint, all-orders calculation of the observable and the distance from the problematic manifold $\mathcal P_N$.
This is precisely the strategy represented by \Eq{eq:sudakov_safe_strategy}, though \Ref{Larkoski:2015lea} did not provide a generic method to identify the companion observable $\mathcal{O}_c$.
In \Sec{sec:observables}, we will establish that the distance from an event to the manifold $\mathcal P_N$ is precisely $N$-(sub)jettiness~\cite{Stewart:2010tn,Thaler:2010tr,Thaler:2011gf}, suggesting that they are universal companion observables for the calculation of Sudakov-safe observables, in a similar spirit to \Refs{Alioli:2012fc,Alioli:2013hqa,Alioli:2015toa}.

It is worth mentioning that, even if an observable is EMD H\"{o}lder continuous everywhere, resummation along the lines of \Eq{eq:sudakov_safe_strategy} may still be beneficial for making reliable predictions.
The $C$-parameter~\cite{Parisi:1978eg,Donoghue:1979vi,Ellis:1980wv} is an example of an EMD H\"{o}lder continuous observable, yet its fixed-order perturbative distribution exhibits discontinuous behavior at $C = \frac{3}{4}$~\cite{Catani:1997xc}.
This perturbative discontinuity can be smoothed through soft-gluon resummation, and such techniques are relevant for other observables that exhibit Sudakov shoulder behavior~\cite{Larkoski:2015uaa}.
This is different, however, from Sudakov-safe observables, where the observable itself (and not just its distribution) is ill-defined on some $\mathcal P_N$.

To summarize, our definition of IRC safety does includes Sudakov-safe observables, but we argue that this is appropriate since such observables are indeed perturbatively accessible via regulation with $N$-(sub)jettiness.
This motivates the following conjecture:
\begin{framed}
\begin{conjecture}\label{conj:pertcalc}
An observable is perturbatively calculable, using a combination of fixed-order and resummation techniques, if it is IRC safe according to the definition in \Sec{sec:ircsafety}.
\end{conjecture}
\end{framed}
\noindent 
Proving this conjecture, or finding a counterexample, would shed considerable light on the structure of perturbative quantum field theory.
Of course, even if an observable is perturbatively calculable, it may suffer from large nonperturbative or detector corrections, and it may be helpful to use the $K$ and $\alpha$ paramters in \Eq{eq:Kalpha} to assess the sensitivity of observables to long-distance effects.

\section{Observables: Distances between events and manifolds}
\label{sec:observables}

In this section, we show that a number of event-level and jet substructure observables can be identified as geometric quantities in the space of events.
Broadly speaking, the observables we consider take the general form of a distance between an event and a manifold, as in \Eq{eq:obsdefemd}.
The illustration in \Fig{fig:space_obs} shows an observable as a distance between geometric objects in the space of events.
While not all IRC-safe observables can be written in this way, a remarkably large family of classic observables take precisely this geometric form.
We will work with unnormalized observables here, but normalized versions can be obtained by dividing by the total energy (or transverse momentum in the hadronic case).

\begin{figure}[t]
\centering
\includegraphics[scale=0.6]{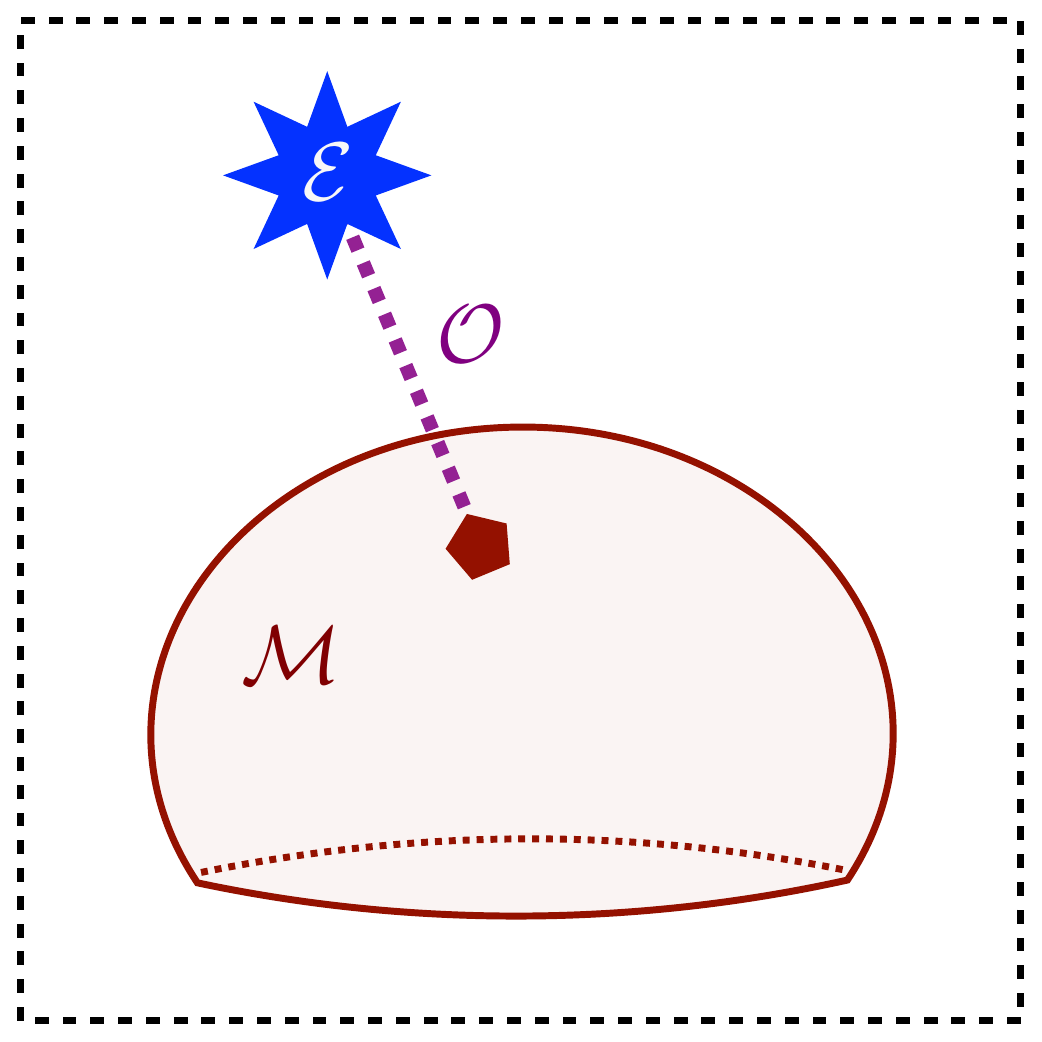}
\caption{\label{fig:space_obs} An illustration of an observable $\mathcal O$ as the distance of closest approach, as measured by the EMD, between the event $\mathcal E$ and a manifold $\mathcal M$ of events.
Many classic collider observables fit into this precise form, stated in \Eq{eq:obsdefemd}, with particular choices of manifold.
}
\end{figure}

We begin by discussing thrust and spherocity, where the manifold is the set of all back-to-back two-particle events.
To understand (recoil-free) broadening, we expand the manifold to all two-particle events, beyond just back-to-back configurations.
Then, to connect to $N$-jettiness, we utilize the idealized $N$-particle manifold defined in \Eq{eq:npmanifold}.
Our geometric language gives clear and intuitive explanations of what physics these observables probe and why they take the forms that they do.
While these EMD formulations do not necessarily lead to practical computational improvements, we do highlight ways to speed up the numerical evaluation of event isotropy using techniques from the optimal transport literature.
Finally, we identify jet angularities and $N$-subjettiness as jet substructure observables obeying similar principles at the level of jets.

\begin{table}[t]
\centering
\begin{tabular}{rccl}
\hline
\hline
  & \multicolumn{3}{l}{$\mathcal O(\mathcal E)=\displaystyle \min_{\mathcal E'\in\mathcal M} \text{EMD}_\beta(\mathcal E, \mathcal E')$} \\
 Name &  & $\beta$ & Manifold $\mathcal M$ \\ \hline\hline
Thrust & $t(\mathcal E)$ & 2& $\mathcal P^{\rm BB}_2$:  2-particle events, back to back \\
Spherocity & $\sqrt{s(\mathcal E)}$ & 1 &  $\mathcal P^{\rm BB}_2$:  2-particle events, back to back \\ 
Broadening & $b(\mathcal E)$ & 1 & $\mathcal P_2$:  2-particle events \\ 
$N$-jettiness & $\mathcal T_N^{(\beta)} (\mathcal E)$& $\beta$ & $\mathcal P_N$:  $N$-particle events \\ 
Isotropy & $\mathcal I^{(\beta)}(\mathcal E)$ & $\beta$ & $\mathcal M_{\mathcal{U}}$: Uniform events\\
\hline
Jet Angularities & $\lambda_\beta(\mathcal J)$ & $\beta$ & $\mathcal P_1$:   1-particle jets\\ 
$N$-subjettiness & $\tau_N^{(\beta)}(\mathcal J)$ & $\beta$ & $\mathcal P_N$:   $N$-particle jets \\ 
\hline
\hline
\end{tabular}
\caption{Observables as the EMD between the event $\mathcal E$ and a manifold $\mathcal{M}$, using the EMD definition in \Eq{eq:emd_noR}.
Several of these observables are illustrated in \Fig{fig:space_manyobs}.
Here, we consider only the ``recoil-free'' versions of these observables.}
\label{tab:obs}
\end{table}

\begin{figure}[t]
\centering
\subfloat[]{\includegraphics[scale=0.685]{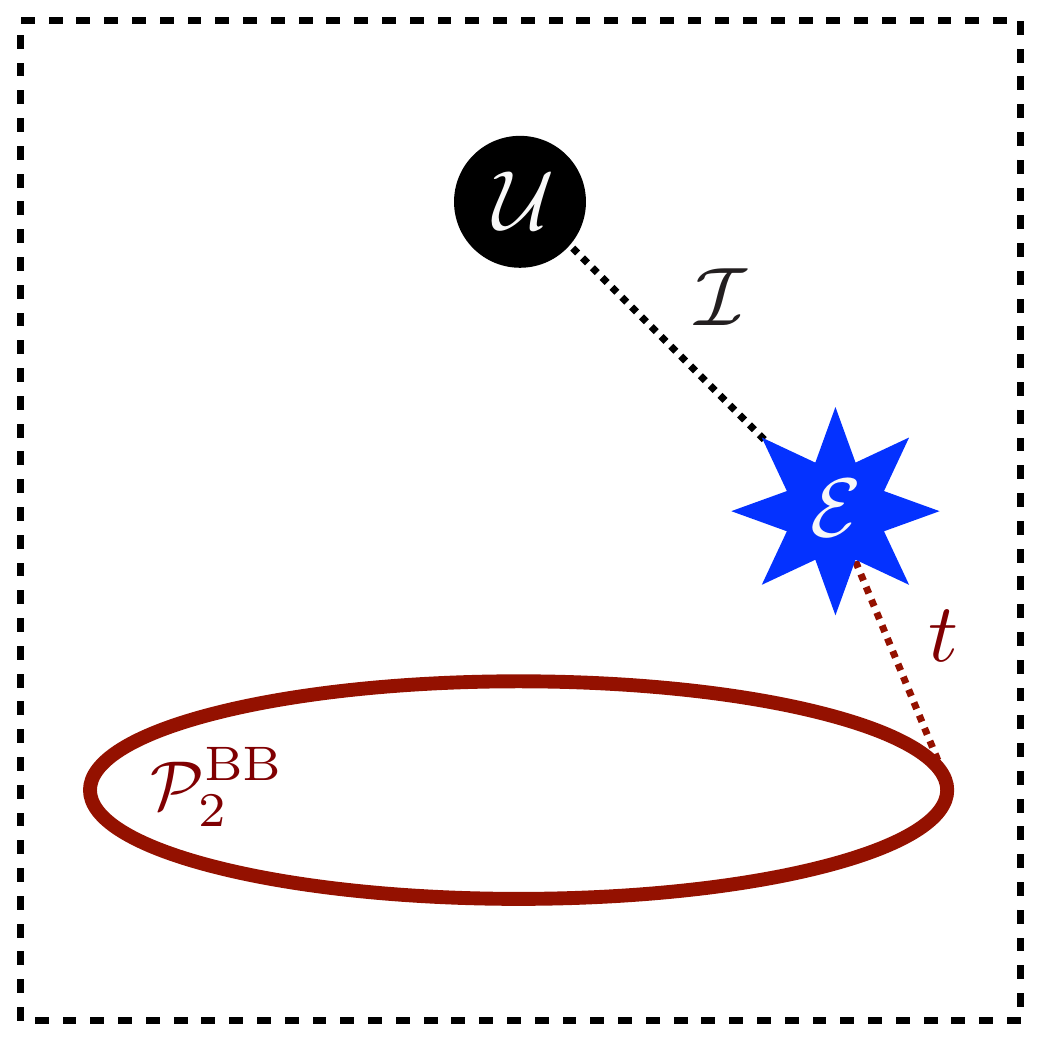}}\hspace{4mm}
\subfloat[]{\includegraphics[scale=0.685]{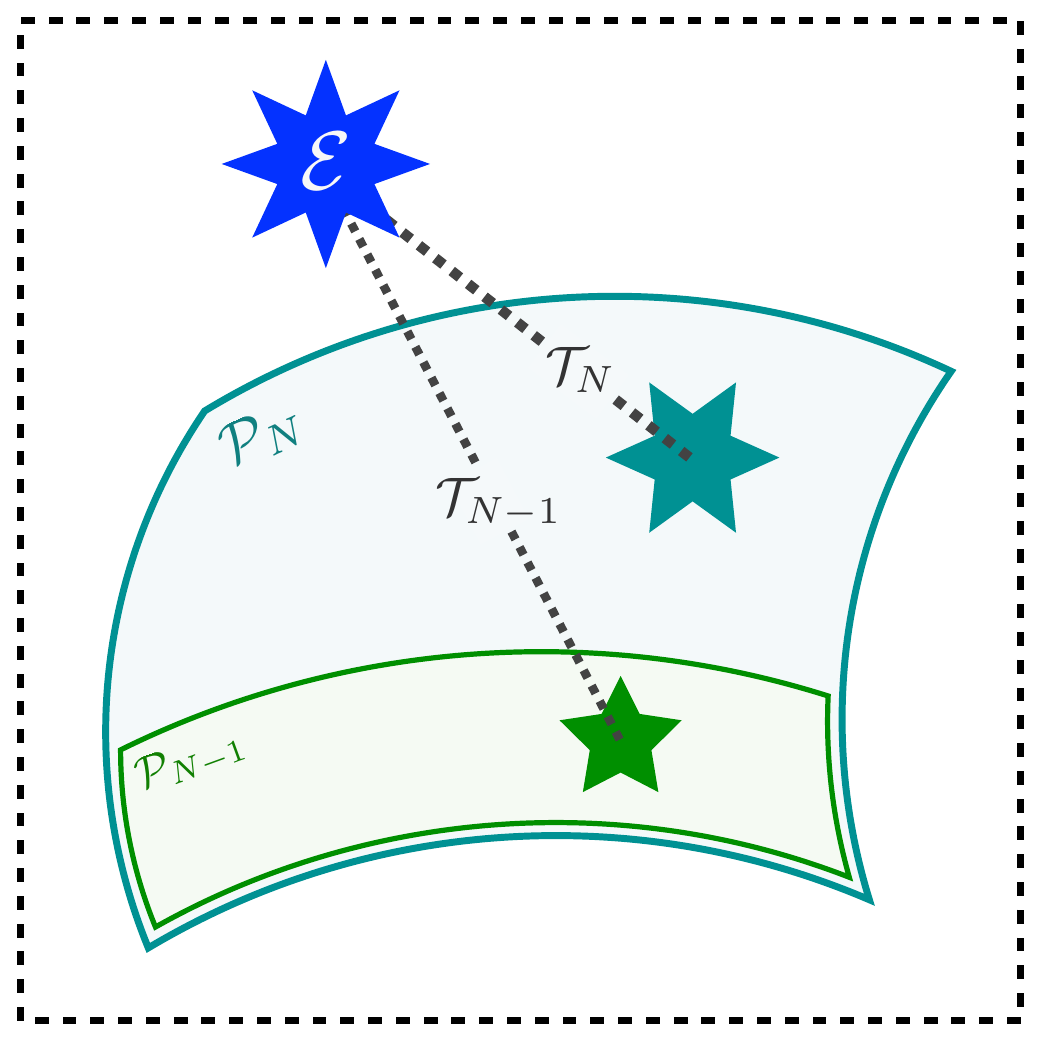}}
\caption{\label{fig:space_manyobs} An illustration of a variety of observables as distances between an event $\mathcal E$ and various manifolds in the space of events, as summarized in \Tab{tab:obs}.
(a) Thrust $t$ is the smallest distance from the event to the manifold $\mathcal P_2^\text{BB}$ of two-particle back-to-back events, while event isotropy $\mathcal I$ is the distance to the uniform event $\mathcal U$.
(b) $N$-jettiness observables $\mathcal T_N$ are the smallest distances from the event to the $N$-particle manifolds $\mathcal P_N$.
}
\end{figure}

For most of the observables in this section, the $R$ parameter is not needed, in which case we define a notion of EMD relevant for comparing events with equal energies:
\begin{equation}
\label{eq:emdRtoinfty}
\EMD_{\beta} (\mathcal E, \mathcal E') = \lim_{R \to \infty} R^\beta \, \EMD_{\beta,R} (\mathcal E, \mathcal E').
\end{equation}
This only has a finite limit if $\mathcal E$ and $\mathcal E'$ have the same total energy, which is a useful property to simplify our analysis.
Explicitly, when comparing events with equal energy, this EMD simplifies to:
\begin{equation}
\label{eq:emd_noR}
\EMD_{\beta} (\mathcal E, \mathcal E') = \min_{\{f_{ij}\ge0\}} \sum_{i=1}^M\sum_{j=1}^{M'} f_{ij} \theta_{ij}^\beta ,
\end{equation}
\begin{equation}
\label{eq:emdconstraints_noR}
\sum_{i=1}^M f_{ij} \le E_j', \quad\quad \sum_{j=1}^{M'} f_{ij} \le E_i, \quad\quad \sum_{i=1}^M\sum_{j=1}^{M'} f_{ij} = \sum_{i=1}^M E_i = \sum_{j=1}^{M'}E_j'.
\end{equation}
This will be the precise notion of EMD we use when the $R$ subscript is suppressed.

In \Tab{tab:obs}, we summarize some of the observables considered below and their geometric interpretations.
In \Fig{fig:space_manyobs}, we illustrate the geometric construction of many of these observables, which we will explore in detail below.

\subsection{Event-level observables}
\label{sec:eventobservables}

\subsubsection{Thrust}
\label{subsubsec:thrust}

Thrust is an observable that quantifies the degree to which an event is pencil-like~\cite{Brandt:1964sa,Farhi:1977sg,DeRujula:1978vmq}.
It has been experimentally measured~\cite{Barber:1979bj,Bartel:1979ut,Althoff:1983ew,Bender:1984fp,Abrams:1989ez,Li:1989sn,Decamp:1990nf,Braunschweig:1990yd,Abe:1994mf,Heister:2003aj,Abdallah:2003xz,Achard:2004sv,Abbiendi:2004qz} and theoretically calculated~\cite{Gehrmann-DeRidder:2007nzq,GehrmannDeRidder:2007hr,Becher:2008cf,Weinzierl:2009ms,Abbate:2010xh,Abbate:2012jh} in detail for electron-positron collisions.
Thrust seeks to find an axis $\hat n$ (the ``thrust axis'') such that most of the radiation lies in the direction of either $\hat n$ or $-\hat n$; i.e.~it maximizes the amount of radiation \emph{longitudinal} to the thrust axis.
While a variety of conventions for defining thrust exist, here we use the following dimensionful definition:
\begin{equation}\label{eq:thrust}
t(\mathcal E) = 2\min_{\hat n}\sum_{i=1}^M|\vec p_i|(1- |\vec n_i \cdot \hat n|),
\end{equation}
where $\hat n_i = \vec p_i/|\vec p_i|$ and other definitions follow by simple rescalings.
A thrust value of zero corresponds to an event consisting of two back-to-back prongs, while its maximum value of the total energy corresponds to a perfectly spherical event.

Interestingly, the value of thrust in \Eq{eq:thrust} is equivalent to the cost of an optimal transport problem.
This connection will allow us to cast thrust as a simple geometric quantity written in terms of the EMD.
Using $E_i=|\vec p_i|$ for massless particles and writing out the absolute value, we can cast \Eq{eq:thrust} as:
\begin{equation}\label{eq:thrust2}
t(\mathcal E) = 2\min_{\hat n}\sum_{i=1}^M E_i \min(1 - \hat n_i \cdot \hat n,\, 1 + \hat n_i \cdot \hat n).
\end{equation}
For a fixed $\hat n$, the summand in \Eq{eq:thrust2} is the transportation cost to move particle $i$ to the closer of $\hat n$ or $-\hat n$ with an angular measure of $\theta_{ij}^2 = 2n_i^\mu n_{j\mu}= 2 (1 - \hat n_i \cdot \hat n_j)$.
The sum is then the EMD between the event and a two-particle event consisting of back-to-back particles directed along $\hat n$, where the energy of each of the two particles is equal to the total energy in the corresponding hemisphere.
The minimization over $\hat n$ is equivalent to a minimization over all such two-particle events.

Thus, thrust is our first example of an observable that can be cast in the form of \Eq{eq:obsdefemd}.
First, we define the manifold of back-to-back two-particle events:
\begin{equation}
\label{eq:bbmanifold}
\mathcal P^{\rm BB}_2 = \left\{\left. \sum_{i=1}^2 E_i\, \delta(\hat n - \hat n_i)\,\, \right| \,\, E_i \ge 0, \,\, \hat{n}_1 = - \hat{n}_2 \right\}.
\end{equation}
Then, using the notation of \Eq{eq:emd_noR} with $\beta=2$,%
\footnote{As mentioned in footnote \ref{footnote:pWasser}, strictly speaking only the square root of $\text{EMD}_2$ is a proper metric.
Because the square root is a monotonic function, though, this has no impact on the interpretation of thrust as an optimal transport problem.}
thrust is the smallest EMD from the event to the $\mathcal P^{\rm BB}_2$ manifold:
\begin{equation}
\label{eq:thrustasEMD}
\begin{boxed}{
t(\mathcal E) = \min_{\mathcal E' \in \mathcal P^{\rm BB}_2} \text{EMD}_2(\mathcal E, \mathcal E'),}
\end{boxed}
\end{equation}
where the minimization is carried out over all back-to-back two-particle configurations.

Because of the $R \to \infty$ limit in \Eq{eq:emdRtoinfty}, the optimal back-to-back configuration is guaranteed to have the same total energy as the event $\mathcal E$, as desired.
Note that even if this analysis is carried out in the center-of-mass frame, the optimal back-to-back configuration will generically not be at rest, since it involves two massless particles with different energies.%
\footnote{We thank Samuel Alipour-fard for discussions related to this point.}
This suggests a possible variant of thrust where one restricts the two-particle manifold to only include events that are physically accessible, either by forcing $E_1 = E_2$ or by considering massive particles as in \App{sec:mass}.

\subsubsection{Spherocity}

Spherocity is an observable that also probes the jetty nature of events~\cite{Georgi:1977sf}.
It seeks to find an axis that minimizes the amount of radiation in the event \emph{transverse} to it according to the following criterion:
\begin{equation}
\label{eq:spherocity_orig}
s(\E) = \min_{\hat n} \left(\sum_{i=1}^ME_i|\vec n_i\times\hat n| \right)^2,
\end{equation}
where the original definition of spherocity is related to this by an overall rescaling.
In the small $s$ limit, where the event configurations are back to back, we can write $|\vec n_i\times\hat n|\simeq \sqrt{2(1- |\hat n_i\cdot\hat n|)}$ and obtain:
\begin{equation}
\label{eq:spherocity}
s(\mathcal E)  \simeq \min_{\hat n} \left( \sum_{i=1}^M E_i \sqrt{2(1 - |\hat n_i \cdot \hat n|)} \right)^2.
\end{equation}
We focus on this limiting form for the following discussion.

Similar to the case of thrust, we can identify the spherocity expression to be minimized as an optimal transport problem.
For a fixed $\hat n$, the summand in \Eq{eq:spherocity} is the cost to transport particle $i$ to the closer of $\hat n$ or $-\hat n$ with an angular measure of $\theta_{ij}=\sqrt{2n^\mu_i n_{j\mu}}$.%
\footnote{In fact, \Eq{eq:spherocity_orig} is already an optimal transport problem, using $\theta_{ij} = \sin \Omega_{ij}$, where $\Omega_{ij}$ is the opening angle between particles $i$ and $j$.  This has the same small angle behavior as $\theta_{ij} = 2 \sin \frac{\Omega_{ij}}{2}$ from \Eq{eq:theta_def}.}
The sum is once again the EMD from the event to the manifold of back-to-back events, with the minimization over $\hat n$ interpreted as a minimization over the manifold.

Spherocity, in the appropriate limit, is therefore the square of the smallest EMD (with $\beta=1$) from the event to the manifold $\mathcal P^{\rm BB}_2$ from \Eq{eq:bbmanifold}:
\begin{equation}\label{eq:spherocityasEMD}
\begin{boxed}{
\sqrt{s(\mathcal E)} = \min_{\mathcal E' \in \mathcal P^{\rm BB}_2} \text{EMD}_1(\mathcal E, \mathcal E').
}\end{boxed}
\end{equation}
Through this lens, spherocity differs from thrust (besides the overall exponent) solely in the angular weighting factor:  $\beta=1$ for spherocity and $\beta=2$ for thrust.
One could continue in this direction, defining the distance of closest approach for general $\beta$.
(This is related to the event shape angularities~\cite{Berger:2003iw}, with a key difference being that angularities are traditionally measured with respect to the thrust axis.)
Instead, we now turn towards enlarging the manifold itself.

\subsubsection{Broadening}

Recoil-free broadening~\cite{Larkoski:2014uqa} is an observable that is sensitive to two-pronged events that are not precisely back-to-back jets.
Here we focus on recoil-free broadening, to be distinguished from the original jet broadening~\cite{Rakow:1981qn,Ellis:1986ig,Catani:1992jc} which is defined in terms of the thrust axis.\footnote{There is an EMD-based definition of the original jet broadening, using the thrust axis defined by $\mathcal E_t = \text{arg min}_{\mathcal E'\in \mathcal P_2^\text{BB}}\text{EMD}_2(\mathcal E, \mathcal E')$. With modified angular measure and normalization, the original jet broadening with respect to the thrust axis is $b_t(\mathcal E) = \text{EMD}_1(\mathcal E,\mathcal E_t)$. Note the two different values of $\beta$ in these expressions.}
It differs from spherocity only in that it minimizes the same quantity over two ``kinked'' axes that need not be antipodal.
Though subtle, this difference gives rise to very important theoretical differences between broadening and spherocity in the treatment of soft recoil effects~\cite{Dokshitzer:1998kz}, as discussed extensively in~\Ref{Larkoski:2014uqa}.

Here, we use the following definition of broadening:
\begin{equation}\label{eq:broadening}
b(\mathcal E) = \min_{\hat n_L,\,\hat n_R} \sum_{i=1}^M E_i \min(\theta_{iL}, \theta_{iR}),
\end{equation}
where $\theta_{iL}$ and $\theta_{iR}$ are the angular distances between particle $i$ and $\hat n_L$ and $\hat n_R$, respectively.
The fact that $\hat n_L$ and $\hat n_R$ are minimized separately (rather than $\hat n_L = -\hat n_R$) is the key distinction between recoil-free broadening and previous observables.
For a fixed $\hat n_L$ and $\hat n_R$, the summand in \Eq{eq:broadening} is the cost to transport particle $i$ to the closer of $\hat n_L$ or $\hat n_R$ with an angular measure of $\theta_{ij} = \sqrt{2n_i^\mu n_{j\mu}}$.
The sum is then the EMD from the event to the manifold of all two-particle events, which need not be back-to-back, namely $\mathcal P_2$ from \Eq{eq:npmanifold}.
The minimization over $\hat n_L$ and $\hat n_R$ is then interpreted as a minimization over this manifold.

Thus, broadening is the smallest EMD with $\beta=1$ from the event to $\mathcal P_2$:
\begin{equation}\label{eq:broadeningasEMD}
\begin{boxed}{
b(\mathcal E) = \min_{\mathcal E' \in \mathcal P_2} \text{EMD}_1(\mathcal E, \mathcal E').
}\end{boxed}
\end{equation}
The geometrical formulation of broadening in \Eq{eq:broadeningasEMD} differs from that of spherocity in \Eq{eq:spherocityasEMD} only in that it does not restrict the manifold to back-to-back configurations.
This distinction is important to extend these ideas beyond the two-particle manifold.

\subsubsection{$N$-jettiness}
\label{sec:njettiness}

$N$-jettiness~\cite{Stewart:2010tn} (see also \Ref{Brandt:1978zm}) is an observable that partitions an event into $N$ jet regions and, for hadronic collisions, a beam region.
Without a beam region, it is defined based on a minimization procedure over $N$ axes:
\begin{equation}\label{eq:Njettiness}
\mathcal T_N^{(\beta)} = \min_{\hat n_1,\cdots,\hat n_N} \sum_{i=1}^M E_i \min\left(\theta_{i1}^\beta, \theta_{i2}^\beta, \cdots, \theta_{iN}^\beta\right),
\end{equation}
where $\theta_{i1}$ through $\theta_{iN}$ are the angular distances between particle $i$ and axes $\hat n_1$ through $\hat n_N$, respectively.

We immediately identify the summand as the cost of transporting particle $i$ to the nearest axis.
For fixed $\hat n_1$ through $\hat n_N$, assigning the energy transported to each axis as the energy of that axis gives rise to an $N$-particle event.
The expression to be minimized is then the EMD between the original event and that $N$-particle event.
The minimization over $\hat n_1$ through $\hat n_N$ is interpreted as a minimization over all such $N$-particle events.

Therefore, $N$-jettiness is the smallest distance between the event and the manifold $\mathcal{P}_N$ of $N$-particle events.
Equivalently, one can view it as the EMD to the best $N$-particle approximation of the event, and we return to this interpretation in \Sec{subsec:xcone}.
Thus, we have:
\begin{equation}
\label{eq:Njettiness_asEMD}
\begin{boxed}{
\mathcal T_N^{(\beta)} = \min_{\mathcal E' \in \mathcal P_N} \text{EMD}_{\beta}(\mathcal E, \mathcal E').
}\end{boxed}
\end{equation}
We see that $N$-jettiness generalizes the geometric interpretation of broadening to a general $N$-particle manifold and a general angular weighting exponent $\beta$.

For hadronic collisions, initial state radiation and underlying event activity require the introduction of a ``beam'' (or out-of-jet) region~\cite{Stewart:2009yx,Stewart:2010tn,Berger:2010xi}.
This can be accomplished via the introduction of a beam distance $\theta_{i,\text{beam}}$ into the minimization of \Eq{eq:Njettiness}.
There are many possible beam measures~\cite{Jouttenus:2013hs,Stewart:2015waa}, including ones that involve optimizing over two beam axes $\hat n_a$ and $\hat n_b$.
For simplicity, we focus on $\theta_{i,\text{beam}} = R^\beta$ which makes no explicit reference to the beam directions~\cite{Thaler:2011gf}.
Dividing by an overall factor of $R^\beta$, this modified version of $N$-jettiness can be written as:
\begin{equation}\label{eq:Njettiness_with_beam}
\mathcal T_N^{(\beta,R)} = \min_{\hat n_1,\cdots,\hat n_N} \sum_{i=1}^M E_i \min\left(1, \frac{\theta_{i1}^\beta}{R^\beta}, \frac{\theta_{i2}^\beta}{R^\beta}, \cdots, \frac{\theta_{iN}^\beta}{R^\beta}\right).
\end{equation}
This definition of $N$-jettiness is similar to \Eq{eq:Njettiness}, though now a particle can be closer to the beam than to any axis.
In this case, we say that the particle is transported to the beam and removed for a cost $E_i$.
The summand is then the cost to transport the event to an $N$-particle event plus the cost of removing any particles beyond $R$ from any axes.

Remarkably, this precisely corresponds to the EMD when formulated for events of different total energy.
Namely, $N$-jettiness with this beam region is simply the smallest distance between the event and the manifold of $N$-particle events, with $R$ smaller than the radius of the space:
\begin{equation}
\label{eq:Njettiness_with_beam_asEMD}
\begin{boxed}{
\mathcal T_N^{(\beta,R)} = \min_{\mathcal E'\in \mathcal P_N} \text{EMD}_{\beta,R}(\mathcal E, \mathcal E').
}\end{boxed}
\end{equation}
Particles removed by the optimal transport procedure are interpreted as being part of the beam region.
This fact will also be relevant in \Sec{sec:seqrec} for understanding sequential recombination jet clustering algorithms as geometric constructions in the space of events.

\subsubsection{Event isotropy}
\label{sec:isotropy}

Our new geometric phrasing of these classic collider observables highlights the types of configurations that they are designed to probe.
Specifically, \Eq{eq:obsdefemd} can be interpreted as how similar an event is to the class of events on the manifold $\mathcal{M}$.
This framework also suggests regions of phase space that are poorly resolved by existing observables and provides a prescription for developing new observables by identifying new manifolds of interest.

Event isotropy~\cite{isotropytemp} is a recently-proposed observable that provides a clear example of this strategy.
It is based on the insight that distances from the $N$-particle manifolds (such as thrust and $N$-jettiness) are not well-suited for resolving isotropic events with uniform radiation patterns.
Having observables with sensitivity to isotropic events can, for instance, improve new physics searches for microscopic black holes or strongly-coupled scenarios.
This motivates event isotropy, which is the distance between the event $\mathcal E$ and an isotropic event $\mathcal U$ of the same total energy:
\begin{equation}\label{eq:isotropy}
\mathcal I^{(\beta)}(\mathcal E) =  \text{EMD}_{\beta}(\mathcal E, \mathcal{U}).
\end{equation}
Since $\mathcal E$ and $\mathcal{U}$ have the same total energy by construction, it is natural to normalize event isotropy by the total energy to make it dimensionless.
The analysis in \Ref{isotropytemp} focused primarily on $\beta = 2$, though this approach can be extended to a general angular exponent.
For practical applications, it is convenient to consider a manifold of quasi-isotropic events of the same total energy and then estimate event isotropy as the average EMD between an event and this manifold.

We can cast \Eq{eq:isotropy} into the form of \Eq{eq:obsdefemd} by introducing a manifold $\mathcal M_{\mathcal U}$ of uniform events with varying total energies:
\begin{equation}\label{eq:isotropy_emd}
\boxed{
\mathcal I^{(\beta)}(\mathcal E) = \min_{\mathcal{E}' \in \mathcal M_{\mathcal U}} \text{EMD}_{\beta}(\mathcal E, \mathcal{E}').
}
\end{equation}
The $R \to \infty$ limit in \Eq{eq:emdRtoinfty} enforces that the optimal isotropic approximation $\mathcal{U}$ has the same total energy as $\mathcal E$, as in the original event isotropy definition.

The particular notion of a uniform distribution depends on the collider context---spherical for electron-positron collisions and cylindrical or ring-like for hadronic collisions---with corresponding choices for the energy and angular measures.
The case of ring-like isotropy at a hadron collider is particularly interesting, since there are known simplifications for one-dimensional circular optimal transport problems.
For $\beta = 1$, ring-like event isotropy can be computed in $\mathcal O(M)$ runtime~\cite{DBLP:journals/jmiv/RabinDG11} and there are fast approximations for any $\beta\ge1$~\cite{DBLP:journals/jmiv/RabinDG11}.
This is much faster than the generic $\mathcal O(M^3 \log M)$ expectation for EMD computations, motivating further studies of these one-dimensional geometries.

\subsection{Jet substructure observables}
\label{sec:jetobservables}

\subsubsection{Jet angularities}

Jet angularities are the energy-weighted angular moments of radiation within a jet~\cite{Ellis:2010rwa} (see also \Refs{Almeida:2008yp,Larkoski:2014uqa,Larkoski:2014pca}).
Here, we use the following definition of a recoil-free jet angularity:
\begin{equation}
\lambda_\beta(\mathcal J) = \min_{\hat n} \sum_{i=1}^M E_i \,\theta_{i}^\beta,
\end{equation}
where $\theta_{i}$ is the angular distance between particle $i$ and an axis $\hat n$.
The summand of an angularity is the EMD from the jet to the axis, so we can follow the analogous logic from our previous discussions of event shapes to reframe this observable in our geometric language.
Specifically, the recoil-free angularities are the closest distance between the jet and the 1-particle manifold $\mathcal{P}_1$:
\begin{equation}
\label{eq:emd_jet_angularities}
\begin{boxed}{
\lambda_\beta(\mathcal J) = \min_{\mathcal J'\in\mathcal P_1} \text{EMD}_{\beta}(\mathcal J, \mathcal J').
}\end{boxed}
\end{equation}
One can alternatively consider a definition of angularities where $\theta_{i}$ is computed with respect to a fixed jet axis.
In that case, the angularities are the EMD from the jet to a 1-particle configuration where the total energy of the jet is placed at the position of the desired axis.

\subsubsection{$N$-subjettiness}
\label{subsec:nsubjettiness}

$N$-subjettiness is a jet substructure observable that applies the ideas of $N$-jettiness at the level of jet substructure~\cite{Thaler:2010tr,Thaler:2011gf}.
$N$ axes are placed within the jet, with a penalty for having energy far away from any axis, and then the positions of the axes are optimized.
The (dimensionful) $N$-subjettiness of a jet can be defined as follows:
\begin{equation}
\tau_N^{(\beta)} (\mathcal J) = \min_{\hat n_1,\cdots,\hat n_N} \sum_{i=1}^M E_i \min\left(\theta_{i1}^\beta, \theta_{i2}^\beta, \cdots, \theta_{iN}^\beta\right),
\end{equation}
where $\theta_{i1}$ through $\theta_{iN}$ are the angular distances between particle $i$ and axes $\hat n_1$ through $\hat n_N$.
The beam region is absent due to the fact that these observables are only defined using the particles already within an identified jet.

\begin{figure}[t]
\centering
\includegraphics[scale=0.7]{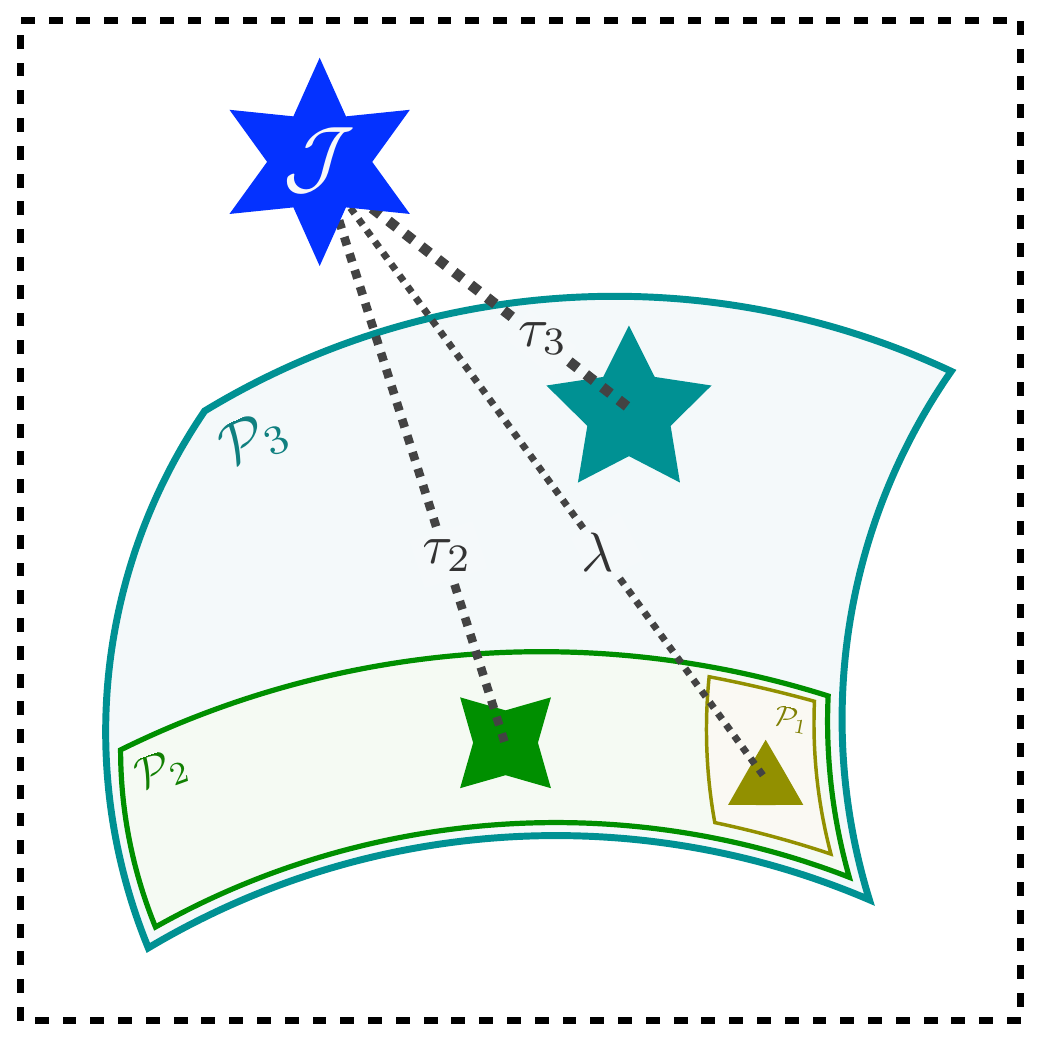}
\caption{\label{fig:space_nsub} An illustration of $N$-subjettiness values as the smallest distances, as measured by EMD, between the event $\mathcal E$ and each of the $N$-particle manifolds $\mathcal P_N$.
The jet angularities are the distances to the 1-particle manifold $\mathcal P_1$.
These observables form a set of ``coordinates'' for the space.
}
\end{figure}

We can find a geometric interpretation for $N$-subjettiness by using the analogous discussion from $N$-jettiness in \Sec{sec:njettiness}.
$N$-subjettiness is the distance between the jet and the manifold of all $N$-particle jets:
\begin{equation}
\begin{boxed}{
\tau_N^{(\beta)}(\mathcal J) = \min_{\mathcal J'\in\mathcal P_N} \text{EMD}_{\beta}(\mathcal J, \mathcal J').
}\end{boxed}
\end{equation}
As a limiting case, $N = 1$ corresponds to the jet angularities in \Eq{eq:emd_jet_angularities}.

In this way, we can view $N$-subjettiness values as ``coordinates'' for the space of jets, defined as distances from each of the $N$-particle manifolds, illustrated in \Fig{fig:space_nsub}.
The $N$-subjettiness ratios $\tau_{N} / \tau_{N-1}$, used ubiquitously for jet substructure studies~\cite{Larkoski:2017jix,Asquith:2018igt,Marzani:2019hun}, are then the relative distances between the manifolds $\mathcal{P}_N$ and $\mathcal{P}_{N-1}$.
This is also an interesting way to interpret existing constructions of observable bases using $N$-(sub)jettiness~\cite{Datta:2017rhs,Datta:2017lxt,Larkoski:2019nwj}; the fact that multiple $\beta$ values are typically needed for these constructions emphasizes that the choice of ground metric affects the geometry of the space induced by the EMD.

\section{Jets: The closest $N$-particle description of an $M$-particle event}
\label{sec:jets}

In this section, we turn our attention to how jets are defined.
We interpret two of the most common classes of jet algorithms as simple geometric constructions in the space of events.
Intuitively, we find that jets are the best $N$-particle approximation to an $M$-particle event. 
Many existing techniques naturally emerge from this simple principle in fascinating ways.

First, we discuss exclusive cone finding, as this technique corresponds exactly to the intuition above that jets approximate the energy flow of an event using a smaller number of particles.
Next, we show that many sequential recombination algorithms can be derived by iteratively approximating an $M$-particle event using $M-1$ particles.
These jet-finding strategies are illustrated in \Fig{fig:space_jets} as projections to $N$-particle manifolds in the space of events.

\begin{figure}[t]
\centering
\subfloat[]{\includegraphics[scale=0.685]{figures/space_xcone}}\hspace{4mm}
\subfloat[]{\includegraphics[scale=0.685]{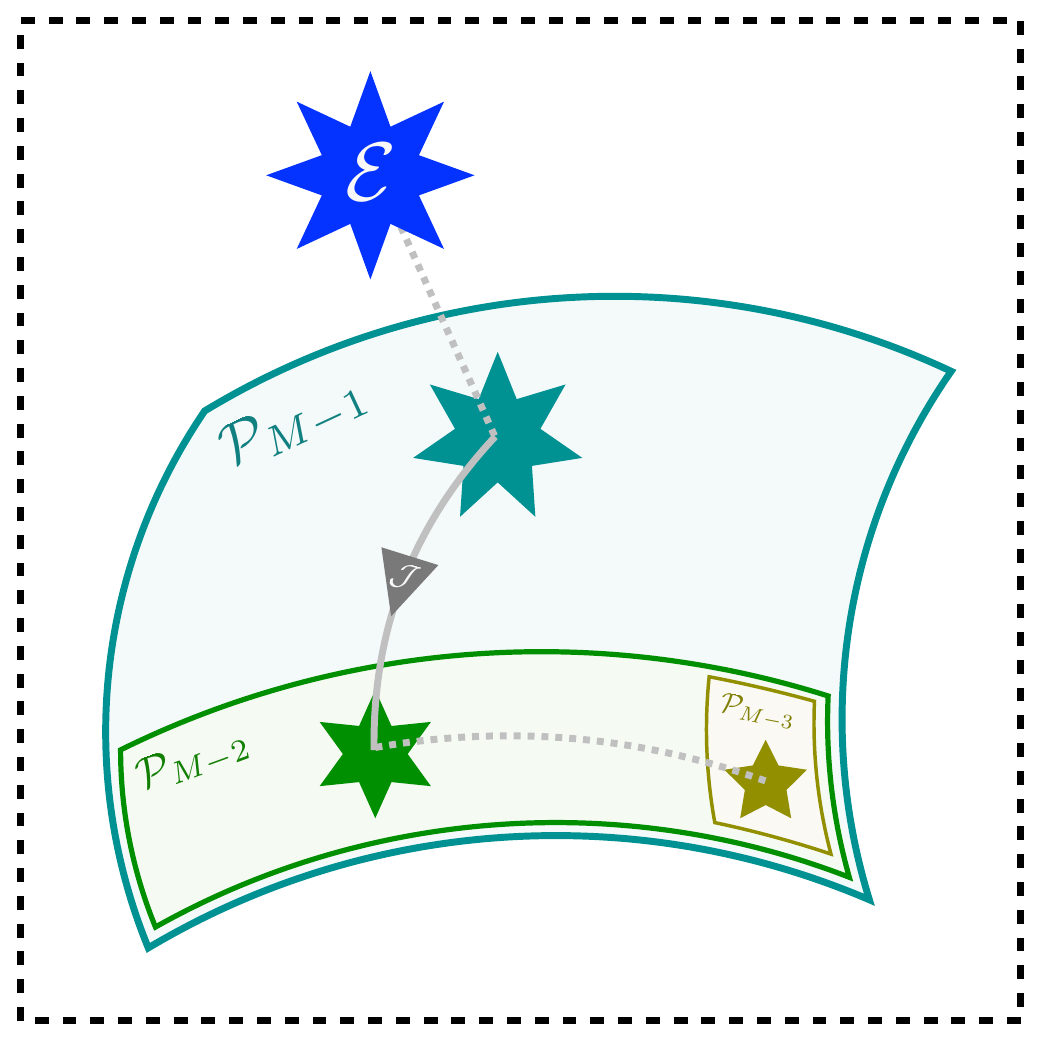}}
\caption{\label{fig:space_jets} An illustration of jet clustering algorithms as projections to $N$-particle manifolds $\mathcal P_N$ in the space of events.
(a) Exclusive cone finding algorithms yield $N$ jets as the closest $N$-particle approximation to the event, as measured by the EMD.
(b) Sequential recombination algorithms iteratively find the best $(M-1)$-particle approximation to the $M$-particle event, either (dashed) merging two particles or (solid) removing a particle and calling it a jet.}
\end{figure}

\subsection{General $N$: Exclusive cone finding}
\label{subsec:xcone}

XCone~\cite{Stewart:2015waa,Thaler:2015xaa} is an exclusive cone finding algorithm that seeks to find jets by minimizing $N$-jettiness.
It returns a fixed number of jets based on the parameters $N$ and $R$, in the same spirit as the exclusive version of the $k_t$ sequential recombination algorithm~\cite{Catani:1993hr}.
XCone proceeds by finding the $N$ axes that minimize $N$-jettiness as defined in \Eq{eq:Njettiness_with_beam}:
\begin{equation}\label{eq:XCone}
\underset{\hat n_1,\cdots,\hat n_N}{\text{argmin}}\sum_{i=1}^M E_i \min\left(1, \frac{\theta_{i1}^\beta}{R^\beta}, \frac{\theta_{i2}^\beta}{R^\beta}, \cdots, \frac{\theta_{iN}^\beta}{R^\beta}\right).
\end{equation}
Together with the energy assigned to those axes, or equivalently the set of particles mapped to each axis, the $N$ axes from \Eq{eq:XCone} define $N$ jets.
The jet radius parameter $R$ controls which particles are not assigned to any jet (i.e.\ assigned to the beam region). 
Following the discussion in \Sec{sec:njettiness}, \Eq{eq:XCone} can be interpreted as finding the $N$-particle configuration that best approximates the event of interest.

In our geometric language, we can cast XCone as identifying the point of closest approach between an event $\mathcal E$ and the $N$-particle manifold $\mathcal{P}_N$:
\begin{equation}\label{eq:XConeasEMD}
\begin{boxed}{
\mathcal J^\text{XCone}_{N,\beta,R}(\mathcal E) = \underset{\mathcal J\in\mathcal P_N}{\text{argmin}}\,\,\text{EMD}_{\beta,R}(\mathcal E, \mathcal J).
}\end{boxed}
\end{equation}
Different variants of XCone correspond to different choices for the energy weight $E_i$ and the angular measure $\theta_{ij}$~\cite{Jouttenus:2013hs,Stewart:2015waa}, which in turn correspond to different choices for what defines the ``best'' $N$-particle approximation to an event.

As discussed in \Ref{Thaler:2015uja}, there is a close relationship between exclusive cone finding algorithms, stable cone algorithms~\cite{Blazey:2000qt,Ellis:2001aa,Salam:2007xv}, and jet maximization algorithms~\cite{Georgi:2014zwa,Ge:2014ova,Bai:2014qca,Bai:2015fka,Wei:2019rqy}.
For the choice of $\beta = 2$, the jet axis aligns with the jet momentum direction, which is known as the stable cone criterion~\cite{Blazey:2000qt,Ellis:2001aa}.
For $N = 1$, one can relate the optimization problem in \Eq{eq:XConeasEMD} to maximizing a ``jet function'' over all possible partitions of an event into one in-jet region and one out-of-jet region~\cite{Georgi:2014zwa}.
Iteratively applying the $N = 1$ procedure is related to the \text{SISCone} algorithm with progressive jet removal~\cite{Salam:2007xv}.
All of these various algorithms can now be interpreted in our geometric picture as different ways to ``project'' the event $\mathcal E$ onto the $N$-particle manifold $\mathcal{P}_N$.

\subsection{$N=M-1$: Sequential recombination}
\label{sec:seqrec}

Sequential recombination algorithms are a class of jet clustering algorithms that have seen tremendous use at colliders, particularly the anti-$k_t$ algorithm~\cite{Cacciari:2008gp} which is the current default jet algorithm at the LHC.
These methods utilize an interparticle distance $d_{ij}$, a particle-beam distance $d_{iB}$, and a recombination scheme for merging two particles.
The algorithm proceeds iteratively by finding the smallest distance, combining particle $i$ and $j$ if it is a $d_{ij}$, or calling $i$ a jet and removing it from further clustering if it is a $d_{iB}$.

There are a variety of distance measures and recombination schemes that appear in the literature, many of which are implemented in the \textsc{FastJet} library~\cite{Cacciari:2011ma}.
The most commonly used distance measures take the form:
\begin{equation}\label{eq:dij}
d_{ij} = \min\left(E_i^{2p}, E_j^{2p}\right) \frac{\theta_{ij}^2}{R^2},\quad \quad d_{iB} = E_i^{2p},
\end{equation}
where $p$ is an energy weighting exponent and $R$ is the jet radius.
The exponent $p=1$ corresponds to $k_t$ jet clustering~\cite{Catani:1993hr,Ellis:1993tq}, $p=0$ corresponds to Cambridge/Aachen (C/A) clustering~\cite{Dokshitzer:1997in,Wobisch:1998wt}, and $p=-1$ corresponds to anti-$k_t$ clustering~\cite{Cacciari:2008gp}.
The recombination scheme determines the energy $E_c$ and direction $\hat n_c$ of the combined particle and typically takes the form:
\begin{equation}
\label{eq:recomb}
E_c = E_i  + E_j,\quad \quad \hat n_c = \frac{E_i^\kappa \, \hat n_i + E_j^\kappa \, \hat n_j}{E_i^\kappa + E_j^\kappa},
\end{equation}
where $\kappa=1$ corresponds the $E$-scheme (most typically used), $\kappa=2$ is the $E^2$-scheme~\cite{Catani:1993hr,Butterworth:2002xg}, and $\kappa\to\infty$ is the winner-take-all scheme~\cite{Bertolini:2013iqa,Larkoski:2014uqa,Salambroadening}.
In the $E$-scheme, the four-momenta of the two particles are simply added.%
\footnote{\label{footnote:Escheme}One has to be a bit careful about the interpretation of jet masses in the $E$-scheme.  In the discussion below, the combined particle is interpreted as a massless four-vector.  For the angular distance in \Eq{eq:theta_def}, the direction $\hat{n}_i$ is the same for massless and massive particles, so one can consistently assign the mass of the jet to be the invariant mass of the summed jet constituents.  For the rapidity-azimuth distance typically used at hadron colliders, though, the rapidity of a particle depends on its mass, so one has to be careful about whether one is talking about a light-like jet axis or a massive jet when discussing the $E$-scheme.  See further discussion in \App{sec:mass}.}
In the winner-take-all scheme, the direction is determined by the more energetic particle.

\begin{table}[t]
\centering
\begin{tabular}{c|ccc|cc}
\hline\hline
EMD$_{\beta,R} $& Name & Measure $d_{ij}$ & $d_{iB}$ & Name & Scheme $\lambda^*$ \\ \hline\hline
$0<\beta < 1$ & Gen.\ $k_t$ & $\min(E_i, E_j) \frac{\theta_{ij}^\beta}{R^\beta}$ & $E_i$ & Winner-take-all  & $\text{argmin}(E_i,E_j)$ \\
$\beta = 1$ & $k_T$ & $\min(E_i, E_j) \frac{\theta_{ij}^{\phantom{\beta}}}{R}$ & $E_i$ & Winner-take-all  & $\text{argmin}(E_i,E_j)$ \\
$\beta = \frac32$ & ? & $\frac{E_i E_j}{\sqrt{E_i^2 + E_j^2}} \frac{\theta_{ij}^\frac32}{R^\frac32}$ & $E_i$ & $E^2$-scheme & $\frac{E_j^2 }{E_i^2 + E_j^2}$ \\
$\beta = 2$ & ? & $\frac{E_i E_j}{E_i + E_j} \frac{\theta_{ij}^2}{R^2}$ &  $E_i$ & $E$-scheme &  $\frac{E_j}{E_i + E_j}$ \\
\hline
$\beta\to\infty$ & C/A & $\frac{\theta^{\phantom{\beta}}_{ij}}{2R}$ & 1 & ? & $\frac12$\\ \hline \hline
\end{tabular}
\caption{Different sequential recombination measures $d_{ij}$ and recombination schemes $\lambda^*$ that emerge from an EMD formulation.
A question mark indicates a method that, to our knowledge, does not yet appear in the literature.
The traditional definitions of generalized $k_t$ and C/A require squaring $d_{ij}$ and $d_{iB}$.
Note the factor of 2 in the C/A effective jet radius parameter.
}
\label{tab:seqrec}
\end{table}

The conceptual and algorithmic richness of these different distance measures and recombination schemes arose from decades of phenomenological studies.
Remarkably, many of these techniques naturally emerge from event space geometry, as finding the point on the $(M-1)$-particle manifold $\mathcal P_{M-1}$ that is closest to configuration $\mathcal{E}$ with $M$ particles.
Note that the sequential recombination algorithms in \Eqs{eq:dij}{eq:recomb} depend on the two parameters $p$ and $\kappa$, whereas \Eq{eq:SRasEMD} depends only on $\beta$, so the logic below will only identify a one-dimensional family of jet algorithms, as summarized in \Tab{tab:seqrec}.

To derive this connection between event geometry and sequential recombination, we need the following simple yet profound lemma, using the suggestive notation of $d_{iB}$ and $d_{ij}$ to refer to the EMD cost of rearrangement.
\begin{lemma}
As measured by the EMD, the closest $(M-1)$-particle event to an $M$-particle event has, without loss of generality, either:
\begin{enumerate}
\item[(a)] Two of the particles in the event merged together.
\item[(b)] One of the particles in the event removed.
\end{enumerate}
\end{lemma}
\begin{proof}
Removing a particle from the event has some EMD cost $d_{iB}$ and merging a pair of particles has a some EMD cost $d_{ij}$.
To reduce the number of particles in the event by one, one can either remove a particle or merge two particles.
Altering more than two particles by (re)moving fractions of additional particles always incurs additional EMD costs.
If there are multiple pairs that are zero distance apart, then we can without loss of generality always choose to only merge one pair.
\end{proof}

The two options in this lemma correspond precisely to the two possible actions at each stage of a sequential recombination algorithm.
The EMD cost of removing a particle is always
\begin{equation}\label{eq:diBemd}
d_{iB} = E_i.
\end{equation}
If this is less than the cost of merging two particles together, then particle $i$ can be identified as a jet.
For one step of a sequential recombination (SR) procedure applied to an event $\mathcal E$ with $M$ particles, we can express this step mathematically as:
\begin{equation}\label{eq:SRasEMD}
\begin{boxed}{
\mathcal J_{\beta, R}^\text{SR}(\mathcal E) = \mathcal E - \underset{\mathcal E' \in \mathcal P_{M-1}}{\text{argmin}}\,\,\text{EMD}_{\beta,R}(\mathcal E, \mathcal E').
}\end{boxed}
\end{equation}
In our geometric picture, if the $M$ particle event is ``far away'' from the $(M-1)$-particle manifold $\mathcal P_{M-1}$, then the projected difference is a jet.

On the other hand, if the cost of merging two particles is less than any of the particle energies, then the event is ``close'' to the $(M-1)$-particle manifold.
Consider a pair of particles with energies $E_i$ and $E_j$ separated by a distance $\theta_{ij}$.
To find the best $(M-1)$-particle approximation, we want to merge these two particles into one combined particle with energy $E_i + E_j$.
Because the EMD is a metric, the optimal transportation plan must occur along a ``geodesic'' connecting the particles, with particle $i$ moving a distance $\lambda\,\theta_{ij}$ and particle $j$ moving a distance $(1-\lambda)\,\theta_{ij}$ for some $\lambda \in [0,1]$.\footnote{This linear decomposition of the distance does not hold for a general ground metric. However, it does hold when using the rapidity-azimuth distance, the opening-angle on the sphere, the small angle limit of \Eq{eq:theta_def}, or the improved $\theta_{ij}$ distance with particle masses in \App{sec:mass}.}
Minimizing this cost with respect to $\lambda$ yields both the cost of merging those two particles as well as the optimal recombination scheme with which to merge them.
Because no energy is removed in this process, \Eq{eq:SRasEMD} yields a zero energy jet, which we can interpret as no jet being found at this step of the sequential recombination.

The cost of merging particles $i$ and $j$ depends on the jet radius parameter $R$ and angular exponent $\beta$:
\begin{equation}\label{eq:emd2ps}
d_{ij} = \min_\lambda\left[E_i \left(\lambda \,\frac{\theta_{ij}}{R}\right)^\beta + E_j \left((1-\lambda) \,\frac{\theta_{ij}}{R}\right)^\beta\right].
\end{equation}
For $\beta \le 1$, the cost in \Eq{eq:emd2ps} is minimized at the endpoints.
This corresponds to moving the less energetic particle the entire distance $\theta_{ij}$ to the more energetic particle, which is the precisely behavior of the winner-take-all recombination scheme.
For $\beta>1$, the optimal value $\lambda^*$ can be found by differentiating \Eq{eq:emd2ps} with respect to $\lambda$ and setting the result equal to zero.
In general, the optimal recombination scheme has:
\begin{align}
0 < \beta \le 1:& \quad \lambda^* = 1 \text{ if $E_i < E_j$, else } 0, \nonumber \\
\beta>1: & \quad \lambda^* = \frac{1}{1 + \left(\frac{E_i}{E_j}\right)^{\frac{1}{\beta-1}}}.
\label{eq:opt_recomb}
\end{align}
To determine the actual cost, we substitute this $\lambda^*$ back into \Eq{eq:emd2ps}:
\begin{align}
\beta \le 1:&\quad d_{ij} = \min(E_i, E_j)\frac{\theta_{ij}^\beta}{R^\beta}, \nonumber\\
\beta>1: & \quad d_{ij} = \frac{E_i E_j^{\frac{\beta}{\beta-1} }+ E_i^{\frac{\beta}{\beta-1}}E_j}{\left(E_i^{\frac{1}{\beta-1}} + E_j^{\frac{1}{\beta-1}}\right)^\beta}\frac{\theta_{ij}^\beta}{R^\beta}.
\label{eq:dijemd}
\end{align}
If all $d_{ij}$ values in \Eq{eq:dijemd} are smaller than all particle energies in \Eq{eq:diBemd}, then the optimal transportation plan is to merge particles $i$ and $j$.

In this way, \Eq{eq:SRasEMD} takes an $M$-particle event and returns a jet (with zero energy if no actual jet is found) plus the remaining $(M-1)$-particle approximation.
This corresponds exactly to one step of a sequential clustering procedure.
Iterating this procedure until $M = 1$, we derive a sequential recombination jet algorithm, where the jets correspond to all of the positive energy configurations obtained from \Eq{eq:SRasEMD}.

Many existing methods reside within the simple framework of \Eq{eq:SRasEMD}.
For instance, $\beta=1$ corresponds to $k_t$ jet clustering with winner-take-all recombination.
The recombination scheme for $\beta=2$ is the $E$-scheme, whereas for $\beta=\frac32$ it is the $E^2$-scheme.
Raising the distance measures to the $1/\beta$ power and taking the $\beta\to\infty$ limit, we obtain the C/A clustering metric, albeit with an effective jet radius that is twice the $R$ parameter.
There are also a number of methods, indicated as question marks in \Tab{tab:seqrec}, that emerge from this reasoning yet do not presently appear in the literature.
Exploring these new methods is an interesting avenue for future work.

Intriguingly, in this geometric picture, the distance measure $d_{ij}$ and the recombination scheme $\lambda^*$ are paired by the $\beta$ parameter.
A similar pairing was noted in \Refs{Stewart:2015waa,Dasgupta:2015lxh} in the context of choosing approximate axes for computing $N$-(sub)jettiness, and it would be interesting to explore the phenomenological implications of these paired choices for jet clustering.
One sequential combination algorithm that does not appear is anti-$k_t$.
Given that anti-$k_t$ is a kind of hybrid between sequential recombination and cone algorithms, there may be a way to combine the logic of \Secs{subsec:xcone}{sec:seqrec} to find a geometric phrasing of anti-$k_t$.
If successful, such a geometric construction would likely illuminate the difference between exclusive jet algorithms like XCone that find a fixed number of jets $N$ and inclusive jet algorithms like anti-$k_t$ that determine $N$ dynamically.

\section{Pileup subtraction: Moving away from uniform events}
\label{sec:pileup}

The LHC era has brought with it new collider data analysis challenges.
One notable example is pileup mitigation~\cite{Soyez:2018opl}, removing the diffuse soft contamination from additional uncorrelated proton-proton collisions.
The radiation from pileup interactions is approximately uniform in the rapidity-azimuth plane, and several existing pileup mitigation strategies seek to remove this uniform distribution of energy from the event~\cite{Cacciari:2007fd,Krohn:2013lba,Cacciari:2014jta,Cacciari:2014gra,Bertolini:2014bba,Berta:2014eza,Komiske:2017ubm,Monk:2018clo,Martinez:2018fwc}.

\begin{figure}[t]
\centering
\includegraphics[width=0.5\columnwidth]{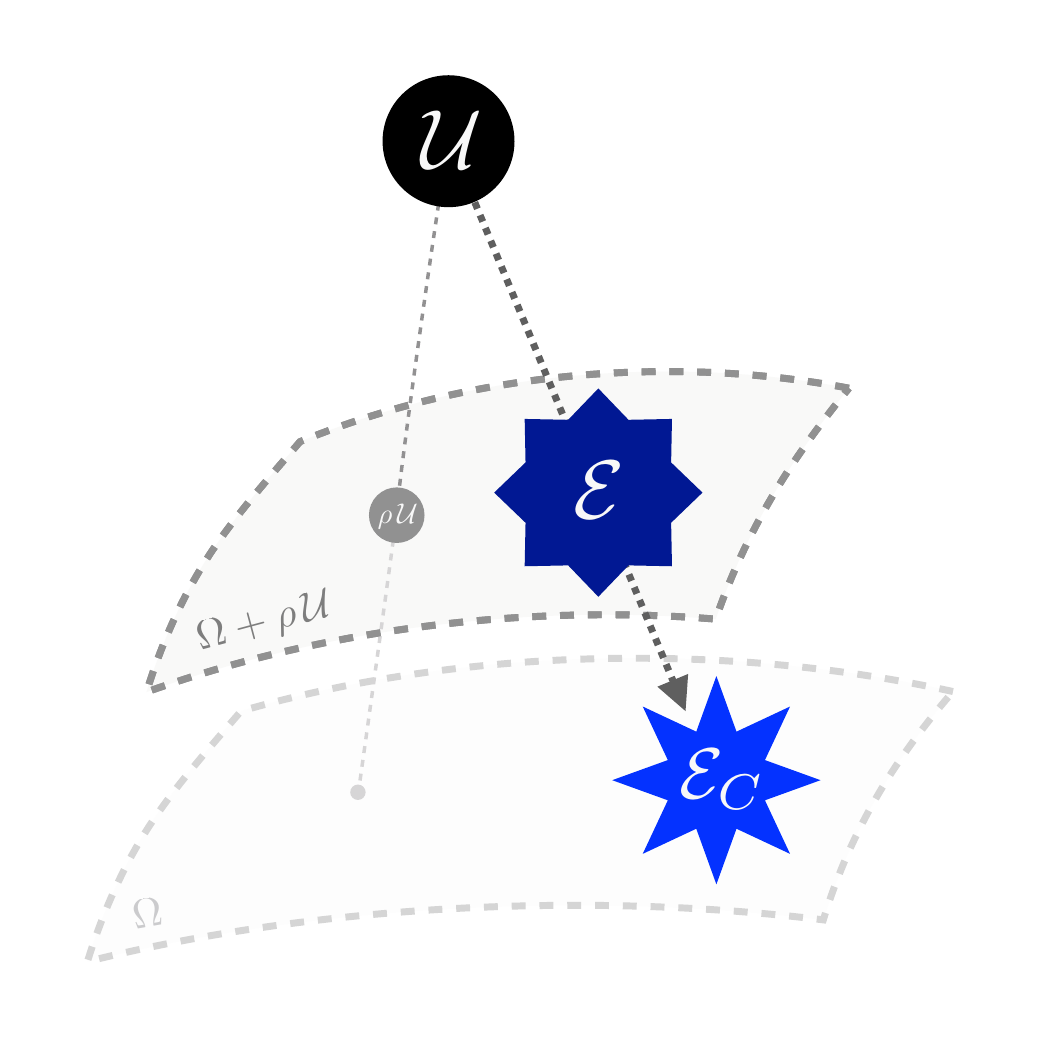}
\caption{
\label{fig:spacepu}
A visualization of pileup subtraction in the space of events as moving away from uniform radiation.
This proceeds by finding the event $\mathcal{E}_C$ that, when combined with uniform contamination $\rho \, \mathcal{U}$, is most similar to the given event $\mathcal{E}$.
Different pileup mitigation strategies implement this removal in different ways.
In the figure above, $\Omega$ refers to the space of all energy flows and $\Omega + \rho \, U$ is a subset of that space obtained by adding uniform contamination to every event configuration (shown as a separate manifold for ease of visualization). 
}
\end{figure}

In this section, inspired by the approximate uniformity of pileup, we consider a class of pileup removal procedures that can be described as ``subtracting'' a uniform distribution of energy with density $\rho$, denoted $\rho\,\mathcal U$, from a given event.
We take the pileup density per unit area $\rho$ to be given, for instance, by the area-median approach~\cite{Cacciari:2007fd}.
Given an event flow $\mathcal E$, the subtracted distribution $\mathcal E - \rho \, \mathcal{U}$ is typically not a valid energy flow, since the local energy density can go negative.
Therefore, to implement this principle at the level of energy distributions, we turn this logic around and declare the corrected event $\mathcal{E}_C$ to be one that is as close as possible to the given event $\E$ when uniform radiation $\rho\,\mathcal U$ is added to it:
\begin{equation}
\label{eq:pileupemd}
\E_C(\E, \rho)  = \argmin_{\E' \in \Omega}\,\EMD_\beta(\E,\E' + \rho\,\mathcal U).
\end{equation}
Here, $\Omega$ refers to the complete space of energy flows, and the $R\to\infty$ limit of the EMD from \Eq{eq:emdRtoinfty} enforces that the corrected distribution $\E_C$ has the correct total energy.

As illustrated in \Fig{fig:spacepu}, one can visualize \Eq{eq:pileupemd} as a procedure that subtracts a uniform component from the energy flow.
To make contact with existing techniques, we show that area-based Voronoi subtraction~\cite{Cacciari:2007fd,Cacciari:2008gn,Cacciari:2011ma} and ghost-based constituent subtraction~\cite{Berta:2014eza} can be cast in the form of \Eq{eq:pileupemd} in the low-pileup limit.
We then develop two new pileup mitigation techniques that have optimal transport interpretations even away from the low-pileup limit: Apollonius subtraction, which corresponds to exactly implementing \Eq{eq:pileupemd} for $\beta=1$, and iterated Voronoi subtraction, which repeatedly applies \Eq{eq:pileupemd} with an infinitesimal $\rho$.
Since pileup is characteristic of a hadron collider, throughout this section we compute the EMD using particle transverse momenta $p_{T,i}$ and rapidity-azimuth coordinates $\hat n_i = (y_i,\phi_i)$, with $\theta_{ij}$ being the rapidity-azimuth distance.
Typically, pileup is taken to be uniform in a bounded region of the plane (e.g.~$|y|< y_{\rm max}$), though the specifics will not significantly affect our analysis.
First, though, we establish an important lemma that justifies why the corrected distribution $\E_C$ has a particle-like interpretation.

\subsection{A property of semi-discrete optimal transport}

There is a direct connection between pileup subtraction in \Eq{eq:pileupemd} and semi-discrete optimal transport~\cite{hartmann2017semi}.
Semi-discrete means that we are comparing a discrete energy flow (i.e.~one composed of individual particles) to a smooth distribution (i.e.~uniform pileup contamination).

Importantly, if $\mathcal{E}$ is discrete, then the corrected distribution $\E_C$ will also be discrete.
This can be proved via the following lemma.
\begin{lemma}
\label{lemma:pileupemd}
$\E_C$ defined according to \Eq{eq:pileupemd} is strictly contained in $\E$, where containment here means that $\E-\E_C$ is a valid distribution with non-negative particle transverse momenta.
\end{lemma}
\begin{proof}
Suppose for the sake of contradiction that $\E_C$ is defined according to \Eq{eq:pileupemd} has some support where $\E$ does not.
Let $\tilde\E$ be the distribution that $\E_C$ flows to when $\E_C+\rho\,\mathcal U$ is optimally transported to $\E$, noting that by definition, $\tilde\E$ must be contained in $\E$.
By the linear sum structure of \Eq{eq:emd_noR}~\cite{OTbook}, we have the following relation:
\begin{equation}
\label{eq:emdineq}
\EMD_\beta(\E,\E_C+\rho\,\mathcal U)=\EMD_\beta(\tilde\E,\E_C)+\EMD_\beta(\E-\tilde\E,\rho\,\mathcal U).
\end{equation}
Now using the following property of $\EMD_\beta$ inherited from Wasserstein distances~\cite{hartmann2017semi}:
\begin{equation}
\label{eq:wassersteinprop}
\EMD_\beta(\mathcal E,\mathcal F) \ge  \EMD_\beta(\mathcal E+\mathcal G,\mathcal F+\mathcal G),
\end{equation}
with equality if $\beta=1$ and the ground metric is Euclidean, we add $\tilde\E$ to both arguments of the last term in \Eq{eq:emdineq} and apply \Eq{eq:wassersteinprop} to find:
\begin{equation}
\label{eq:emdineq2}
\EMD_\beta(\E,\E_C+\rho\,\mathcal U)\ge\EMD_\beta(\tilde\E,\E_C)+\EMD_\beta(\E,\tilde\E+\rho\,\mathcal U).
\end{equation}
Now using that $\EMD_\beta(\tilde\E,\E_C)>0$ by the assumption that they have different supports as well as the non-negativity of the EMD, we find:
\begin{equation}
\label{eq:contradiction}
\EMD_\beta(\E,\E_C+\rho\,\mathcal U)>\EMD_\beta(\E,\tilde\E+\rho\,\mathcal U),
\end{equation}
which contradicts the assumption that $\E_C$ is found according to \Eq{eq:pileupemd}.
Thus, we conclude that $\E_C$ has no support outside of the support of $\E$, verifying the claim.
\end{proof}

This lemma establishes that pileup mitigation strategies defined by \Eq{eq:pileupemd} act by scaling the energies of the particles in the original event $\mathcal{E}$, not by producing new particles.
Indeed, this is a desirable feature of many popular pileup mitigations schemes, including two well-known methods that we describe next.

\subsection{Voronoi area subtraction}
\label{sec:voronoi}

\begin{figure}[t]
\centering
\subfloat[Voronoi]{\includegraphics[width=0.33\columnwidth]{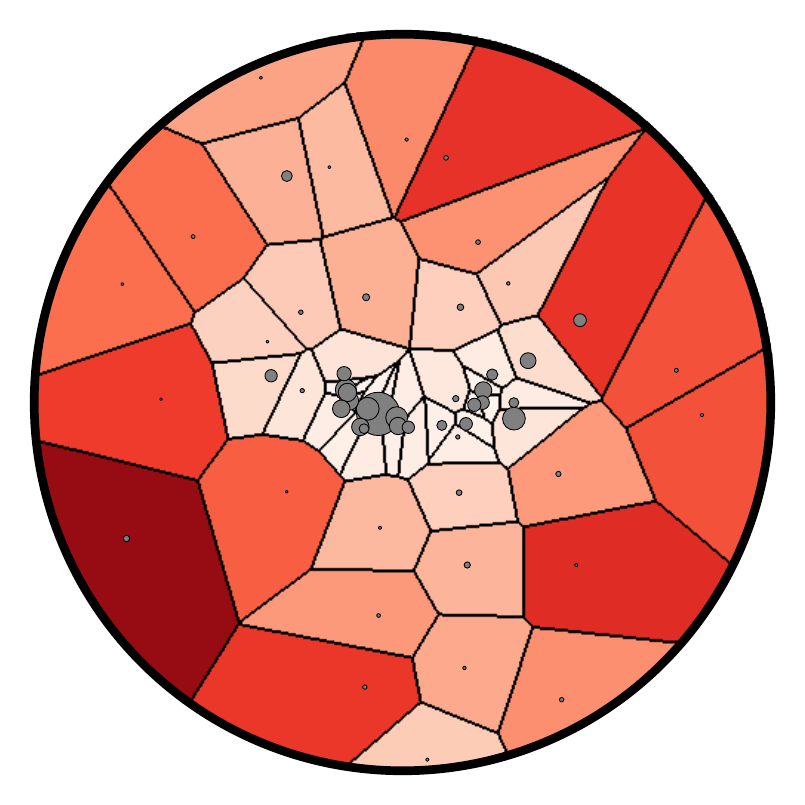}\label{fig:voronoi}}
\subfloat[Constituent Subtraction]{\includegraphics[width=0.33\columnwidth]{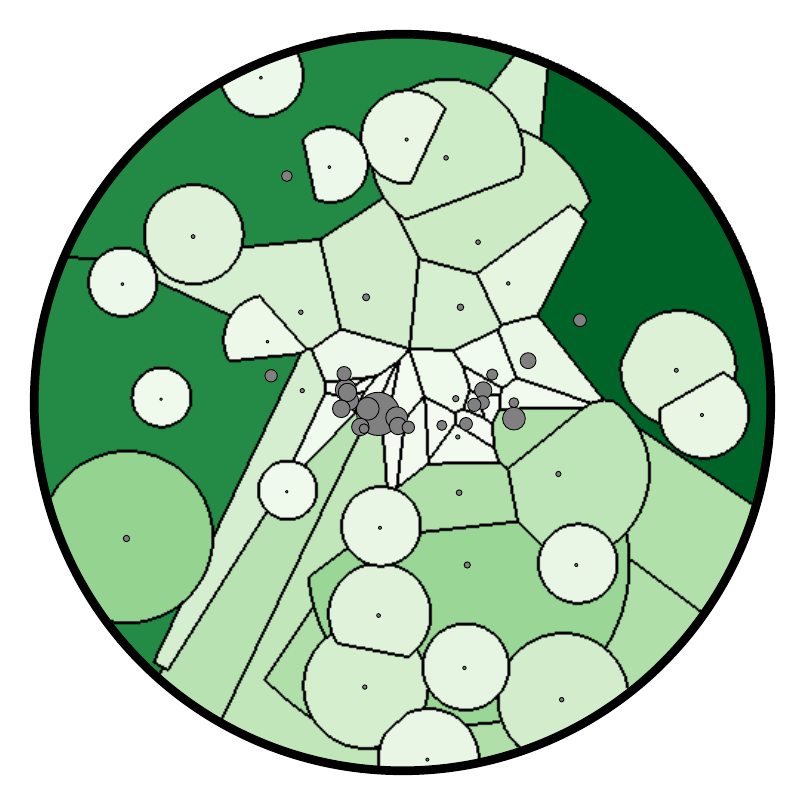}\label{fig:constsub}}
\subfloat[Apollonius]{\includegraphics[width=0.33\columnwidth]{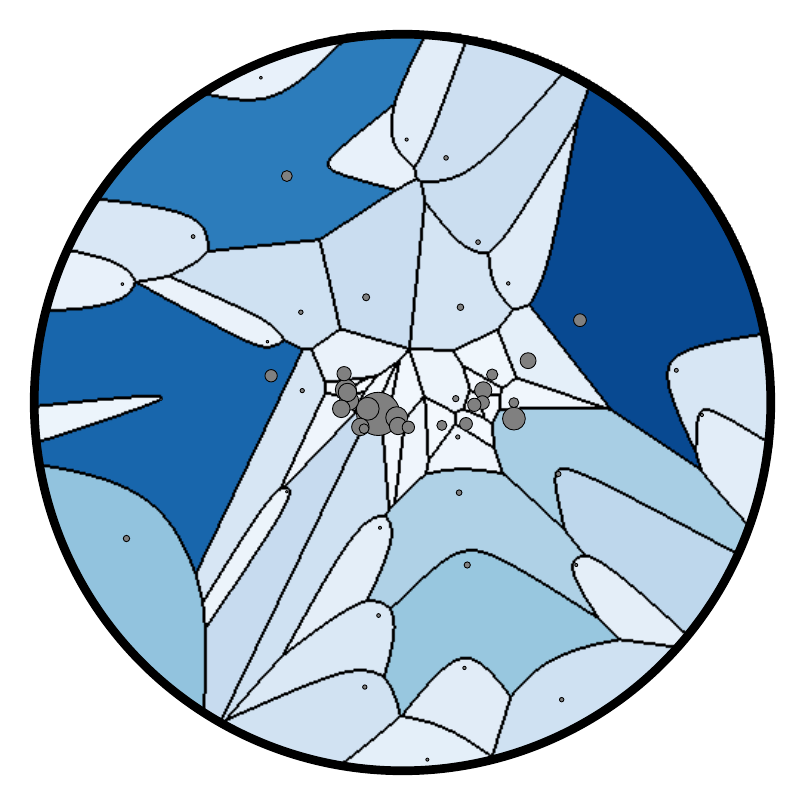}\label{fig:apollonius}}
\caption{
\label{fig:pileupexample}
Three pileup subtraction procedures shown for a jet from the 2011A CMS Jet Primary Dataset.
The jet constituents are shown as gray disks at their locations in the rapidity-azimuth plane with sizes proportional to their transverse momenta.
The boundary is at a distance $R=0.5$ from the jet axis at the center.
%
%
The Voronoi diagram (a) is independent of $\rho$.
The constituent subtraction (b) and Apollonius (c) diagrams are determined using a $\rho$ that corresponds to subtracting one-tenth of the total scalar $p_T$ of the jet.
}
\end{figure}

Voronoi area subtraction~\cite{Cacciari:2007fd,Cacciari:2008gn,Cacciari:2011ma} is a pileup mitigation technique that estimates a particle's pileup contamination by associating it with an area determined by its corresponding \emph{Voronoi region}, or the set of points in the plane closer to that particle than any other~\cite{Aurenhammer2013:book}.
Letting $A_i^\text{Vor.}$ be the area of the Voronoi region of particle $i$, Voronoi subtraction then simply removes $\rho A_i^\text{Vor.}$ from each particle's transverse momentum, without letting the particle $p_T$ become negative.
If $\rho A_i^\text{Vor.}\ge p_{T,i}$ then the particle is removed entirely.
In \Fig{fig:voronoi}, we show the Voronoi regions for an example jet recorded by the CMS detector~\cite{CMS:JetPrimary2011A,Komiske:2019jim}.

Voronoi area subtraction (VAS) can be thought of as carving up the uniform event $\rho\,\mathcal U$ according to the original event's Voronoi diagram and transporting this energy to the location of the corresponding particle, yielding the corrected energy flow:%
\begin{equation}
\label{eq:evas}
\E_\text{VAS}(\hat n)=\sum_{i=1}^M\max\left(p_{T,i}-\rho A_i^\text{Vor.},0\right)\delta(\hat n - \hat n_i).
\end{equation}
Strictly speaking, Voronoi area subtraction does not satisfy exact IRC invariance (see \Eqs{eq:exactirsafety}{eq:exactcsafety}) and thus it cannot in general be written as operating on energy flows.
The reason is that an exact IRC splitting changes the number of Voronoi regions as well as their areas.
In order for \Eq{eq:evas} to be valid, we therefore assume that particles with exactly zero transverse momentum are removed and exactly coincident particles are combined before applying the Voronoi area subtraction procedure.

In the limit that $\rho\le p_{T,i}/A_i^\text{Vor.}$ for all particles $i$, the max in \Eq{eq:evas} evaluates to just its first argument.
In this case, since no particle is assigned a larger correction than its own transverse momentum, the Voronoi diagram gives the optimal transportation plan that minimizes the \EMD of moving the uniform event with density $\rho$ onto the event of interest:
\begin{equation}\label{eq:evas2}
\boxed{
\rho \le \min\left\{p_{T,i}/A_i^\text{Vor.}\right\}: \quad \mathcal E_\text{VAS} = \argmin_{\mathcal E' \in \Omega} \EMD_1(\mathcal E, \mathcal E' + \rho\,\mathcal U).
}
\end{equation}
Thus, in this small-pileup limit, \Eq{eq:evas} agrees with \Eq{eq:pileupemd} with $\beta=1$.
Despite this attractive geometric interpretation, Voronoi area subtraction beyond this limit is sensitive to arbitrarily soft particles: the amount that is subtracted depends only on particle positions, through their Voronoi areas, and not their transverse momenta.

\subsection{Constituent subtraction}
\label{subsec:const_sub}

Constituent subtraction~\cite{Berta:2014eza} is another pileup mitigation method that resolves several pathologies of Voronoi area subtraction by correcting the particles in a manner that depends on both their positions and their transverse momenta.%
\footnote{In this discussion, we focus on the $\alpha=0$ case of constituent subtraction, as recommended by \Ref{Berta:2014eza}.}
This comes at the cost of requiring a fine grid of low energy ``ghost'' particles with $p_T^g=\rho A^\text{ghost}$, where $A^\text{ghost}$ is the area assigned to each ghost, as a proxy for the pileup contamination.
The algorithm is applied by considering the geometrically closest ghost-particle pair $k,i$ and modifying them via:
\begin{equation}
p_{T,i} \to \max(p_{T,i} - p_{T,k}^g, 0), \quad \quad p_{T,k}^g \to \max(p_{T,k}^g - p_{T,i}, 0),
\end{equation}
continuing until all such pairs have been considered.
Since the number of ghosts is typically large in order to have fine angular granularity, this iteration through all ghost-particle pairs can be computationally expensive.

Constituent subtraction (CS) in the continuum ghost limit can be geometrically described by placing circles around each point in the rapidity-azimuth plane and simultaneously increasing their radii.
Each point in the plane is assigned to the particle whose circle reaches it first.
Circles stop growing when $A_i^\text{CS}$, the area assigned to particle $i$, grows larger than $p_{T,i}/\rho$.
We can write the resulting distribution as:
\begin{equation}
\label{eq:ecs}
\E_\text{CS}(\hat n)=\sum_{i=1}^M\left(p_{T,i}-\rho A_i^\text{CS}\right)\delta(\hat n-\hat n_i).
\end{equation}
Unlike naive Voronoi area subtraction, continuum constituent subtraction satisfies exact IRC invariance, since a zero energy particle has zero $A^\text{CS}$ and an exact collinear splitting yields two areas that sum to the original $A^\text{CS}$.
Constituent subtraction is also better suited for intermediate values of $\rho$, where particles can be fully removed, since further corrections are distributed to the next closest particle instead of being ignored as in Voronoi area subtraction.

Due to the complicated shapes of the corresponding regions, it is difficult to describe the areas $A_i^\text{CS}$ analytically and in practice they need to be estimated using numerical ghosts.
An example of constituent subtraction is shown in \Fig{fig:constsub}, where it can be seen that some region boundaries are straight and thus contained in the Voronoi diagram of \Fig{fig:voronoi}.
Indeed, growing circles from a set of points and assigning points in the plane according to which circle reaches them first is another way of describing the construction of a Voronoi diagram.
Regions with circular boundaries correspond to softer particles that are fully subtracted by the constituent subtraction procedure.

When $\rho$ is sufficiently small such that no particle's region has a circular boundary (i.e.\ no circle stops growing), constituent subtraction is exactly equivalent to Voronoi area subtraction.
Constituent subtraction in the low-pileup limit is then also equivalent to optimally transporting the uniform event with density $\rho$ to the event of interest and subtracting accordingly, again in line with \Eq{eq:pileupemd} with $\beta=1$:
\begin{equation}\label{eq:ecs2}
\boxed{
\rho \le \min\left\{p_{T,i}/A_i^\text{Vor.}\right\}: \quad \mathcal E_\text{CS} = \argmin_{\mathcal E' \in \Omega} \EMD_1(\mathcal E, \mathcal E' + \rho\,\mathcal U).
}
\end{equation}
Constituent subtraction can also be extended with a $\Delta R^\text{max}$ parameter to restrict ghosts from affecting distant particles.
Our geometric formalism can also encompass this locality by re-introducing the $R$-parameter to the EMD in \Eq{eq:ecs2} with $R = \Delta R^\text{max}$.

\subsection{Apollonius subtraction}
\label{sec:apollonius}

Voronoi area subtraction and constituent subtraction both make contact with \Eq{eq:pileupemd} in the small-$\rho$ limit, but we would like to explore pileup subtraction based on optimal transport for all values of $\rho$.
By Lemma~\ref{lemma:pileupemd}, we know that the corrected event is contained in the original event, and by the decomposition properties of the EMD in \Eq{eq:emdineq}, we only need to consider the transport of $\rho\,\mathcal U$ to $\E$.
Since the total transverse momenta of $\rho\,\mathcal U$ and $\E$ are generally different, this is now an example of a semi-discrete, \emph{unbalanced} optimal transport problem~\cite{OTtheory,bourne2018semi}.

The problem of minimizing the EMD between a uniform distribution and an event is solved, for general $\beta$, by a \emph{generalized Laguerre diagram}~\cite{bourne2018semi}.
For the special case of $\beta=1$, which we focus on here, this is also known as the \emph{Apollonius diagram} (or additively weighted Voronoi diagram)~\cite{DBLP:conf/esa/KaravelasY02,DBLP:journals/comgeo/GeissKPR13,hartmann2017semi}, and for $\beta=2$ it is a \emph{power diagram}~\cite{DBLP:journals/tog/XinLCCYTW16}.
An Apollonius diagram in the plane is constructed from a set of points $\hat n_i$ that each carry a non-negative weight $w_i$ that is the $i^\text{th}$ component of a vector ${\bf w}\in\mathbb{R}_+^M$.
In the two-dimensional Euclidean plane, the Apollonius region associated to particle $i$ depending on ${\bf w}$ is:
\begin{equation}
\label{eq:apollonius}
R^\text{Apoll.}_i({\bf w})=\left\{\hat n\in\mathbb{R}^2\,\big|\,\|\hat n-\hat n_i\|-w_i\le\|\hat n-\hat n_j\|-w_j,\,\,\forall\,\,j\neq i\right\},
\end{equation}
where particle indices $i,j=1,\ldots,M$ and $\|\cdot\|$ is the Euclidean norm.
One interpretation of \Eq{eq:apollonius} is that region $i$ is all points closer to a circle of radius $w_i$ centered at $\hat n_i$ than to the corresponding circle for any other particle.
The boundaries of the Apollonius regions are contained in the set $\{\hat n\in\mathbb{R}^2\,|\,\|\hat n-\hat n_i\|-\|\hat n-\hat n_j\|=w_i-w_j\}$, which is a union of hyperbolic segments.
Note that adding the same constant to all of the weights does not change the resulting Apollonius diagram.
Hence, if all the weights are equal, they can equivalently be set to zero and we attain the Voronoi diagram as a limiting case of an Apollonius diagram.

We can now specify the action of Apollonius subtraction on an event using the areas of the Apollonius regions subject to the minimal EMD requirement:
\begin{align}
\label{eq:corrapollonius}
\E_\text{Apoll.}(\hat n) &=\sum_{i=1}^M\left(p_{T,i} - \rho A_i^\text{Apoll.}({\bf w}^*)\right)\delta(\hat n-\hat n_i),
\\ \label{eq:wstar}
{\bf w}^* &=\argmin_{{\bf w}\in\mathbb{R}_+^M}\sum_{i=1}^M\EMD_{1,R}\left(\{p_{T,i},\hat n_i\},\rho R_i^\text{Apoll.}(\bf w)\right),
\end{align}
treating $\rho R_i^\text{Apoll.}({\bf w})$ as an event with uniform energy density $\rho$ in that Apollonius region.
Here, \Eq{eq:corrapollonius} is analogous to \Eqs{eq:evas}{eq:ecs}, and \Eq{eq:wstar} implements the requirement that the EMD of the subtraction is minimal.
Note that the $R$ parameter in \Eq{eq:wstar} serves only to guarantee that it is more efficient to transport energy rather than create/destroy it.
As long as $2R$ is greater than the diameter of the space, $R$ has no impact on the solution other than to guarantee that $\rho A_i^\text{Apoll.}({\bf w}^*)$ does not exceed $p_{T,i}$, as this would be less efficient than transporting the excess energy elsewhere.
An example of an Apollonius diagram is shown in \Fig{fig:apollonius}, where hyperbolic boundaries of the Apollonius regions are clearly seen in the outer part of the jet and straight boundaries, matching those of the Voronoi diagram, are seen near the core.

In this way, Apollonius subtraction generalizes Voronoi area and constituent subtraction beyond the small-pileup limit, directly implementing \Eq{eq:pileupemd} for $\beta=1$ for all values of $\rho$:
\begin{equation}\label{eq:eapoll2}
\boxed{
\mathcal E_\text{Apoll.} = \argmin_{\mathcal E' \in \Omega} \EMD_1(\mathcal E, \mathcal E' + \rho\,\mathcal U).
}
\end{equation}
While the optimal solution in \Eq{eq:wstar} is based on an unbalanced optimal transport problem, the restatement in \Eq{eq:eapoll2} corresponds to balanced transport.
This same connection underpins Lemma~\ref{lemma:pileupemd}, guaranteeing that the corrected event in \Eq{eq:corrapollonius} involves the same $M$ directions as the original event, just with different weights.

To turn \Eq{eq:eapoll2} into a practical algorithm, we would need an efficient way to compute the weights according to \Eq{eq:wstar}.
While \Refs{OTtheory,bourne2018semi} have developed the theoretical framework of semi-discrete, unbalanced optimal transport needed to solve this convex minimization problem, they stop short of describing easily-implementable algorithms to attain practical solutions.
In order to create \Fig{fig:apollonius}, we were limited to using numerical ghosts to directly solve for the transport plan that minimizes the EMD cost of subtracting the uniform energy component from the event, which is too computationally costly for LHC applications.

If the target areas $A_i^\text{Apoll.}$ are previously specified, then the solution to \Eq{eq:wstar} simplifies~\cite{hartmann2017semi}.
Given that the areas depend nontrivially on the resulting weight vector, though, the only case where we know them ahead of time is when $\rho$ is such that \emph{all} of the energy will be exactly subtracted, in which case $A_i^\text{Apoll.}=p_{T,i}/\rho$.
Though this is not so useful for pileup, where we typically want to subtract an amount of energy less than the total, it does indicate that an Apollonius diagram can be found and used to compute the event isotropy from \Sec{sec:isotropy} without the use of numerical ghosts.
We leave the implementation of such a procedure to future work, though we note that \Ref{hartmann2017semi} has already built an implementation that relies on numerical ghosts to estimate the areas of the Apollonius regions rather than solving for them analytically.

\subsection{Iterated Voronoi subtraction}
\label{sec:ivs}

Given the difficulty of analytically solving \Eq{eq:wstar} and thus implementing Apollonius subtraction, we now develop an alternative method called iterated Voronoi subtraction that gives up a global notion of minimizing EMD but retains a local one.
In all three methods described above, the difficulty comes when a particle is removed in the course of subtracting pileup.
Otherwise, the above methods all reduce to subtracting transverse momentum according to the Voronoi areas of the regions corresponding to the particles, as in \Eq{eq:evas2}.
This suggests a procedure in which pileup is subtracted according to \Eq{eq:pileupemd} an infinitesimal amount at a time, thus ensuring that \Eq{eq:evas2} can be used at every stage of the procedure.

The area of the Voronoi cell of particle $i$ is now a function of the total amount of energy density that has been subtracted thus far, a quantity that starts at zero and will be integrated up to the target $\rho^\text{tot}$ over the course of the procedure.
When a particle loses all of its transverse momentum, it is removed from the Voronoi diagram and is considered to have zero area associated to it.
The removal of a particle from the diagram changes the Voronoi regions of all of its neighbors, and their areas are updated accordingly.
Denoting the area associated to particle $i$ after $\rho$ worth of energy density has been subtracted as $A_i^\text{IVS}(\rho)$, we can write the corrected distribution for iterated Voronoi subtraction (IVS) as:
\begin{equation}
\label{eq:corrivs}
\E_\text{IVS}(\hat n)=\sum_{i=1}^M\left(p_{T,i} - \int_0^{\rho^\text{tot}} d\rho\,A_i^\text{IVS}(\rho)\right)\delta(\hat n-\hat n_i).
\end{equation}

Unlike \Eqs{eq:corrapollonius}{eq:wstar}, \Eq{eq:corrivs} naturally lends itself to a simple and efficient implementation.
We can iteratively solve for $A_i^\text{IVS}(\rho)$ using the fact that the areas correspond to Voronoi regions, and furthermore that these regions change only when a particle is removed.
Let $\E^{(0)}$ be the initial event consisting of particles with transverse momenta $p_{T,i}^{(0)}$ in Voronoi regions with area $A_i^{(0)}$.
We subtract a total energy density $\rho^\text{tot}$ by breaking up the integral in \Eq{eq:corrivs} starting with $\rho^{(0)}=0$ and determining the boundaries from:
\begin{equation}
\label{eq:rhon}
\rho^{(n)}=\max\left\{\rho\,\Bigg|\,0\le\rho\le\rho^\text{tot}-\sum_{i=0}^{n-1}\rho^{(i)}\,\text{ s.t. }\rho A_j^{(n-1)}\le p_{T,j}^{(n-1)}\,\forall\,j\right\},
\end{equation}
where $n$ starts at 1 and goes up to at most $M$.
The values of $\rho^{(n)}$ can be expressed simply as:
\begin{equation}
\label{eq:rhon2}
\rho^{(n)}=\min\left(\min\left\{\frac{p_{T,j}^{(n-1)}}{A_j^{(n-1)}}\right\}, \rho^\text{tot} - \sum_{i=0}^{n-1} \rho^{(i)}\right),
\end{equation}
where the inner minimum is taken over all remaining particles with $p_{T,j}^{(n-1)}>0$.

The updated particle momenta from each piece of the integral in \Eq{eq:corrivs} are then:
\begin{equation}
p_{T,i}^{(n)}=p_{T,i}^{(n-1)}-\rho^{(n)}A_i^{(n-1)},
\end{equation}
and a particle is considered removed if its transverse momentum is zero, in which case it is also considered to have zero area.
The areas $A_i^{(n)}$ are determined by the Voronoi diagram of $\E^{(n)}$.
The above procedure terminates either when the total amount of energy density removed is equal to $\rho^\text{tot}$ or there are no more particles left in the event.
Thus, iterated Voronoi subtraction makes contact with the geometric perspective of \Eq{eq:pileupemd}, applying it in infinitesimal increments, resulting in the discrete steps:
\begin{equation}\label{eq:eitvor2}
\boxed{
\mathcal E_\text{IVS}^{(n)} = \argmin_{\mathcal E' \in \Omega} \EMD_1(\mathcal E^{(n-1)}, \mathcal E' + \rho^{(n)}\,\mathcal U).
}
\end{equation}
Said another way, this is simply a repeated application of Voronoi area subtraction:  subtract until a particle reaches zero momentum, and repeat until the desired energy density has been removed.

Iterated Voronoi subtraction is made even more attractive computationally when one considers that the Voronoi diagram of $\E^{(n)}$ does not need to be recomputed from scratch.
Rather, it can be obtained from the Voronoi diagram of $\E^{(n-1)}$ by removing a site and updating only the neighboring regions.
Thus, we only need to construct the Voronoi diagram of $\E^{(0)}$ and each removal can be done in constant (amortized) time as the average number of neighbors of any cell is no more than 6~\cite{Aurenhammer2013:book}.
We have constructed an implementation of iterated Voronoi subtraction that interfaces with \textsc{FastJet} and will explore its phenomenological properties in future work.

\section{Theory space}
\label{sec:theory}

When do two theories give rise to similar signatures?
In this section, we seek to generalize the intuition behind the EMD to obtain a metric between theories using their predicted cross sections in energy flow space.
A construction of such a distance and the induced ``theory space'' is conceptually useful and, in fact, naturally underpins several recently introduced techniques for collider physics.

We introduce the cross section mover's distance ($\Sigma$MD) as a metric for the space of theories.
Here, we treat a ``theory'' as an ensemble of event energy flows with corresponding cross sections, encompassing both the predictions of quantum field theories as well as the structure of collider datasets.
To accomplish this, we again use an EMD-like construction, except the $\Sigma$MD uses the EMD itself as the ``angles'' and the event cross sections as the ``energies'', as mentioned in \Tab{tab:emdsmdcomp}.
The resulting space of theories with the $\Sigma$MD as a metric is illustrated in \Fig{fig:theory_space}.
Interestingly, \Ref{Erdmenger:2020abc} also put a metric on theory space by using the Fisher information matrix.

\begin{figure}[t]
\centering
\includegraphics[scale=0.7]{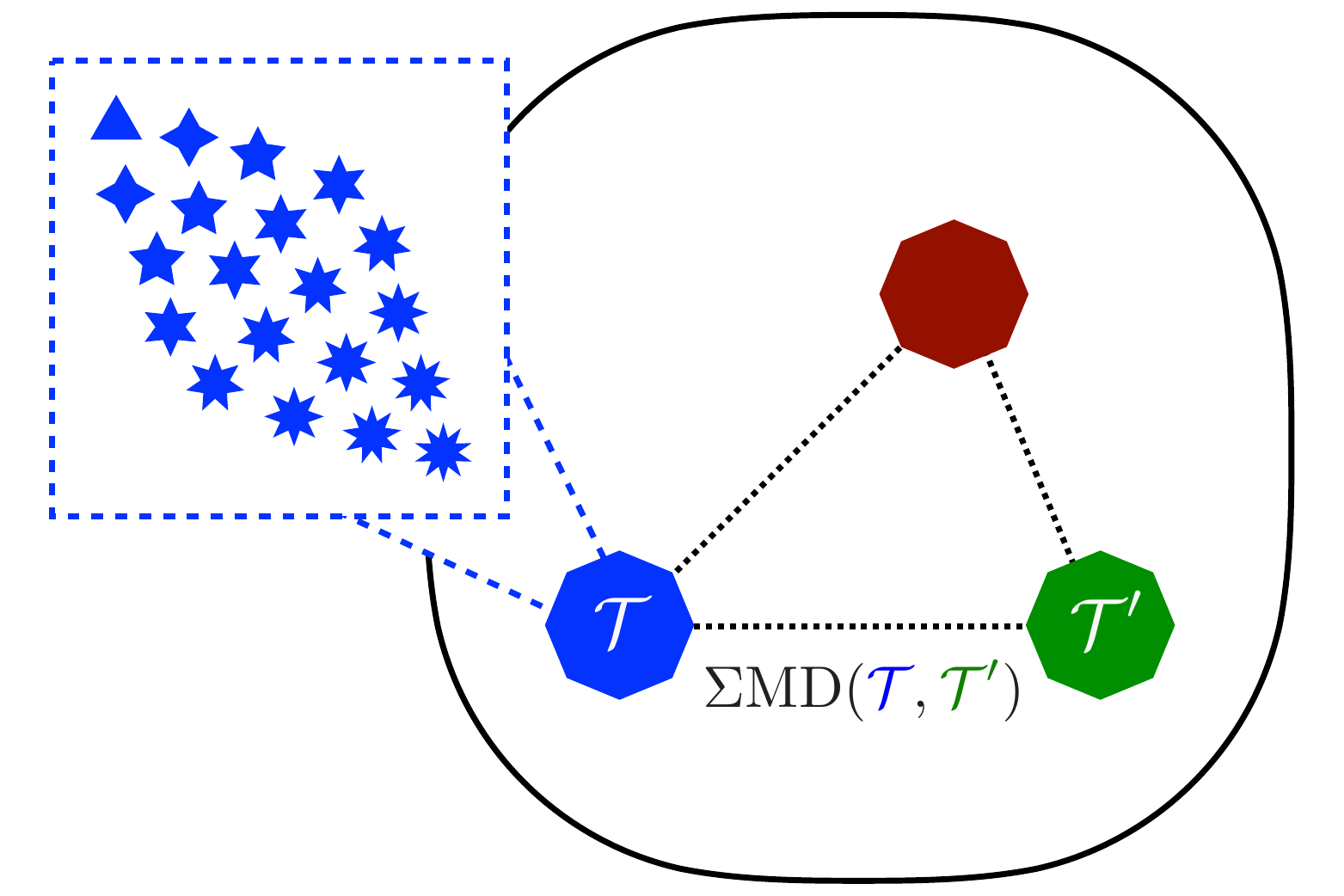}
\caption{\label{fig:theory_space} An illustration of the space of theories.
Each point in the space is a ``theory'': a distribution over (or collection of) events, as indicated by the blue point.
The distance between theories is quantified by the $\Sigma$MD, giving rise to a metric space.}
\end{figure}

\subsection{Introducing a distance between theories}

A ``theory'' $\mathcal T$ is taken to be a (finite, for now) set of events with associated cross sections $\{(\mathcal E_i,\sigma_i)\}_{i=1}^N$.
We can equivalently view $\mathcal T$ as a distribution over the space of event configurations $\mathcal E$:
\begin{equation}
\label{eq:T_def}
\mathcal{T}(\mathcal E) = \sum_{i = 1}^N \sigma_i \, \delta(\mathcal E - \mathcal E_i).
\end{equation}
In the case of unweighted events, the cross sections are simply $\sigma_i = 1/\mathcal{L}$, where $\mathcal{L}$ is the total integrated luminosity.
One can think of these $N$ events as being produced by a perfect event generator or measured by a perfect collider experiment in the context of this theory.
In the $\mathcal{L} \to \infty$ or $N\to\infty$ limit, \Eq{eq:T_def} becomes a smooth distribution.

The $\Sigma$MD is the minimum ``work'' required to rearrange one theory $\mathcal T$ into another $\mathcal T'$ by moving cross section $\mathcal F_{ij}$ from event $i$ in one theory to event $j$ in the other:
\begin{equation}\nonumber
\text{$\Sigma$MD}_{\gamma,S; \beta, R}(\mathcal T, \mathcal T') = \min_{\{\mathcal F_{ij}\ge0\}} \sum_{i=1}^N\sum_{j=1}^{N'}  \mathcal F_{ij} \left(\frac{\text{EMD}_{\beta, R}(\mathcal E_i, \mathcal E_j')}{S} \right)^\gamma + \left|\sum_{i=1}^N \sigma_i - \sum_{j=1}^{N'}\sigma_j'\right|,
\end{equation}
\begin{equation}\label{eq:SigmaMD}
\sum_{i=1}^N \mathcal F_{ij} \le \sigma_j', \quad \sum_{j=1}^{N'} \mathcal F_{ij} \le \sigma_i, \quad \sum_{i=1}^N\sum_{j=1}^{N'} \mathcal F_{ij} = \min\left(\sum_{i=1}^N \sigma_i,\sum_{j=1}^{N'}\sigma_j'\right).
\end{equation}
Here, $i$ and $j$ index the events in theories $\mathcal T$ and $\mathcal T'$, respectively. 
The parameter $S$, which has the same units as the EMD, controls the relative importance of the two terms, analogous to the jet radius parameter $R$ in the EMD.
We have also introduced a possible $\gamma$ exponent, analogous to the $\beta$ angular exponent in the EMD.

The $\Sigma$MD has dimensions of cross section, where the first term quantifies the difference in event distributions and the second term accounts for the creation or destruction of cross section.
For $\gamma > 0$, it is a true metric as long as the underlying EMD is a metric and $2S$ is larger than the largest attainable EMD between two events.%
\footnote{The analogous discussion to footnote \ref{footnote:pWasser} holds for $\gamma > 1$.}
In the limit as $S\to\infty$, the $\Sigma$MD reduces simply to the difference in total cross section between the two theories.
The natural continuum notion of \Eq{eq:SigmaMD} can be used whenever such an analysis is analytically tractable.

The $\Sigma$MD from a theory to itself is zero in the continuum limit or with infinite data.
Further, two theories that differ in their Lagrangians yet give rise to identical scattering cross sections for all energy flows will have a $\Sigma$MD of zero.
This includes, for instance, theories that are equivalent up to field redefinitions~\cite{Cheung:2017pzi} or rearrangements of the asymptotic states~\cite{Frye:2018xjj,Hannesdottir:2019rqq,Hannesdottir:2019opa}.
Finally, if two theories have any observable differences in energy flow, then the $\Sigma$MD between them will be non-zero.
Note that the $\Sigma$MD inherits the flavor and charge insensitivity of the EMD, but it is interesting to consider extending the $\Sigma$MD to account for additional quantum numbers that particles may carry.

\subsection{Jet quenching via quantile matching}

Quantile matching~\cite{Brewer:2018dfs} is an analysis strategy to study the modification of jets as they traverse the quark-gluon plasma in heavy-ion collisions.
We now show that, surprisingly, this technique can be cast naturally in a $\Sigma$MD formulation.
The optimal ``theory moving'' transport between two otherwise-equivalent datasets provides a proxy for the jet modification by the quark-gluon plasma.

Intuitively, the idea is to select a set of statistically equivalent jets from both proton-proton (pp) collisions and heavy-ion (AA) collisions.
This gives a snapshot of jets and their energies before and after modification by the quark-gluon plasma, respectively.
Such a selection can be achieved by selecting jets with the same upper cumulative effective cross section, after appropriately normalizing the AA cross section to account for the average number of nucleon-nucleon collisions.

With such a selection, a quantile matching can be used to specify $p_T^\text{quant}$ for a given heavy-ion jet with reconstructed transverse momentum $p_T^\text{AA}$:
\begin{equation}
\Sigma^\text{eff}_\text{pp}(p_T^\text{quant}) \equiv \Sigma_\text{AA}^\text{eff}(p_T^\text{AA}),
\end{equation}
where $p_T^\text{quant}$ gives a proxy for the jet $p_T$ prior to modification by the quark-gluon plasma.
The ratio between the heavy-ion and proton-proton jet transverse momentum in the same quantile then gives a physically-motivated quantification of the medium jet modification.
\begin{equation}
Q_\text{AA}(p_T^\text{quant}) = \frac{p_T^\text{AA}}{p_T^\text{quant}}.
\end{equation}

We now turn to explaining the intriguing connection between this quantile matching procedure and the $\Sigma$MD through optimal transport.
We use transverse momenta in place of energies and take $R\to\infty$ in \Eq{eq:emd}, where the EMD becomes simply the difference in jet transverse momenta.
Further, we set $\gamma = 1$ and $S=1$ in \Eq{eq:SigmaMD} and note that the normalization of the cross sections makes the second term in that equation vanish.

The theory moving problem now becomes a simple one-dimensional optimal transport problem of moving the pp jet $p_T$ distribution to the AA jet $p_T$ distribution.
Remarkably, this is mathematically equivalent to quantile matching.
We use the notation TM to represent the optimal theory movement $\mathcal F^*$ in the $\Sigma$MD.
Letting $\mathcal T_\text{AA}$ be the set of heavy-ion jets and $\mathcal T_\text{pp}$ be the set of proton-proton jets, we have that:
\begin{equation}
\begin{boxed}{
p_T^\text{quant} = \text{TM}\left(\mathcal T_\text{AA}, \mathcal T_\text{pp}\right)[p_{T}^\text{AA}],
}\end{boxed}
\end{equation}
where we can define this formally using a ``ghost'' heavy-ion jet with transverse momentum $p_T^\text{AA}$ and infinitesimal cross section $\sigma \sim 0$.

Quantile matching can therefore be seen as a matching induced by the optimal theory movement between the heavy-ion and proton-proton jets.
In this sense, it operationally defines the modification by the quark-gluon plasma in terms of the theory-movement of the jet transverse momentum spectrum.
It would be interesting to follow this connection further and explore this procedure using the full EMD beyond the $R\to\infty$ limit to study the medium modification as a function of the jet substructure.

\subsection{Event clustering and coresets}

One of the essential unsupervised methods for probing a dataset is to analyze its complexity.
A method to do this for collider physics datasets is that of $k$-eventiness, recently introduced in~\Ref{Komiske:2019jim}.
Here, one seeks to find $k$ representative events that minimize the EMD from each event in the dataset to the nearest representative event:
\begin{equation}
\mathcal V^{(\gamma)}_k = \min_{\mathcal K_1, \ldots,\mathcal K_k}\sum_{i=1}^N \sigma_i \min\{\text{EMD}(\mathcal E_i, \mathcal K_1)^\gamma, \text{EMD}(\mathcal E_i, \mathcal K_2)^\gamma, \ldots, \text{EMD}(\mathcal E_i, \mathcal K_k)^\gamma\},
\end{equation}
where we have dropped the $\beta$ and $R$ subscripts on \text{EMD} for compactness.
The value of $\mathcal V_k$ probes how well the dataset is approximated by the $k$ events.
This gives rise to the notion of $\mathcal V_k$ as the ``$k$-eventiness'' of the dataset, in analogy with $N$-(sub)jettiness, where smaller values of $\mathcal V_k$ indicate better approximations.

From a geometric perspective, $\mathcal V_k$ is the smallest $\Sigma$MD to the manifold of $k$-event datasets.
Analogous to \Eq{eq:emd_noR}, we can introduce the $S \to \infty$ version of the $\Sigma$MD:
\begin{equation}
\label{eq:smd_noR}
\text{$\Sigma$MD}_{\gamma} (\mathcal T, \mathcal T') = \lim_{S \to \infty} S^\gamma \, \text{$\Sigma$MD}_{\gamma,S} (\mathcal T, \mathcal T'),
\end{equation}
which yields an infinite distances between theories of different total cross sections.
Following the identical logic to \Secs{sec:njettiness}{subsec:nsubjettiness}, we have:
\begin{equation}
\label{eq:keventiness_EMD}
\begin{boxed}{
\mathcal V^{(\gamma)}_k = \min_{|\mathcal T'|=k} \text{$\Sigma$MD}_\gamma(\mathcal T, \mathcal T').
}\end{boxed}
\end{equation}
Here, we use the $|\cdot|$ notation to count the number of events in $\mathcal T'$.
Just like for $N$-(sub)jettiness, different values of $\gamma$ highlight different aspects of theory space geometry.

Following the logic in \Sec{subsec:xcone} of lifting the $N$-jettiness observable into the XCone jet algorithm, we can lift $k$-eventiness into an event clustering algorithm.
The representative events $\mathcal K$ (i.e.\ the point of closest approach on the $k$-event manifold), has the interpretation of the ``$k$-geometric-medians'' for $\gamma = 1$ or ``$k$-means'' for $\gamma = 2$.
For practical applications, it is often convenient to restrict the representative events to be within the dataset $\mathcal T$, i.e. the ``$k$-medoids'', giving only an approximate value of $\mathcal V_k$.
While the full problem of finding the representative jets may be computationally intractable, fast approximations to find the medoids exist and have been explored in \Ref{Komiske:2019jim}.

Inspired by \Sec{sec:seqrec}, one might consider implementing sequential clustering algorithms by iterating $\Sigma$MD computations to approximate $M$ events with $M-1$ events and so forth.
Such a clustering may be helpful for rigorous data compression of large collider datasets or, if implemented efficiently, for tasks such as triggering.
These ideas are closely related to the notion of finding a coreset (see \Ref{2019arXiv191008707J} for a recent review), for which techniques from quantum information and quantum computation may also find use~\cite{harrow2020small}.
Additionally, \Ref{DBLP:journals/corr/abs-1805-07412} uses the Wasserstein metric to construct ``measure coresets'' that take into account the underlying data distribution and which may prove useful for high-energy physics applications.
We leave further exploration of theory geometry and theory space algorithms to future work.

\section{Conclusions}
\label{sec:conc}

In this chapter, we have explored the metric space of collider events from a theoretical perspective.
Beginning from the EMD between final states, namely the ``work'' required to rearrange one into another, we have cast a multitude of diverse collider algorithms and analysis techniques in a geometric language.
First, we connected this metric to the fundamental notion of IRC safety in massless quantum field theories, with the EMD providing a sharp language to define IRC safety and even Sudakov safety.
We extended this connection by highlighting that a wide variety of collider observables, including thrust and $N$-jettiness, can be cast as distances between events and manifolds in this space.
Further, we demonstrated that many jet clustering algorithms, such as exclusive cone finding and sequential clustering, can be exactly derived in full detail from the simple principle that jets are the best $N$-particle approximation to the event.
Even pileup mitigation techniques developed to face the LHC-era challenge of high luminosity running can be cast in the language of subtracting a uniform radiation pattern, connecting this field to semidiscrete unbalanced optimal transport.
Finally, we generalized our reasoning to define a distance between ``theories'' as sets of events with cross sections, proving a new lens to understand several existing techniques and a roadmap for future developments in the geometry of theory space.

From the perspective of massless quantum field theories, our metric space of events is the natural space for understanding observables, as the only truly observable quantities are IRC safe.
More speculatively, it would be interesting to circumvent the (unphysical) particle-level stage of calculations and make theoretical predictions directly in the space of events.
Understanding and expanding in this direction would require natural notions of volume and integration in this space, perhaps aided by recent developments in Wasserstein spaces~\cite{DBLP:journals/siamma/BianchiniB10,DBLP:journals/siamma/AguehC11,DBLP:conf/gsi/BertrandK13,DBLP:conf/gsi/GouicL15}, though we leave this fascinating exploration to future work.
Nonetheless, it has already been established that the energy flows themselves obey factorization theorems in effective field theory contexts~\cite{Bauer:2008jx}, and give rise to rich behavior in correlators~\cite{Basham:1978zq,Belitsky:2013ofa,Dixon:2019uzg,Chen:2019bpb}.
Going directly from first principles and symmetries to observables (i.e.\ energy flows, not particles) suggests a natural extension to the philosophy driving the present scattering amplitudes program (see \Refs{Elvang:2013cua,Carrasco:2015iwa,Cheung:2017pzi} for reviews).
It is also interesting to extend this logic to massive quantum field theories where observable quantities can depend on flavor and charge.
One promising strategy is to treat events as collections of objects (jets, electrons, muons, etc.) and to use a ground distance that penalizes converting one type of object into another~\cite{Romao:2020ojy}.
Alternatively, it may be possible to find an EMD variant that mimics the flavor-sensitive clustering behavior of \Ref{Banfi:2006hf}.

It is also useful to discuss these developments in the broader context of machine learning and the physical sciences~\cite{Larkoski:2017jix,Radovic:NatureML,Carleo:2019ptp}.
Typically, problems in the natural sciences can be cast as machine learning problems such as classification and regression, whereby the relevant tools from machine learning can be applied to achieve improved performance on those tasks.
It is far rarer for machine learning to enhance our theoretical or conceptual understanding of physics directly.
This story provides an interesting case where new insights and questions exposed by machine learning have impacted purely theoretical and phenomenological collider physics.
The question of when two collider events are similar, for which the EMD was introduced, was originally motivated by unsupervised learning methods and autoencoders~\cite{Hajer:2018kqm,Heimel:2018mkt,Farina:2018fyg,Roy:2019jae}, which require a distance matrix or reconstruction loss.
By providing an answer to this simple question, which itself involved familiar machine learning tools such as optimal transport, we uncovered a new mathematical formalism to better understand and express concepts in quantum field theory and collider physics.
We hope that this will be just one example of many future profound insights into the natural sciences facilitated by this perspective.

\chapter{A Basis of Infrared- and Collinear-Safe Observables}

\section{Introduction}
\label{sec:intro}

Jet substructure is the analysis of radiation patterns and particle distributions within the collimated sprays of particles (jets) emerging from high-energy collisions~\cite{Seymour:1991cb,Seymour:1993mx,Butterworth:2002tt,Butterworth:2007ke,Butterworth:2008iy}.
Jet substructure is central to many analyses at the Large Hadron Collider (LHC), finding applications in both Standard Model measurements~\cite{Chatrchyan:2012sn,CMS:2013cda,Aad:2015lxa,Aad:2015cua,TheATLAScollaboration:2015ynv,Aad:2015hna,Aad:2015rpa,ATLAS:2016dpc,ATLAS:2016wlr,ATLAS:2016wzt,CMS:2016rtp,Rauco:2017xzb,ATLAS:2017jiz,Sirunyan:2017dgc} and in searches for physics beyond the Standard Model~\cite{CMS:2011bqa,Chatrchyan:2012ku,Fleischmann:2013woa,Pilot:2013bla,TheATLAScollaboration:2013qia,CMS:2014afa,CMS:2014aka,Khachatryan:2015axa,TheATLAScollaboration:2015elo,TheATLAScollaboration:2015xqi,Khachatryan:2015bma,CMS:2016qwm,CMS:2016bja,CMS:2016ehh,CMS:2016rfr,Aaboud:2016okv,Khachatryan:2016mdm,CMS:2016flr,CMS:2016pod,Aaboud:2016qgg,Sirunyan:2016wqt,Sirunyan:2017usq,Sirunyan:2017acf,Sirunyan:2017nvi}. 
An enormous catalog of jet substructure observables has been developed to tackle specific collider physics tasks~\cite{Abdesselam:2010pt,Altheimer:2012mn,Altheimer:2013yza,Adams:2015hiv,Larkoski:2017jix}, such as the identification of boosted heavy particles or the classification of quark- from gluon-initiated jets.

The space of possible jet substructure observables is formidable, with few known complete and systematic organizations.
Previous efforts to define classes of observables around organizing principles include:
the jet energy moments and related Zernike polynomials to classify energy flow observables~\cite{GurAri:2011vx};
a pixelated jet image~\cite{Cogan:2014oua} to represent energy deposits in a calorimeter;
the energy correlation functions (ECFs)~\cite{Larkoski:2013eya} to highlight the $N$-prong substructure of jets;
the generalized energy correlation functions (ECFGs)~\cite{Moult:2016cvt} based around soft-collinear power counting~\cite{Larkoski:2014gra};
and a set of $N$-subjettiness observables~\cite{Thaler:2010tr,Thaler:2011gf,Stewart:2010tn} to capture $N$-body phase space information~\cite{Datta:2017rhs}.
With any of these representations, there is no simple method to combine individual observables, so one typically uses sophisticated multivariate techniques such as neural networks to fully access the information contained in several observables~\cite{Gallicchio:2010dq,Gallicchio:2011xq,Gallicchio:2012ez,Baldi:2014kfa,Baldi:2014pta,deOliveira:2015xxd,Barnard:2016qma,Komiske:2016rsd,Kasieczka:2017nvn,Almeida:2015jua,Baldi:2016fql,Guest:2016iqz,Louppe:2017ipp,Pearkes:2017hku,Butter:2017cot,Aguilar-Saavedra:2017rzt,Datta:2017rhs}.
Furthermore, the sense in which these sets ``span'' the space of jet substructure is often unclear, sometimes relying on the existence of complicated nonlinear functions to map observables to kinematic phase space.

In this chapter, we introduce a powerful set of jet substructure observables organized directly around the principle of infrared and collinear (IRC) safety.
These observables are multiparticle energy correlators with specific angular structures which directly result from IRC safety.
Since they trace their lineage to the hadronic energy flow analysis of \Ref{Tkachov:1995kk}, we call these observables the \emph{energy flow polynomials} (\Bs) and we refer to the set of \Bs as the \emph{energy flow basis}.
In the language of \Ref{Tkachov:1995kk}, the \Bs can be viewed as a discrete set of $C$-correlators, though our analysis is independent from the original $C$-correlator logic.
Crucially, the EFPs form a \emph{linear} basis of all IRC-safe observables, making them suitable for a wide variety of jet substructure contexts where linear methods are applicable.

There is a one-to-one correspondence between \Bs and loopless multigraphs, which helps to visualize and calculate the \Bs.
A multigraph is a graph where any two vertices can be connected by multiple edges; in this context, a \emph{loop} is an edge from a vertex to itself, while a closed chain of edges is instead referred to as a \emph{cycle}.
For a multigraph $G$ with $N$ vertices and edges $(k,\ell)\in G$, the corresponding \B takes the form:
\begin{equation}\label{eq:introefp}
\B_{G} = \sum_{i_1=1}^M\cdots\sum_{i_N=1}^M z_{i_1}\cdots z_{i_N}\prod_{(k,\ell) \in G} \theta_{i_k i_\ell},
\end{equation}
where the jet consists of $M$ particles, $z_i \equiv E_i/\sum_{j=1}^M E_j$ is the energy fraction carried by particle $i$, and $\theta_{ij}$ is the angular distance between particles $i$ and $j$.
The precise definitions of $E_i$ and $\theta_{ij}$ will depend on the collider context, with energy and spherical ($\theta,\phi$) coordinates typically used for $e^+e^-$ collisions, and transverse momentum $p_T$ and rapidity-azimuth $(y,\phi)$ coordinates for hadronic collisions.
For brevity, we often use the multigraph $G$ to represent the formula for $\B_G$ in \Eq{eq:introefp}, e.g.:
\begin{equation}\label{eq:wedgegraphex}
\begin{gathered}
\includegraphics[scale=.3]{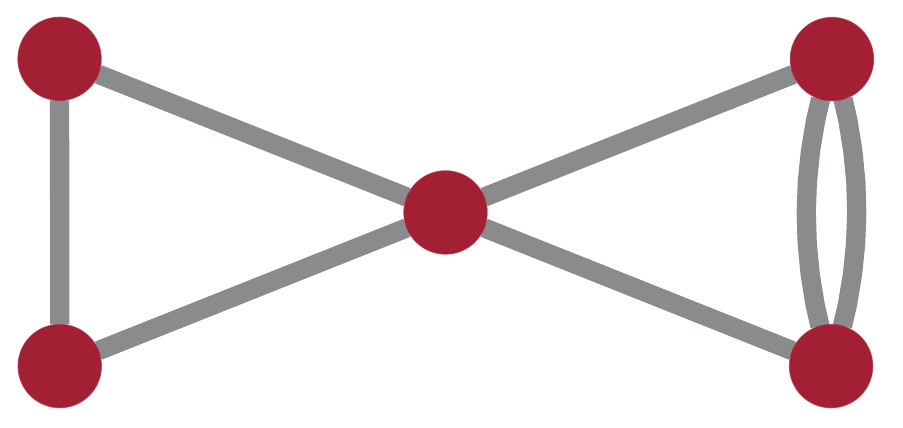}
\end{gathered}
= \sum_{i_1=1}^M\sum_{i_2 = 1}^M \sum_{i_3 = 1}^M \sum_{i_4=1}^M\sum_{i_5=1}^M  z_{i_1} z_{i_2} z_{i_3}z_{i_4}z_{i_5}\theta_{i_1i_2} \theta_{i_2i_3}\theta_{i_1i_3}\theta_{i_1i_4}\theta_{i_1i_5}\theta_{i_4i_5}^2.
\end{equation}

This chapter is a self-contained introduction to the energy flow basis, with the following organization.
\Sec{sec:efps} contains a general overview of the \Bs, with more detailed descriptions of \Eq{eq:introefp} and the correspondence to multigraphs.
We also discuss a few different choices of measure for $z_i$ and $\theta_{ij}$.
As already mentioned, EFPs are a special case of $C$-correlators~\cite{Tkachov:1995kk}, so not surprisingly, we find a close relationship between \Bs and other classes of observables that are themselves $C$-correlators, including jet mass, ECFs~\cite{Larkoski:2013eya}, certain generalized angularities~\cite{Larkoski:2014pca}, and energy distribution moments~\cite{GurAri:2011vx}.
We also highlight features of the \Bs which are less well-known in the $C$-correlator-based literature.
 
In \Sec{sec:basis}, we give a detailed derivation of the \Bs as an (over)complete linear basis of all IRC-safe observables in the case of massless particles.
Because this section is rather technical, it can be omitted on a first reading, though the logic just amounts to systematically imposing the constraints of IRC safety.
In \Sec{sec:Eexpansion}, we use an independent (and arguably more transparent) logic from \Ref{Tkachov:1995kk} to show that any IRC-safe observable can be written as a linear combination of $C$-correlators:
\begin{equation}\label{eq:genccorr}
\mathcal C_N^{f_N}=\sum_{i_1=1}^M \cdots\sum_{i_N=1}^M E_{i_1}\cdots E_{i_N} \, f_N(\hat p_{i_1}^\mu,\ldots,\hat p_{i_N}^\mu),
\end{equation}
where $f_N$ is an angular weighting function that is only a function of the particle directions $\hat{p}_i^\mu=p_i^\mu/E_i$ (and not their energies $E_i$).
To derive \Eq{eq:genccorr}, we use the Stone-Weierstrass theorem~\cite{stone1948generalized} to expand an arbitrary IRC-safe observable in polynomials of particle energies, and then directly impose IRC safety and particle relabeling invariance.
In \Sec{sec:angleexpansion}, we determine the angular structures of the \Bs by expanding $f_N$ in terms of a discrete set of polynomials in pairwise angular distances.
Remarkably, the discrete set of polynomials appearing in this expansion is in one-to-one correspondence with the set of non-isomorphic multigraphs, which facilitates indexing the \Bs by multigraphs to encode the geometric structure in \Eq{eq:introefp}.

In \Sec{sec:complexity}, we investigate the complexity of computing \Bs.
Naively, \Eq{eq:introefp} has complexity $\mathcal O(M^N)$ due to the $N$ nested sums over $M$ particles.
However, the rich analytic structure of \Eq{eq:introefp} and the graph representations of \Bs allow for numerous algorithmic speedups.
Any \B with a disconnected graph can be computed as the product of the \Bs corresponding to its connected components.
Furthermore, we find that the Variable Elimination (VE) algorithm~\cite{zhang1996exploiting} can be used to vastly speed up the computation of many \Bs compared to the naive $\mathcal O(M^N)$ algorithm.
VE uses the factorability of the summand to systematically determine a more efficient order for performing nested sums.
For instance, all tree graphs can be computed in $\mathcal O(M^2)$ using VE. 
As an explicit example, consider an EFP with naive $\mathcal O(M^6)$ scaling:
\begin{equation}\label{eq:extree}
\begin{gathered}
\includegraphics[scale=.25]{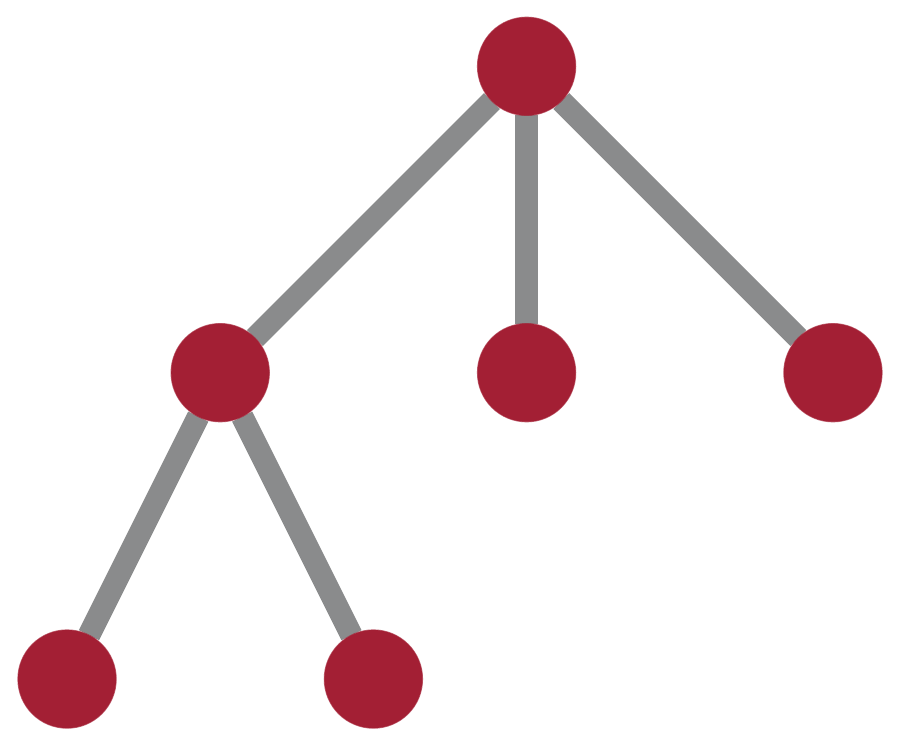}
\end{gathered}
\,=\sum_{i_1=1}^M\sum_{i_2=1}^Mz_{i_1}z_{i_2}\theta_{i_1i_2}\left(\sum_{i_3=1}^Mz_{i_3}\theta_{i_1i_3}\right)^2\left(\sum_{i_4=1}^Mz_{i_4}\theta_{i_2i_4}\right)^2.
\end{equation}
The quantities in parentheses are computable in $\mathcal O(M^2)$, since they are length $M$ lists with each element a sum over $M$ objects, making the overall expression in \Eq{eq:extree} computable in $\mathcal O(M^2)$.
The efficient computation of the \Bs overcomes one of the main previous challenges in using higher-$N$ multiparticle correlators in collider physics applications.\footnote{Sadly, fully-connected graphs, which correspond to the original ECFs~\cite{Larkoski:2013eya}, cannot be simplified using VE.} 

In \Sec{sec:linreg}, we perform numerical linear regression with \Bs for various jet observables.
The linear spanning nature of the energy flow basis means that any IRC-safe observable $\mathcal S$ can be linearly approximated by \Bs, which we write as:
\begin{equation}\label{eq:introefpsspan}
\mathcal S\simeq\sum_{G\in\mathcal G}s_G \, \B_G,
\end{equation}
for some finite set of multigraphs $\mathcal G$ and some real coefficients $s_G$.  
One might worry that the number of \Bs needed to achieve convergence could be intractably large.
In practice, though, we find that the required set of $\mathcal G$ needed for convergence is rather reasonable in a variety of jet contexts.
While we find excellent convergence for IRC-safe observables, regressing with IRC-unsafe observables does not work as well, demonstrating the importance of IRC safety for the energy flow basis.

In \Sec{sec:linclass}, we perform another test of \Eq{eq:introefpsspan} by using linear classification with \Bs to distinguish signal from background jets.
We consider three representative jet tagging problems: quark/gluon classification, boosted $W$ tagging, and boosted top tagging.
In this study, the observable appearing on the left-hand side of \Eq{eq:introefpsspan} is the optimal IRC-safe discriminant for the two classes of jets.
Remarkably, linear classification with \Bs performs comparably to multivariate machine learning techniques, such as jet images with convolutional neural networks (CNNs)~\cite{Cogan:2014oua,deOliveira:2015xxd,Barnard:2016qma,Komiske:2016rsd,Kasieczka:2017nvn} or dense neural networks (DNNs) with a complete set of $N$-subjettiness observables~\cite{Datta:2017rhs}.
Both the linear regression and classification models have few or no hyperparameters, illustrating the power and simplicity of linear learning methods combined with our fully general linear basis for IRC-safe jet substructure.

Our conclusions are presented in \Sec{sec:conclusion}, where we highlight the relevance of the energy flow basis to machine learning and discuss potential future applications and developments.
A review of $C$-correlators and additional tagging plots are left to the appendices.

\section{Energy flow polynomials}
\label{sec:efps}

IRC-safe observables have long been of theoretical and experimental interest because observables which lack IRC safety are not well defined~\cite{Kinoshita:1962ur,Lee:1964is,Weinberg:1995mt,sterman1995handbook}, or require additional care to calculate~\cite{Larkoski:2013paa,Larkoski:2015lea,Waalewijn:2012sv,Chang:2013rca,Elder:2017bkd}, in perturbative quantum chromodynamics (pQCD).
More broadly, though, IRC safety is a simple and natural organizing principle for high-energy physics observables, since IRC-safe observables probe the high-energy structure of an event while being insensitive to low-energy and collinear modifications.
IRC safety is also an important property experimentally as IRC-safe observables are more robust to noise and finite detector granularity.

As argued in \Refs{Tkachov:1995kk,Sveshnikov:1995vi,Cherzor:1997ak,Tkachov:1999py}, the $C$-correlators in \Eq{eq:genccorr} are a generic way to capture the IRC-safe structure of a jet, as long as one chooses an appropriate angular weighting function $f_N$.
Later in \Sec{sec:basis}, we give an alternative proof that $C$-correlators span the space of IRC-safe observables and go on to give a systematic expansion for $f_N$.
This expansion results in the \Bs, which yield an (over)complete linear basis for IRC-safe observables.
In this section, we highlight the basic features of the \Bs and their relationship to previous jet substructure observables.

\subsection{The energy flow basis}
\label{sec:efbasis}

One can think of the \Bs as $C$-correlators that make specific, discrete choices for the angular weighting function $f_N$ in \Eq{eq:genccorr}.  
True to their name, \Bs have angular weighting functions that are polynomial in pairwise angular distances $\theta_{ij}$.
The energy flow basis is therefore all $C$-correlators with angular structures that are unique monomials in $\theta_{ij}$, meaning monomials that give algebraically different expressions once the sums in \Eq{eq:genccorr} are performed.
Since we intend to apply the energy flow basis for jet substructure, we remove the dependence on the overall jet kinematics by normalizing the particle energies by the total jet energy, $E_J\equiv\sum_{i=1}^M E_i$, leading to the \Bs written in terms of the energy fractions $z_i\equiv E_i/E_{J}$ as in \Eq{eq:introefp}. 

The uniqueness requirement on angular monomials can be better understood by developing a correspondence between monomials in $\theta_{ij}$ and multigraphs:
\begin{EFPMultiCorre}
The set of loopless multigraphs on $N$ vertices corresponds exactly to the set of angular monomials in $\{\theta_{i_k i_\ell}\}_{k<\ell\in\{1,\cdots,N\}}$. Each edge $(k,\ell)$ in a multigraph is in one-to-one correspondence with a term $\theta_{i_ki_\ell}$ in an angular monomial; each vertex $j$ in the multigraph corresponds to a factor of $z_{i_j}$ and summation over $i_j$ in the \B:
\begin{align}\label{eq:correspondence}
&\begin{gathered}
\includegraphics[scale=0.04]{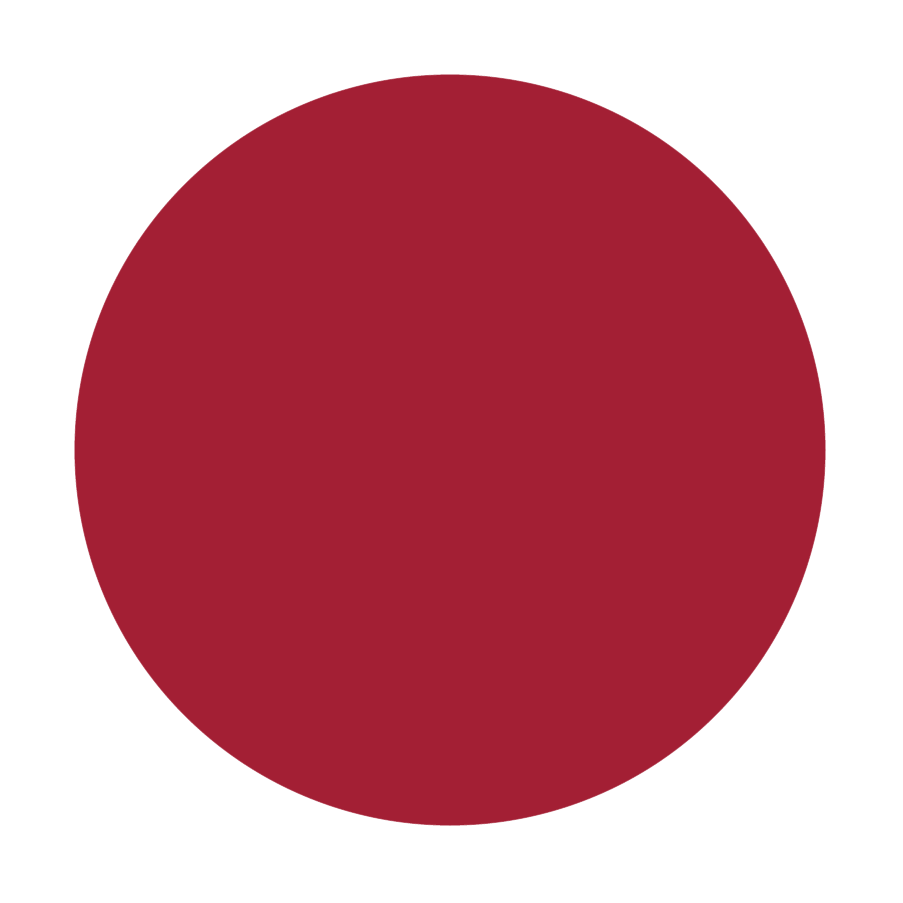}
\end{gathered}_j
 \,\,\Longleftrightarrow\,\, \sum_{i_j=1}^Mz_{i_j},
& k \begin{gathered}
\includegraphics[scale=0.3]{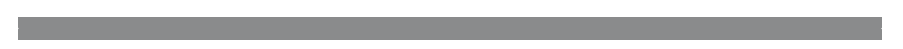}
\end{gathered} \,\ell
 \,\,\Longleftrightarrow\,\, \theta_{i_ki_\ell}.
\end{align}
\end{EFPMultiCorre}

Using \Eq{eq:correspondence}, the \Bs can be directly encoded by their corresponding multigraphs. For instance:
\begin{equation}\label{eq:flyswatter}
\begin{gathered}
\includegraphics[scale=0.32]{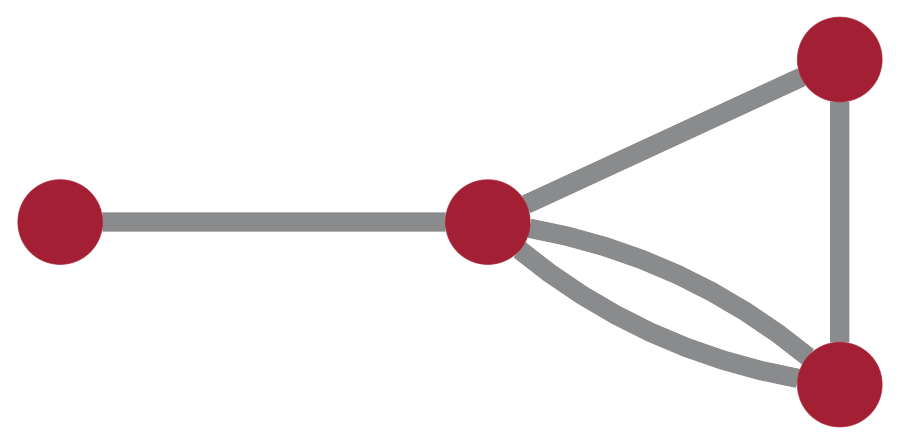}
\end{gathered}
= \sum_{i_1=1}^M\sum_{i_2 = 1}^M \sum_{i_3 = 1}^M \sum_{i_4=1}^M  z_{i_1} z_{i_2} z_{i_3}z_{i_4}\theta_{i_1i_2}\theta_{i_2i_3}\theta_{i_2i_4}^2\theta_{i_3i_4}.
\end{equation}
Since any labeling of the vertices gives an equivalent algebraic expression, we represent the graphs as unlabeled.
The specification that the \Bs are unique monomials translates into the requirement that the corresponding multigraphs are non-isomorphic.
Versions of these multigraphs have previously appeared in the physics literature in the context of many-body configurations~\cite{dallot1992reduced, mandal2017determination}, encoding all local scalar operators of a free theory~\cite{Hogervorst:2014rta}, and in graphically depicting ECFs for jets~\cite{Moult:2016cvt,Larkoski:2017iuy}.

\begin{table}[t]
\centering
\begin{tabular}{|rrcl|}
\hline
& \textbf{Multigraph}  & & \textbf{Energy Flow Polynomial}  \\ \hline \hline
$N$ : & Number of vertices & $\Longleftrightarrow$ &  $N$-particle correlator \\ 
$d$ : & Number of edges &$\Longleftrightarrow$ & Degree of angular monomial \\ 
$\chi$ : & Treewidth $+\,1$  &$\Longleftrightarrow$ & Optimal VE complexity $\mathcal O(M^\chi)$  \\
& Chromatic number &$\Longleftrightarrow$ & Minimum number of prongs to not vanish\\
\hline
& Connected& $\Longleftrightarrow$ &Prime  \\
& Disconnected& $\Longleftrightarrow$ &Composite  \\
\hline
\end{tabular}
\caption{Corresponding properties of multigraphs and \Bs.}
\label{tab:correspondence}
\end{table}

\Tab{tab:correspondence} contains a summary of the correspondence between the properties of \Bs and multigraphs.
The number of graph vertices $N$ corresponds to the number of particle sums in the EFP, and the number of graph edges $d$ corresponds to the \emph{degree} of the \B (i.e.\ the degree of the underlying angular monomial).
The number of separated prongs for which an individual \B is first non-vanishing is the \emph{chromatic number} of the graph: the smallest number of colors needed to color the vertices of the graph with no two adjacent vertices sharing a color.
For computational reasons discussed further in \Sec{sec:complexity}, we also care about the treewidth of the graph, which is related to the computational complexity $\chi$ of an EFP.
Also for computational reasons, we make a distinction between connected or \emph{prime} multigraphs and disconnected or \emph{composite} multigraphs; the value of a composite \B is simply the product of the prime \Bs corresponding to its connected components.

\begin{table}[t]
\centering
\subfloat[]{\label{tab:efpcounts:a}
\begin{tabular}{|rc||*{10}{r}|}\hline
\multicolumn{2}{|c}{Maximum degree $d$}&\multicolumn{1}{c}{\bf0}&\multicolumn{1}{c}{\bf1}&\multicolumn{1}{c}{\bf2}&\multicolumn{1}{r}{\bf3}&\multicolumn{1}{c}{\bf4}&\multicolumn{1}{c}{\,\,\bf5}&\multicolumn{1}{c}{\bf6}&\multicolumn{1}{c}{\bf7}&\multicolumn{1}{c}{\bf8}&\multicolumn{1}{c|}{\bf9}\\\hhline{:==:t:*{10}{=}:}
\multirow{2}{*}{{\bf Prime \Bs}}
    & \href{https://oeis.org/A076864}{A076864} & 1 & 1 & 2 & 5 & 12 & 33 & 103 & 333 & 1\,183   & 4\,442    \\ 
    & Cumul.    & 1 & 2 & 4 & 9 & 21 & 54 & 157 & 490 & 1\,673 & 6\,115 \\ \hhline{|--||*{10}{-}}
\multirow{2}{*}{{\bf All \Bs}}
    & \href{https://oeis.org/A050535}{A050535} & 1 & 1 & 3 & 8   & 23 & 66   & 212 & 686     & 2\,389   & 8\,682      \\
    & Cumul.    & 1 & 2 & 5 & 13 & 36 & 102 & 314 & 1\,000 & 3\,389 & 12\,071  \\ \hhline{|--||*{10}{-}}
\end{tabular}}
\\\vspace{.25in}
\subfloat[]{\label{tab:efpcounts:b}
\begin{tabular}{|cc||rrrrrrrrrr|}\hline
\multicolumn{2}{|c}{$d$}&\bf1&\bf2&\bf3&\bf4&\bf5&\bf6&\bf7&\bf8&\bf9&\bf10 \\ \hhline{:==:t:*{10}{=}:}
\multirow{10}{*}{$N$} 
        & \bf2 & 1 & 1 & 1 & 1 & 1   & 1   & 1     & 1     & 1         & 1         \\
        & \bf3 &    & 1 & 2 & 3 & 4   & 6   & 7     & 9     & 11       & 13       \\
        & \bf4 &    &    & 2 & 5 & 11 & 22 & 37   & 61   & 95       & 141     \\
        & \bf5 &    &    &    & 3 & 11 & 34 & 85   & 193 & 396     & 771     \\
        & \bf6 &    &    &    &    & 6   & 29 & 110 & 348 & 969     & 2\,445 \\
        & \bf7 &    &    &    &    &      & 11 & 70   & 339 & 1\,318 & 4\,457 \\
        & \bf8 &    &    &    &    &      &      & 23   & 185 & 1\,067 & 4\,940 \\
        & \bf9 &    &    &    &    &      &      &        & 47   & 479     & 3\,294 \\
        & \bf10 &  &    &    &    &      &      &        &        & 106     & 1\,279 \\
        & \bf11 &  &    &     &    &     &       &       &        &            & 235     \\ \hhline{|--||*{10}{-}|} 
\end{tabular}}
\caption{(a) The number of \Bs (prime and all) organized by degree $d$, for $d$ up to 9. The cumulative rows tally the number of \Bs with degree at most $d$, i.e.\ the number of basis elements truncated at that $d$. While these sequences grow quickly, the total number of all basis elements is at most 1000 for $d\le7$, which is computationally tractable. (b) The number of prime \Bs broken down by number of vertices $N$ and number of edges $d$ in the multigraph, for $d$ up to 10.  All connected graphs (prime \Bs) for $d$ up to 5 are shown explicitly in \Tab{tab:graphs}.}
\label{tab:efpcounts}
\end{table}

Because the \B basis is infinite, a suitable organization and truncation scheme is necessary to use the basis in practice.
In this chapter, we usually truncate by restricting to the set of all multigraphs with at most $d$ edges.
This is a natural choice because it corresponds to truncating the approximation of the angular function $f_N$ at degree $d$ polynomials.
Furthermore, this truncation results in a \emph{finite} number of \Bs at each order of truncation, which is not true for truncation by the number of vertices.
The number of multigraphs with exactly $d$ edges is Sequence A050535 in the On-Line Encyclopedia of Integer Sequences (OEIS)~\cite{sloane2007line,harary2014graphical}; the number of connected multigraphs with exactly $d$ edges is Sequence A076864 in the OEIS~\cite{sloane2007line}.
The numbers of \Bs in our truncation of the energy flow basis are the partial sums of these sequences, which are listed in \Tab{tab:efpcounts:a} up to $d=9$.
\Tab{tab:efpcounts:b} tabulates the number of prime \Bs of degree $d$ binned by $N$ up to $d=10$.
\Tab{tab:graphs} illustrates all connected multigraphs with $d\le 5$ edges.

\begin{table}[t]
\centering
\begin{tabular}{| >{\centering}m{.5in} | >{\centering}m{5in} |}\hline
\textbf{Degree} & \textbf{Connected Multigraphs} \tabularnewline\hline \hline
$d=0$ & \includegraphics[scale=0.02]{graphs/dot}\tabularnewline\hline
$d=1$ & \vspace{.08in}\includegraphics[scale=0.15]{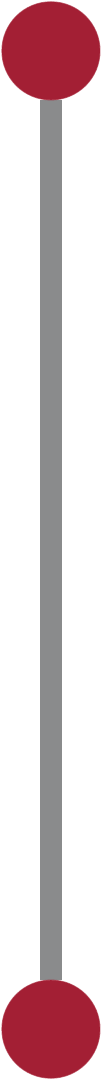}
\tabularnewline\hline
$d=2$ & \vspace{.08in}\includegraphics[scale=0.15]{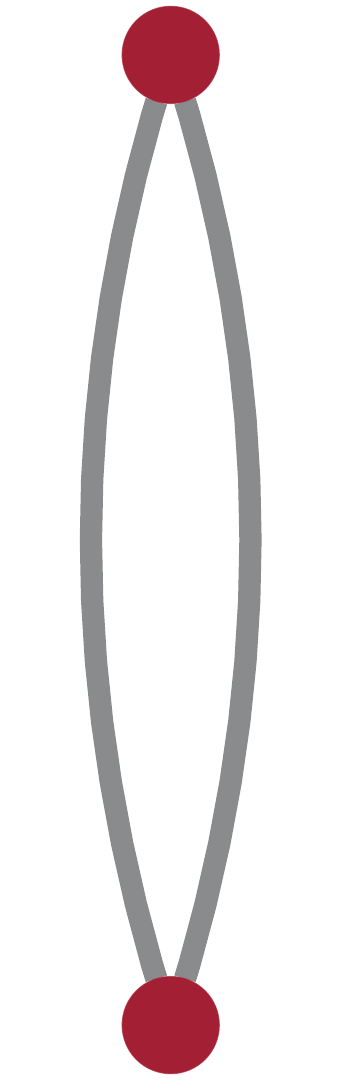}
                                     \includegraphics[scale=0.15]{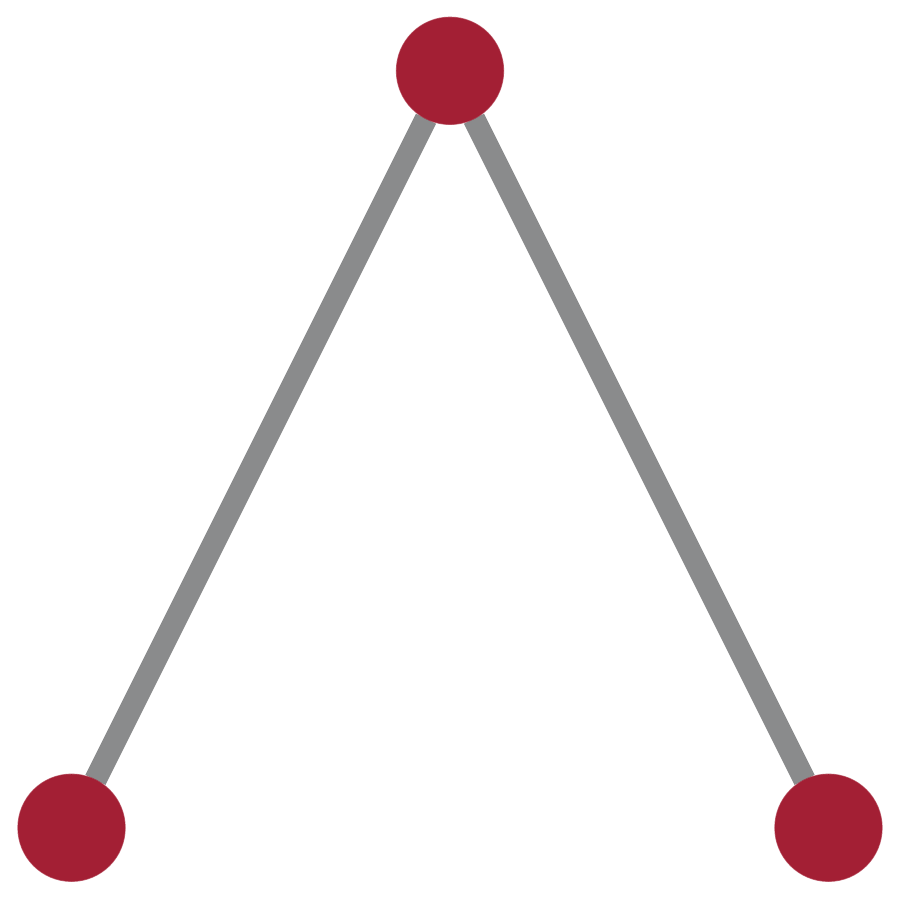} 
\tabularnewline\hline
$d=3$ & \vspace{.08in}\includegraphics[scale=0.15]{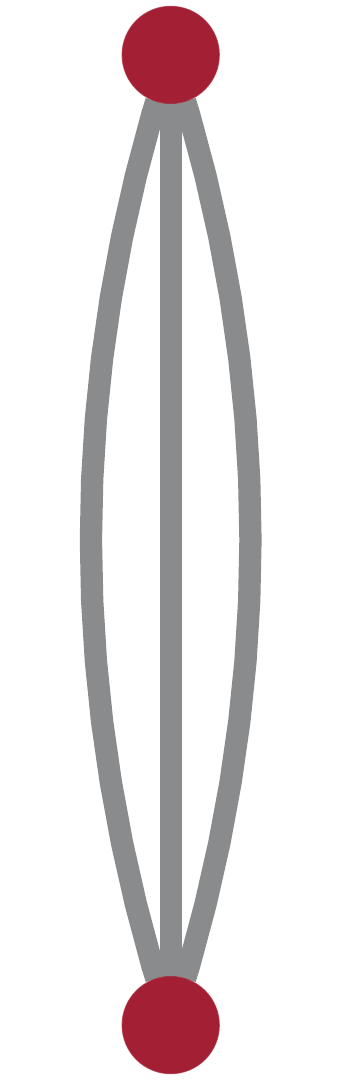}
                                     \includegraphics[scale=0.15]{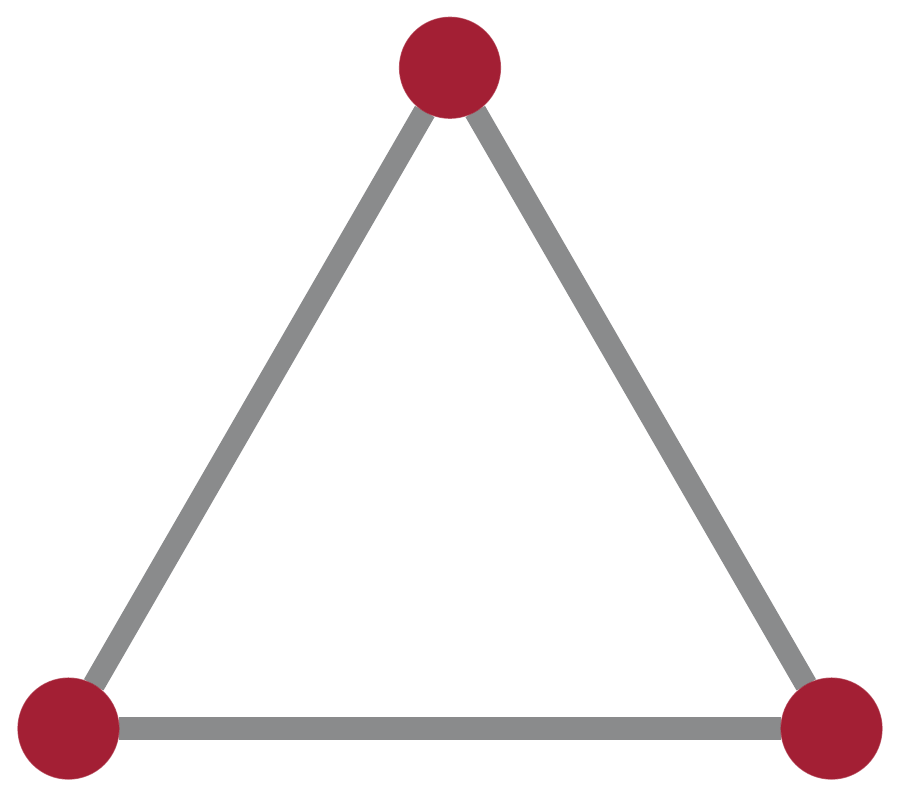}
                                     \includegraphics[scale=0.15]{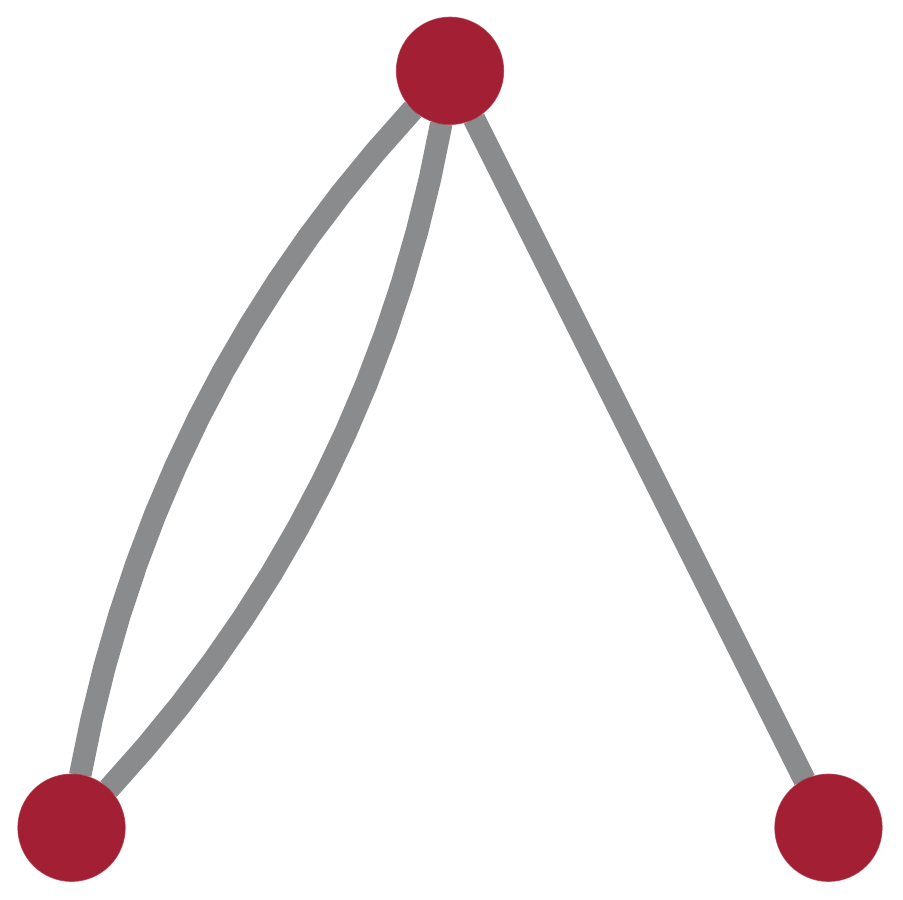}
                                     \includegraphics[scale=0.2]{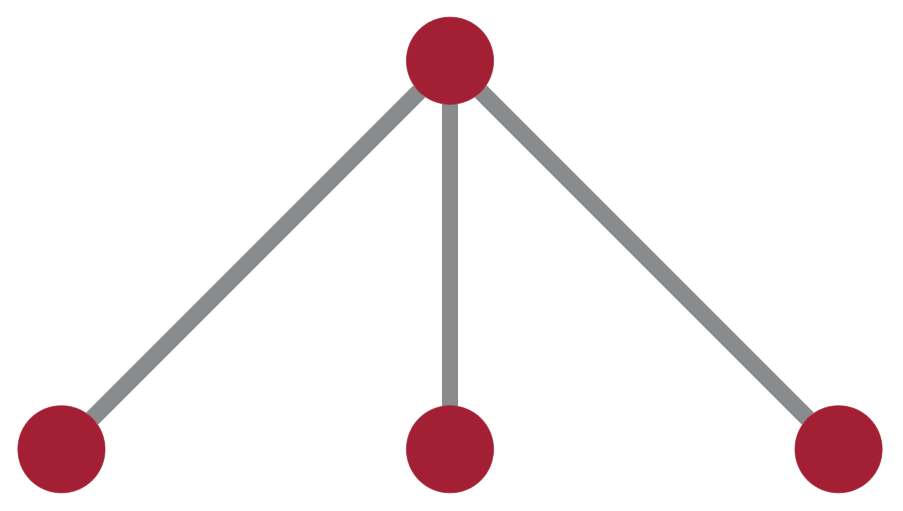}
                                     \includegraphics[scale=0.15]{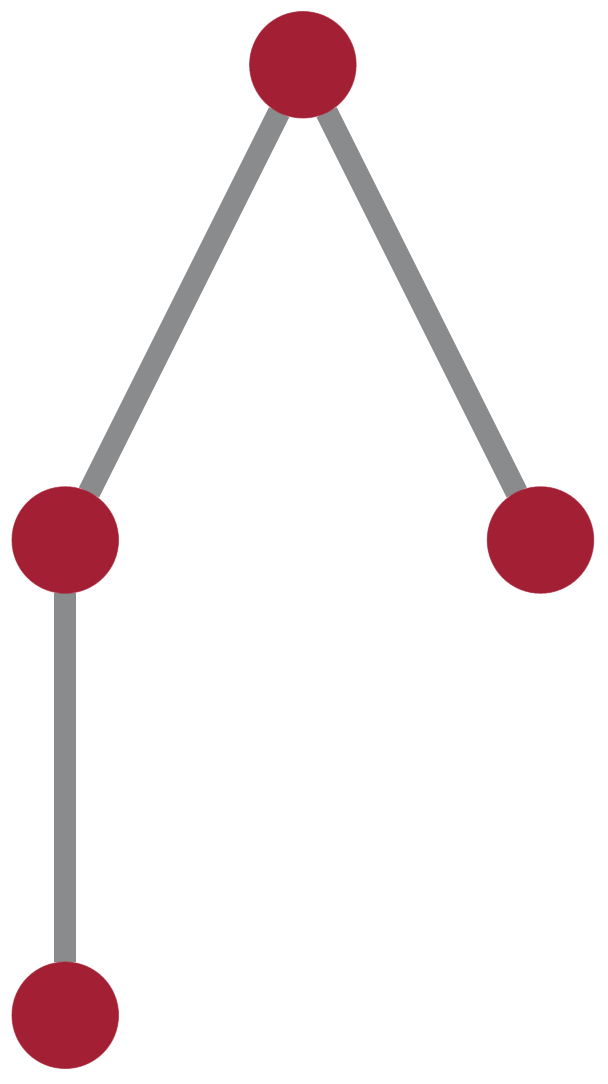}
\tabularnewline\hline
\multirow{5}{*}{$d=4$} & \vspace{.08in}
                                     \includegraphics[scale=0.15]{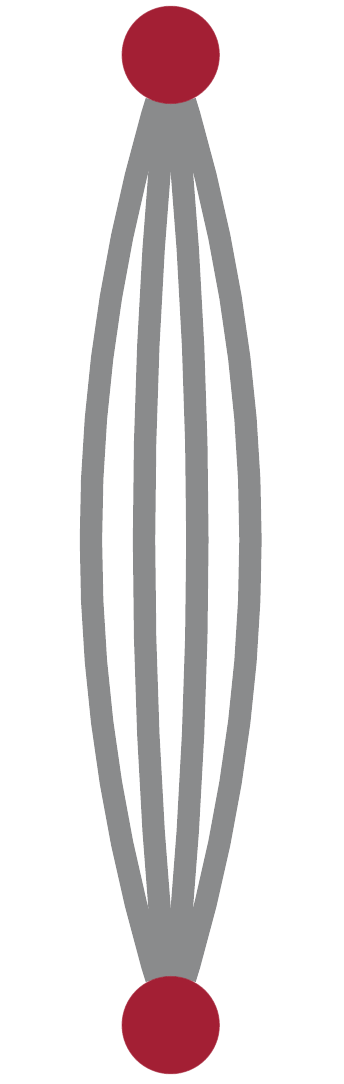}
                                     \includegraphics[scale=0.15]{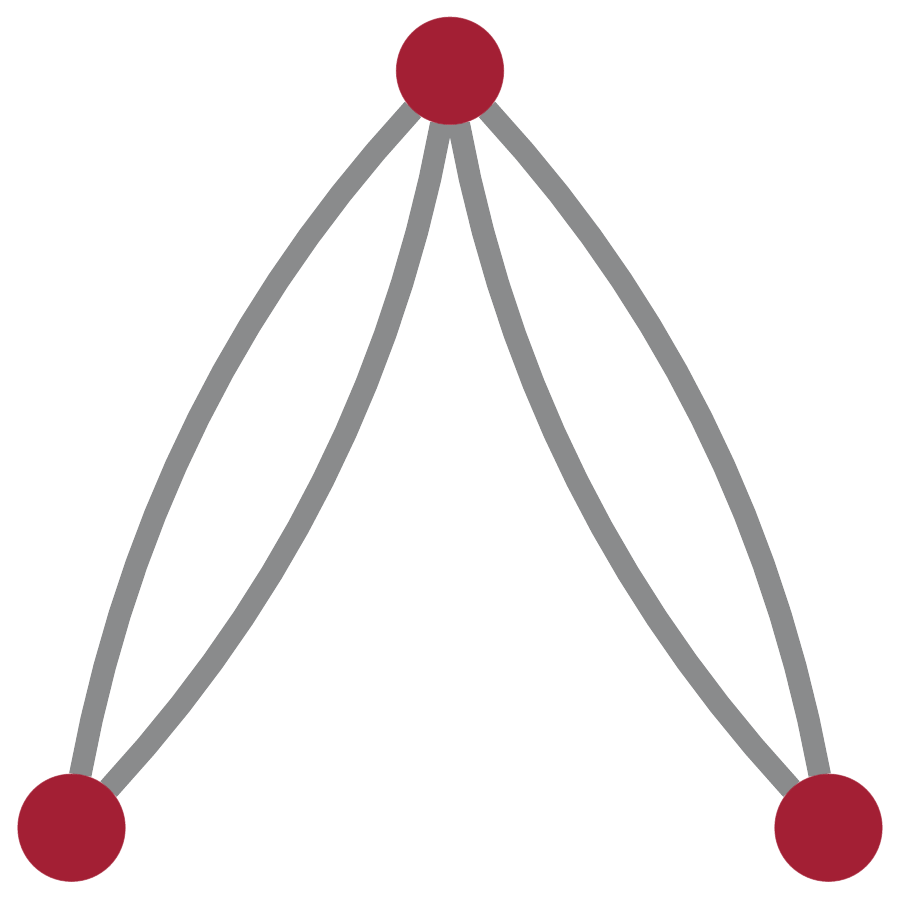}
                                     \includegraphics[scale=0.15]{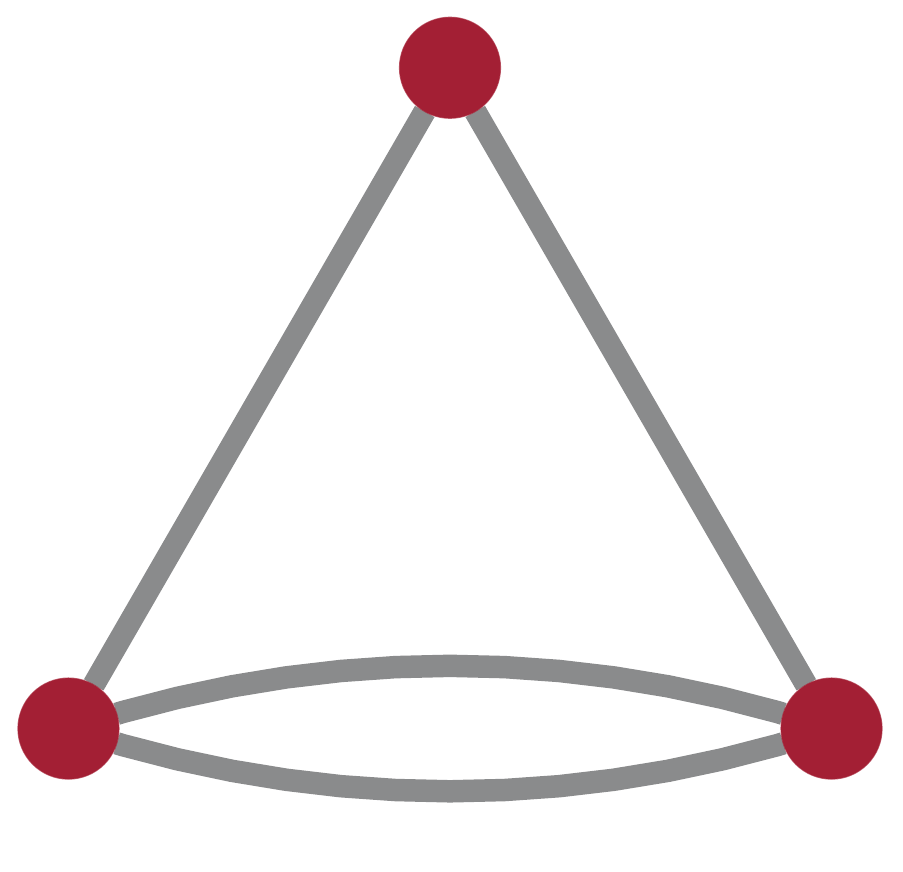}
                                     \includegraphics[scale=0.15]{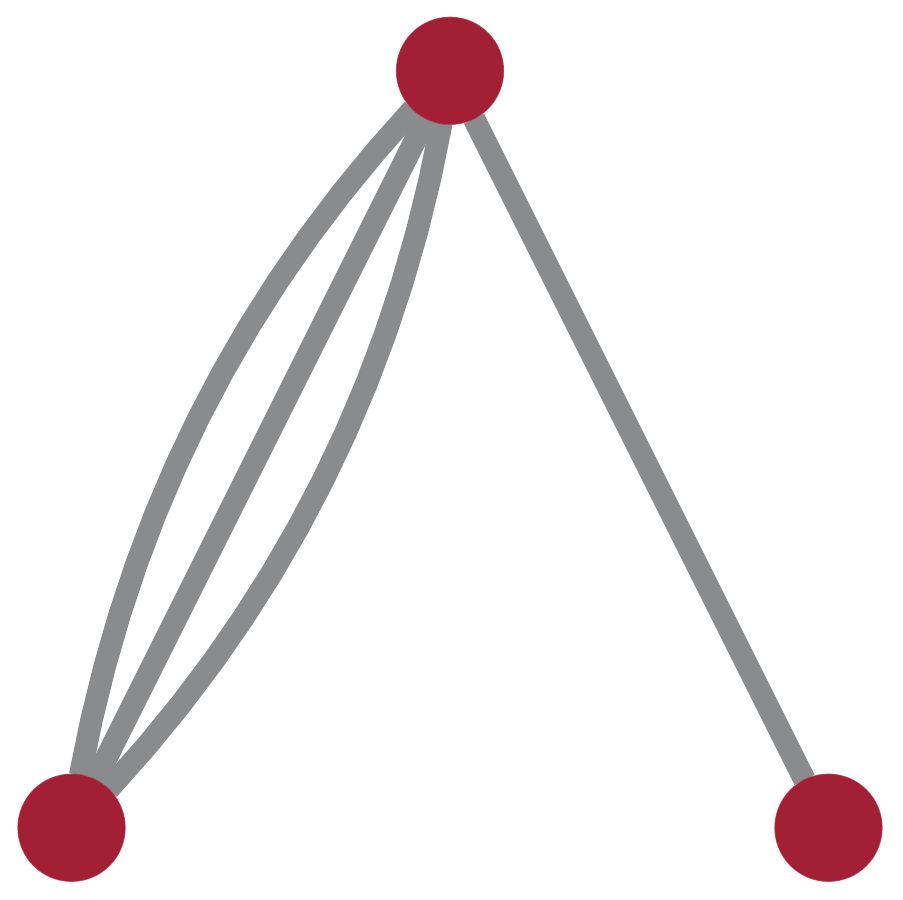}
                                     \includegraphics[scale=0.15]{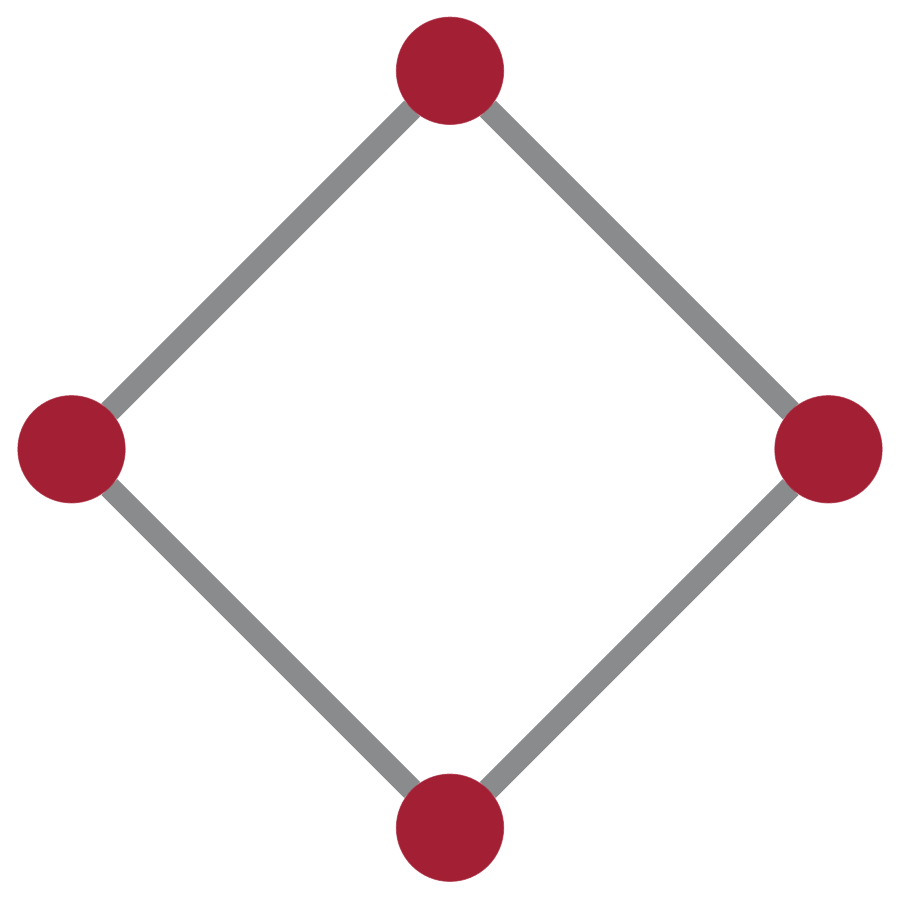}
                                     \includegraphics[scale=0.2]{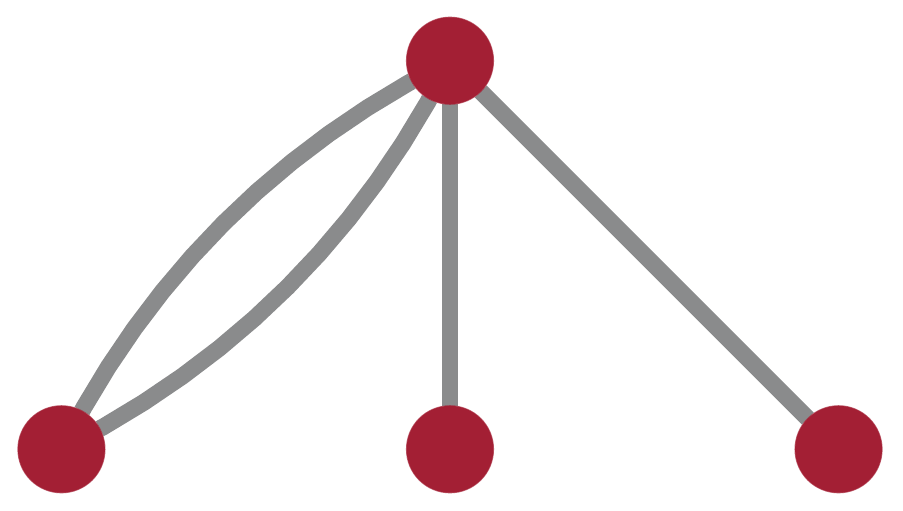}
                                     \tabularnewline &
                                     \includegraphics[scale=0.15]{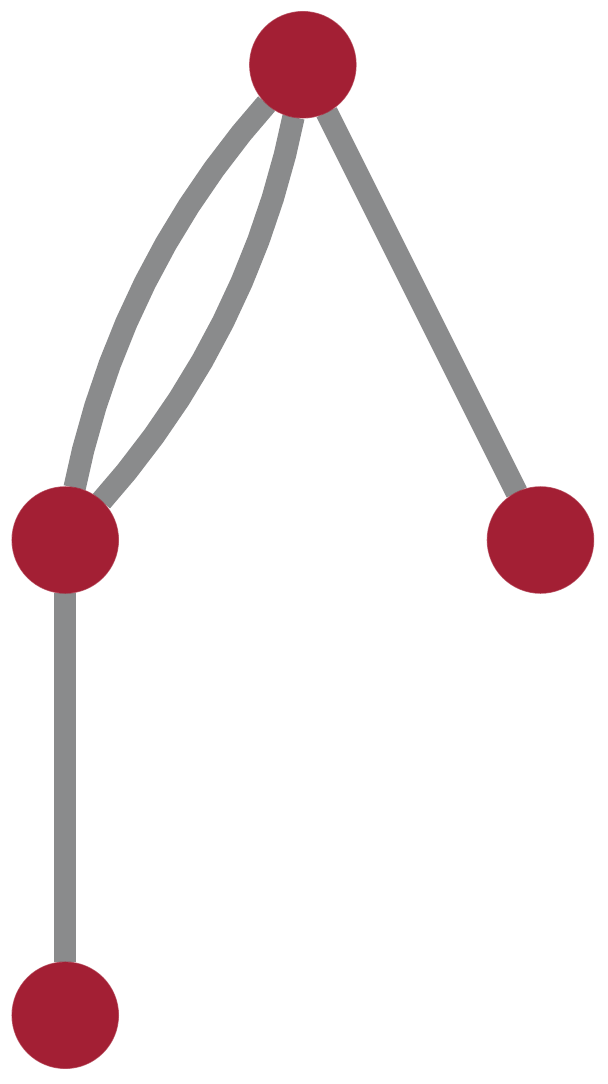}
                                     \rotatebox[origin=t]{270}{\includegraphics[scale=0.2]{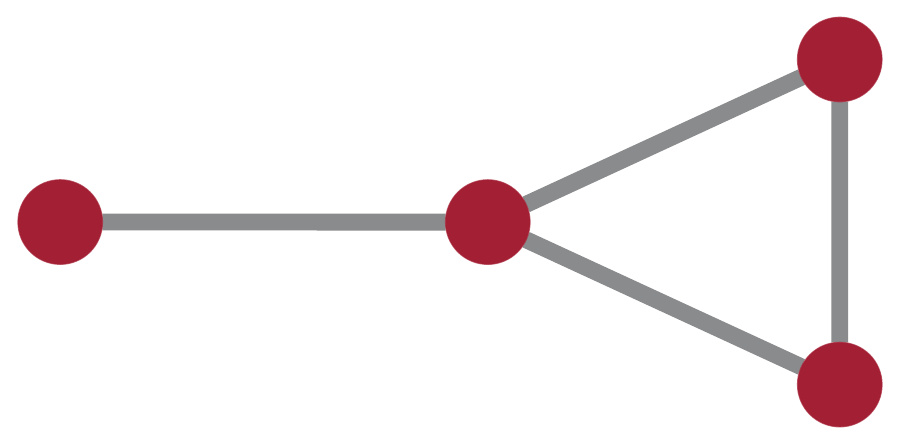}}
                                     \includegraphics[scale=0.15]{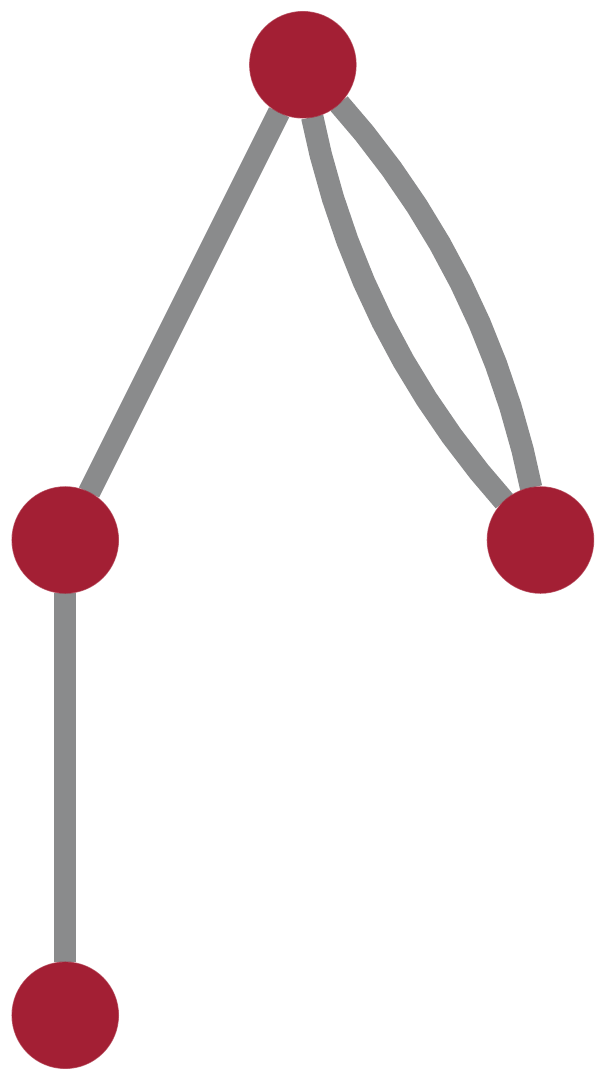}
                                     \includegraphics[scale=0.2]{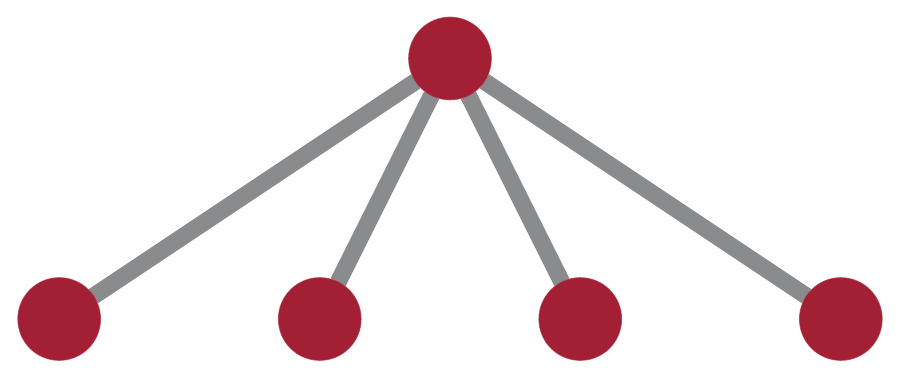}
                                     \includegraphics[scale=0.15]{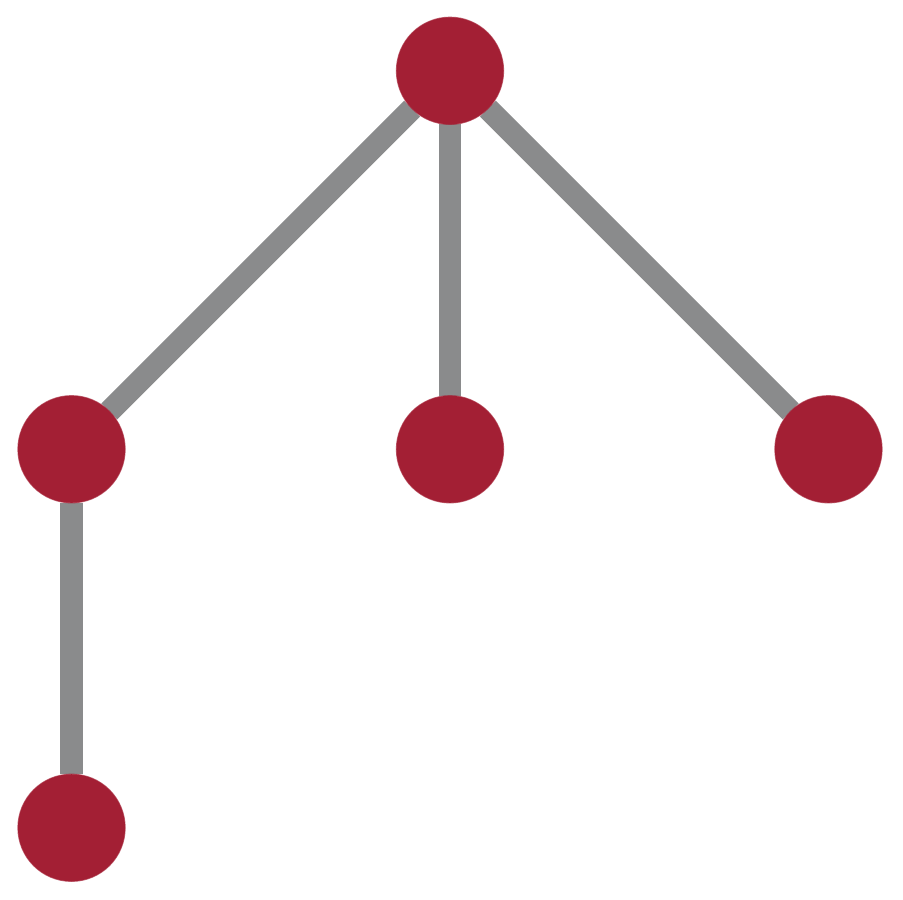}
                                     \includegraphics[scale=0.15]{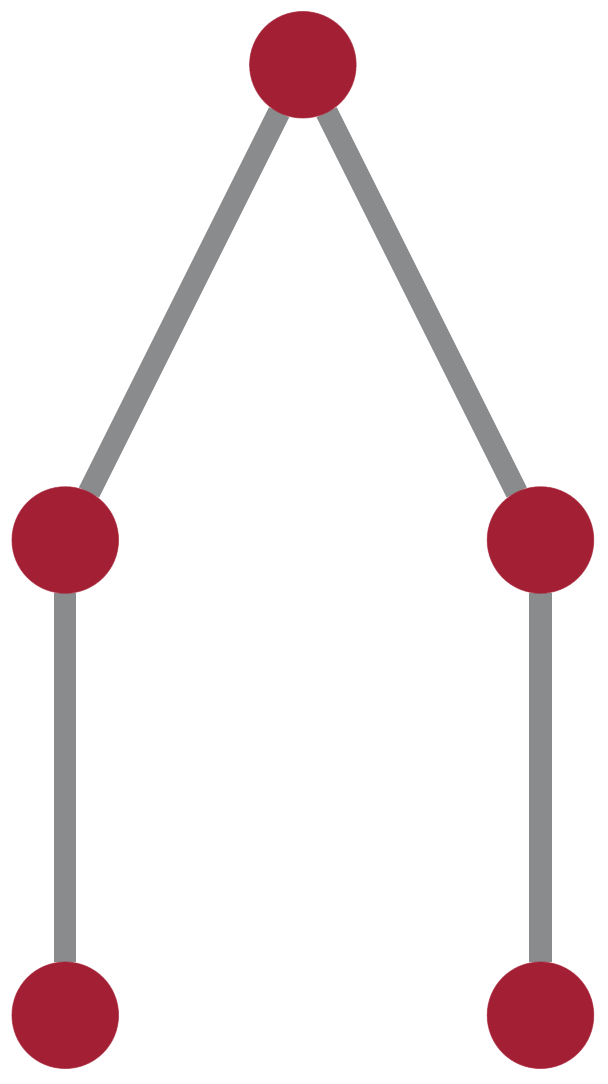}
\tabularnewline\hline
\multirow{12}{*}{$d=5$} & \vspace{.08in}
                                     \includegraphics[scale=0.15]{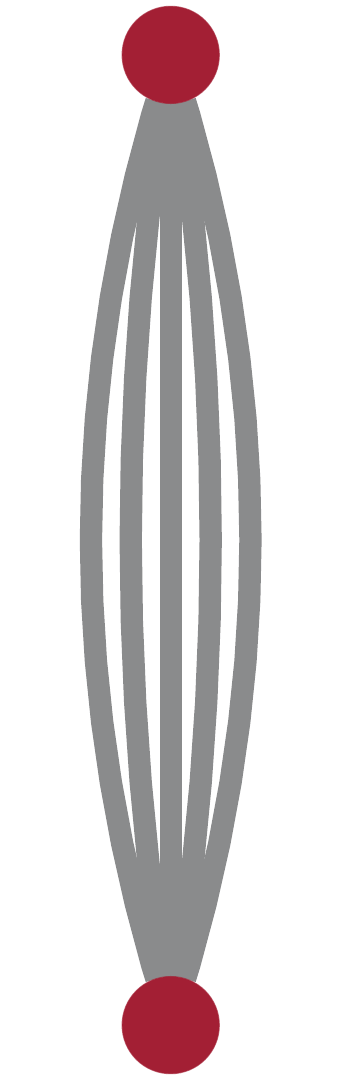}
                                     \includegraphics[scale=0.15]{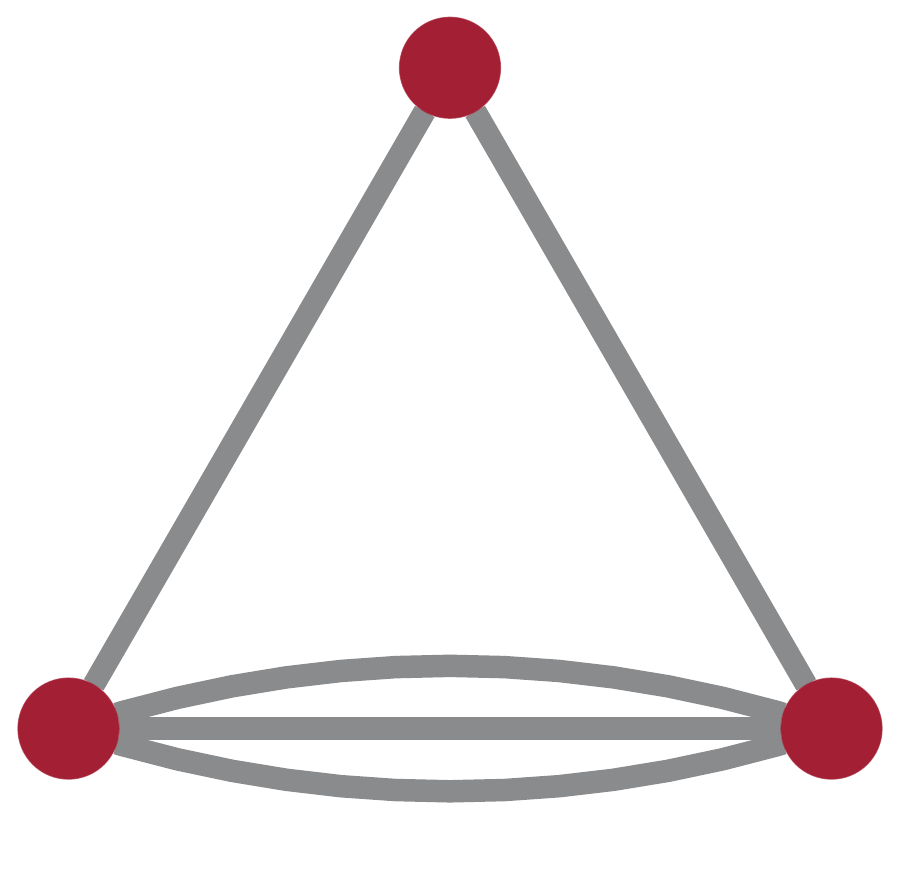}
                                     \includegraphics[scale=0.15]{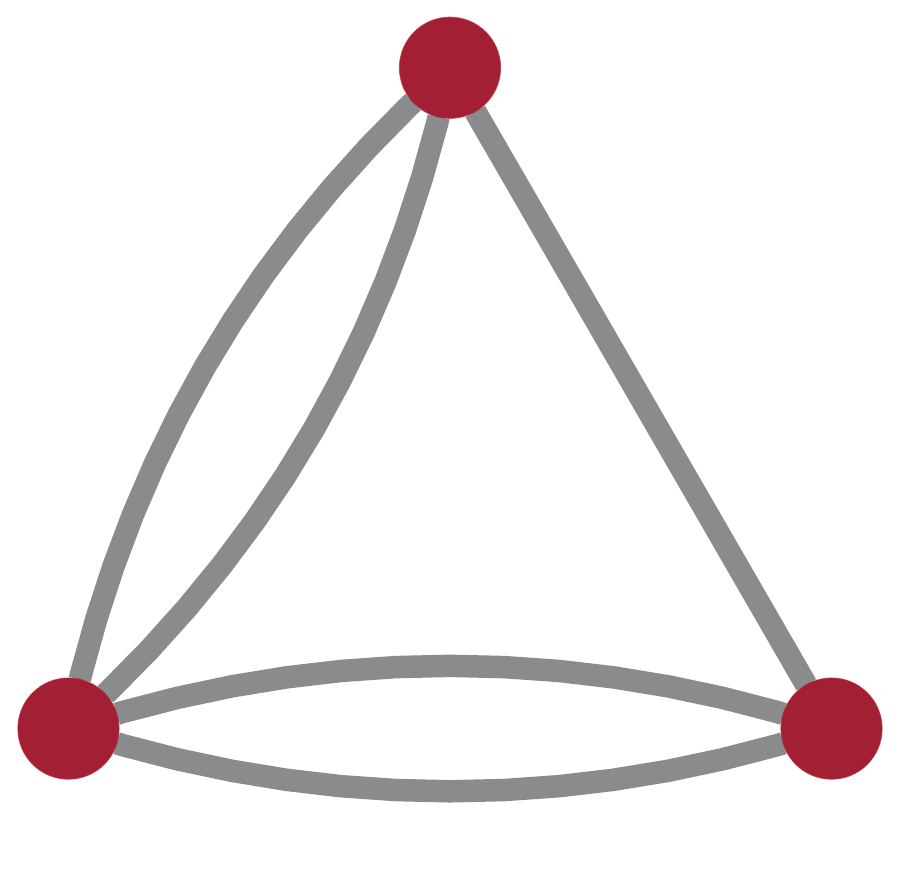}
                                     \includegraphics[scale=0.15]{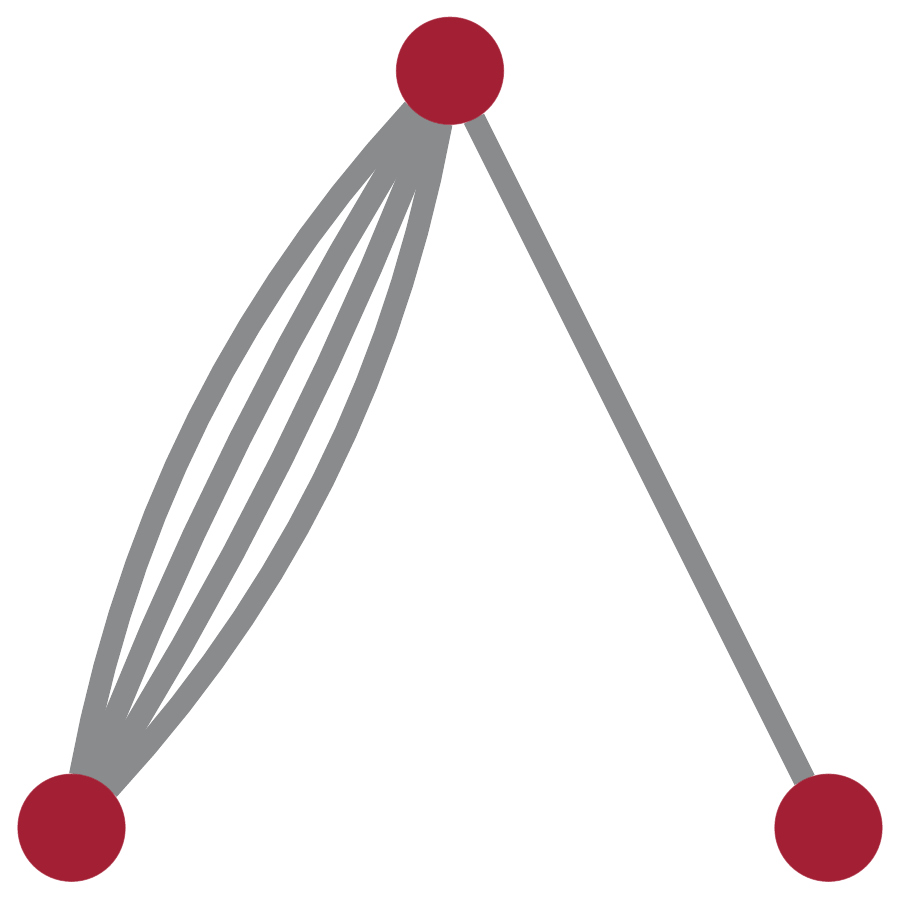}
                                     \includegraphics[scale=0.15]{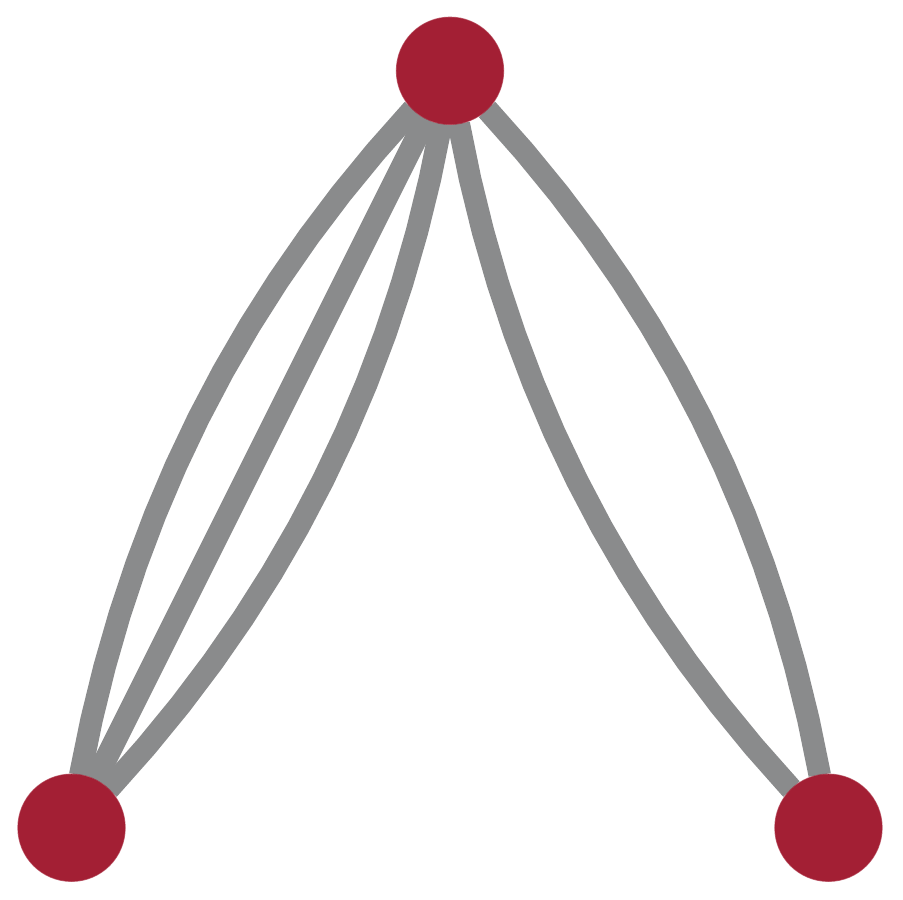}
                                     \rotatebox[origin=t]{270}{\includegraphics[scale=0.165]{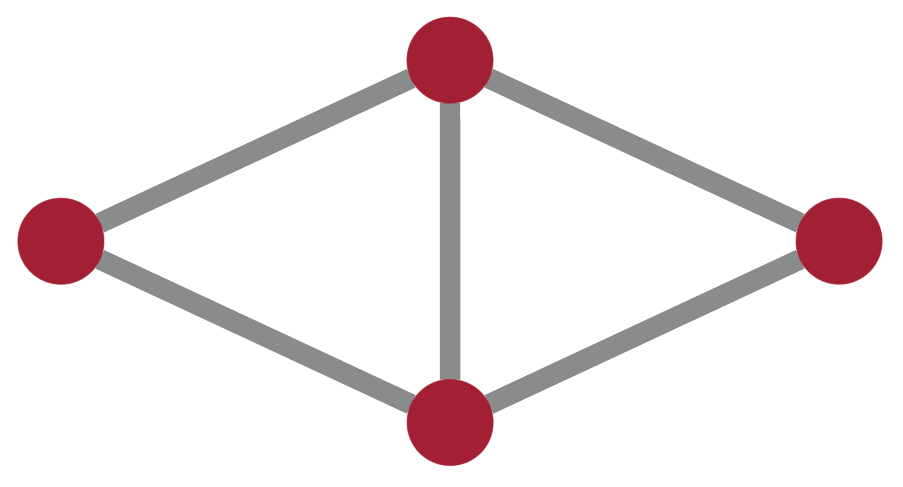}}
                                     \includegraphics[scale=0.175]{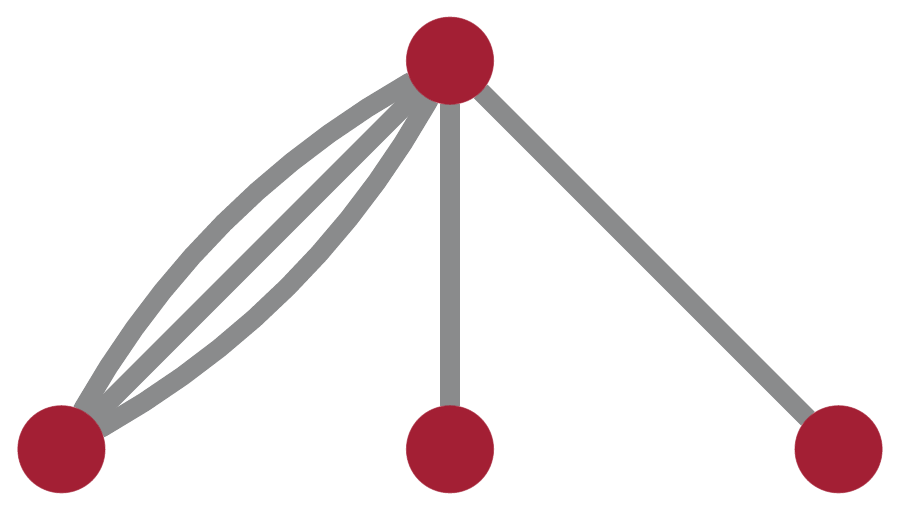}
                                     \includegraphics[scale=0.175]{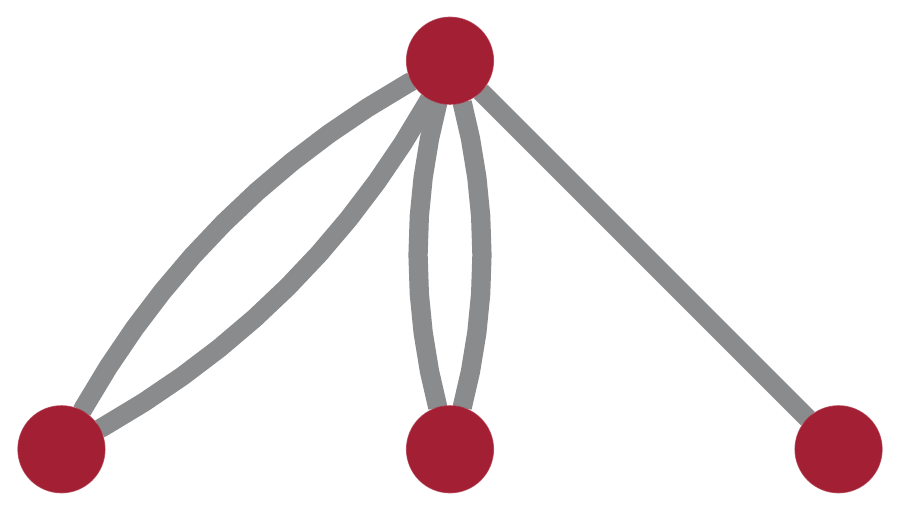}
                                     \tabularnewline &
                                     \includegraphics[scale=0.15]{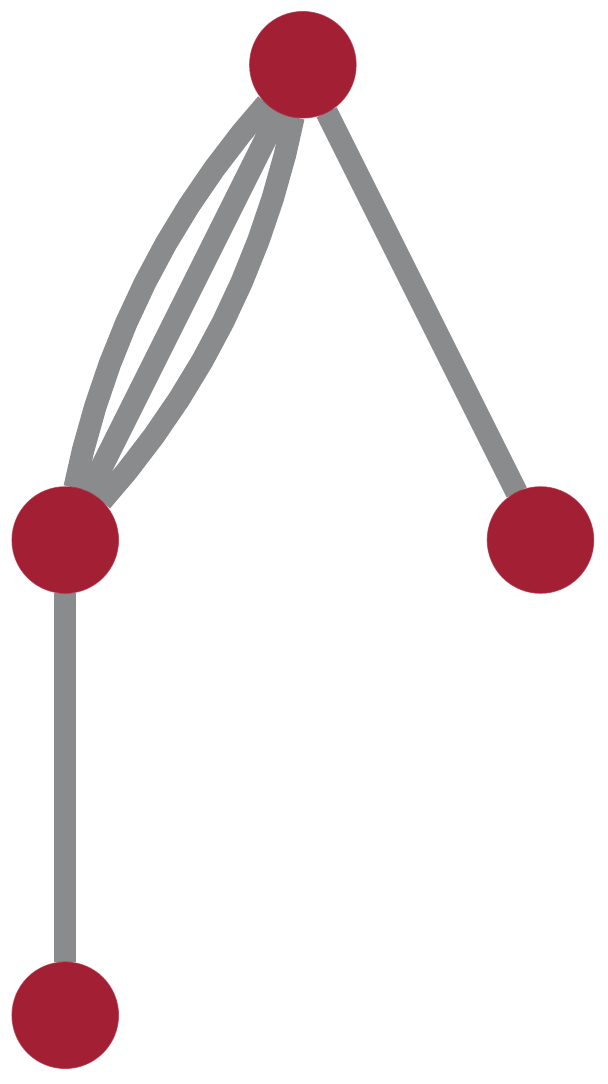}
                                     \rotatebox[origin=t]{270}{\includegraphics[scale=0.175]{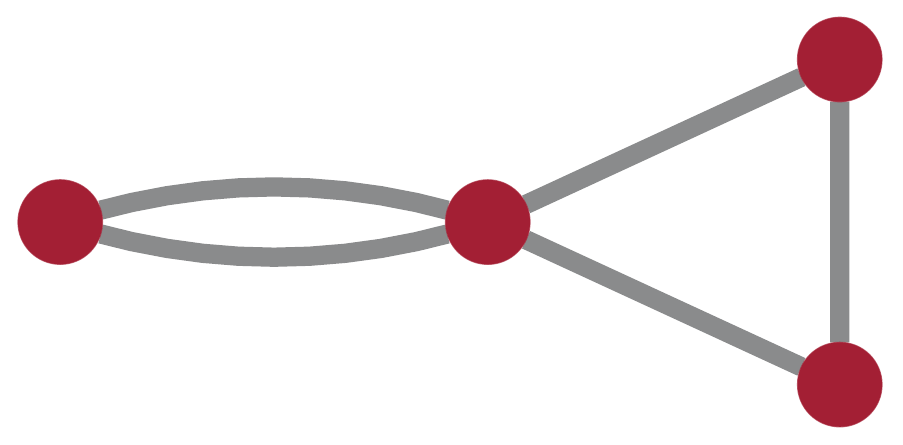}}
                                     \rotatebox[origin=t]{270}{\includegraphics[scale=0.175]{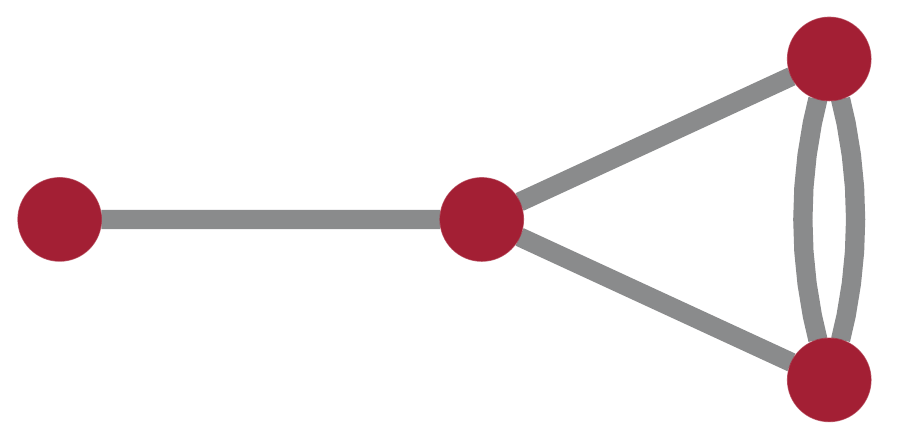}}
                                     \rotatebox[origin=t]{270}{\includegraphics[scale=0.175]{graphs/4_5_9}}
                                     \includegraphics[scale=0.15]{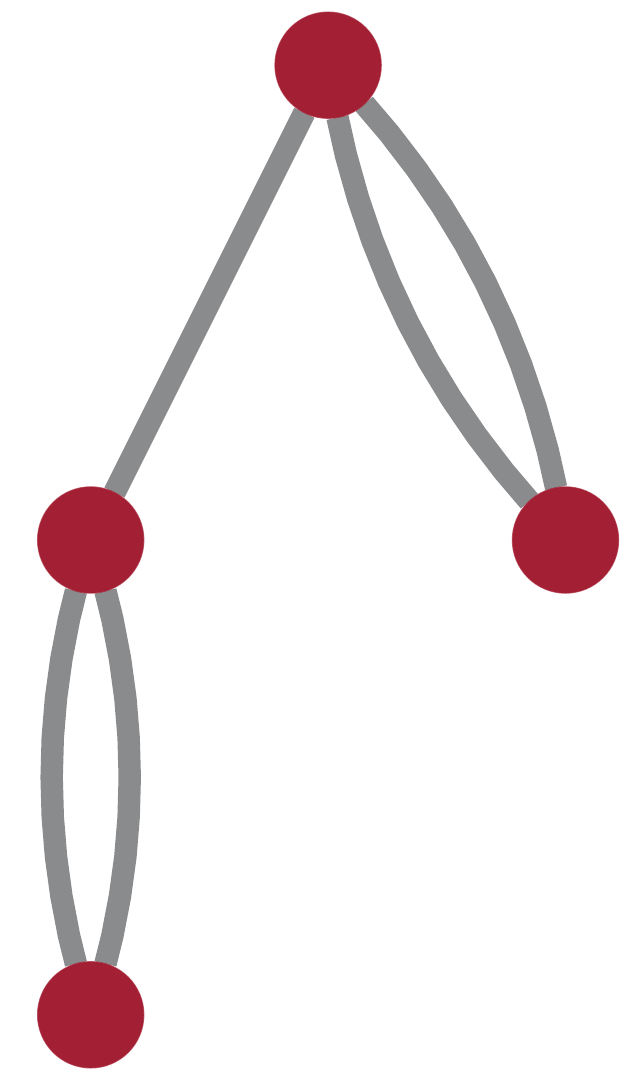}
                                     \includegraphics[scale=0.15]{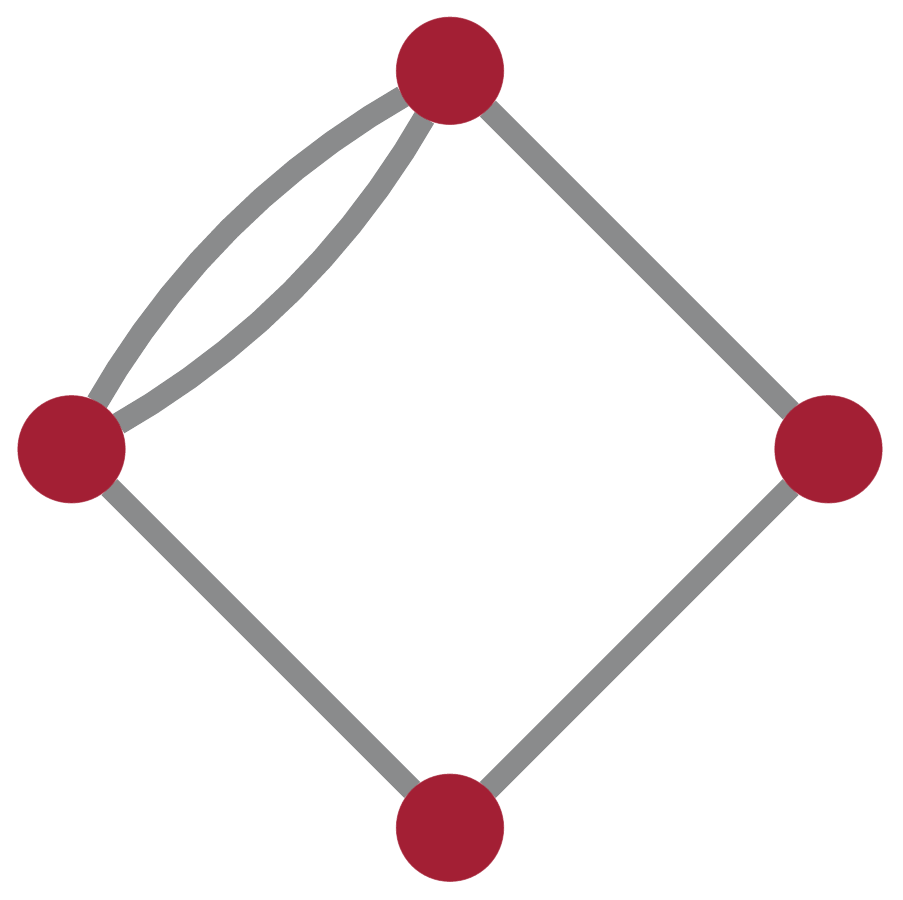}
                                     \includegraphics[scale=0.15]{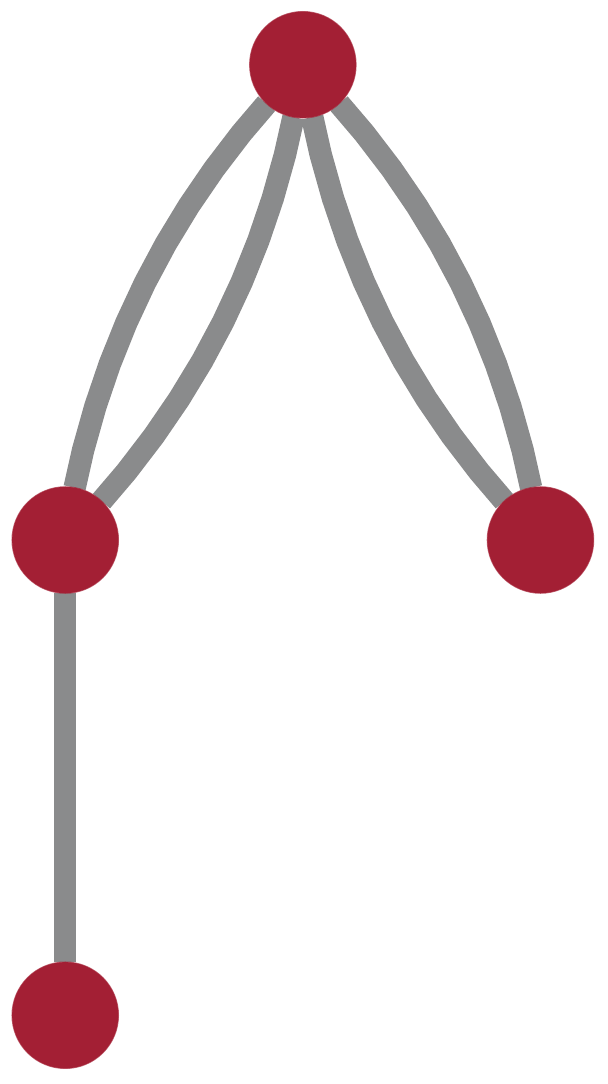}
                                     \includegraphics[scale=0.15]{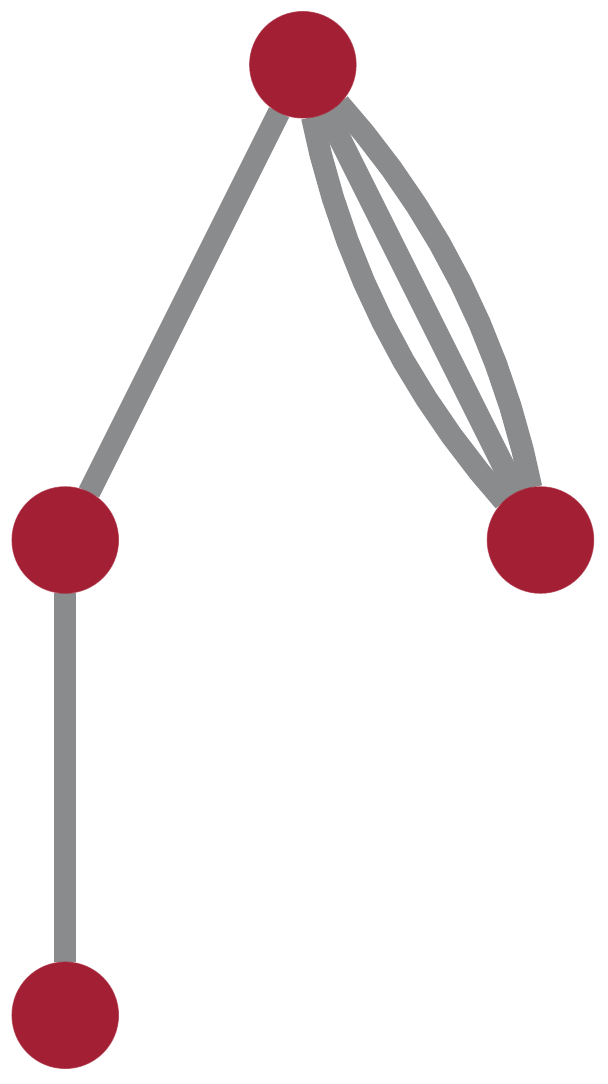}
                                     \includegraphics[scale=0.15]{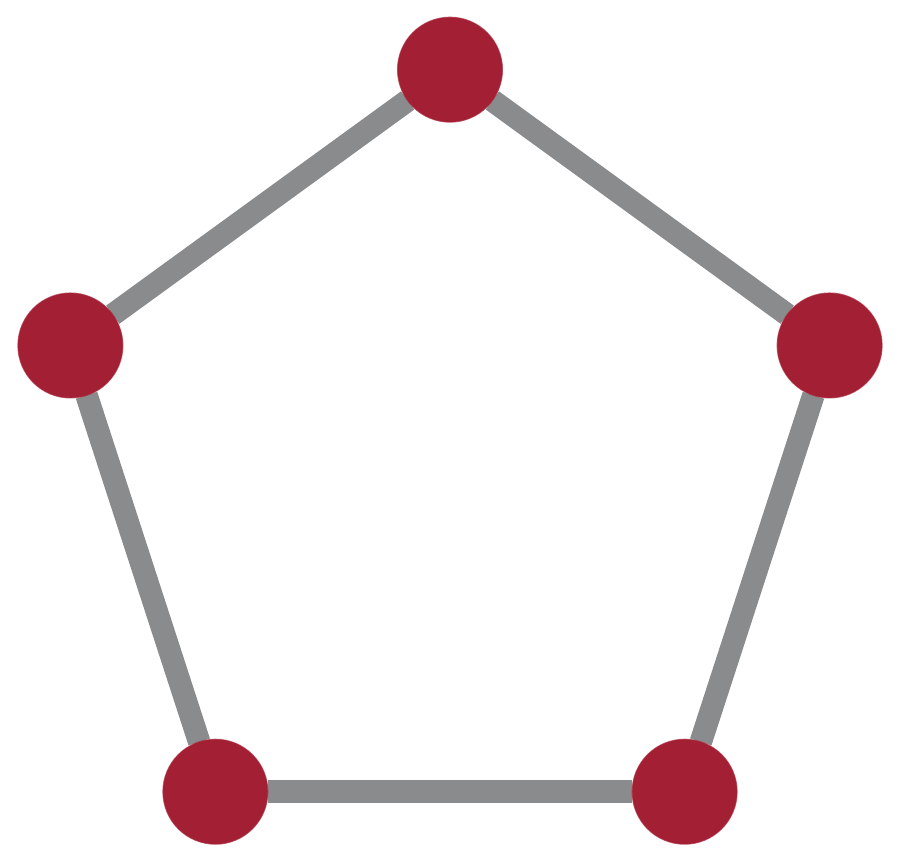}
                                     \includegraphics[scale=0.175]{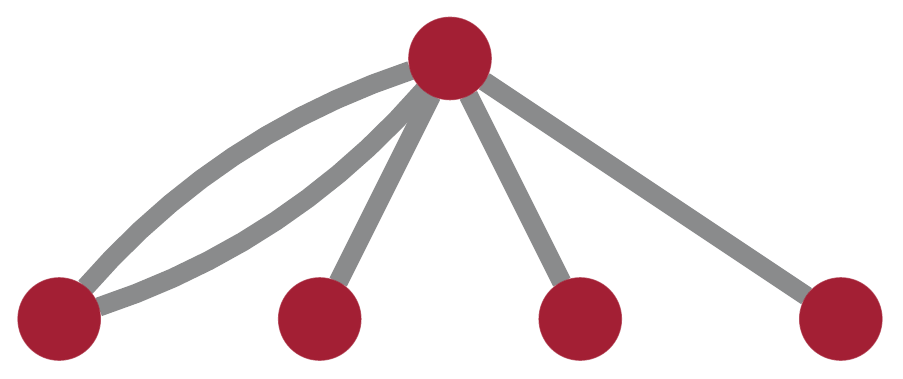}
                                     \tabularnewline &
                                     \rotatebox[origin=c]{270}{\includegraphics[scale=0.15]{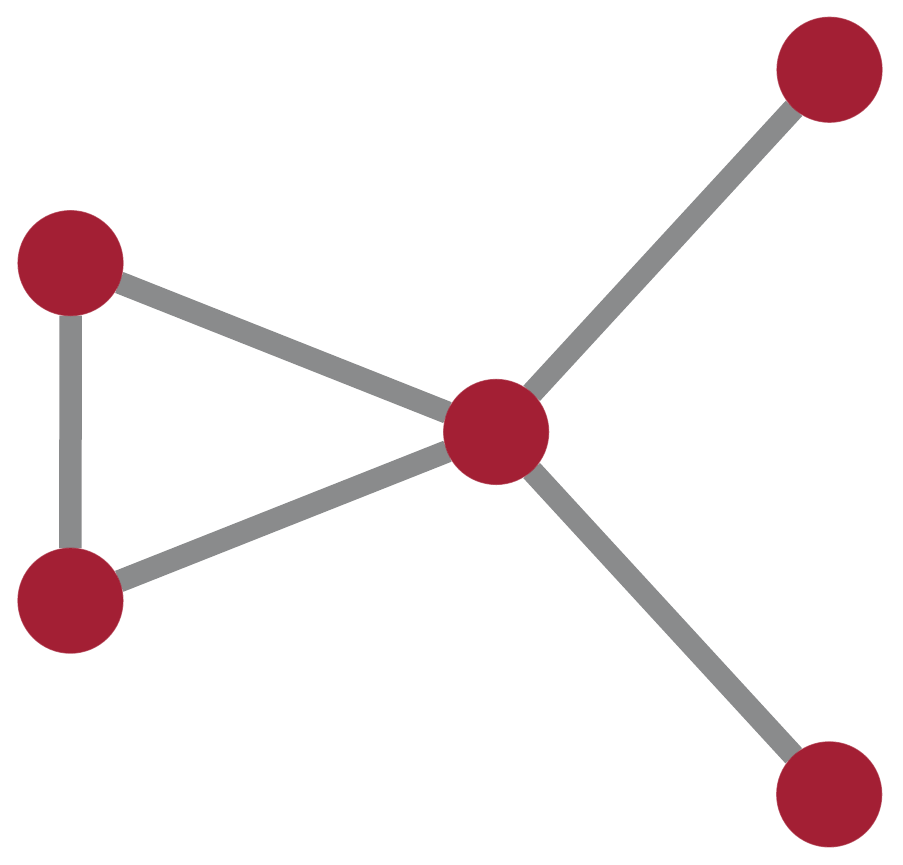}}
                                     \includegraphics[scale=0.15]{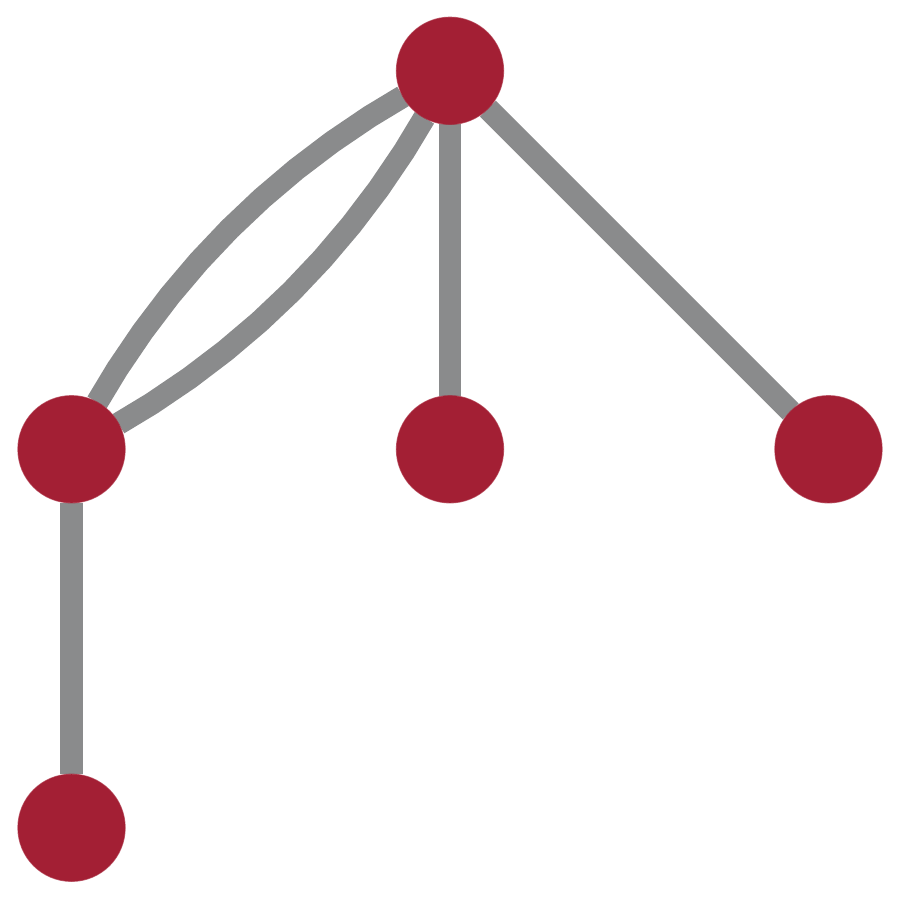}
                                     \includegraphics[scale=0.15]{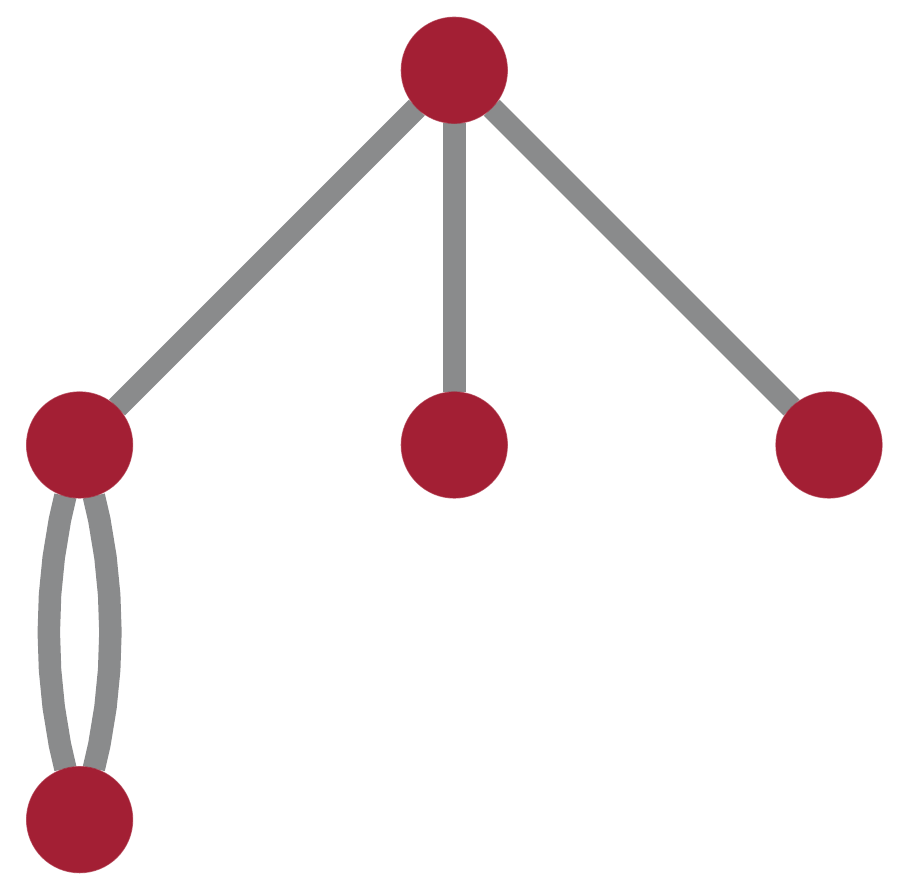}
                                     \rotatebox[origin=t]{270}{\includegraphics[scale=0.175]{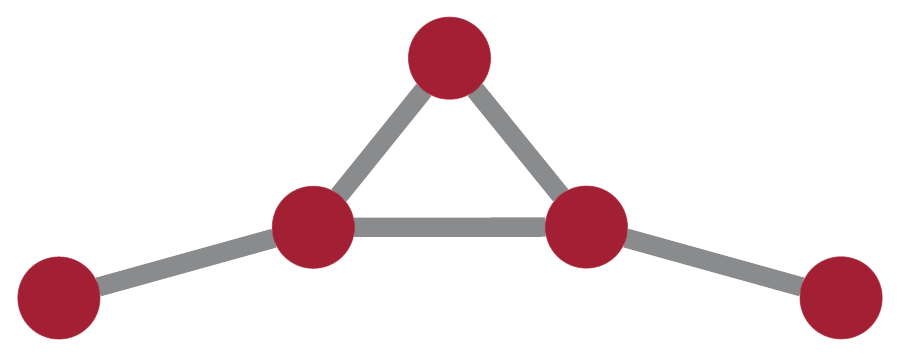}}
                                     \rotatebox[origin=t]{270}{\includegraphics[scale=0.175]{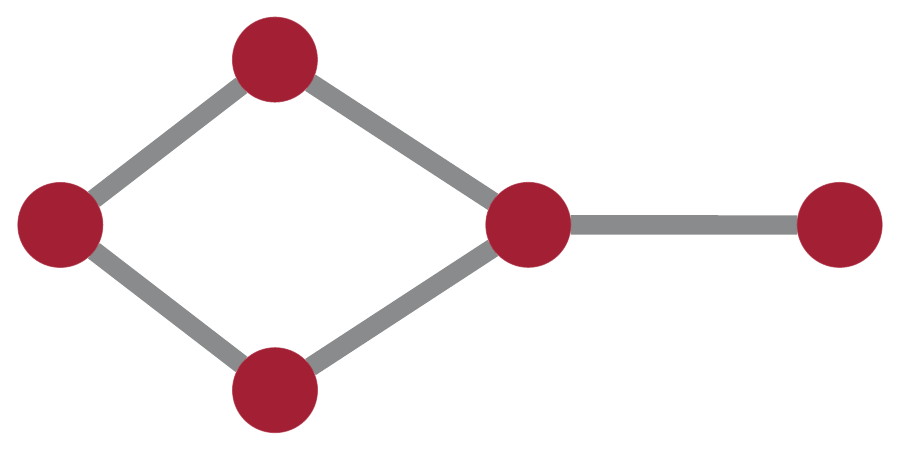}}
                                     \rotatebox[origin=t]{90}{\includegraphics[scale=0.175]{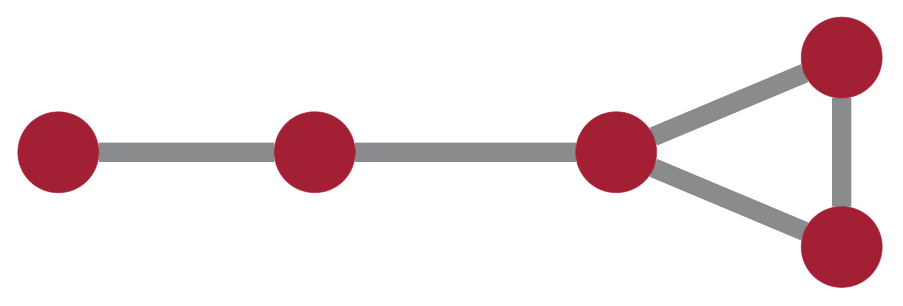}}
                                     \includegraphics[scale=0.15]{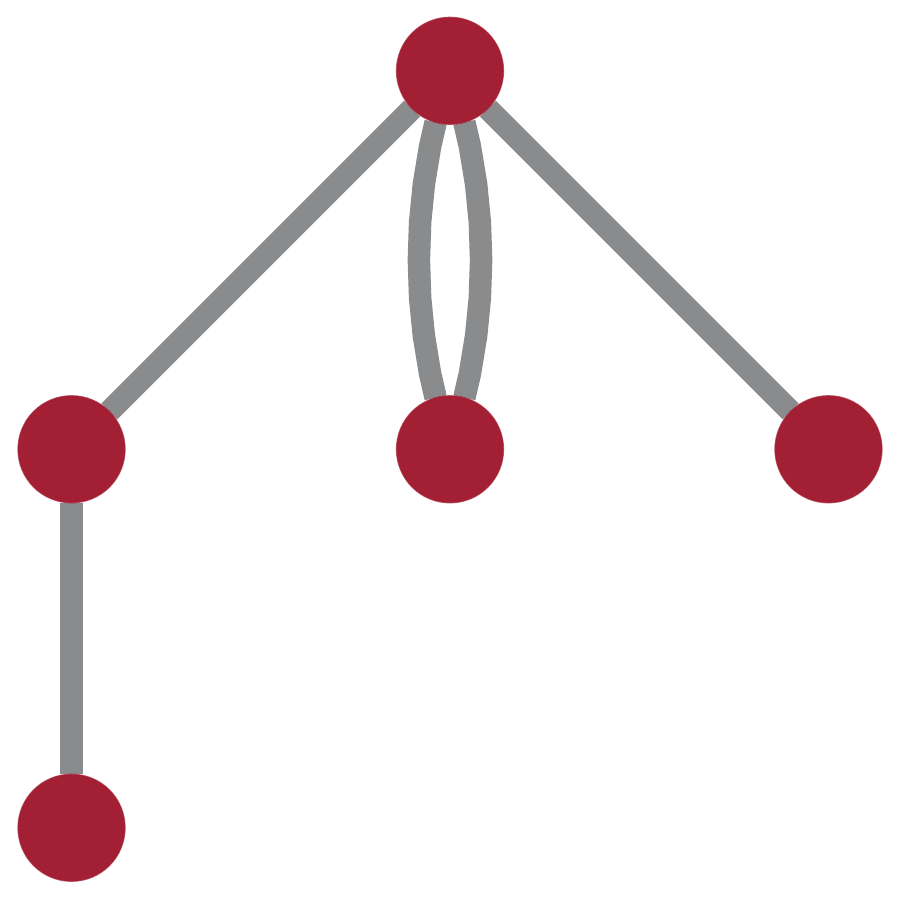}
                                     \includegraphics[scale=0.15]{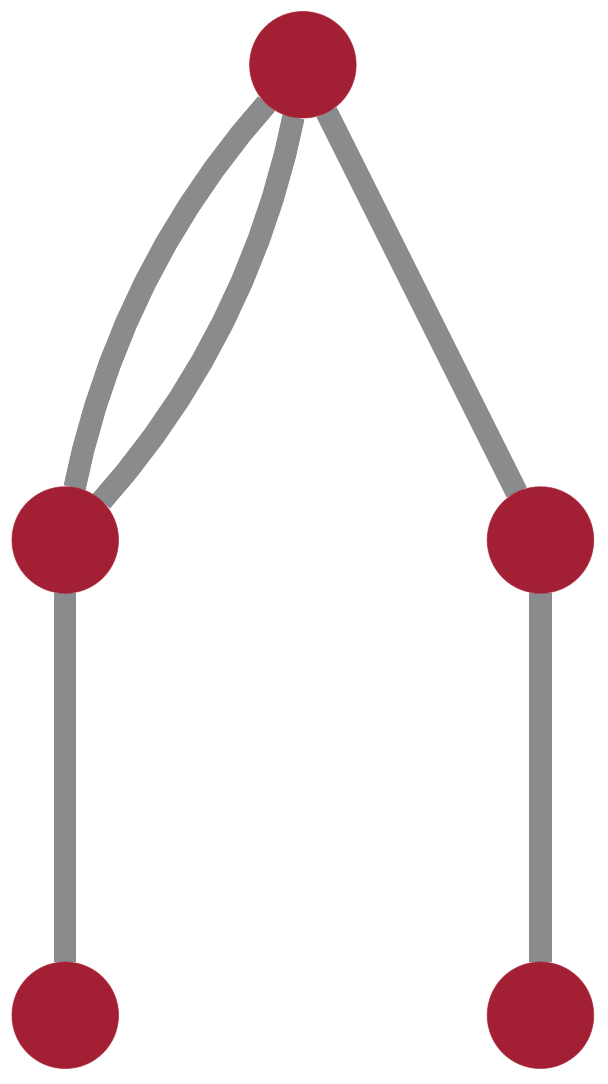}
                                     \includegraphics[scale=0.15]{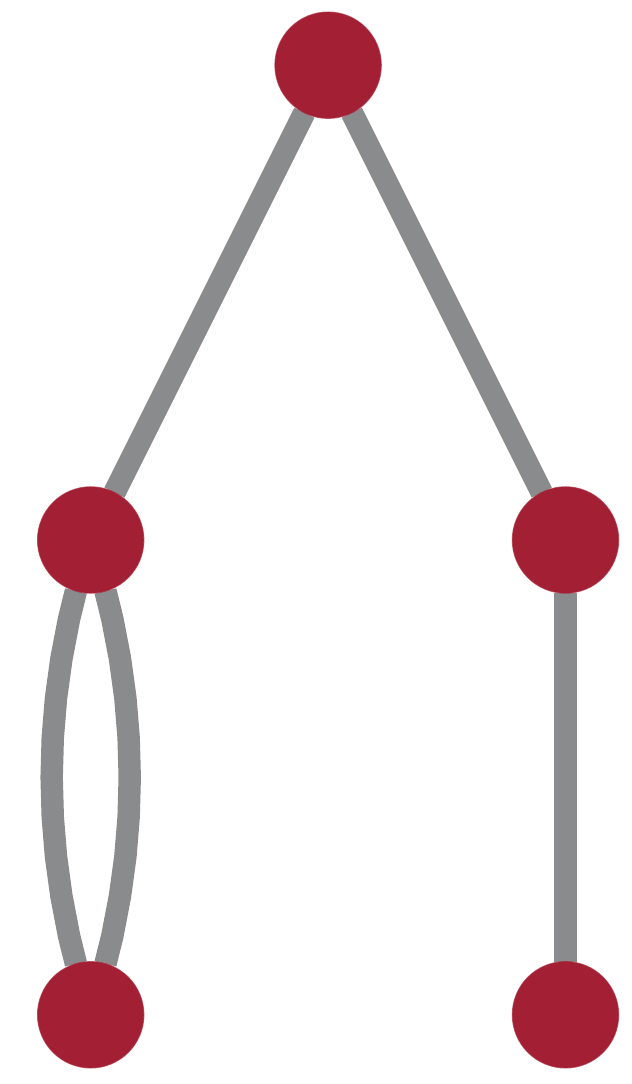}
                                     \includegraphics[scale=0.175]{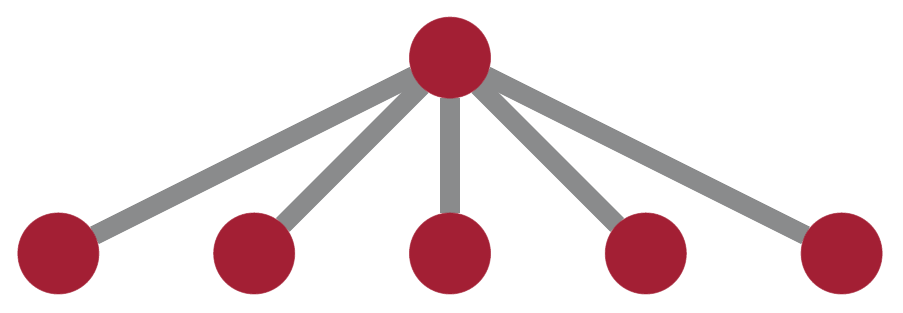}
                                     \tabularnewline & 
                                     \includegraphics[scale=0.15]{graphs/6_5_2}
                                     \includegraphics[scale=0.175]{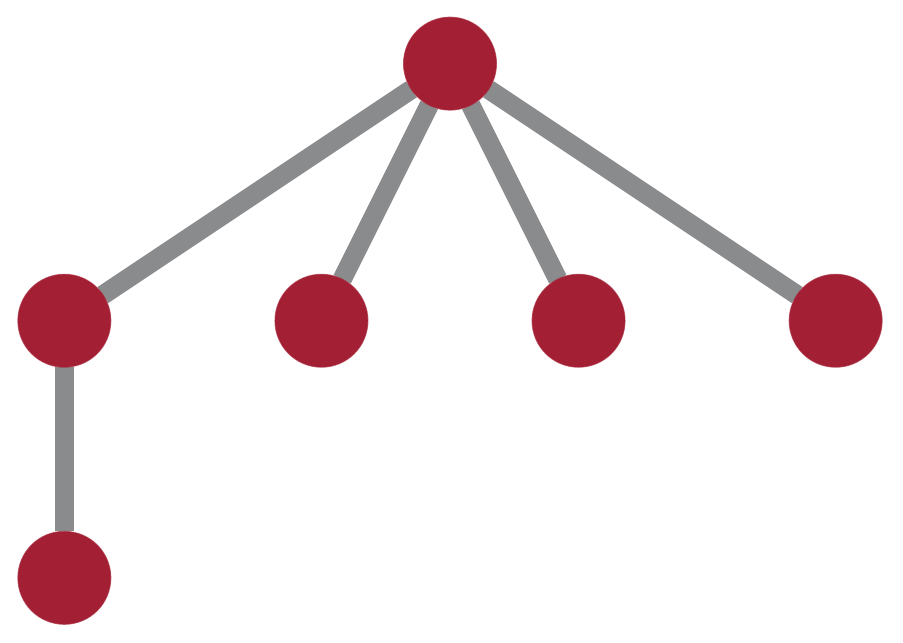}
                                     \includegraphics[scale=0.15]{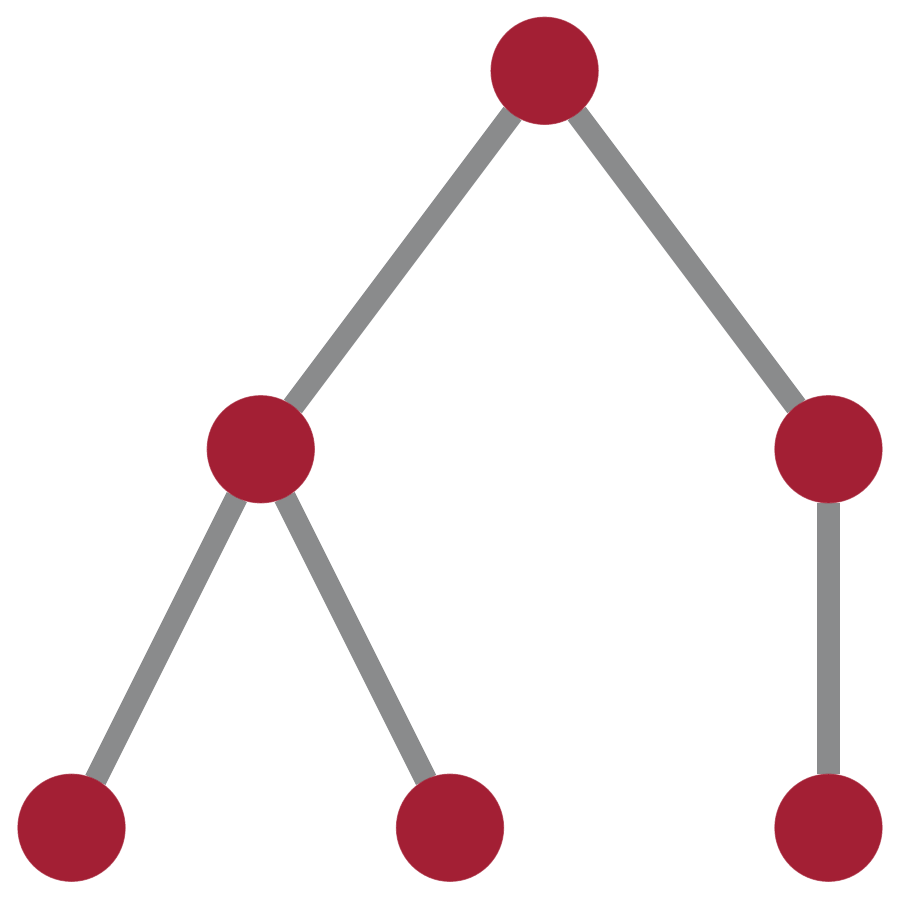}
                                     \includegraphics[scale=0.15]{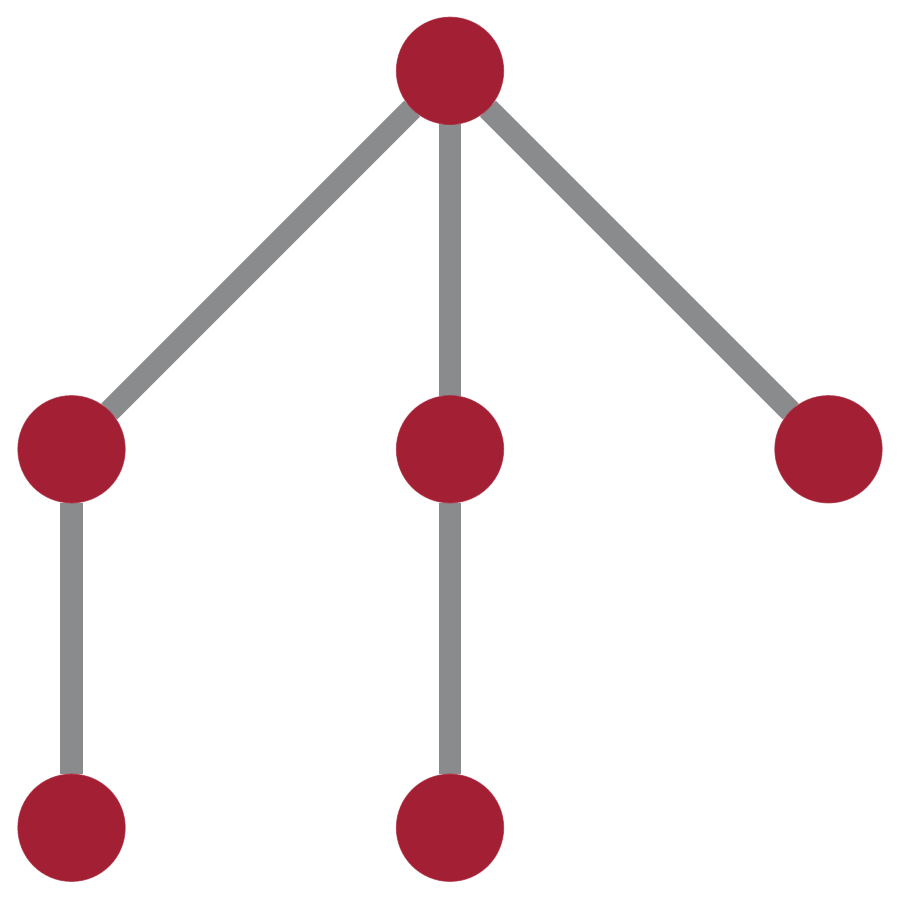}
                                     \includegraphics[scale=0.15]{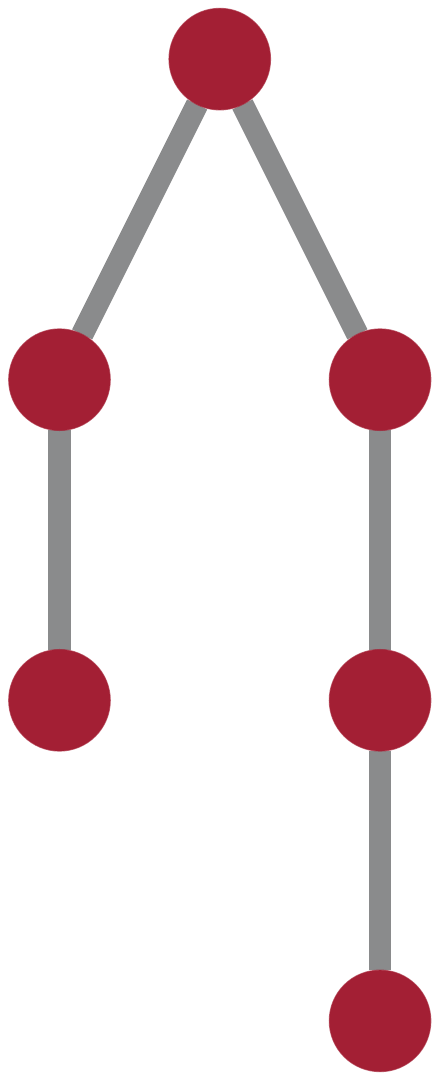}
\tabularnewline\hline
\end{tabular}
\caption{All non-isomorphic, loopless, connected multigraphs organized by the total number of edges $d$, up to $d=5$, sorted by their number of vertices $N$.  Note that for a fixed number of edges $d$, the total number of multigraphs (connected or not) is finite. These graphs correspond to the $d\le5$ prime \Bs counted in \Tab{tab:efpcounts:a}. Image files for all of the prime \B multigraphs up to $d=7$ are available \href{https://github.com/pkomiske/EnergyFlow/tree/images/graphs}{here}.}
\label{tab:graphs}
\end{table}

\afterpage{\clearpage}

\subsection{Energy and angular measures }
\label{sec:measures}

There are many possible choices for the energy fraction $z_i$ and angular measure $\theta_{ij}$ used to define the \Bs.
In the analysis of \Sec{sec:basis}, this choice arises because there are many systematic expansions of IRC-safe observables in terms of energy-like and angular-like quantities.
Typically, one wants to work with observables that respect the appropriate Lorentz subgroup for the collision type of interest.
For $e^+e^-$ colliders, the symmetries are the group of rotations about the interaction point, and for hadron colliders they are rotations about and boosts along the beam axis (sometimes with a reflection in the plane perpendicular to the beam). 
Therefore, the energy fractions $z_i$ usually use particle energies $E_i$ at an $e^+e^-$ collider and particle transverse momenta $p_{T,i}$ at a hadron collider.

For the angular weighting function $f_N$, though, there are many different angular structures one can build out of the particle directions $\hat{p}^\mu_i$.
The \Bs use the simplest and arguably most natural choice to expand the angular behavior:  \emph{pairwise} angular distances $\theta_{ij}$,  determined using spherical coordinates $(\theta,\phi)$ at an $e^+e^-$ collider and rapidity-azimuth coordinates $(y,\phi)$ at a hadron collider.
Other classes of observables, such as ECFs~\cite{Larkoski:2013eya} and ECFGs~\cite{Moult:2016cvt}, also use pairwise angles since they manifestly respect the underlying Lorentz subgroup.
For building the \Bs, is important that the $\theta_{ij}$, or any other choice of geometric object, be sufficient to reconstruct the value of the original function $f_N$ in terms of the $\hat p_i^\mu$.
For pairwise angles, this property can be shown by triangulation, under the assumption that the observable in question does not depend on the overall jet direction nor on rotations or reflections about the jet axis.
Since jets are collimated sprays of particles, the $\theta_{ij}$ are typically small and are good expansion parameters.

At various points in this chapter, we explore three different energy/angular measures.  For $e^+e^-$ collisions, our default is:
\begin{equation}\label{eq:eemeasure}
\boxed{\text{$e^+e^-$ Default}} \qquad
\begin{aligned}
\quad z_i &= \frac{E_i}{E_J}, \qquad  E_J\equiv\sum_{i=1}^M E_i, \\
\theta_{ij} &= \left(\frac{2\,p_i^\mu p_{j\mu}}{E_i  E_j} \right)^{\beta/2},
\end{aligned}
\end{equation}
where $\beta>0$ is an angular weighting factor. 
For the hadron collider studies in \Secs{sec:linreg}{sec:linclass}, we use:
\begin{equation}\label{eq:hadronicmeasure}
\boxed{\text{Hadronic Default}} \qquad
\begin{aligned}
\quad z_i &= \frac{p_{T,i}}{p_{T,J}}, \qquad  p_{T,J}\equiv\sum_{i=1}^Mp_{T,i}, \\
\theta_{ij} &= \left(\Delta y_{ij}^2+\Delta\phi_{ij}^2\right)^{\beta/2},
\end{aligned}
\end{equation}
where $\Delta y_{ij}\equiv y_i-y_j$, $\Delta\phi_{ij}\equiv\phi_i-\phi_j$ are determined by the rapidity $y_i$ and azimuth $\phi_i$ of particle $i$.
This measure is rotationally-symmetric in the $(y,\phi)$ plane, which is the most commonly used case in jet substructure.
For situations where this rotational symmetry is not desirable (such as for jet pull~\cite{Gallicchio:2010sw}), we can instead use a two-dimensional measure that treats the rapidity and azimuthal directions separately:
\begin{equation}\label{eq:2Dhadronicmeasure}
\boxed{\text{Hadronic Two-Dimensional}} \qquad
\begin{aligned}
\quad z_i &= \frac{p_{T,i}}{p_{T,J}}, \qquad  p_{T,J}\equiv\sum_{i=1}^Mp_{T,i}, \\
\theta_{ij} &= \Delta y_{ij} \text{ or } \Delta\phi_{ij},
\end{aligned}
\end{equation}
where each line on the multigraph now has an additional decoration to indicate whether it corresponds to $\Delta y$ or $\Delta \phi$.

We emphasize that the choice of measure is not unique, though it is constrained by the IRC-safety arguments in \Sec{sec:basis}.
For example, IRC safety requires that the energy-like quantities appear linearly in $z_i$.
For the default measures, the angular exponent $\beta$ can take on any positive value and still be consistent with IRC safety.
Depending on the context, different choices of $\beta$ can lead to faster or slower convergence of the \B expansion, with $\beta<1$ emphasizing smaller values of $\theta_{ij}$ and $\beta>1$ emphasizing larger values of $\theta_{ij}$.
For special choices of $z_i$ and $\theta_{ij}$, some \Bs may be linearly related, a point we return to briefly in \Sec{sec:algebraic}.

\subsection{Relation to existing substructure observables}
\label{sec:jetobs}

Many familiar jet observables can be nicely interpreted in the energy flow basis.
When an observable can be written as a simple expression in terms of particle four-momenta or in terms of energies and angles, the energy flow decomposition can often be performed exactly.
Some of the most well-known observables, such as jet mass and energy correlation functions, are exactly finite linear combinations of \Bs (with appropriate choice of measure), which one might expect since they also correspond to natural $C$-correlators.
Unless otherwise specified, the analysis below uses the default hadronic measure in \Eq{eq:hadronicmeasure} with $\beta=1$ and treats all particles as massless.\footnote{A proper treatment of non-zero particle masses would require an additional expansion in the \emph{velocities} of the particles (see related discussion in \Refs{Salam:2001bd,Mateu:2012nk}). To avoid these complications, one can interpret all particles as being massless in the $E$-scheme~\cite{Salam:2001bd}, i.e.\ $p_{\text{rescaled}}^\mu = E \, (1, \hat{p})$ with $\hat{p} = \vec{p} / |\vec{p}^{}|$.}

\subsubsection{Jet mass}
\label{sec:relation_jetmass}

Jet mass is most basic jet substructure observable, and not surprisingly, it has a nice expansion in the energy flow basis.
In particular, the squared jet mass divided by the jet energy squared is an exact $N=2$ \B using the $e^+e^-$ measure in \Eq{eq:eemeasure} with $\beta = 1$:
\begin{equation}\label{eq:massexp}
e^+e^-: \qquad \frac{m_J^2}{E_J^2}=\frac12\sum_{i_1 = 1}^M \sum_{i_2 = 1}^M z_{i_1}z_{i_2} \left(\frac{2\,p_{i_1}^\mu p_{i_2\mu}}{E_{i_1} E_{i_2}} \right) =\frac12\times
\begin{gathered}
\includegraphics[scale=.2]{graphs/2_2_1}
\end{gathered}.
\end{equation}
Note that mass is exactly an \B for any $\beta = 2 / N$ measure choice.

For the hadronic measure in \Eq{eq:hadronicmeasure} with $\beta= 1$, there is an approximate equivalence with the squared jet mass divided by the jet (scalar) transverse momentum:
\begin{equation}\label{eq:massexph}
\text{Hadronic}: \qquad \frac{m_J^2}{p_{TJ}^2}=\sum_{i_1 = 1}^M\sum_{i_2 = 1}^M z_{i_1} z_{i_2} (\cosh(\Delta y_{i_1i_2}) - \cos(\Delta\phi_{i_1i_2})) = \frac12\times
\begin{gathered}
\includegraphics[scale=.2]{graphs/2_2_1}
\end{gathered}
+ \cdots.
\end{equation}
Since the jet mass is not exactly rotationally symmetric in the rapidity-azimuth plane, the subleading terms in \Eq{eq:massexph} are not fully encompassed by the simplified set of hadronic observables depending only on $\{\Delta y_{ij}^2+\Delta\phi_{ij}^2\}$, but could be fully encompassed by using an expansion in $\{\Delta y_{ij},\Delta\phi_{ij}\}$ as in \Eq{eq:2Dhadronicmeasure}.
For narrow jets, these higher-order terms in the expansion become less relevant since $\Delta y_{ij},\,\Delta\phi_{ij}\ll1$.\footnote{Alternatively, we could use a measure with $\theta_{ij} = \left(\frac{2\,p_i^\mu p_{j\mu}}{p_{T,i}  p_{T,j}} \right)^{\beta/2}$, similar in spirit to the Conical Geometric measure of \Ref{Stewart:2015waa}, to exactly recover the jet mass.}

\subsubsection{Energy correlation functions}
\label{sec:equivECF}

The ECFs are designed to be sensitive to $N$-prong jet substructure~\cite{Larkoski:2013eya}.
They can be written as a $C$-correlator, \Eq{eq:genccorr}, with a particular choice of angular weighting function:
\begin{equation}
\label{eq:ecfs}
f_N^{(\beta)}(\{\theta_{ij}\})= \prod_{i<j} \theta^\beta_{ij},
\end{equation}
where $\theta_{ij}=(\Delta y_{ij}^2+\Delta\phi_{ij}^2)^{1/2}$.
In terms of multigraphs, the ECFs correspond to complete graphs on $N$ vertices:
\begin{equation}\label{eq:ecfgraphs}
e_2^{(\beta)} = \begin{gathered}
\includegraphics[scale=.2]{graphs/2_1_1}
\end{gathered}\,,\hspace{.5in}
e_3^{(\beta)} = \begin{gathered}
\includegraphics[scale=.2]{graphs/3_3_1}
\end{gathered}\,,\hspace{.5in}
e_4^{(\beta)} = \begin{gathered}
\includegraphics[scale=.2]{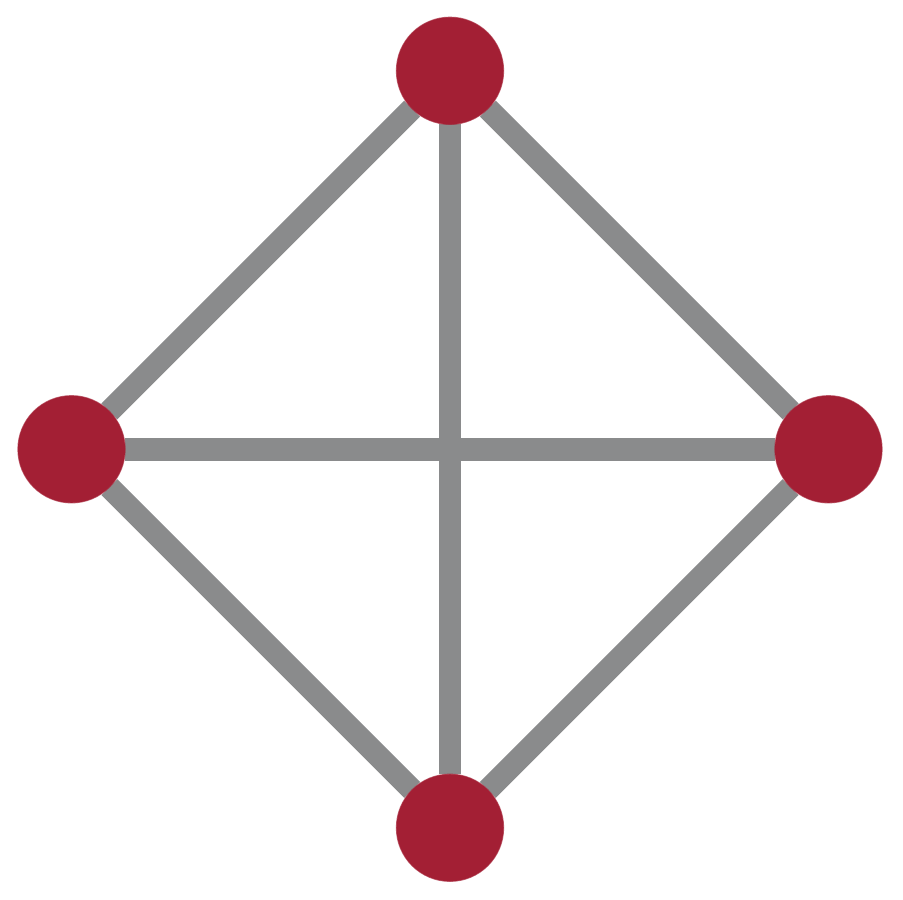}
\end{gathered},
\end{equation}
which are \Bs using the measure in \Eq{eq:hadronicmeasure} with exponent $\beta$.

The ECFs have since been expanded to a more flexible set of observables referred to as the ECFGs~\cite{Moult:2016cvt}. 
Letting $\min^{(m)}$ indicate the $m$-th smallest element in a set, the ECFGs are also $C$-correlators with angular weighting function:
\begin{equation}\label{eq:ecfgs}
\,_vf_N^{(\beta)}(\{\theta_{ij}\}) =  \prod_{m=1}^v \min_{i<j}^{(m)}\{ \theta^\beta_{ij}\}.
\end{equation}
The ECFGs do not have an exact multigraph correspondence due to the presence of the min function, but are evidently closely related to the \Bs since they share a common energy structure. 
The min function itself can be approximated by polynomials in its arguments, which induces an approximating series for the ECFGs in terms of \Bs when plugged into the common energy structure.

Both the \Bs and the ECFGs represent natural extensions of the ECFs but in different directions.
From our graph-theoretic perspective, the \Bs extend the ECFs to non-fully-connected graphs.
The ECFGs extend the scaling properties of the ECFs into observables with independent energy and angular scalings.
As discussed in \Sec{subsec:goingbeyond}, there are angular structures possible in the \Bs that are not possible in the ECFGs.
As with any jet substructure analysis, the choice of which set of observables to use depends on the physics of interest, with the \Bs designed for linear completeness and the ECFGs designed for nice power-counting properties.

\subsubsection{Angularities}
\label{sec:angularities}

Next, we consider the IRC-safe jet angularities~\cite{Larkoski:2014pca} (see also \Refs{Berger:2003iw,Almeida:2008yp,Ellis:2010rwa,Larkoski:2014uqa}) defined by:
\begin{equation}\lambda^{(\alpha)} = \sum_{i = 1}^M z_{i} \, \theta_{i}^\alpha,\end{equation}
where $\alpha>0$ is an angular exponent and $\theta_i$ denotes the distance of particle $i$ to the jet axis.
For concreteness and analytic tractability, we take the jet axis to be the $p_T$-weighted centroid in $( y,\phi)$-space, such that the jet axis is located at:
\begin{equation}
\label{eq:jetaxis}
y_J = \sum_{j=1}^M z_j  y_j, \qquad \phi_J = \sum_{j=1}^M z_j \phi_j.
\end{equation}
With this, the angularities can be expressed as:
\begin{align}
\lambda^{(\alpha)} &=\sum_{i_1=1}^Mz_{i_1}\left(( y_{i_1}- y_J)^2+(\phi_{i_1}-\phi_J)^2\right)^{\alpha/2}\nonumber
\\&=\sum_{i_1=1}^M z_{i_1} \left(\left(\sum_{i_2 = 1}^Mz_{i_2}\Delta y_{i_1i_2}\right)^2 + \left(\sum_{i_2 = 1}^Mz_{i_2}\Delta\phi_{i_1i_2}\right)^2\right)^{\alpha/2}\nonumber
\\&= \sum_{i_1=1}^M z_{i_1} \left(\sum_{i_2 = 1}^Mz_{i_2}\theta_{i_1i_2}^2 - \frac12 \sum_{i_2 = 1}^M\sum_{i_3 = 1}^M z_{i_2}z_{i_3} \theta_{i_2i_3}^2\right)^{\alpha/2}.\label{eq:angthetaijs}
\end{align}

For even $\alpha$, the parenthetical in \Eq{eq:angthetaijs} can be expanded and identified to be a linear combination of \Bs with $N = \alpha$ and $d=\alpha$ (see \Ref{GurAri:2011vx} for a related discussion).
For $\alpha = 2$, \Eq{eq:angthetaijs} implies:
\begin{equation}\label{eq:lam2}
\lambda^{(2)} = \frac12\sum_{i\in J} \sum_{j\in J} z_{i}z_{j} \theta_{ij}^2 =\frac12\times \begin{gathered}
\includegraphics[scale=.2]{graphs/2_2_1}
\end{gathered}.
\end{equation}
For $\alpha = 4$ and $\alpha = 6$,  \Eq{eq:angthetaijs} implies:
\begin{align}\label{eq:lam4}
\lambda^{(4)} &=
\begin{gathered}
\includegraphics[scale=.2]{graphs/3_4_1}
\end{gathered}
-\frac34 \times
\begin{gathered}
\includegraphics[scale=.2]{graphs/2_2_1}\hspace{0mm}
\includegraphics[scale=.2]{graphs/2_2_1}
\end{gathered}, 
\\
\label{eq:lam6}
\lambda^{(6)}&=
\begin{gathered}
\includegraphics[scale=.25]{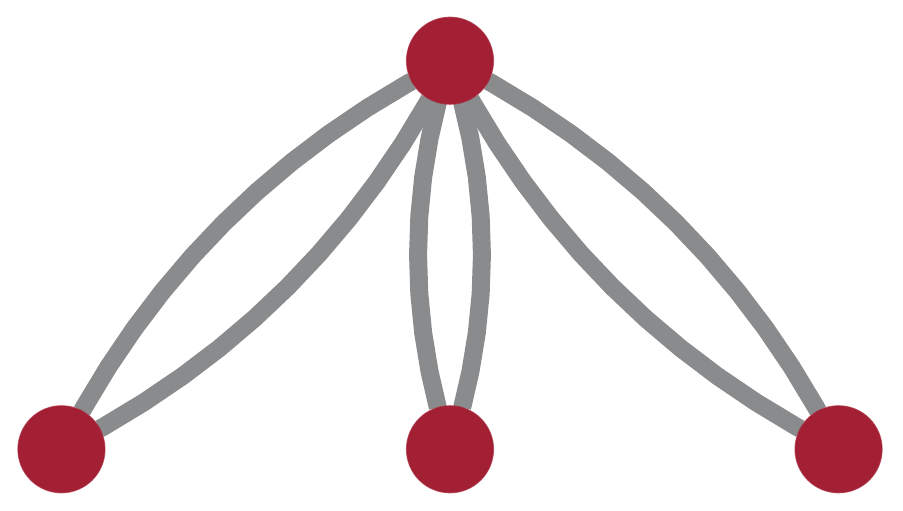}
\end{gathered}
-\frac32 \times
\begin{gathered}
\includegraphics[scale=.2]{graphs/3_4_1}
\end{gathered}
\hspace{0mm}
\begin{gathered}
\includegraphics[scale=.2]{graphs/2_2_1}
\end{gathered}\,
+\frac58\times
\begin{gathered}
\includegraphics[scale=.2]{graphs/2_2_1}\hspace{0mm}
\includegraphics[scale=.2]{graphs/2_2_1}\hspace{0mm}
\includegraphics[scale=.2]{graphs/2_2_1}
\end{gathered}\,.
\end{align}
This can be continued for arbitrarily high, even $\alpha$.
Thus, the even $\alpha$ angularities are exact, non-trivial linear combinations of \Bs, illustrating the close connections between the two classes of observables.
While angularities with odd or non-integer $\alpha$ do not have the same analytic tractability, the specific case of $\alpha=1/2$ is shown to be numerically well approximated by \Bs in \Sec{sec:regression}.

\subsubsection{Geometric moment tensors}
\label{sec:geometricmomenttensors}

Next, we consider observables based on the two-dimensional geometric moment tensor of the energy distribution in the $(y,\phi)$-plane~\cite{GurAri:2011vx,Gallicchio:2012ez}: 
\begin{align}
{\bf C}& = \sum_{i\in J} z_{i} \begin{pmatrix} \Delta y_i^2 & \Delta y_i\Delta\phi_i \\ \Delta\phi_i\Delta y_i & \Delta\phi_i^2 \end{pmatrix}= \begin{pmatrix}\frac12\sum_{i,j} z_{i}z_{j} \Delta y_{ij}^2 &\frac12\sum_{i,j} z_{i}z_{j} \Delta y_{ij}\Delta\phi_{ij} \\ \frac12\sum_{i,j} z_{i}z_{j} \Delta\phi_{ij}\Delta y_{ij} &\frac12\sum_{i,j} z_{i}z_{j} \Delta\phi_{ij}^2\end{pmatrix},
\label{eq:2dmomtens}
\end{align}
where the distances are measured with respect to the $p_T$-weighted centroid axis $(y_J,\phi_J)$ from \Eq{eq:jetaxis}.
Useful observables can be constructed from the trace and determinant of {\bf C}, such as planar flow $\text{Pf} = 4 \det{\bf C}/(\tr {\bf C})^2$~\cite{Thaler:2008ju,Almeida:2008yp}, which is a ratio of two IRC-safe observables.

We see that \Eq{eq:2dmomtens} is exactly a matrix of EFPs with $N=2$ and the two-dimensional hadronic measure from \Eq{eq:2Dhadronicmeasure}.
The trace $\tr {\bf C}$ and determinant $\det{\bf C}$ have the rotational symmetry in the $(y,\phi)$-plane of the default hadronic measure from \Eq{eq:hadronicmeasure}, allowing them to be written as linear combinations of EFPs with that measure:
\begin{align}\label{eq:momentpoly}
\tr \bf C&=\frac12\times
\begin{gathered}
\includegraphics[scale=.2]{graphs/2_2_1}
\end{gathered}\,,&\hspace{-.5in}
4\,\det{\bf C}&= 
\begin{gathered}
\includegraphics[scale=.2]{graphs/3_4_1}
\end{gathered}\,
-\frac12\times
\begin{gathered}
\includegraphics[scale=.2]{graphs/2_4_1}
\end{gathered}\,.
\end{align}

In \Ref{GurAri:2011vx}, a general class of energy flow moments was explored and categorized, with the goal of classifying observables according to their energy flow distributions.
These energy flow moments are defined with respect to a specified jet axis:
\begin{equation}\label{eq:hightens}
I_{k_1\cdots k_N} \equiv \sum_{i=1}^M z_i \, x_{k_1}^{(i)}\cdots x_{k_N}^{(i)},
\end{equation}
where $k_i \in \{1,2\}$, $x_1^{(i)} = \Delta  y_i= y_i- y_J$ and $x_2^{(i)} = \Delta\phi_i=\phi_i-\phi_J$.
Using the $p_T$-weighted centroid axis, this is the natural generalization of \Eq{eq:2dmomtens}, with the special case of $I_{k_1k_2} = ({\bf C})_{k_1k_2}$.
By performing a similar analysis to the one used to arrive at \Eq{eq:momentpoly}, one can show that any scalar constructed by contracting the indices of a product of objects in \Eq{eq:hightens} can be decomposed into an exact linear combination of \Bs.

\subsection{Going beyond existing substructure observables}
\label{subsec:goingbeyond}

Because the \Bs are $C$-correlators that span the space of IRC-safe observables, their angular structures should encompass all possible behaviors of $C$-correlators.
By contrast, the ECFs and ECFGs mentioned in \Sec{sec:equivECF} have more restricted behaviors, and it is illuminating to understand the new kinds of structures present in the \Bs.

Without loss of generality, the angular weighting function $f_N$ in \Eq{eq:genccorr} can be taken to be a symmetric function of the particle directions $\hat{p}_i^\mu$ due to the symmetrization provided by the sum structure (see \Eq{eq:symf} below).
The ECFs and ECFGs exhibit a stronger symmetry, though, since the angular functions in \Eqs{eq:ecfs}{eq:ecfgs} are invariant under the swapping any two pairwise angles $\theta_{ij}$.
This symmetry is manifested in the ECFs multigraphs in \Eq{eq:ecfgraphs} by the fact that all pairs of indices are connected by the same number of edges.

We can easily see that the pairwise swap symmetry of the ECFs is stronger than the full permutation symmetry of the \Bs: the group of permutations of the angular distances $\theta_{ij}$ has $\binom{N}{2}!$ elements, whereas the group of permutations of the indices $\{i_a\}$ has $N!$ elements.
An example of an \B that does not satisfy the stronger symmetry is the following $N=4$ graph:
\begin{equation}\label{eq:thebirdfoot}
\begin{gathered}
\includegraphics[scale=0.25]{graphs/4_3_1}
\end{gathered}
= \sum_{i_1=1}^M\sum_{i_2=1}^M\sum_{i_3=1}^M \sum_{i_4=1}^M z_{i_1}z_{i_2}z_{i_3} z_{i_4} \theta_{i_1i_2}\theta_{i_1 i_3}\theta_{i_1i_4}.
\end{equation}
The angular weighting function of the \B in \Eq{eq:thebirdfoot} is symmetric under the $4!$ permutations in the indices (vertices) $i_a \to i_{\sigma(a)}$ but not under the exchange of pairwise angles (edges) $\theta_{i_1 i_3} \to \theta_{i_2 i_3}$ which would result in a different \B, namely:
\begin{equation}
\begin{gathered}
\includegraphics[scale=0.2,angle=90]{graphs/4_3_2}
\end{gathered}
\neq
\begin{gathered}
\includegraphics[scale=0.25,angle=0]{graphs/4_3_1}
\end{gathered}.
\end{equation}

Another feature of the ECFs and ECFGs is that their angular weighting function $f_N$ vanishes whenever two of its arguments become collinear.
Indeed, \Ref{Moult:2016cvt} made the erroneous claim that this vanishing behavior was required by collinear safety.\footnote{If the sums are taken over distinct $N$-tuples as in \Ref{Moult:2016cvt}, then the angular function does have to vanish on collinearity for C safety. In general, non-collinearly-vanishing angular functions are C safe if the sum is taken over all $N$-tuples of particles, including sets with repeated indices.}
Instead, the argument in \Sec{sec:collinearsafety} shows this not to be the case, and observables defined by \Eq{eq:genccorr} are IRC safe for any sufficiently smooth and non-singular $f_N$.
An example of an \B that does not necessarily vanish when two of its arguments become collinear is the following $N=3$ graph:
\begin{equation}\label{eq:thewedge}
\begin{gathered}
\includegraphics[scale=0.2]{graphs/3_2_1}
\end{gathered}
= \sum_{i_1=1}^M\sum_{i_2=1}^M\sum_{i_3=1}^M z_{i_1}z_{i_2}z_{i_3}  \theta_{i_1i_2}\theta_{i_1 i_3},
\end{equation}
which does not vanish when $\hat p_{i_2}^\mu\to\hat p_{i_3}^\mu$.  
More generally, any non-fully-connected graph will not vanish in every collinear limit, but the corresponding \B will still be collinear safe.

By relaxing the restrictions on the angular weighting function $f_N$ to those minimally required by IRC safety, the energy flow basis captures all topological structures which can possibly appear in a $C$-correlator, beyond just the ones described by ECFs and ECFGs.

\section{Constructing a linear basis of IRC-safe observables}
\label{sec:basis}

Having introduced the \Bs, we now give a detailed argument that they linearly span the space of IRC-safe observables.
Due to its more technical nature, this section can be omitted on a first reading, and the reader may skip to \Sec{sec:complexity}.
\Refs{Tkachov:1995kk,Sveshnikov:1995vi,Cherzor:1997ak,Tkachov:1999py} argue that, from the point of view of quantum field theory, all IRC-safe information about the jet structure should be contained in the $C$-correlators.
In \Sec{sec:Eexpansion}, we independently arrive at the same conclusion by a direct application of IRC safety.
We then go on in \Sec{sec:angleexpansion} to expand the angular structure of the $C$-correlators to find a correspondence between multigraphs and \Bs.

An IRC-safe observable $\mathcal S$ depends only on the unordered set of particle four-momenta $\{p^\mu_i\}_{i=1}^M$, and not any non-kinematic quantum numbers.
An observable defined on $\{p^\mu_i\}_{i=1}^M$ can alternatively be thought of as a collection of functions, one for each number of particles $M$.
IRC safety then imposes constraints on this collection and thereby induces relations between the functions.
The requirement of IR safety imposes the constraint~\cite{sterman1995handbook}:
\begin{align}\label{eq:IRsafety}
\mathcal S(\{p_1^\mu,\ldots,p_M^\mu\}) & = \lim_{\varepsilon \to 0}\mathcal S(\{p_1^\mu,\ldots, p_M^\mu, \varepsilon \, p_{M+1}^\mu\}),&&\hspace{-.5in}\forall p_{M+1}^\mu,
\end{align}
while the requirement of C safety imposes the constraint:
\begin{align}
\label{eq:Csafety}
\mathcal S(\{p_1^\mu,\ldots,p_M^\mu\}) & = \mathcal S(\{p_1^\mu,\ldots,(1 - \lambda) p_M^\mu, \lambda p_M^\mu\}),&&\hspace{-.5in}\forall \lambda\in[0,1].
\end{align}
\Eq{eq:IRsafety} says that the observable is unchanged by the addition of infinitesimally soft particles, while \Eq{eq:Csafety} guarantees that the observable is insensitive to a collinear splitting of particles.

As written, only particle $M$ is affected in \Eq{eq:Csafety}.
The indexing used to identify particles, however, is arbitrary and these properties continue to hold when the particles are reindexed.
This \emph{particle relabeling symmetry} is not an additional constraint that is imposed but rather a consequence of assigning labels to an unordered set of particles.
These three restrictions---IR safety, C safety, and particle relabeling symmetry---are necessary and sufficient conditions for obtaining the energy flow basis.

Throughout this analysis, particles are treated as massless, $p^\mu_i = E_i \,\hat p_i^\mu$, where $\hat p_i^\mu$ is purely geometric.
Note that we could replace $E_i$ with any quantity linearly dependent on energy, such as the transverse momentum $p_{T,i}$, which corresponds to making a different choice of measure in \Sec{sec:measures}.

\subsection{Expansion in energy}
\label{sec:Eexpansion}

Consider an arbitrary IRC-safe observable $\mathcal{S}$, expanded in terms of the particle energies.
If the observable has a simple analytic dependence on the energies, then the usual Taylor expansion can be used:
\begin{align}\label{eq:taylorexpand}
\mathcal S=\left.\mathcal S_M\right|_{E=0}+\sum_{i_1=1}^ME_i\left.\pd{\mathcal S_M}{E_{i_1}}\right|_{E=0}+\frac{1}{2}\sum_{i_1=1}^M\sum_{i_2=1}^ME_{i_1}E_{i_2}\left.\pd{^2\mathcal S_M}{E_{i_1}\partial E_{i_2}}\right|_{E=0}+\cdots,
\end{align}
where $M$ is the particle multiplicity and the derivatives are evaluated at vanishing energies. 
An example of this is the jet mass from \Eq{eq:massexp}:
\begin{align}\label{eq:jetmass}
m_J^2&=\sum_{i=1}^M\sum_{j=1}^M \eta_{\mu\nu}p_i^\mu p_{j}^\nu=\sum_{i=1}^M\sum_{j=1}^ME_iE_j\eta_{\mu\nu}\, \hat p_i^\mu\hat p_j^\nu,
\end{align}
where $\eta_{\mu \nu}$ is the Minkowski metric.  
This expression is already in the form of \Eq{eq:taylorexpand} with:
\begin{align}
\pd{^2m_J^2}{E_{i_1}\partial E_{i_2}}=2\eta_{\mu\nu}\hat p_{i_1}^\mu\hat p_{i_2}^\nu,
\end{align}
and all other Taylor coefficients zero.
See \Sec{sec:jetobs} for additional examples of observables with explicit formulas for which \Eq{eq:taylorexpand} can be applied.

For some observables, though, a Taylor expansion may be difficult or impossible to obtain. 
The simplest example is a non-differentiable observable.
This is the case for $m_J$ (rather than $m_J^2$); the presence of the square root spoils the existence of a Taylor expansion, but the square root can be nonetheless approximated by polynomials arbitrarily well in a bounded interval.
A more complicated case is if the observable is defined in terms of an algorithm, such as a groomed jet mass~\cite{Butterworth:2008iy,Ellis:2009su,Ellis:2009me,Krohn:2009th,Dasgupta:2013ihk,Larkoski:2014wba}, and an explicit formula in terms of particle four-momenta would not be practical to differentiate or write down.
Similarly, the observable could be a non-obvious function of the particles, i.e.\ the optimal observable to accomplish some task.  

In cases without a Taylor expansion, the Stone-Weierstrass theorem~\cite{stone1948generalized} still guarantees that the observable can be approximated over some bounded energy range by polynomials in the energies.\footnote{A version of this theorem that suffices for our purposes can be phrased as follows: for any continuous, real-valued function $f$ defined on a compact subset  $X\subset\mathbb R^n$, for all $\epsilon>0$ there exists a polynomial $p$ of finite degree at most $N_{\rm max}$ such that $|p(\mathbf x)-f(\mathbf x)|<\epsilon$ for all $\mathbf x\in X$. Conceptually, this theorem is used to approximate any continuous function on a bounded region by a polynomial.} 
We write down such an expansion by considering all possible polynomials in the energies and multiplying each one by a different geometric function.
Combining all terms of degree $N$ into $\mathcal C_N$, the expansion is:
\begin{align}\label{eq:SWexpand}
\mathcal S&\simeq\sum_{N=0}^{N_{\rm max}}\mathcal C_N,\qquad
\mathcal C_N\equiv\sum_{i_1=1}^M \cdots \sum_{i_N=1}^M C^{(M)}_{i_1 \cdots i_N}(\hat p_1^\mu,\ldots, \hat p_M^\mu)\prod_{j=1}^N E_{i_j},
\end{align}
where $C_{i_1\cdots i_n}^{(M)}(\hat p_1^\mu,\ldots,\hat p_M^\mu)$ are geometric angular functions, which depend on the indices of the energy factors $i_1\cdots i_n$ and could in general be different for different multiplicities $M$.
The Stone-Weierstrass theorem guarantees that there is a maximum degree $N_{\rm max}$ in this energy expansion for any given desired accuracy, but places no further restrictions on the $\mathcal C_N$.

To derive constraints on these angular functions $C^{(M)}_{i_1 \cdots i_N}$, we impose the three key properties of IR safety in \Sec{sec:irsafety}, particle relabeling invariance in \Sec{sec:relabelsym}, and C safety in \Sec{sec:collinearsafety}, which we summarize in \Sec{subsec:final}.
In applying these properties, we will often use the fact that when setting two expressions for the observable $\mathcal S$ equal to each other, we can read off term-by-term equality by treating the particle energies as independent quantities:
\begin{equation}\label{eq:equality}
\mathcal S=\mathcal S'\quad\implies\quad \mathcal C_N=\mathcal C'_N,\quad\forall N\le N_{\rm max}.
\end{equation}
Note that the sum structure in \Eq{eq:SWexpand} implies that, without loss of generality, the angular functions can be taken to depend only on the labels $i_1,\ldots,i_N$ as an unordered set.

\subsubsection{Infrared safety}
\label{sec:irsafety}

IR safety constrains the angular functions appearing in the expansion of \Eq{eq:SWexpand} in two ways: by restricting which particle directions contribute to a particular term in the sum and by relating angular functions of different multiplicities.

First, consider a particular angular function, $C^{(M)}_{i_1\cdots i_N}$ in \Eq{eq:SWexpand}, and some particle $j\not\in\{ i_1,\ldots, i_N\}$. Consider particle $j$ in the soft limit: if $C^{(M)}_{i_1\cdots i_N}$ depends on $\hat p_j^\mu$ in any way, then IR safety is violated because $E_j$ does not appear in the product of energies but the value of the observable changes as the direction of $j$ is changed.
Hence, IR safety imposes the requirement that
\begin{align}\label{eq:IRresult1}
C^{(M)}_{i_1\cdots i_N}(\hat p_1^\mu,\ldots,\hat p_M^\mu)=C^{(M)}_{i_1\cdots i_N}(\hat p_{i_1}^\mu,\ldots,\hat p_{i_N}^\mu),
\end{align}
namely the indices of the arguments must match those of the angular function.
Note that we must always write $C^{(M)}_{i_1\cdots i_N}$ with $N$ arguments, even if some are equal due to indices coinciding.

Next, consider two polynomial approximations of the same observable: one as a function of $M$ particles and the other as a function of $M+1$ particles.
In the soft limit of particle $M+1$, $E_{M+1}\to 0$, the IR safety of $\mathcal S$, written formally in \Eq{eq:IRsafety}, guarantees that the function of $M+1$ particles approaches the function of $M$ particles. In terms of the corresponding polynomial approximations, we have that:
\begin{align}\label{eq:IRM+1}
\sum_{i_1=1}^{M+1} \hspace{-.05in}\cdots \hspace{-.05in}\sum_{i_N=1}^{M+1}  C^{(M+1)}_{i_1 \cdots i_N}&(\hat p_{i_1}^\mu,\ldots, \hat p_{i_N}^\mu) \prod_{j=1}^N E_{i_j} \\& =\sum_{i_1=1}^{M} \hspace{-.05in}\cdots \hspace{-.05in}\sum_{i_N=1}^{M} C^{(M)}_{i_1 \cdots i_N} (\hat p_{i_1}^\mu,\ldots, \hat p_{i_N}^\mu)\prod_{j=1}^N E_{i_j}+\mathcal O(E_{M+1}).
\end{align}

We see from \Eq{eq:IRM+1} that the same angular coefficients from the polynomial approximation of the function of $M+1$ particles can be validly chosen for the approximation of the function of $M$ particles, with the following equality of angular functions:
\begin{align}\label{eq:IRresult2}
C^{(M+1)}_{i_1 \cdots i_N} (\hat p_{i_1}^\mu,\ldots,\hat p_{i_N}^\mu)=C^{(M)}_{i_1\cdots i_N} (\hat p_{i_1}^\mu, \ldots, \hat p_{i_N}^\mu)\equiv C_{i_1\cdots i_N}(\hat p_{i_1}^\mu,\ldots,\hat p_{i_N}^\mu),
\end{align}
which says that the multiplicity label on the angular functions can be dropped.

As a result of enforcing IR safety, the dependence of the angular functions on multiplicity has been eliminated, as well as the dependence of a given angular function on any particles with indices not appearing in its subscripts.

\subsubsection{Particle relabeling symmetry }
\label{sec:relabelsym}

Now, using particle relabeling symmetry, for all $\sigma\in S_M$, where $S_M$ is the group of permutations of $M$ objects, we have that $\mathcal C_N$ is unchanged by the replacement $E_{i_j} \to E_{\sigma(i_j)}$ and $\hat p^\mu_{i_j} \to \hat p^\mu_{\sigma(i_j)}$.
With the angular functions as constrained by IR safety, the particle relabeling invariance of $\mathcal C_N$ can be written as:
\begin{align}\label{eq:updatedCn}
\mathcal C_N&=\sum_{i_1=1}^M\cdots\sum_{i_N=1}^MC_{i_1\cdots i_N}(\hat p_{i_1}^\mu,\ldots,\hat p_{i_N}^\mu)\prod_{j=1}^NE_{i_j}\\
&=\sum_{i_1=1}^M\cdots\sum_{i_N=1}^MC_{i_1\cdots i_N}(\hat p_{\sigma(i_1)}^\mu,\ldots, \hat p_{\sigma(i_N)}^\mu)\prod_{j=1}^NE_{\sigma(i_j)}\nonumber
\\&=\sum_{i_1=1}^M\cdots\sum_{i_N=1}^MC_{\sigma^{-1}(i_1)\cdots \sigma^{-1}(i_N)}(\hat p_{i_1}^\mu ,\ldots,\hat p_{i_N}^\mu)\prod_{j=1}^NE_{i_j},\label{eq:relabelsym1}
\end{align}
where the sums were reindexed according to $\sigma^{-1}$. In particular, from \Eq{eq:relabelsym1}, we have for any $\sigma \in S_M$ that:
\begin{align}\label{eq:relabelsym2}
C_{i_1\cdots i_N}(\hat p_{i_1}^\mu,\ldots,\hat p_{i_N}^\mu)=C_{\sigma(i_1)\cdots \sigma(i_N)}(\hat p_{i_1}^\mu,\ldots, \hat p_{i_N}^\mu).
\end{align}
\Eq{eq:relabelsym2} allows us to permute the indices of $C_{i_1\cdots i_N}$ within $S_M$, equating previously unrelated angular functions.

As written, $C_{i_1\cdots i_N}$ is not necessarily symmetric in its arguments.
Without loss of generality, though, we can symmetrize $C_{i_1\cdots i_N}$ without changing the value of $\mathcal C_N$ as follows:
\begin{align}\label{eq:symf}
\mathcal C_N&=\sum_{i_1=1}^M\cdots\sum_{i_N=1}^M\underbrace{\frac{1}{N!}\sum_{\sigma\in S_N}C_{i_1\cdots i_N}(\hat p_{\sigma(i_1)}^\mu,\ldots,\hat p_{\sigma(i_N)}^\mu)}_{C'_{i_1\cdots i_N}(\hat p_{i_1}^\mu,\ldots,\hat p_{i_N}^\mu)}\prod_{j=1}^NE_{i_j},
\end{align}
where $C'_{i_1\cdots i_N}$ is now symmetric in its arguments.
We assume in the next step of the derivation that the angular weighting functions are symmetric in their arguments.

\subsubsection{Collinear safety}
\label{sec:collinearsafety}

The key requirement for restricting the form of $C_{i_1\cdots i_N}$ is C safety.
If the angular weighting function(s) were required to vanish whenever two of the inputs were collinear, then the observable would be manifestly C safe~(see e.g.\ \cite{Moult:2016cvt}); this is a sufficient condition for C safety but not a necessary one.
More generally, one can have non-zero angular functions of $N$ arguments even when subsets of the arguments are collinear.

Using the IR safety argument of \Eq{eq:IRresult2} and the particle relabeling symmetry of \Eq{eq:relabelsym2}, we can relate any angular function $C_{i_1\cdots i_N}$ to one of the following:
\begin{align}\label{eq:Cexample}
&C_{123\cdots N}(\hat p_{i_1}^\mu,\hat p_{i_2}^\mu,\hat p_{i_3}^\mu,\ldots,\hat p_{i_N}^\mu),\nonumber\\
&C_{1123\cdots (N-1)}(\hat p_{i_1}^\mu,\hat p_{i_1}^\mu,\hat p_{i_2}^\mu,\hat p_{i_3}^\mu,\ldots,\hat p_{i_{N-1}}^\mu),\nonumber\\
&C_{112234\cdots (N-2)}(\hat p_{i_1}^\mu,\hat p_{i_1}^\mu,\hat p_{i_2}^\mu,\hat p_{i_2}^\mu,\hat p_{i_3}^\mu,\hat p_{i_4}^\mu,\ldots,\hat p_{i_{N-2}}^\mu),\nonumber\\
&C_{1112234\cdots(N-3)}(\hat p_{i_1}^\mu,\hat p_{i_1}^\mu,\hat p_{i_1}^\mu,\hat p_{i_2}^\mu,\hat p_{i_2}^\mu,\hat p_{i_3}^\mu,\hat p_{i_4}^\mu,\ldots,\hat p_{i_{N-3}}^\mu),\nonumber\\
&\vdots\nonumber\\
&C_{11\cdots 1}(\hat p_{i_1}^\mu,\hat p_{i_1}^\mu,\ldots,\hat p_{i_1}^\mu),
\end{align}
where there is one of these ``standard'' angular functions for each integer partition of $N$.
In particular, the length of the integer partition is how many unique indices appear in the subscript and the values of the partition indicate how many times each index is repeated.

The role of C safety is to impose relationships between these standard angular functions, eventually showing that the only required function is $C_{123\cdots N}$.
Intuitively, this means that as any set of particles become collinear, the angular dependence is that of collinear limit of $N$ arbitrary directions.
The proof that this follows from C safety, however, is the most technically involved step of this derivation.

The requirement of C safety in \Eq{eq:Csafety} implies that $\mathcal S$ is unchanged whether one considers $\{E_i,\hat p_i^\mu\}_{i=1}^M$ or the same particles with a collinear splitting of the first particle, $\{\tilde E_i, \hat p_i^\mu\}_{i=0}^M$, where:
\begin{equation}
\label{eq:csplit}
\tilde E_0=(1- \lambda) E_1,\qquad \tilde E_1 = \lambda E_1,\qquad \hat p_0^\mu=\hat p_1^\mu,\qquad \tilde E_i = E_i,
\end{equation}
for all $\lambda\in[0,1]$ and $i>1$.  Rewriting \Eq{eq:updatedCn}, we can explicitly separate out the terms of the sums involving $k$ collinearly split indices $\{0,1\}$:
\begin{align}
\mathcal C_N&=\sum_{i_1=0}^M\cdots\sum_{i_N=0}^MC_{i_1\cdots i_N}(\hat p_{i_1}^\mu,\ldots,\hat p_{i_N}^\mu)\prod_{j=1}^N\tilde E_{i_j}\label{eq:step1}\\
&=\sum_{k=0}^N\binom{N}{k}\sum_{i_1=0}^1\cdots\sum_{i_k=0}^1\sum_{i_{k+1}=2}^M\cdots\sum_{i_N=2}^MC_{i_1\cdots i_N}(\hat p_{i_1}^\mu,\ldots,\hat p_{i_N}^\mu)\prod_{j=1}^N\tilde E_{i_j}, \label{eq:step2}
\end{align}
where in going to this last expression, we have used the symmetry of $C_{i_1\cdots i_N}$ in its arguments and accounted for the degeneracy of such terms using the binomial factor $\binom{N}{k}$.
We then insert the collinear splitting kinematics of \Eq{eq:csplit} into \Eq{eq:step2},
\begin{align}
\mathcal C_N&=\sum_{k=0}^N\binom{N}{k}\sum_{i_1=0}^1\cdots\sum_{i_k=0}^1\lambda^{\sum_{a=1}^ki_a}(1-\lambda)^{k-\sum_{a=1}^ki_a}E_1^k\label{eq:step3}\\
&\hspace{.5in}\times\sum_{i_{k+1}=2}^M\cdots\sum_{i_N=2}^MC_{i_1\cdots i_N}(\hat p_{i_1}^\mu,\ldots,\hat p_{i_N}^\mu)\prod_{j=k+1}^N E_{i_j} \nonumber\\
&=\sum_{k=0}^N\binom{N}{k}\sum_{\ell=0}^k\binom{k}{\ell}\lambda^\ell(1-\lambda)^{k-\ell}E_1^k\label{eq:collinear}\\
&\hspace{.5in}\times\sum_{i_{k+1}=2}^M\cdots\sum_{i_N=2}^M C_{\underbrace{{}_{0\cdots0}}_{\ell}\underbrace{{}_{1\cdots1}}_{k-\ell}i_{k+1}\cdots i_n}(\hat p_{1}^\mu,\ldots,\hat p_{1}^\mu,\hat p_{i_{k+1}}^\mu,\ldots,\hat p_{i_N}^\mu)\prod_{j=k+1}^N E_{i_j},\nonumber
\end{align}
where in going to this last expression, we have used the particle relabeling symmetry of \Eq{eq:relabelsym2} to sort the $\{0,1\}$ subscript indices of the angular functions.

The constraint of C safety says that \Eq{eq:collinear} is equal to \Eq{eq:updatedCn} on the non-collinearly split event.
To make this constraint more useful, we use the binomial theorem to write 1 in a suggestive way:
\begin{align}\label{eq:binomial1}
1=(\lambda+1-\lambda)^k=\sum_{\ell=0}^k\binom{k}{\ell}\lambda^\ell(1-\lambda)^{k-\ell},
\end{align}
and insert this expression into \Eq{eq:updatedCn}, separating out factors where $k$ of the indices are equal to 1:
\begin{align}\label{eq:Fnmanip}
\mathcal C_N&=\sum_{k=0}^N\binom{N}{k}\sum_{\ell=0}^k\binom{k}{\ell}\lambda^\ell(1-\lambda)^{k-\ell}E_1^k\\
&\nonumber\hspace{.5in}\times \sum_{i_{k+1}=2}^M\cdots\sum_{i_N=2}^MC_{\underbrace{{}_{1\cdots1}}_{k}i_{k+1}\cdots i_N}(\hat p_{1}^\mu,\ldots,\hat p_{1}^\mu,\hat p_{i_{k+1}}^\mu,\ldots,\hat p_{i_N}^\mu)\prod_{j=k+1}^NE_{i_j}.
\end{align}
Subtracting~\Eq{eq:Fnmanip} from~\Eq{eq:collinear} and treating the energies as independent quantities, the following constraint can be read off:
\begin{align}\label{eq:Cconstraint}
\sum_{\ell=0}^k\binom{k}{\ell}\lambda^\ell(1-\lambda)^{k-\ell}\left(C_{\underbrace{{}_{0\cdots0}}_{\ell}\underbrace{{}_{1\cdots1}}_{k-\ell}i_{k+1}\cdots i_N}-C_{\underbrace{{}_{1\cdots 1}}_{k}i_{k+1}\cdots i_N}\right)=0,
\end{align}
where the identical arguments of the angular functions are suppressed for compactness.  

We would like to obtain that the quantity in parentheses in \Eq{eq:Cconstraint} vanishes since the equation holds for all $\lambda$. 
To see this, suppose that the quantity in parentheses does not vanish, and let $\hat\ell$ be the smallest such $\ell$ where this happens.
Consider the regime $0<\lambda\ll1$: by the definition of $\hat\ell$, there are no $\mathcal O(\lambda^\ell)$ terms for $\ell<\hat \ell$ and thus the left-hand side of \Eq{eq:Cconstraint} is $\mathcal O(\lambda^{\hat\ell})\neq 0$, contradicting \Eq{eq:Cconstraint}. 
We thus obtain:
\begin{align}\label{eq:Cresult}
C_{\underbrace{{}_{0\cdots0}}_{\ell}\underbrace{{}_{1\cdots1}}_{k-\ell}i_{k+1}\cdots i_N}&(\hat p_{1}^\mu,\ldots,\hat p_{1}^\mu,\hat p_{i_{k+1}}^\mu,\ldots,\hat p_{i_n}^\mu) \\& =C_{\underbrace{{}_{1\cdots 1}}_{k}i_{k+1}\cdots i_N}(\hat p_{1}^\mu,\ldots,\hat p_{1}^\mu,\hat p_{i_{k+1}}^\mu,\ldots,\hat p_{i_N}^\mu),
\end{align}
for $0\le\ell\le k$.  Note that in this expression, the first $k$ arguments of the functions are identical.

The constraint in~\Eq{eq:Cresult} is very powerful, especially when combined with the relabeling symmetry of~\Eq{eq:relabelsym2}.
While we obtained \Eq{eq:Cresult} using the collinear limit, the particle direction $\hat{p}_{0}$ appears nowhere in this expression, so the $0$ subscript is simply an index on the angular function.
Therefore, when any $k$ arguments of one of the angular functions become collinear, any $\ell\le k$ of the corresponding subscript labels may be swapped out for values not appearing anywhere else in the indices.  
A concrete example of this is
\begin{align}\label{eq:relabelex}
C_{1123\cdots N-1}(\hat p_{i_1}^\mu,\hat p_{i_1}^\mu,\hat p_{i_2}^\mu,\ldots,\hat p_{i_{N-1}}^\mu)=C_{1234\cdots N}(\hat p_{i_1}^\mu,\hat p_{i_1}^\mu,\hat p_{i_2}^\mu,\ldots,\hat p_{i_{N-1}}^\mu),
\end{align}
where the $N$ index here plays the role of the $0$ index in \Eq{eq:Cresult}.
This then implies that all of the angular functions in~\Eq{eq:Cexample} can related to a single function:
\begin{align}\label{eq:fNresult}
C_{i_1\cdots i_N}(\hat p_{i_1}^\mu, \ldots, \hat p_{i_N}^\mu) = C_{123\cdots N}(\hat p_{i_1}^\mu,\ldots,\hat p_{i_N}^\mu) \equiv f_N(\hat p_{i_1}^\mu,\ldots,\hat p_{i_N}^\mu),
\end{align}
yielding the intuitive result that the angular dependence when some number of particles become collinear should follow from the collinear limit of $N$ arbitrary directions.

\subsubsection{A new derivation of $C$-correlators}
\label{subsec:final}

Finally, substituting \Eq{eq:fNresult} into \Eq{eq:updatedCn} implies that 
\begin{align}\label{eq:Eexpansionresult}
\mathcal S\simeq\sum_{N=0}^{N_{\rm max}}\mathcal C_N^{f_N},\qquad \mathcal C_N^{f_N}=\sum_{i_1=1}^M\cdots\sum_{i_N=1}^M E_{i_1}\cdots E_{i_N} f_N(\hat p_{i_1}^\mu,\ldots,\hat p_{i_N}^\mu),
\end{align}
where we recognize $\mathcal C_N^{f_N}$ as the $C$-correlators of \Eq{eq:genccorr}.
This expression says that an arbitrary IRC-safe observable can be approximated arbitrarily well by a linear combination of $C$-correlators.
In this way, we have given a new derivation that $C$-correlators linearly span the space of IRC-safe observables by directly imposing the constraints of IRC safety and particle relabeling symmetry on an arbitrary observable.

The argument presented here suffices to show the IRC-safety of the $C$-correlators with any continuous angular weighting function, even if it is not symmetric.
Though we used the symmetrization in \Eq{eq:symf} to aid the C-safety derivation in \Sec{sec:collinearsafety}, it is now perfectly valid to relax this constraint on $f_N$.
In particular, we can simply consider \Eq{eq:symf} applied in reverse and select a single term in the symmetrization sum to represent $f_N$.
Thus we are not constrained merely to symmetric $f_N$, which will be helpful in obtaining the \Bs.

\subsection{Expansion in geometry}
\label{sec:angleexpansion}

Having now established that the $C$-correlators linearly span the space of IRC-safe observables, we now expand the angular weighting function $f_N$ in \Eq{eq:Eexpansionresult} in terms of a discrete linear angular basis.\footnote{Our approach here turns out to be similar to the construction of kinematic polynomial rings for operator bases in \Ref{Henning:2017fpj}.}
By virtue of the sum structure of the $C$-correlators, this angular basis directly translates into a basis of IRC-safe observables, i.e.\ the energy flow basis.

Following the discussion in \Sec{sec:measures}, we take the angular function $f_N$ to depend only on the pairwise angular distances $\theta_{ij}$.
Note that the results of \Sec{sec:Eexpansion} continue to hold with pairwise angular distances in place of particle directions, as long as $\theta_{ij}$ is a dimensionless function of $\hat{p}_i^\mu$ and $\hat{p}_j^\mu$ with no residual dependence on energy.
Of course, this choice would not be valid for expanding IRC-safe observables that do not respect the symmetries implied by $\theta_{ij}$, such as trying to use the default hadronic measure in \Eq{eq:hadronicmeasure} for observables that depend on the overall jet rapidity.
In such cases, one can perform an expansion directly in the $\hat p_{i}^\mu$, though we will not pursue that here.

Expanding the angular function $f_N$ in terms of polynomials up to order $d_{\rm max}$ in the pairwise angular distances yields:
\begin{align}\label{eq:angexp}
f_N(\hat p_{i_1}^\mu,\ldots,\hat p_{i_N}^\mu)&\simeq\sum_{d=0}^{d_{\rm max}} \sum_{\mathcal M\in \Theta_d} b_{\mathcal M} \,\mathcal M,
\end{align}
where $\Theta_d$ is the set of monomials in $\{\theta_{ij}\,|\,i<j\in \{i_1,\ldots,i_N\}\}$ of degree $d$, $\mathcal M$ is one of these monomials, and the $b_{\mathcal M}$ are numerical coefficients.
While this is a perfectly valid expansion, it represents a vast overcounting of the number of potential angular structures.

Our goal is to substitute \Eq{eq:angexp} into the definition of a $C$-correlator in \Eq{eq:Eexpansionresult} and identify the \emph{unique} analytic structures that emerge.
Note that two monomials $\mathcal M_1,\mathcal M_2\in \Theta_d$ that are related by a permutation $\sigma\in S_N$ with action $\theta_{i_ai_b}\to\theta_{i_{\sigma(a)} i_{\sigma(b)}}$ give rise to identical $C$-correlators, $\mathcal C^{\mathcal M_1}=\mathcal C^{\mathcal M_2}$, as a result of the relabeling symmetry in \Sec{sec:relabelsym}.
Thus, we can greatly simplify the angular expansion by summing only over equivalence classes of monomials not related by permutations, which we write as $\Theta_d/S_N$.
Writing this out in terms of $\mathcal{E} \in\Theta_d/S_N$:
\begin{align}
\mathcal C_N^{f_N} & \simeq \sum_{i_1 = 1}^M \cdots \sum_{i_N = 1}^M E_{i_1}\cdots E_{i_N} \sum_{d=0}^{d_{\rm max}} \sum_{\mathcal E\in\Theta_d/S_N}\sum_{\mathcal M\in\mathcal E} \, b_{\mathcal M} \, \mathcal M\label{eq:sub2}\\
&=\sum_{d=0}^{d_{\rm max}}\sum_{\mathcal E\in\Theta_d/S_N}b_{\mathcal E}\sum_{i_1 = 1}^M \cdots \sum_{i_N = 1}^M E_{i_1}\cdots E_{i_N}\mathcal M_{\mathcal E}\label{eq:sub3},
\end{align}
where, by the relabeling symmetry, $\mathcal M_{\mathcal E}$ can be any representative monomial in the equivalence class $\mathcal{E}$, and the coefficient $b_{\mathcal E} = |\mathcal E | \, b_{\mathcal M}$ absorbs the size $|\mathcal E |$ of the equivalence class.

As described in \Sec{sec:efbasis}, the set of monomials $\Theta_d$ is in bijection with the set of multigraphs with $d$ edges and $N$ vertices, and the set of equivalence classes $\Theta_d/S_N$ is in bijection with the set of non-isomorphic multigraphs with $d$ edges and $N$ vertices. 
In particular, each edge $(k,\ell)$ in a multigraph $G$ corresponds to a factor of $\theta_{i_ki_\ell}$ in the monomial $\mathcal M_{\mathcal E}$:
\begin{equation}\label{eq:mon2graph}
\mathcal M_{\mathcal E}=\prod_{(k,\ell)\in G}\theta_{i_ki_\ell},
\end{equation}
where $G$ corresponds to the equivalence class $\mathcal E$. 
By substituting \Eq{eq:mon2graph} into \Eq{eq:sub3} and relabeling the coefficient $b_{\mathcal E}$ to $b_G$, we can identify the resulting analytic structures that linearly span the space of $C$-correlators as the (unnormalized) \Bs:
\begin{equation}\label{eq:efpspanCN}
\mathcal C_N^{f_N}\simeq \sum_{d=0}^{d_{\rm max}}\sum_{G\in\mathcal G_{N,d}}b_G\,\B_G,\qquad \B_G\equiv\sum_{i_1=1}^M\cdots\sum_{i_N=1}^ME_{i_1}\cdots E_{i_N}\prod_{(k,\ell)\in G}\theta_{i_ki_\ell},
\end{equation}
where $\mathcal G_{N,d}$ is the set of non-isomorphic multigraphs with $d$ edges on $N$ vertices.

In \Sec{subsec:final}, it was shown that the set of IRC-safe observables is linearly spanned by the set of $C$-correlators, summarized in \Eq{eq:Eexpansionresult}. 
In this section, we have shown in \Eq{eq:efpspanCN} that the $C$-correlators themselves are linearly spanned by the \Bs, whose angular structures are efficiently encoded by multigraphs. 
By linearity, the \Bs therefore form a complete linear basis for all IRC-safe observables, completing our argument.

\section{Computational complexity of the energy flow basis}
\label{sec:complexity}

Since we would like to apply the energy flow basis in the context of jet substructure, the efficient computation of \Bs is of great practical interest. 
Naively, calculating an \B whose graph has a large number of vertices requires a prohibitively large amount of computation time, especially as the number of particles in the jet grows large.
In practice, though, we can dramatically speed up the implementation of the \Bs by making use of the correspondence with multigraphs.
Code to calculate the \Bs using these methods is available through our {\tt \href{https://pkomiske.github.io/EnergyFlow}{EnergyFlow}} module.

\subsection{Algebraic structure}
\label{sec:algebraic}

The set of \Bs has a rich algebraic structure which will allow in some cases for faster computation. 
Firstly, they form a monoid (a group without inverses) under multiplication. 
In analogy with the natural numbers, the composite \Bs, those with disconnected multigraphs, can be expressed as a product of the prime \Bs corresponding to the connected components of a disconnected graph:
\begin{equation}\label{eq:efpprod}
\B_G=\prod_{g\in C(G)}\B_g,
\end{equation}
where $C(G)$ is the set of connected components of the multigraph $G$. 

As a concrete example of \Eq{eq:efpprod}, consider:
\begin{align}
\begin{gathered}
\includegraphics[scale=.2]{graphs/3_3_2}
\includegraphics[scale=.2]{graphs/2_4_1}
\end{gathered}
&= \left(\sum_{i_1=1}^M\sum_{i_1=1}^M\sum_{i_3=1}^Mz_{i_1}z_{i_2}z_{i_3}\theta_{i_1i_2}^2\theta_{i_2i_3}\right)\left(\sum_{i_4=1}^M\sum_{i_5=1}^Mz_{i_4}z_{i_5}\theta_{i_4i_5}^4\right).
\end{align}
Thus, we only need to perform summations for the computation of prime \Bs, with the composite ones given by \Eq{eq:efpprod}. 
Note that if one were combining \Bs with a nonlinear method, such as a neural network, the composite \Bs would not be needed as separate inputs since the model could in principle learn to compute them on its own. 
The composite \Bs are, however, required to have a linear basis and should be included when linear methods are employed, such as those in \Secs{sec:linreg}{sec:linclass}.

The relationship between prime and composite \Bs is just the simplest example of the algebraic structure of the energy flow basis. 
The \Bs depend on $M$ energies and $M\choose2$ pairwise angles, but there are only $3M-4$ degrees of freedom for the phase space of $M$ massless particles, leading generically to additional (linear) relations among the \Bs. 
Hence, the \Bs are an \emph{over}complete linear basis. 
We leave further analysis and exploration of these relations to future work, and simply remark here that linear methods continue to work even if there are redundancies in the basis elements.

\subsection{Dispelling the $\mathcal O(M^N)$ myth for $N$-particle correlators}
\label{sec:dispel}

\begin{table}[t]
\centering
\subfloat[]{\label{tab:numchi:a} 
\begin{tabular}{|c|cc||*{10}{r}|}
\hline
\multicolumn{3}{|c}{$d$} & \multicolumn{1}{c}{\bf1}&\multicolumn{1}{c}{\bf2}&\multicolumn{1}{c}{\bf3}&\multicolumn{1}{r}{\bf4}&\multicolumn{1}{c}{\bf5}&\multicolumn{1}{c}{\,\,\bf6}&\multicolumn{1}{c}{\bf7}&\multicolumn{1}{c}{\,\,\,\bf8}&\multicolumn{1}{c}{\bf9}&\multicolumn{1}{c|}{\bf10} \\ \hhline{:===:t:*{10}{=}:}
\multirow{5}{*}{Prime} & \multirow{4}{*}{$\chi$} 
  & 2 & 1 & 2 & 4 & 9 & 21 & 55 & 146 & 415 & 1\,212 & 3\,653   \\
&& 3 &    &    & 1 & 3 & 12 & 47 & 185 & 757 & 3\,181 & 13\,691 \\
&& 4 &    &    &    &    &      & 1   & 2     & 11   & 49       & 231       \\ 
&& 5 &    &    &    &    &      &      &        &        &            & 1           \\ 
\hhline{|~|--||*{10}{-}|}
&\multicolumn{2}{c||}{Total} & 1 & 2 & 5 & 12 & 33 & 103 & 333 & 1\,183 & 4\,442 & 17\,576 \\ 
\hhline{:===::*{10}{=}:}
\multirow{5}{*}{All} & \multirow{4}{*}{$\chi$}
  & 2 & 1 & 3 & 7 & 19 & 48 & 135 & 371 & 1\,077 & 3\,161 & 9\,539  \\
&& 3 &    &    & 1 & 4  & 18 & 76   & 312 & 1\,296 & 5\,447 & 23\,268 \\
&& 4 &    &    &    &     &      & 1     & 3     & 16       & 74       & 352       \\
&& 5 &    &    &    &     &      &        &        &            &            & 1           \\ 
\hhline{|~|--||*{10}{-}|}
&\multicolumn{2}{c||}{Total} & 1 & 3 & 8 & 23 & 66 & 212 & 686 & 2\,389 & 8\,682 & 33\,160 \\ 
\hhline{---||*{10}{-}}
\end{tabular}}
\\
\subfloat[]{\label{tab:numchi:b} 
\begin{tabular}{|cc||*{15}{r}|}
\hhline{--::*{15}{-}}
&& \multicolumn{15}{c|}{$N$} \\ 
$d$ & $\chi$
\unskip\textcolor{white!5}{\makebox[0pt]{\smash{\rule[27pt]{20pt}{3pt}}}}  
& \bf2 & \bf3 & \bf4 & \bf5 & \bf6 & \bf7 & \bf8 & \bf9 & \bf10 & \bf11 & \bf12 & \bf13 & \bf14 &\hspace{-2mm} \bf15 \hspace{-2mm} & \hspace{-2mm} \bf16 \\ 
\hhline{|--|:*{15}{=}:}
\multirow{1}{*}{\bf1} & 2 & 1 & & & & & & & & & & & & & & \\ \hhline{|-|-||*{15}{-}|}
\multirow{1}{*}{\bf2} & 2 & 1 & 1 & 1 & & & & & & & & & & & & \\ \hhline{|-|-||*{15}{-}|}
\multirow{2}{*}{\bf3} & 2 & 1 & 1 & 3 & 1 & 1 & & & & & & & & & & \\ 
                            & 3 & & 1 & & & & & & & & & & & & & \\ \hhline{|-|-||*{15}{-}|}
\multirow{2}{*}{\bf4} & 2 & 1 & 2 & 5 & 5 & 4 & 1 & 1 & & & & & & & & \\ 
                            & 3 & & 1 & 2 & 1 & & & & & & & & & & & \\ \hhline{|-|-||*{15}{-}|}
\multirow{2}{*}{\bf5} & 2 & 1 & 2 & 8 & 10 & 14 & 7 & 4 & 1 & 1 & & & & & & \\ 
                            & 3 & & 2 & 5 & 7 & 3 & 1 & & & & & & & & & \\ \hhline{|-|-||*{15}{-}|}
\multirow{3}{*}{\bf6} & 2 & 1 & 3 & 12 & 21 & 33 & 30 & 21 & 8 & 4 & 1 & 1 & & & & \\ 
                            & 3 & & 3 & 12 & 23 & 23 & 11 & 3 & 1 & & & & & & & \\
                            & 4 & & & 1 & & & & & & & & & & & & \\ \hhline{|-|-||*{15}{-}|}
\multirow{3}{*}{\bf7} & 2 & 1 & 3 & 16 & 35 & 71 & 82 & 81 & 45 & 23 & 8 & 4 & 1 & 1 & & \\ 
                            & 3 & & 4 & 23 & 65 & 92 & 76 & 36 & 12 & 3 & 1 & & & & & \\
                            & 4 & & & 1 & 1 & 1 & & & & & & & & & & \\ \hhline{|-|-||*{15}{-}|}
\multirow{3}{*}{\bf8} & 2 & 1 & 4 & 21 & 58 & 134 & 205 & 245 & 197 & 122 & 52 & 24 & 8 & 4 & 1 & 1 \\ 
                            & 3 & & 5 & 41 & 153 & 311 & 355 & 257 & 118 & 40 & 12 & 3 & 1 & & & \\
                            & 4 & & & 3 & 5 & 5 & 2 & 1 & & & & & & & & \\ \hhline{|-|-||*{15}{-}|}
\end{tabular}}
\caption{(a) The number of prime/all \Bs binned by degree $d$ and complexity $\chi$ up to $d=10$. The complexity is that of our \href{https://pkomiske.github.io/EnergyFlow}{{\tt EnergyFlow}} implementation, running in time $\mathcal O(M^\chi)$. The partial sums of the ``Total'' rows are the entries of \Tab{tab:efpcounts:a}. (b) The number of \Bs binned by degree $d$, complexity $\chi$, and $N$ up to $d=8$. Note that the majority of \Bs shown here have $N>4$, which would be computationally intractable without algorithmic speedups such as VE.}
\label{tab:numchi} 
\end{table}

\afterpage{\clearpage}

It is useful to analyze the complexity of computing an \B.\footnote{The title of this section is inspired by \Ref{Cacciari:2005hq}.}
A naive implementation of \Eq{eq:introefp} runs in $\mathcal O(M^N)$ due to the $N$ nested sums over $M$ particles.
There is a computational simplification, however, that can be used to tremendously speed up calculations of certain \Bs by making use of the graph structure of $G$.
As an example, consider the following EFP:
\begin{equation}\label{eq:wedgesimple}
\begin{gathered}
\includegraphics[scale=.3]{graphs/4_3_1}
\end{gathered}
=\sum_{i_1=1}^M \sum_{i_2=1}^M \sum_{i_3=1}^M\sum_{i_4=1}^M z_{i_1}z_{i_2}z_{i_3}z_{i_4}\theta_{i_1i_2}\theta_{i_1i_3}\theta_{i_1i_4} = \sum_{i_1=1}^M z_{i_1} \left(\sum_{i_2=1}^M z_{i_2}\theta_{i_1i_2}\right)^3,
\end{equation}
which can be computed in $\mathcal O(M^2)$ rather than $\mathcal O(M^4)$ by first computing the $M$ objects in parentheses in \Eq{eq:wedgesimple} and then performing the overall sum.

In general, since the summand is a product of factors, the distributive property allows one to put parentheses around combinations of sum operators and factors. 
A clever choice of such parentheses, known as an \emph{elimination ordering}, can often be used to perform the $N$ sums of \Eq{eq:introefp} in a way which greatly reduces the number of operations needed to obtain the value of the \B for a given set of particles. 
This technique is known as the Variable Elimination (VE) algorithm~\cite{zhang1996exploiting} (see also \Ref{murphy2012machine} for a review). 

When run optimally, the VE algorithm reduces the complexity of computing $\B_G$ to $\mathcal{O}(M^{\text{tw}(G)+1})$ where $\text{tw}(G)$ is the \emph{treewidth} of the graph $G$, neglecting multiple edges in the case of multigraphs.
The treewidth is a measure which captures how tangled a graph is, with trees (graphs with no cycles) being the least tangled (with treewidth 1) and complete graphs the most tangled (with treewidth $N-1$). 
Additionally, we have that for graphs with a single cycle the treewidth is 2 and for complete graphs minus one edge the treewidth is $N-2$. 
Thus the \Bs corresponding to tree multigraphs can be computed with VE in $\mathcal O(M^2)$ whereas complete graphs do require the naive $\mathcal O(M^N)$ to compute with VE. 
Since the ECFs correspond to complete graphs (see \Eq{eq:ecfgraphs}), they do not benefit from VE.  Similarly, VE cannot speed up the computation of ECFGs, since the ECFGs do not have a factorable summand.

Finding the optimal elimination ordering and computing the treewidth for a graph $G$ are both NP-hard. 
In practice, heuristics are used to decide on a pretty-good elimination ordering (which for the small graphs we consider here is often optimal) and to approximate the treewidth. 
In principle, these orderings need only be computed once for a fixed set of graphs of interest. 
Similarly, many algebraic structures reappear when computing a set of \Bs for the same set of particles, making dynamic programming a viable technique for further improving the computational complexity of the method.

\begin{figure}[t]
\centering
\includegraphics[scale=.76]{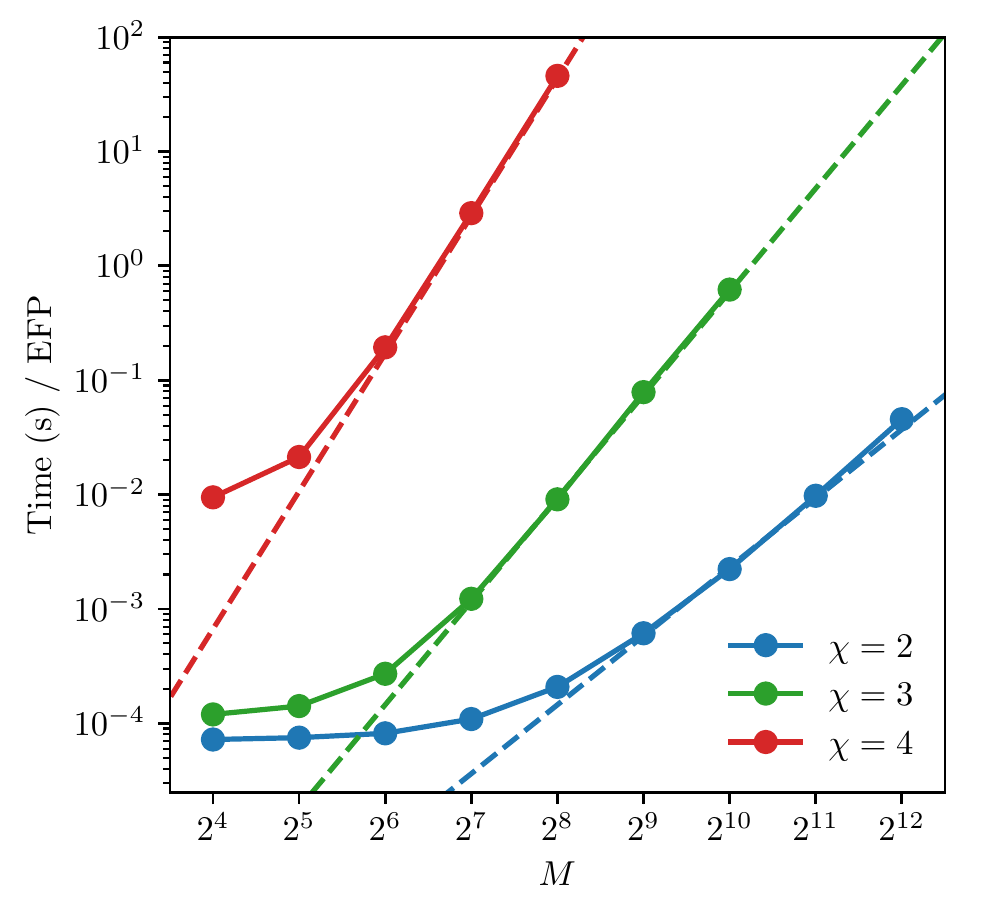}
\caption{Compute time (in seconds) per \B for different VE complexities $\chi$ as a function of the number of inputs $M$.  The quoted value is based on all \Bs with $d\le7$, and each data point is the average of 10 computations.  The dashed lines show the expected $\mathcal{O}(M^\chi)$ scaling behavior.  As $\chi$ increases, the relative amount of overhead decreases and the asymptotic behavior is achieved more rapidly than for smaller $\chi$. Computations were run with Python 3.5.2 and {\tt NumPy} 1.13.3 on a 2.3 GHz Intel Xeon E5-2673 v4 (Broadwell) processor on Microsoft Azure using our \href{https://pkomiske.github.io/EnergyFlow}{{\tt EnergyFlow}} module.}
\label{fig:times}
\end{figure}

\Tab{tab:numchi:a} shows the number of \Bs listed by degree $d$ and VE complexity $\chi$ (with respect to the heuristics used in our implementation), and \Tab{tab:numchi:b} further breaks up the \B counts by $N$. 
\Fig{fig:times} shows the time to compute the average $d\le7$ \B as a function of multiplicity $M$ for different VE complexity $\chi$. 
Finally, we note that though VE often provides a significant speedup over the naive algorithm, there may be even faster ways of computing the \Bs.\footnote{At the risk of burying the lede in a footnote, we have found that with certain choices of the angular measure, it is possible to compute all \Bs in $\mathcal O(M)$.  An exploration of these interesting special cases is performed in \Ref{Komiske:2019asc}.}

\section{Linear regression with jet observables}
\label{sec:linreg}

Regression, classification, and generation are three dominant machine learning paradigms. 
Machine learning applications in collider physics have been largely focused on classification (e.g.\ jet tagging)~\cite{Komiske:2016rsd,Almeida:2015jua,Baldi:2016fql,Kasieczka:2017nvn,Pearkes:2017hku,Butter:2017cot,Aguilar-Saavedra:2017rzt,Guest:2016iqz,Louppe:2017ipp} with recent developments in regression~\cite{Komiske:2017ubm} and generation~\cite{deOliveira:2017pjk, Paganini:2017hrr}. 
For a more complete review of modern machine learning techniques in jet substructure, see \Ref{Larkoski:2017jix}. 
The lack of established regression problems in jet physics is due in part to the difficulty of theoretically probing multivariate combinations as well as the challenges associated with extracting physics information from trained regressions models.

In this section, we show that the linearity of the energy flow basis mitigates many of these problems, providing a natural regression framework using simple linear models, probing the learned observable combinations, and gaining insight into the physics of the target observables. 
Since regression requires training samples, we observe how the regression performance compares on jets with three characteristic phase-space configurations: one-prong QCD jets, two-prong boosted $W$ jets, and three-prong boosted top jets. 
We use linear regression to demonstrate convergence of the energy flow basis on IRC-safe observables, while illustrating their less-performant behavior for non-IRC-safe observables.

\subsection{Linear models with the energy flow basis}
\label{sec:linmod}

Linear models assume a linear relationship between the input and target variables, making them the natural choice for (machine) learning with the energy flow basis for both regression and classification. 
A linear model $\mathscr M$ with \Bs as the inputs is defined by a finite set $\mathcal G$ of multigraphs and numerical coefficients ${\bf w}=\{w_G\}_{G\in\mathcal G}$:
\begin{equation}\label{eq:linmod}
\mathscr M = \sum_{G\in\mathcal G}w_G \,  \B_G.
\end{equation}
The fundamental relationship between \Bs, linear models, and IRC-safe observables is highlighted by comparing \Eq{eq:linmod} to \Eq{eq:introefpsspan}, where the linear model $\mathscr M$ in \Eq{eq:linmod} takes the place of the IRC-safe observable $\mathcal S$ in \Eq{eq:introefpsspan}. 
Because the \Bs are a complete linear basis, $\mathscr M$ is capable of approximating any $\mathcal S$ for a sufficiently large set of \Bs.

The linear structure of \Eq{eq:linmod} allows for an avenue to ``open the box'' and interpret the learned coefficients as defining a unique multiparticle correlator for each $N$. 
To see this, partition the set $\mathcal G$ into subsets $\mathcal G_N$ of graphs with $N$ vertices.
The sum in \Eq{eq:linmod} can be broken into two sums, one over $N$ and the other over all graphs in $\mathcal G_N$. 
The linear energy structure of the \Bs in \Eq{eq:introefp} allows for the second sum to be pushed inside the product of energies onto the angular weighting function:
\begin{equation}\label{eq:linmodefp}
\mathscr M=\sum_{N=0}^{N_{\rm max}}\sum_{i_1=1}^M\cdots\sum_{i_N=1}^Mz_{i_1}\cdots z_{i_N}\left(\sum_{G\in\mathcal G_N}w_G\prod_{(k,\ell)\in G}\theta_{i_ki_\ell}\right),
\end{equation}
where $N_{\rm max}$ is the maximum number of vertices of any graph in $\mathcal G$. 
The quantity in parentheses in \Eq{eq:linmodefp} may be though of as a \emph{single} angular weighting function. 
The linear model written in this way reveals itself to be a sum of $C$-correlators (similar to \Eq{eq:Eexpansionresult}), one for each $N$, where the linear coefficients within each $\mathcal G_N$ parameterize the angular weighting function $f_N$ of that $C$-correlator. 
This arrangement of the learned parameters of the linear model into $N_{\rm max}$ $C$-correlators contrasts sharply with the lack of a physical organization of parameters in nonlinear methods such as neural networks or boosted decision trees.

\subsection{Event generation and \B computation}
\label{sec:eventgen}

For the studies in this section and in \Sec{sec:linclass}, we generate events using \textsc{Pythia} 8.226~\cite{Sjostrand:2006za,Sjostrand:2007gs,Sjostrand:2014zea} with the default tunings and shower parameters at $\sqrt{s}=\SI{13}{TeV}$.
Hadronization and multiple parton interactions (i.e.\ underlying event) are included, and a $\SI{400}{GeV}$ parton-level $p_T$ cut is applied.
For quark/gluon distribution, quark (signal) jets are generated through $pp \to q Z(\to \nu\bar\nu)$, and gluon (background) jets through $pp \to g Z(\to \nu\bar\nu)$, where only light-quarks ($uds$) appear in the quark sample.
For $W$ and top tagging, signal jets are generated through $pp\to W^+W^-(\to\text{hadrons})$ and $pp\to t\bar t(\to\text{hadrons})$, respectively.
For both $W$ and top events, the background consists of QCD dijets.

Final state, non-neutrino particles were made massless, keeping $y$, $\phi$, and $p_T$ fixed,\footnote{Using massless inputs is not a requirement for using the \Bs, but for these initial \B studies, we wanted to avoid the caveats associated with massive inputs for the validity of \Sec{sec:basis}.} and then were clustered with \textsc{FastJet 3.3.0}~\cite{Cacciari:2011ma} using the anti-$k_T$ algorithm~\cite{Cacciari:2008gp} with a jet radius of $R=0.4$ for quark/gluon samples and $R = 0.8$ for $W$ and top samples (and the relevant dijet background). 
The hardest jet with rapidity $|y|<1.7$ and $500\text{ GeV}\le p_T\le550$ GeV was kept. 
For each type of sample, 200k jets were generated. 
For the regression models, 75\% were used for training and 25\% for testing.

For these events, all \Bs up to degree $d\le7$ were computed in Python using our \href{https://pkomiske.github.io/EnergyFlow}{{\tt EnergyFlow}} module making use of {{\tt NumPy}}'s einsum function. 
See \Tabs{tab:efpcounts}{tab:numchi} for counts of \Bs tabulated by various properties such as $N$, $d$, and $\chi$. 
Note that all but 4 of the 1000 $d\le 7$ \Bs can be computed in $\mathcal O(M^2)$ or $\mathcal O(M^3)$ in the VE paradigm, making the set of \Bs with $d\le 7$ efficient to compute.

\subsection{Spanning substructure observables with linear regression}
\label{sec:regression}

We now consider the specific case of training linear models to approximate substructure observables with linear combinations of \Bs. 
For an arbitrary observable $O$, we use least-squares regression to find a suitable set of coefficients ${\bf w}^*$:
\begin{equation}\label{eq:linregdef}
{\bf w}^* = \underset{\bf w}{\arg\min} \left\{\sum_{J\in\text{jets}}\left(O(J) - \sum_{G\in\mathcal G} w_G \,\B_G(J)\right)^2\right\},
\end{equation}
where $O(J)$ is the value of the observable and $\B_G(J)$ the value of the \B given by multigraph $G$ on jet $J$. 
There are possible modifications to \Eq{eq:linregdef} which introduce penalties proportional to $\|{\bf w}\|_1$ or $\|{\bf w}\|_2^2$ where $\|\cdot\|_1$ is the 1-norm and $\|\cdot\|_2$ is the 2-norm. 
The first of these choices, referred to as lasso regression~\cite{tibshirani1996regression}, may be particularly interesting because of the variable selection behavior of this model, which would aid in selecting the most important \Bs to approximate a particular observable. 
We leave such investigation to future work. See \Ref{bishop2006pattern} for a review of linear models for regression.

We use the {\tt LinearRegression} class of the {\tt scikit-learn} python module~\cite{scikit-learn} to implement \Eq{eq:linregdef} with no regularization on the samples described in \Sec{sec:eventgen}. 
In general, the smallest possible regularization which prevents overfitting (if any) should be used.
Because of the linear nature of linear regression and the analytic tractability of \Eq{eq:linregdef}, the ${\bf w}^*$ corresponding to the global minimum of the squared loss function can be found efficiently using convex optimization techniques. 
Such techniques include closed-form solutions or convergent iterative methods.

\begin{table}[t]
\centering
\begin{tabular}{|r||l|l|}
\hhline{~|-|-|}
\multicolumn{1}{c|}{}& \textbf{Observable} & \textbf{Properties}  \\ \hhline{-::==:}
$\frac{m_J}{p_{T,J}}$ 
\unskip\textcolor{white!5}{\makebox[0pt]{\smash{\rule[12pt]{14pt}{3pt}}}}
& Scaled jet mass &  No Taylor expansion about zero energy limit \\ 
$\lambda^{(\alpha=1/2)}$ & Les Houches angularity & No analytic relationship beyond even integers \\ 
$\tau_2^{(\beta = 1)}$ & $2$-subjettiness & Algorithmically defined IRC-safe observable \\
\hhline{|-||-|-|}
$\tau_{21}^{(\beta=1)}$ & $N$-subjettiness ratio & Sudakov safe, safe for two-prong kinematics\\
$\tau_{32}^{(\beta = 1)}$ & $N$-subjettiness ratio & Sudakov safe, safe for three-prong kinematics\\
\hhline{|-||-|-|}
$M$ & Particle multiplicity & IRC unsafe\\
\hhline{|-||-|-|}
\end{tabular}
\caption{The six substructure observables used as targets for linear regression, listed with relevant properties. The first three are IRC safe, the next two are Sudakov safe in general (and IRC safe in the noted regions of phase space), and particle multiplicity is IRC unsafe. The Les Houches Angularity is calculated with respect to the $p_T$-weighted centroid axis in \Eq{eq:jetaxis}, and the $N$-subjettiness observables~ are calculated using $k_T$ axes.}
\label{tab:observables}
\end{table}

As targets for the regression, we consider the six jet observables in \Tab{tab:observables} to highlight some interesting test cases. 
As our measure of the success of the regression, we use a variant of the correlation coefficient between the true and predicted observables that is less sensitive to outliers than the unadulterated correlation coefficient.
When evaluating the trained linear model on the test set, only test samples with predicted values within the 5$^\text{th}$ and 95$^\text{th}$ percentiles of the predictions are included.
In the contexts considered in this chapter, narrowing this percentile range lowers the correlation coefficient and widening the range out toward all of the test set increases the correlation coefficient.
The qualitative nature of the results are insensitive to the specific choice of percentile cutoffs. 
We perform this regression using \Bs of degree up to $d$ for $d$ from 2 to 7 on all three jet samples, with the results shown in \Fig{fig:robustcorr}.
Histograms of the true and predicted distributions for a subset of these observables are shown in \Fig{fig:reghists} for the three types of jets considered here. 

\begin{figure}[t]
\centering
\subfloat[]{\includegraphics[scale=.52]{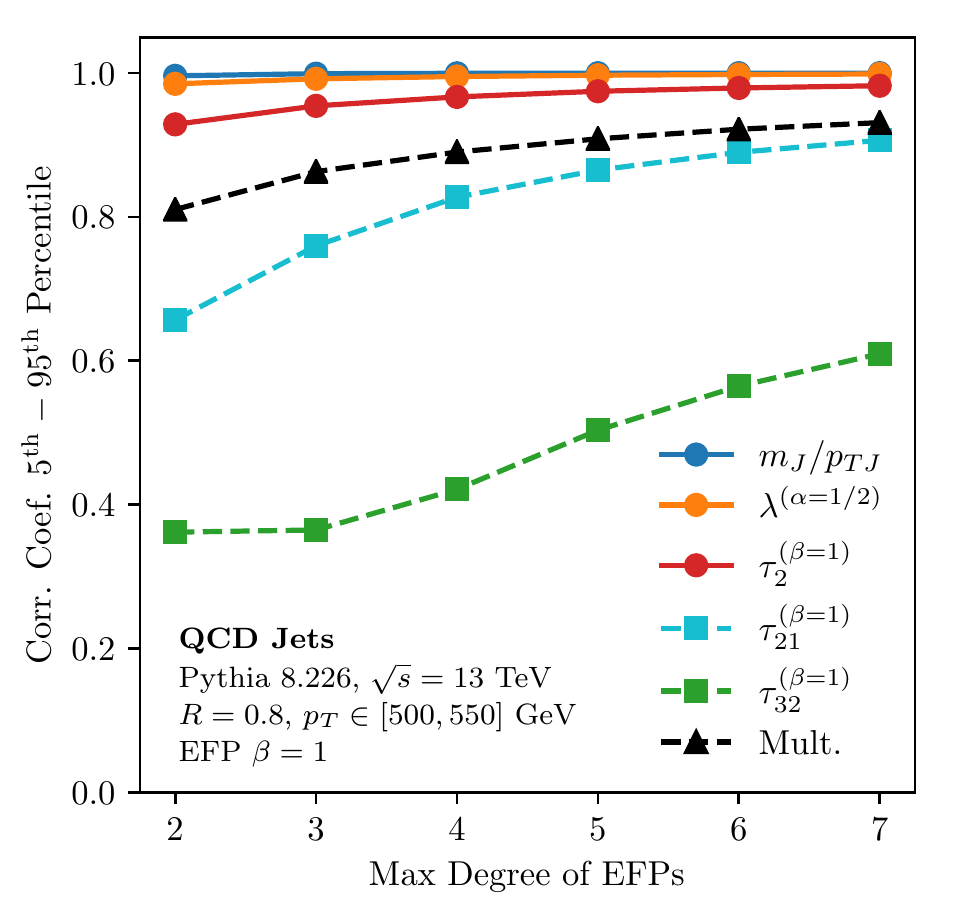}}
\subfloat[]{\includegraphics[scale=.52]{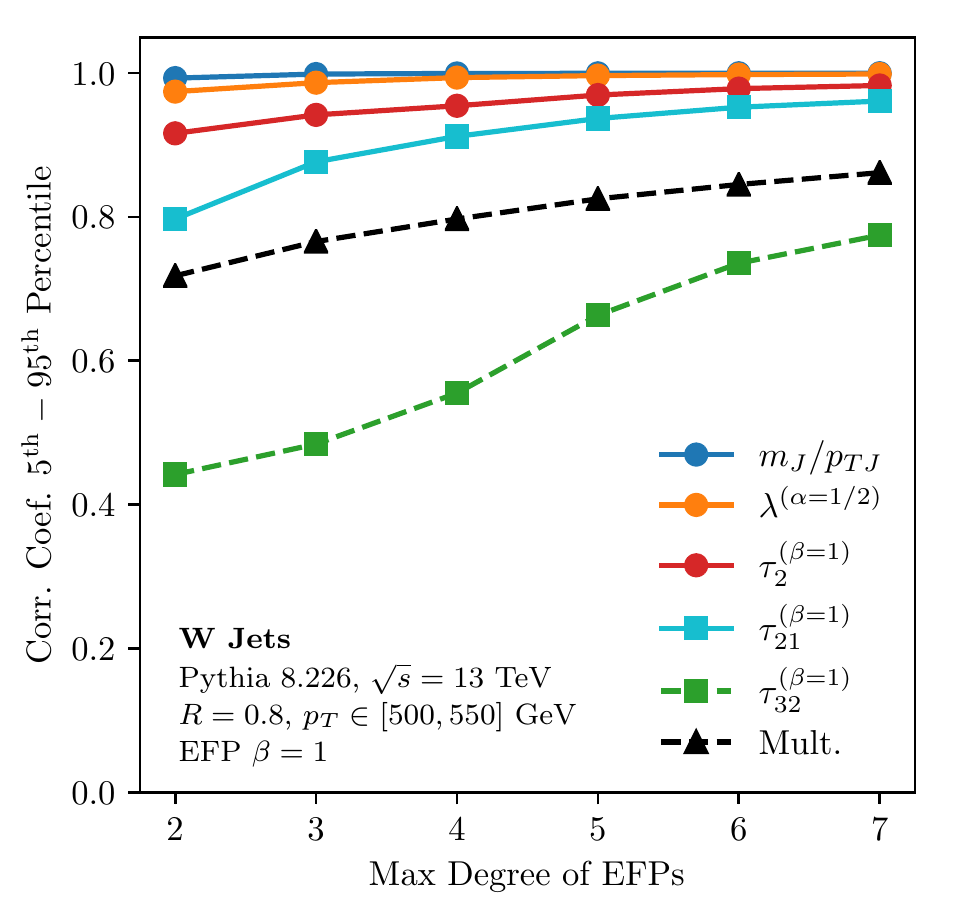}}
\subfloat[]{\includegraphics[scale=.52]{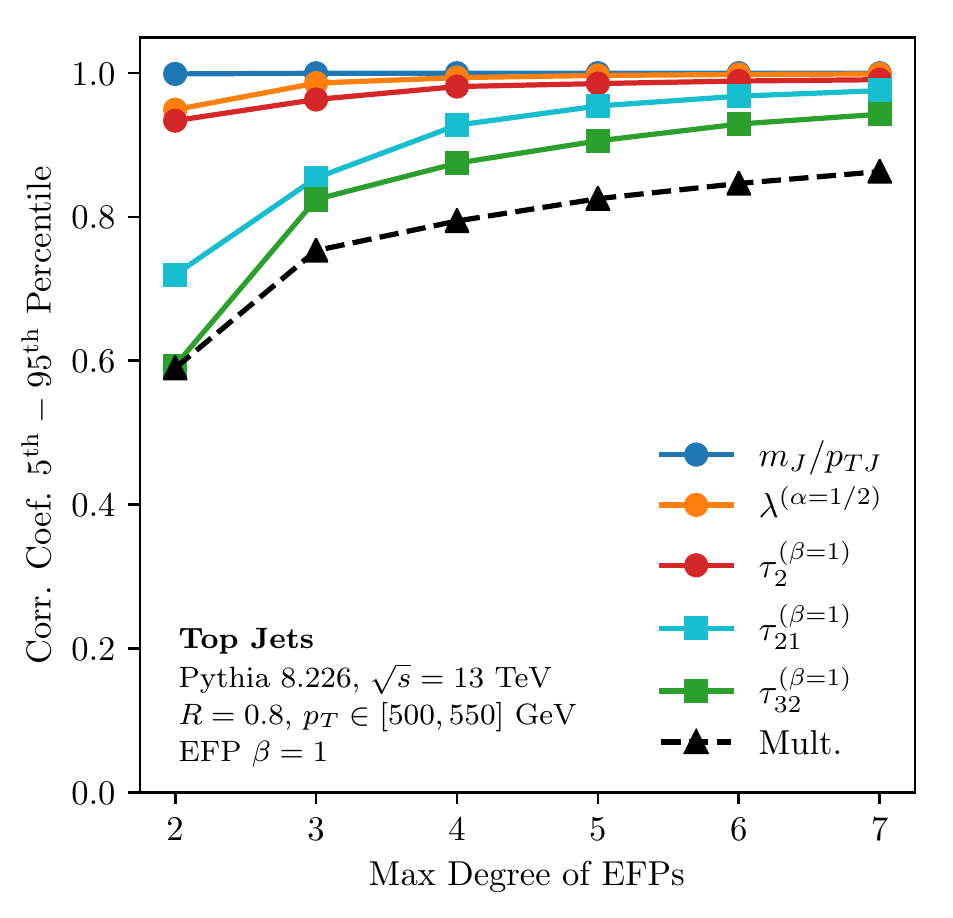}}
\caption{Correlation coefficients between true and predicted values for the jet observables in \Tab{tab:observables}, plotted as a function of maximum \B degree. Shown are the (a) QCD dijet, (b) $W$ jet, and (c) top jet samples, and as explained in the text, we restrict to predictions in the 5$^\text{th}$--95$^\text{th}$ percentiles. Observables in IRC-safe regions of phase space are shown with solid lines and those in IRC-unsafe regions (including Sudakov-safe regions) are shown with dashed lines. The IRC-safe observables are all learned with correlation coefficient above 0.98 in all three cases by $d=7$. Multiplicity (black triangles) sets the scale for the regression performance on IRC-unsafe observables. Note that $\tau_{21}$ has performance similar to the IRC-safe observables only when jets are characteristically two-pronged or higher ($W$ and top jets), and similarly for $\tau_{32}$ when the jets are characteristically three-pronged (top jets).}
\label{fig:robustcorr}
\end{figure}

\begin{figure}[t]
\centering
\subfloat[]{\includegraphics[scale=.5]{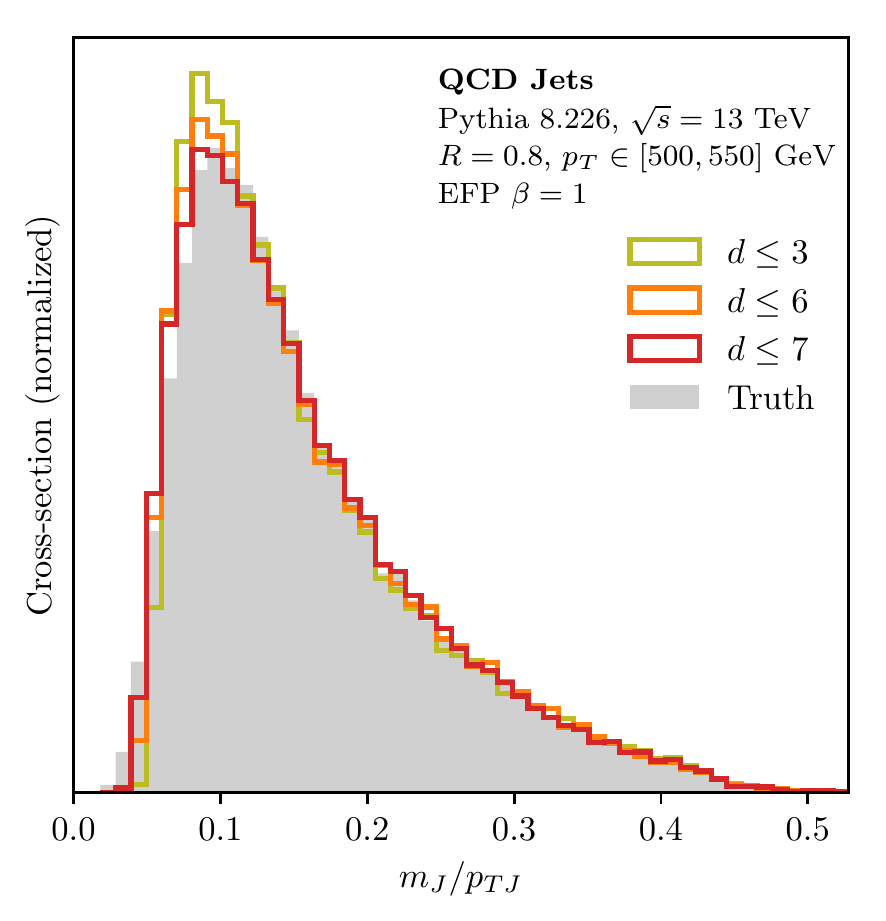}}
\subfloat[]{\includegraphics[scale=.5]{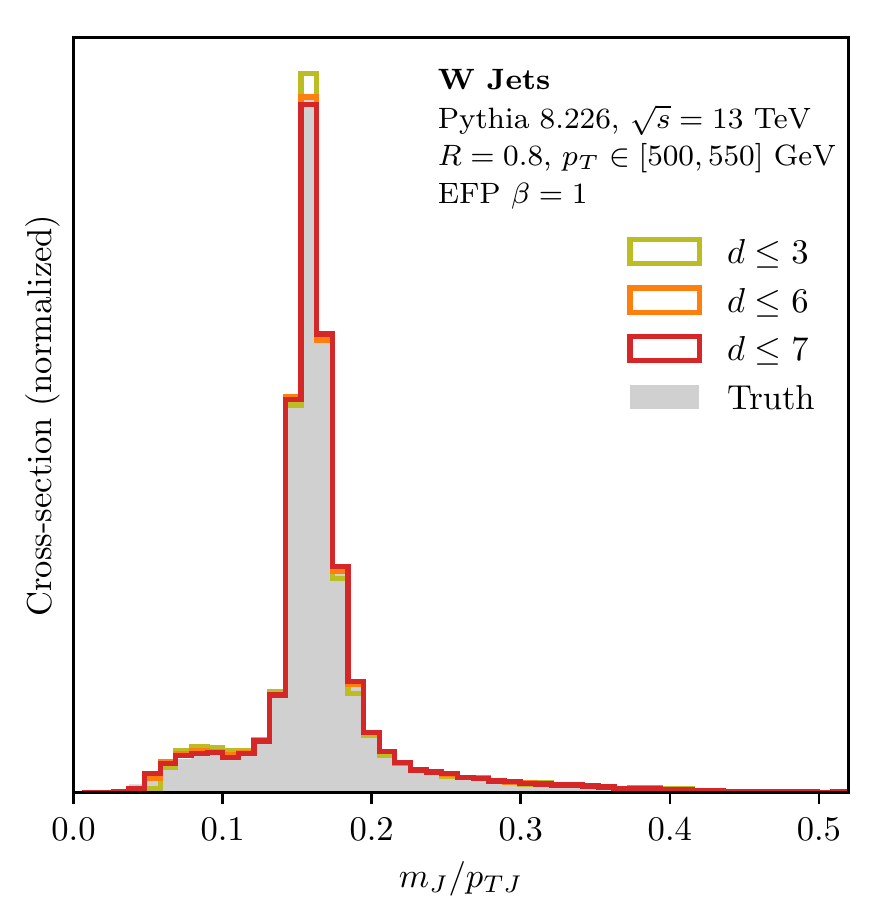}}
\subfloat[]{\includegraphics[scale=.5]{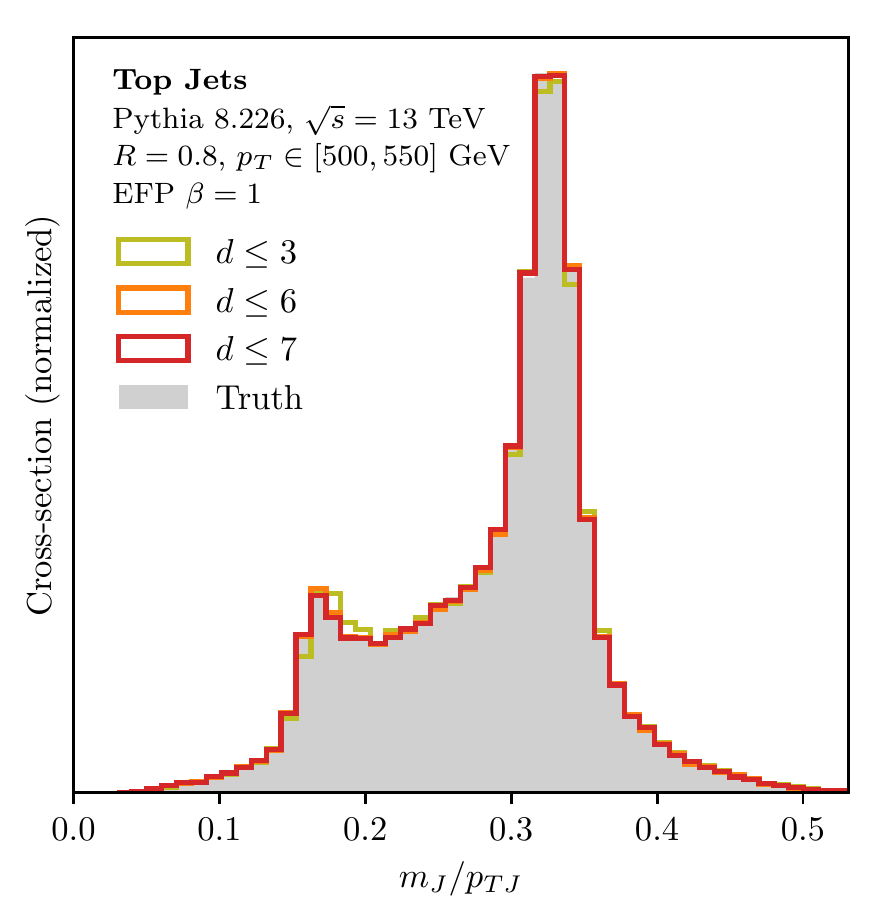}}
\\
\subfloat[]{\includegraphics[scale=.5]{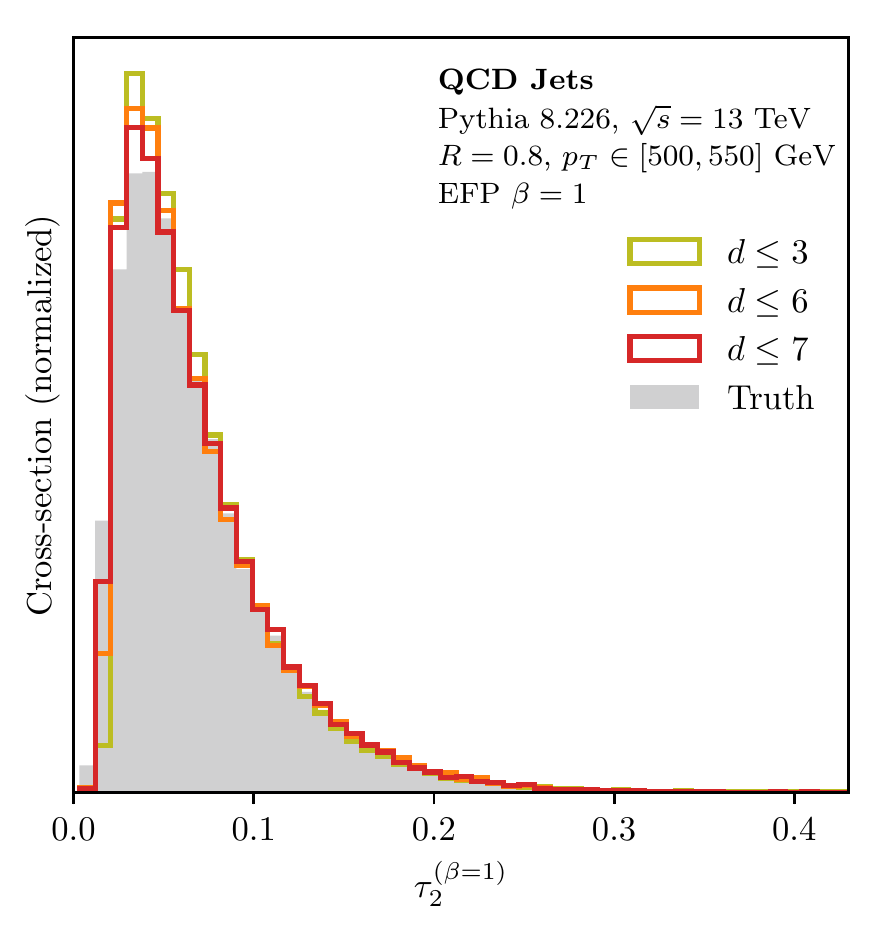}}
\subfloat[]{\includegraphics[scale=.5]{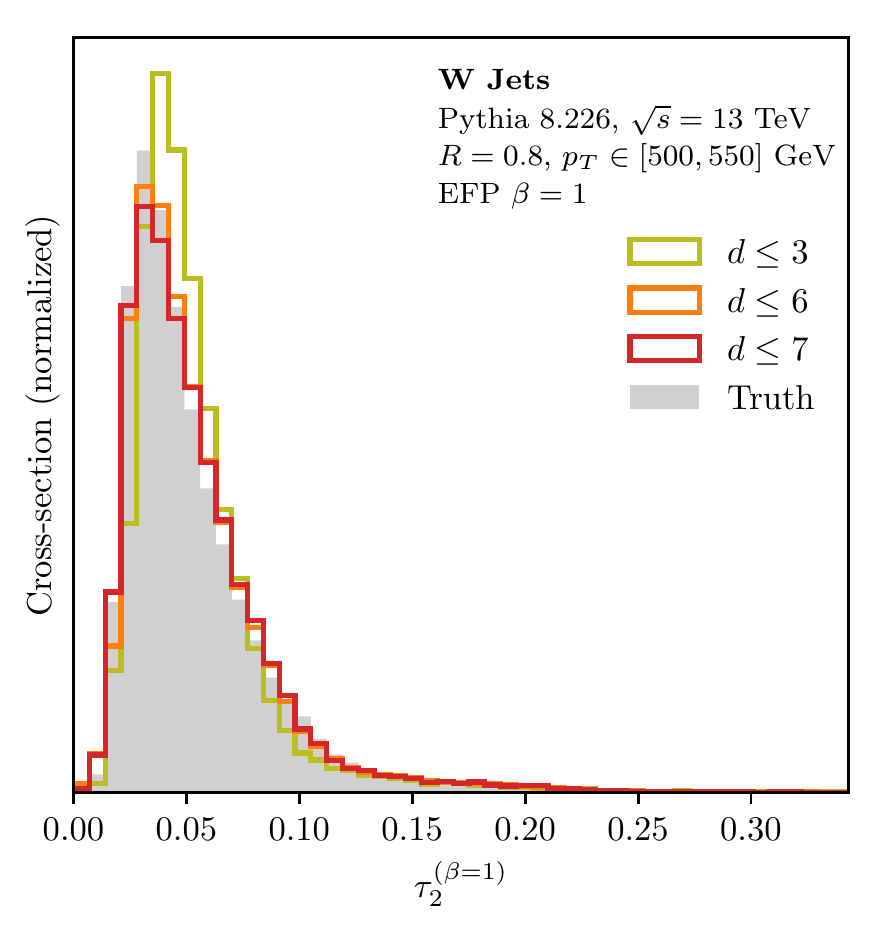}}
\subfloat[]{\includegraphics[scale=.5]{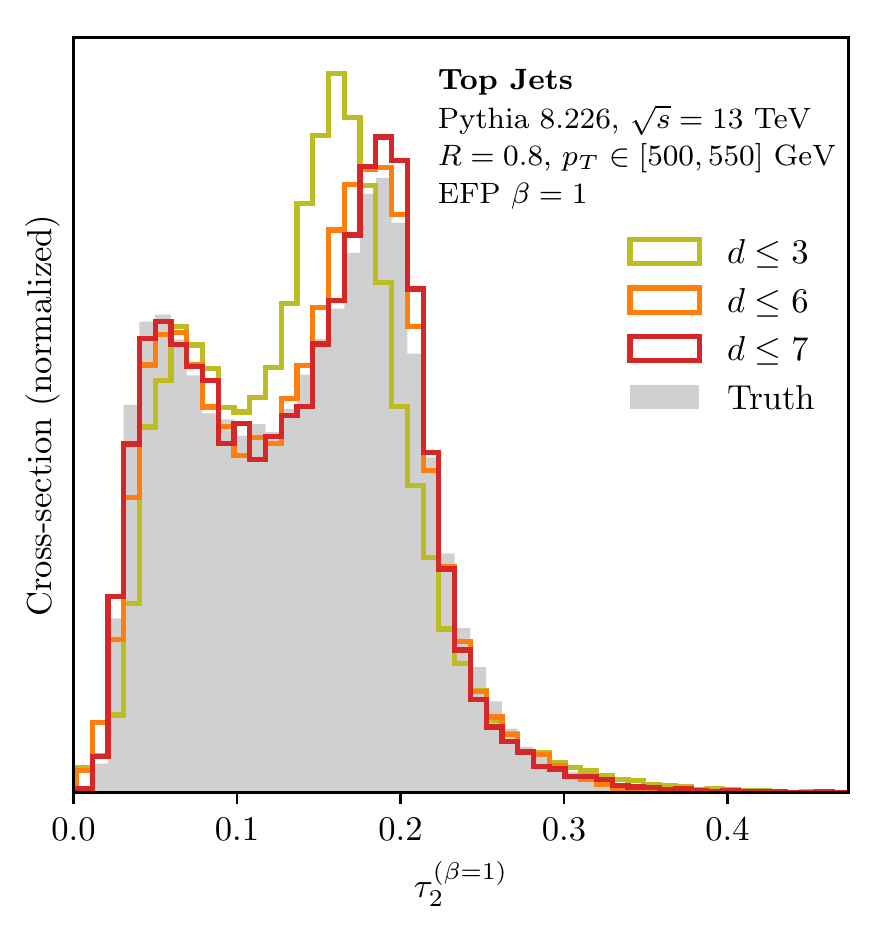}}
\\
\subfloat[]{\includegraphics[scale=.5]{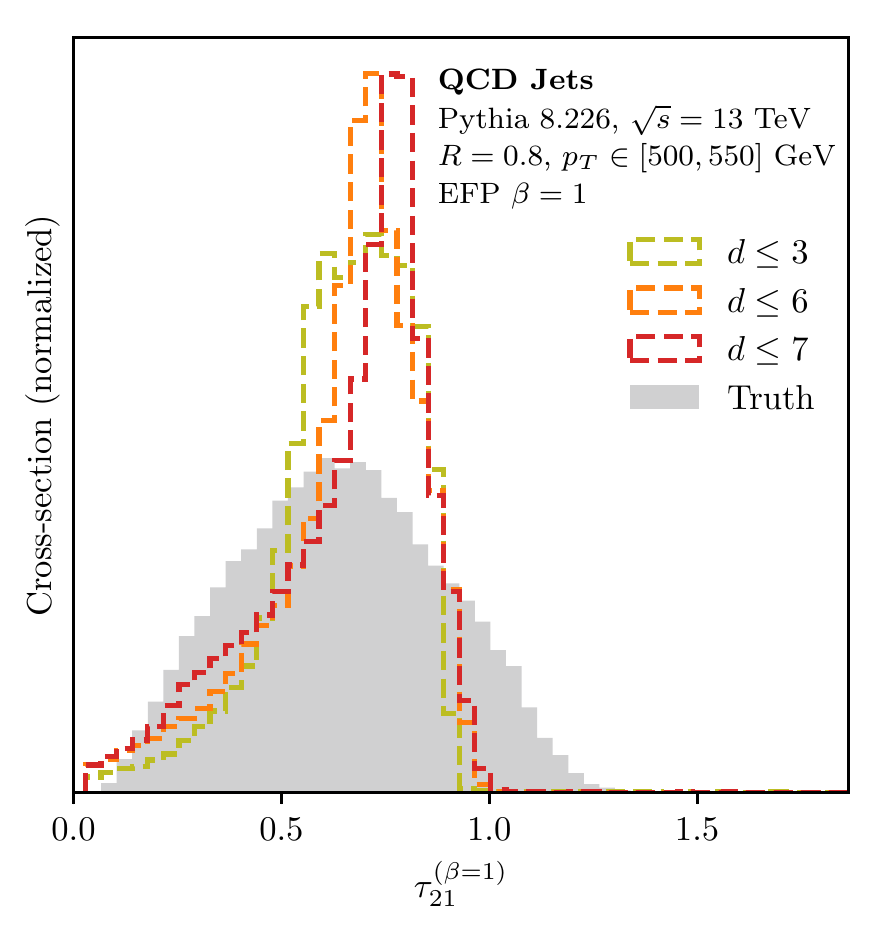}}
\subfloat[]{\includegraphics[scale=.5]{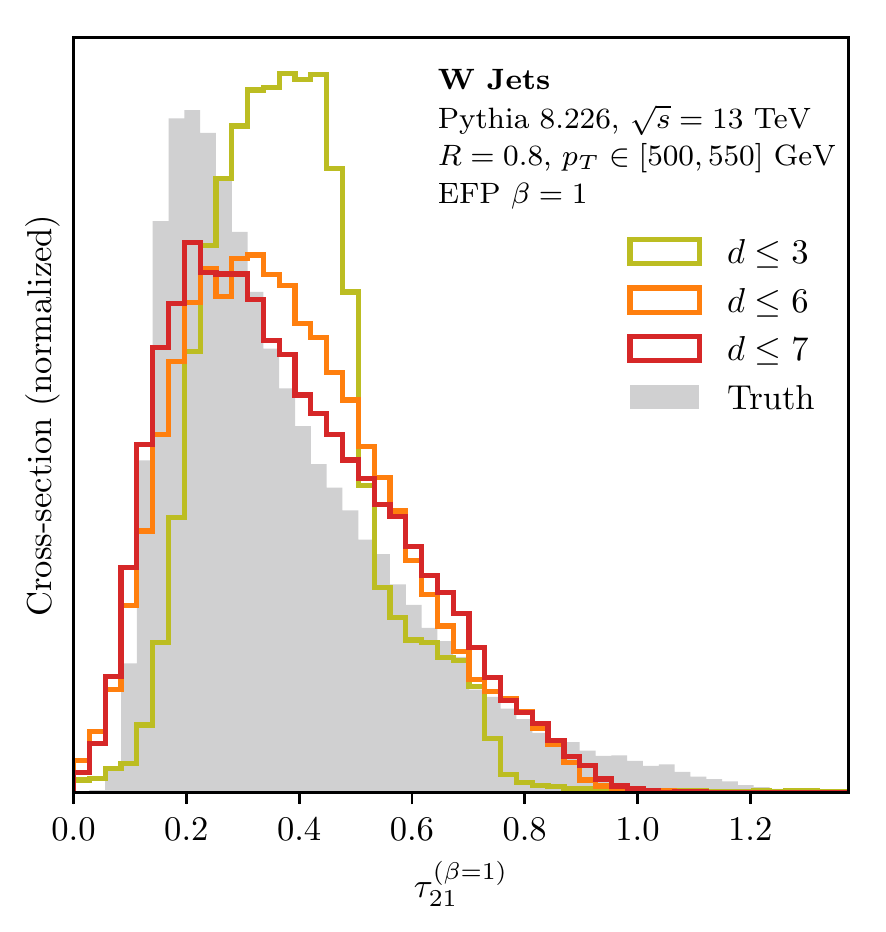}}
\subfloat[]{\includegraphics[scale=.5]{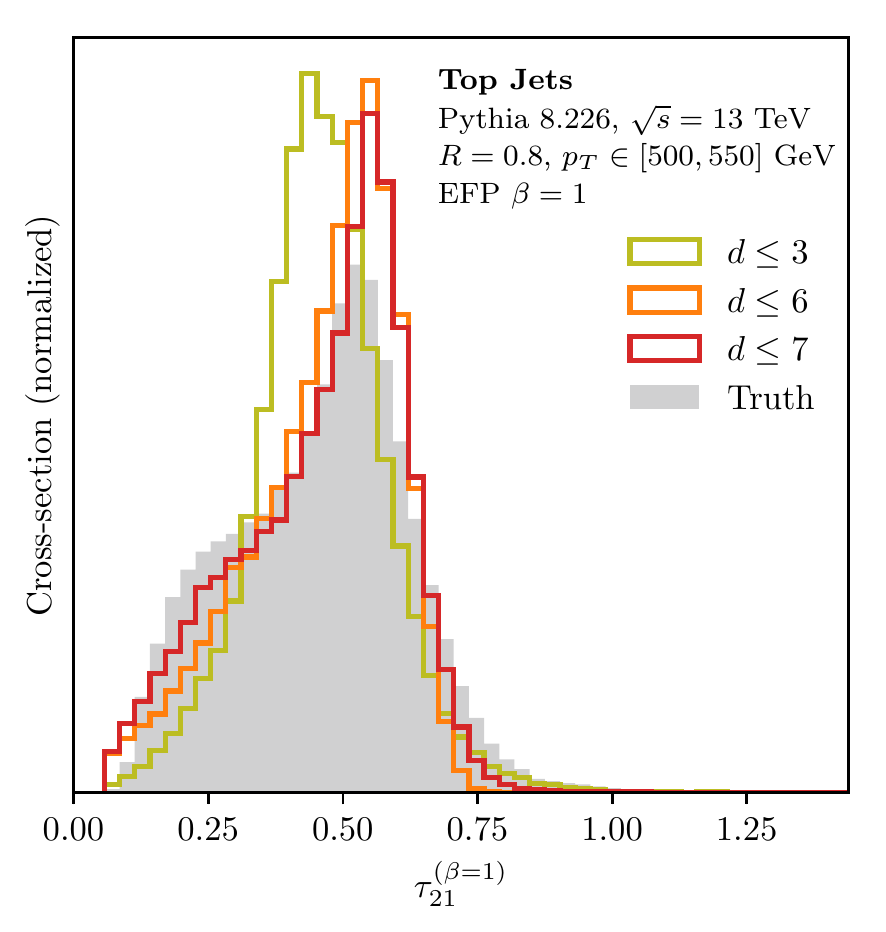}}
\caption{The distributions of true and predicted scaled jet mass (top), $\tau_2^{(\beta=1)}$ (middle), and $\tau_{21}^{(\beta=1)}$ (bottom) using linear regression with \Bs up to different maximum degrees $d$ on QCD jets (left), $W$ jets (center), and top jets (right). Note the excellent agreement for the IRC-safe observables in the first two rows. Observables in IRC-safe regions of phase space are shown with solid lines and those in IRC-unsafe regions are shown with dashed lines. The Sudakov-safe $\tau_{21}^{(\beta=1)}$ predicted distributions match the true distributions for jets typically with two or more prongs ($W$ and top jets) better than for typically one-pronged (QCD) jets.}
\label{fig:reghists}
\end{figure}
\afterpage{\clearpage}

Since the learned coefficients depend on the training set, in principle different linear combinations may be learned to approximate the substructure observables in different jet contexts.
This stands in contrast to the analysis in \Sec{sec:jetobs}, where many jet substructure observables were identified as exact linear combinations of \Bs, independent of the choice of inputs.
The IRC-safe observables---mass, Les Houches angularity, and $2$-subjettiness---are all learned with a correlation coefficient above 0.98 in all three cases by $d=7$.

The IRC-unsafe multiplicity sets the scale of performance for observables that are not IRC safe.
For the $N$-subjettiness ratios, the regression performance depends on whether the observable is IRC safe or only Sudakov safe~\cite{Larkoski:2013paa,Larkoski:2015lea}.
The ratio $\tau_{21}$ is only IRC safe for regions of phase space with two prongs or more (i.e.\ the $W$ and top samples), and $\tau_{32}$ is only IRC safe for three prongs or more (i.e.\ just the top sample).
In cases where the $N$-subjettiness ratio is IRC safe, the regression performs similarly to the other IRC-safe observables, whereas for the cases where the $N$-subjettiness ratio is only Sudakov safe, the regression performance is poor (even worse than for multiplicity).
It is satisfying to see the expected behavior between the safety of the observable and the quality of the regression with \Bs.

As a final cross check of the regression, we can use the linear model in \Eq{eq:linmod} to confirm some of the analytic results of \Sec{sec:jetobs}.
Specifically, we perform a linear regression with the target observable being the even-$\alpha$ angularities with respect to the $p_T$-weighted centroid axis.
These were shown to be non-trivial linear combinations of \Bs in \Sec{sec:angularities}.
Regressing onto $\lambda^{(2)},\,\lambda^{(4)},$ and $\lambda^{(6)}$, the linear model learned the observables with effectively 100\% accuracy and the learned linear combination was exactly that predicted by \Eqss{eq:lam2}{eq:lam4}{eq:lam6}, up to a precision of $10^{-6}$.
\Fig{fig:linspec} shows the learned linear combinations of \Bs for the $W$ jet sample.

\begin{figure}[t]
\centering
\subfloat[]{\includegraphics[scale=.51]{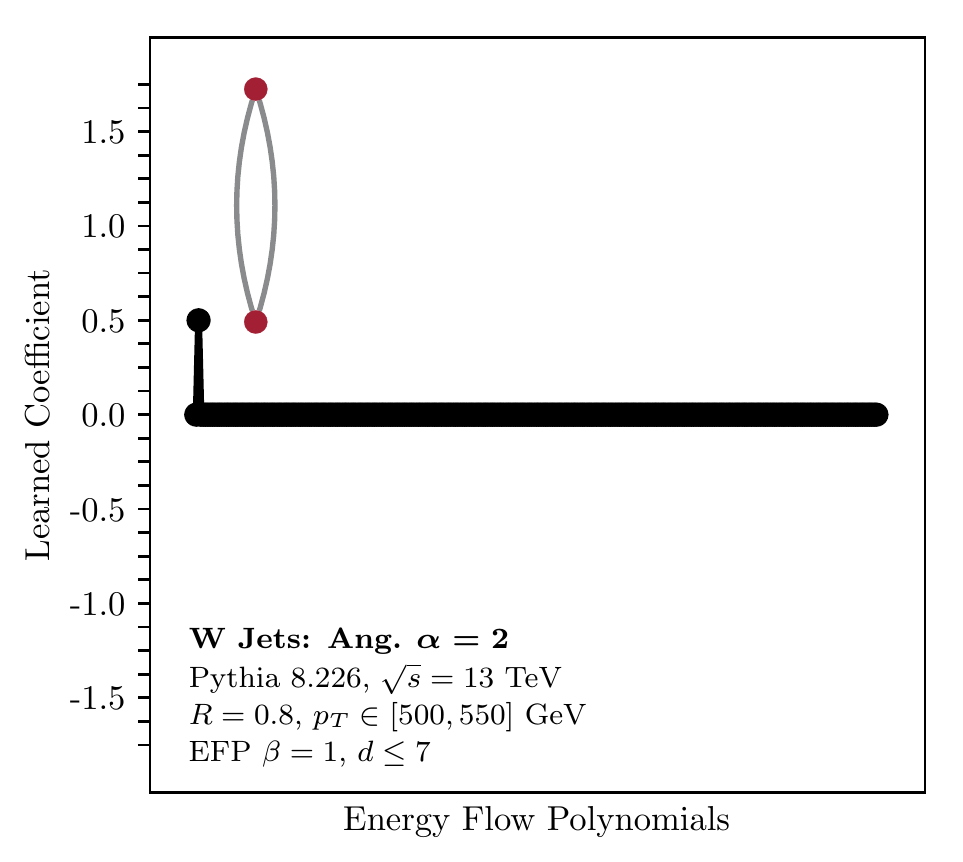}}
\subfloat[]{\includegraphics[scale=.51]{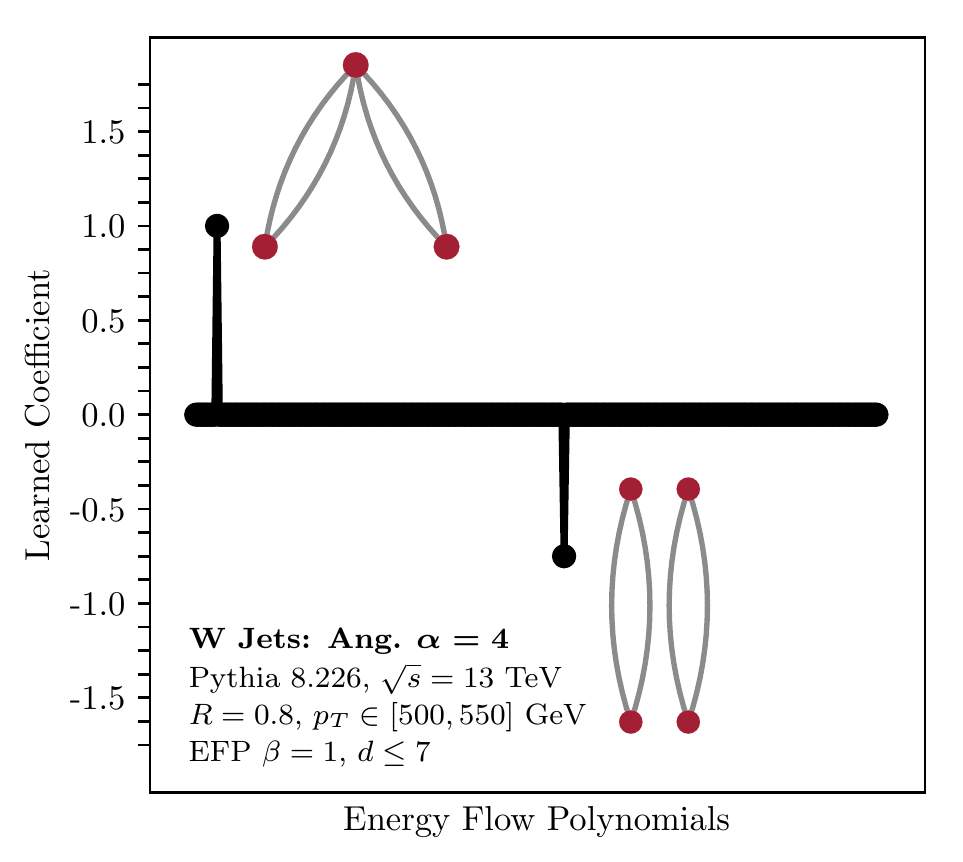}}
\subfloat[]{\includegraphics[scale=.51]{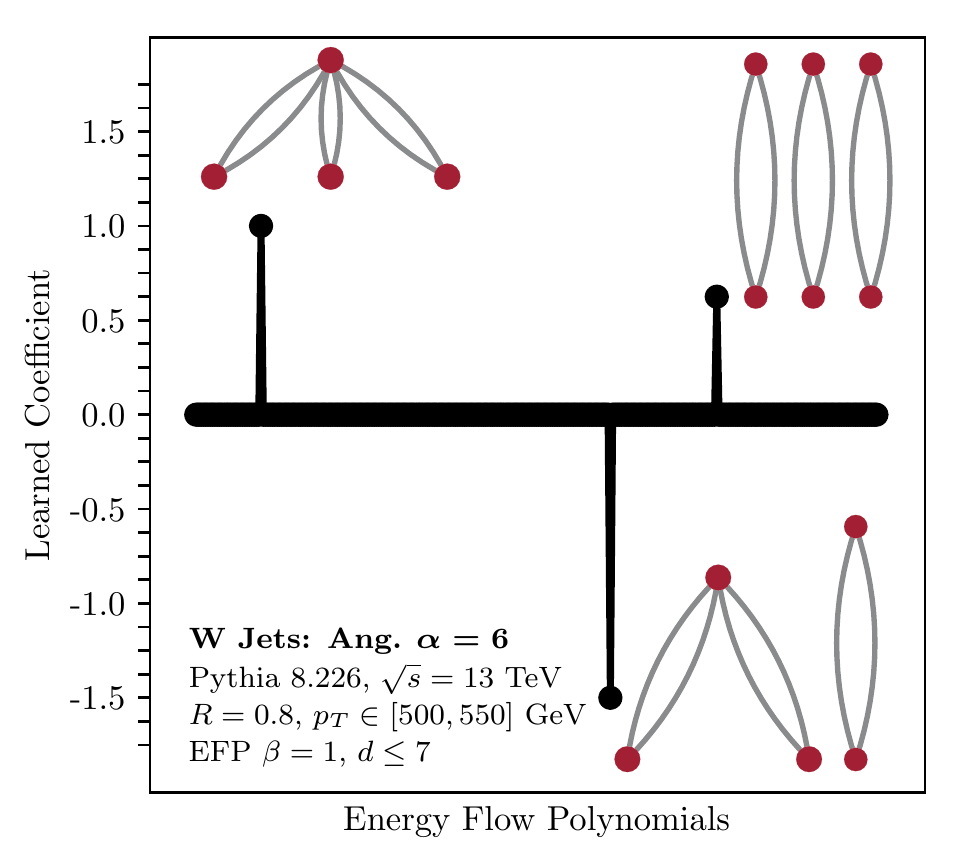}}
\caption{The linear combinations of \Bs learned by linear regression for even-$\alpha$ angularities with the $W$ jet samples.   Shown are (a) $\alpha = 2$, (b) $\alpha = 4$, and (c) $\alpha = 6$.  All but the highlighted \B coefficients are learned to be near zero. The \Bs corresponding to those non-zero coefficients are illustrated directly on the figure. The learned linear coefficients are exactly those predicted analytically in \Eqss{eq:lam2}{eq:lam4}{eq:lam6}. The same behavior is found with the QCD and top jet samples.}
\label{fig:linspec}
\end{figure}

\section{Linear jet tagging}
\label{sec:linclass}

We now apply the energy flow basis to three representative jet tagging problems---light-quark/gluon classification, $W$ tagging, and top tagging---providing a broad set of contexts in which to study the \Bs.
Since the energy flow basis is linear, we can (in principle) access the optimal IRC-safe observable for jet tagging by training a linear classifier for this problem.
As mentioned in \Sec{sec:regression}, one benefit of linear models, in addition to their inherent simplicity, is that they are typically convex problems which can be solved exactly or with gradient descent to a global minimum.
See \Ref{bishop2006pattern} for a review of linear models for classification.

A (binary) linear classifier learns a vector ${\bf w}^*$ that defines a hyperplane orthogonal to the vector. 
A bias term, which can be related to the distance of this hyperplane from the origin, sets the location of the decision boundary, which is the hyperplane translated away from the origin. 
The decision function for a particular point in the input space is the normal distance to the decision boundary. 
In contrast with regression, where the target variable is usually continuous, classification predictions are classes, typically 0 or 1 for a binary classifier. 

Different methods of determining the vector ${\bf w}^*$---such as logistic regression, support vector machines, or linear discriminant analysis---may learn different linear classifiers since the methods optimize different loss functions. 
For our linear classifier, we use Fisher's linear discriminant~\cite{fisher1936use} provided by the {\tt LinearDiscriminantAnalysis} class of the {\tt scikit-learn} python module~\cite{scikit-learn}.  
The choice of logistic regression was also explored, and jet tagging performance was found to be insensitive to which type of linear classifier was used.

The details of the event generation and \B computation are the same as in \Sec{sec:eventgen}. 
To avoid a proliferation of plots, we present only the case of $W$ tagging in the text and refer to \App{app:moretagging} for the corresponding results for quark/gluon classification and top tagging.
Qualitatively similar results are obtained on all three tagging problems, with the conclusion that linear classification with \Bs yields comparable classification performance to other powerful machine learning techniques.
This is good evidence that the \Bs provides a suitable linear expansion of generic IRC-safe information relevant for practical jet substructure applications.

\subsection{Alternative jet representations}

In order to benchmark the \Bs, we compare them to two alternative jet tagging paradigms:
\begin{itemize}
\item The \textbf{jet images} approach~\cite{Cogan:2014oua} treats calorimeter deposits as pixels and the jet as an image, often using convolutional neural networks to determine a classifier. 
Jet images have been applied successfully to the same tagging problems considered here:  quark/gluon classification~\cite{Komiske:2016rsd}, $W$ tagging~\cite{deOliveira:2015xxd}, and top tagging~\cite{Kasieczka:2017nvn,Baldi:2016fql}.
\item The \textbf{$\boldsymbol{N}$-subjettiness basis} was introduced for $W$ tagging in \Ref{Datta:2017rhs} and later applied to tagging non-QCD jets~\cite{Aguilar-Saavedra:2017rzt}.
We use the same choice of $N$-subjettiness basis elements as \Ref{Datta:2017rhs}, namely:
\begin{equation}\label{eq:nsubsspan}
\{\tau_1^{(1/2)}, \tau_1^{(1)}, \tau_1^{(2)}, \tau_2^{(1/2)}, \tau_2^{(1)}, \tau_2^{(2)}, \cdots, \tau_{N-2}^{(1/2)}, \tau_{N-2}^{(1)}, \tau_{N-2}^{(2)}, \tau_{N-1}^{(1)}, \tau_{N-1}^{(2)}\},
\end{equation}
with $3N - 4$ elements needed to probe $N$-body phase space.  These are then used as inputs to a DNN.\end{itemize}
Both of these learning paradigms are expected to perform well, and we will see below that this is the case.
As a strawman, we also consider linear classification with the $N$-subjettiness basis elements in \Eq{eq:nsubsspan}, which is not expected to yield good performance.
For completeness, we also perform DNN classification with the energy flow basis.

We now summarize the technical details of these alternative jet tagging approaches.
For jet images, we create $33\times33$ jet images spanning $2R\times 2R$ in the rapidity-azimuth plane.
Motivated by \Ref{Komiske:2016rsd}, both single-channel ``grayscale'' jet images of the $p_T$ per pixel and two-channel ``color'' jet images consisting of the $p_T$ channel and particle multiplicity per pixel were used.
The $p_T$-channel of the jet image was normalized such that the sum of the pixels was one. 
Standardization was used to ensure that each pixel had zero mean and unit standard deviation by subtracting the training set mean and dividing by the training set standard deviation of each pixel in each channel.
A jet image CNN architecture similar to that used in \Ref{Komiske:2016rsd} was employed: three 36-filter convolutional layers with filter sizes of $8\times 8$, $4\times 4$, and $4\times 4$, respectively, followed by a 128-unit dense layer and a 2-unit softmaxed output. 
A rectified linear unit (ReLU) activation~\cite{nair2010rectified} was applied to the output of each internal layer.
Maxpooling of size $2\times2$ was performed after each convolutional layer with a stride length of 2.
The dropout rate was taken to be 0.1 for all layers.
He-uniform initialization~\cite{heuniform} was used to initialize the model weights. 

For the DNN (both for the $N$-subjettiness basis and for the \Bs), we use an architecture consisting of three dense layers of 100 units each connected to a 2-unit softmax output layer, with ReLU activation functions applied to the output of each internal layer.
For the training of all networks, 300k samples were used for training, 50k for validation, and 50k for testing.
Networks were trained using the Adam algorithm~\cite{adam} using categorical cross-entropy as a loss function with a learning rate of $10^{-3}$ and a batch size of 100 over a maximum of 50 epochs. 
Early stopping was employed, monitoring the validation loss, with a patience parameter of 5.
The python deep learning library {\tt Keras}~\cite{keras} with the {\tt Theano} backend~\cite{bergstra2010theano} was used to instantiate and train all neural networks.
Training of the CNNs was performed on Microsoft Azure using NVIDIA Tesla K80 GPUs and the NVIDIA CUDA framework.
Neural network performance was  checked to be mildly insensitive to these parameter choices, but these parameter choices were not tuned for optimality.
As a general rule, the neural networks used here are employed to give a sense of scale for the performance attainable with jet images and the $N$-subjettiness basis using out-of-the-box techniques; improvements in classification accuracy may be possible for these methods with additional hyperparameter tuning.

\subsection{$W$ tagging results and comparisons}
\label{sec:perf}

We present results for the $W$ tagging study here, with the other two classification problems discussed in \App{app:moretagging}.
The performance of a binary classifier is encapsulated by the background mistag rate $\varepsilon_b$ at a given signal efficiency $\varepsilon_s$.
For all of the figures below, we plot inverse receiver operator characteristic (ROC) curves, $1/\varepsilon_b$ as a function of $\varepsilon_s$, on a semi-log scale; a higher ROC curve indicates a better classifier.
The corresponding standard ROC ($\varepsilon_b$ vs.\ $\varepsilon_s$) and significance improvement ($\varepsilon_s/\sqrt{\varepsilon_b}$ vs.\ $\varepsilon_s$) curves are available in the source files of the {\tt arXiv} preprint as additional pages in the figure.

\begin{figure}[t]
\centering
\includegraphics[scale=.76]{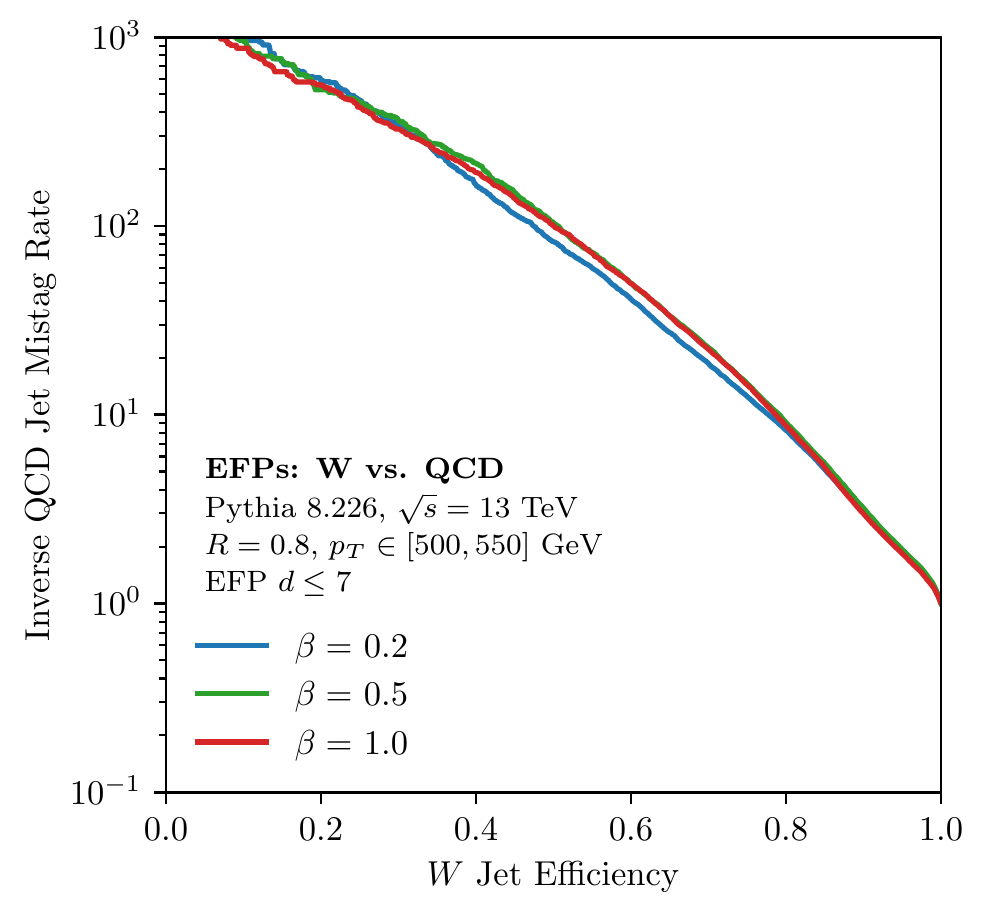}
\caption{Inverse ROC curves for linear $W$ tagging with the energy flow basis using different choices of angular exponent $\beta$ in \Eq{eq:hadronicmeasure}. Though the improvement is mild, $\beta=0.5$ shows the best overall performance. See \Fig{fig:appbetasweep} for the corresponding quark/gluon classification and top tagging results, where $\beta=0.5$ is also the best choice by a slight margin.}
\label{fig:Wbetasweep} 
\end{figure}

We begin by studying the performance for different choices of angular exponent $\beta$ in the default hadronic measure from \Eq{eq:hadronicmeasure}.
\Fig{fig:Wbetasweep} shows ROC curves for the choices of $\beta=0.2$, $\beta=0.5$, and $\beta=1$, using all \Bs with $d \le 7$.
The differences in performance are mild, but $\beta=0.5$ slightly improves the ROC curves for $W$ tagging, so we use $\beta = 0.5$ for the remainder of our studies.
The choice of $\beta = 0.5$ was also found to be optimal for the cases of quark/gluon and top tagging discussed in \App{app:moretagging}.

\begin{figure}[t]
\centering
\subfloat[]{\label{fig:Wefpsweep}\includegraphics[scale=.6]{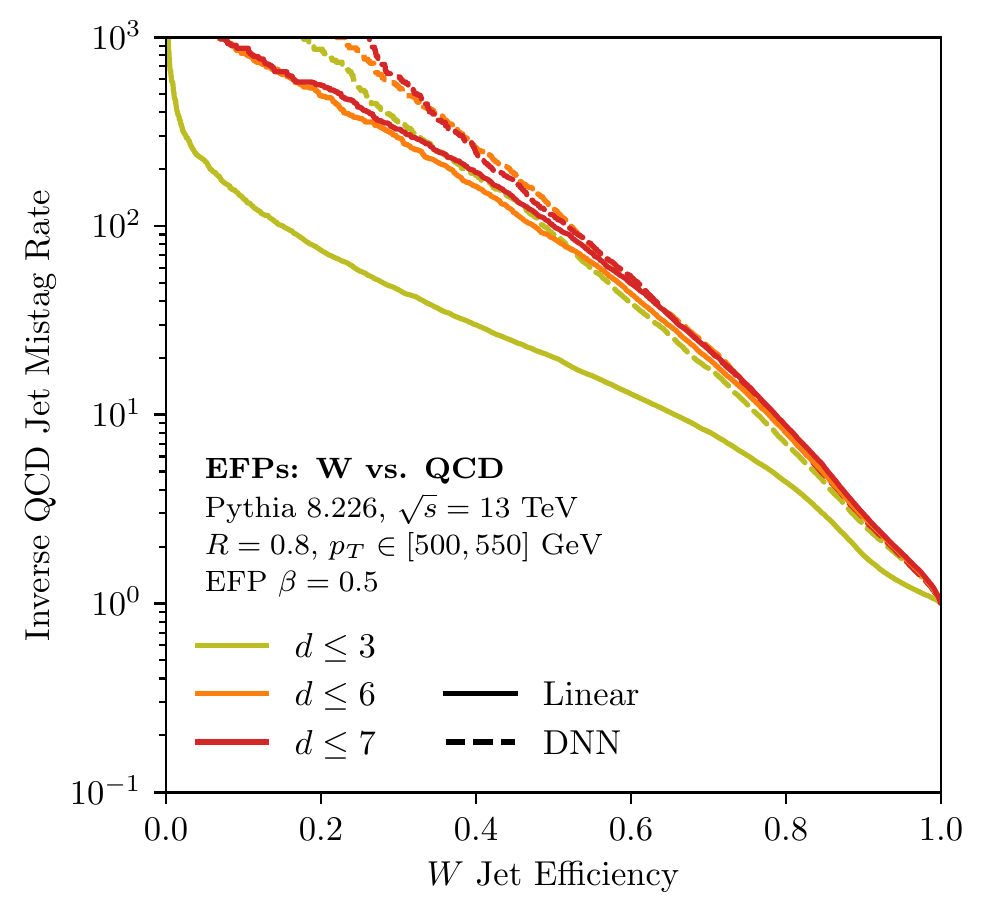}}
\subfloat[]{\label{fig:Wnsubsweep}\includegraphics[scale=.6]{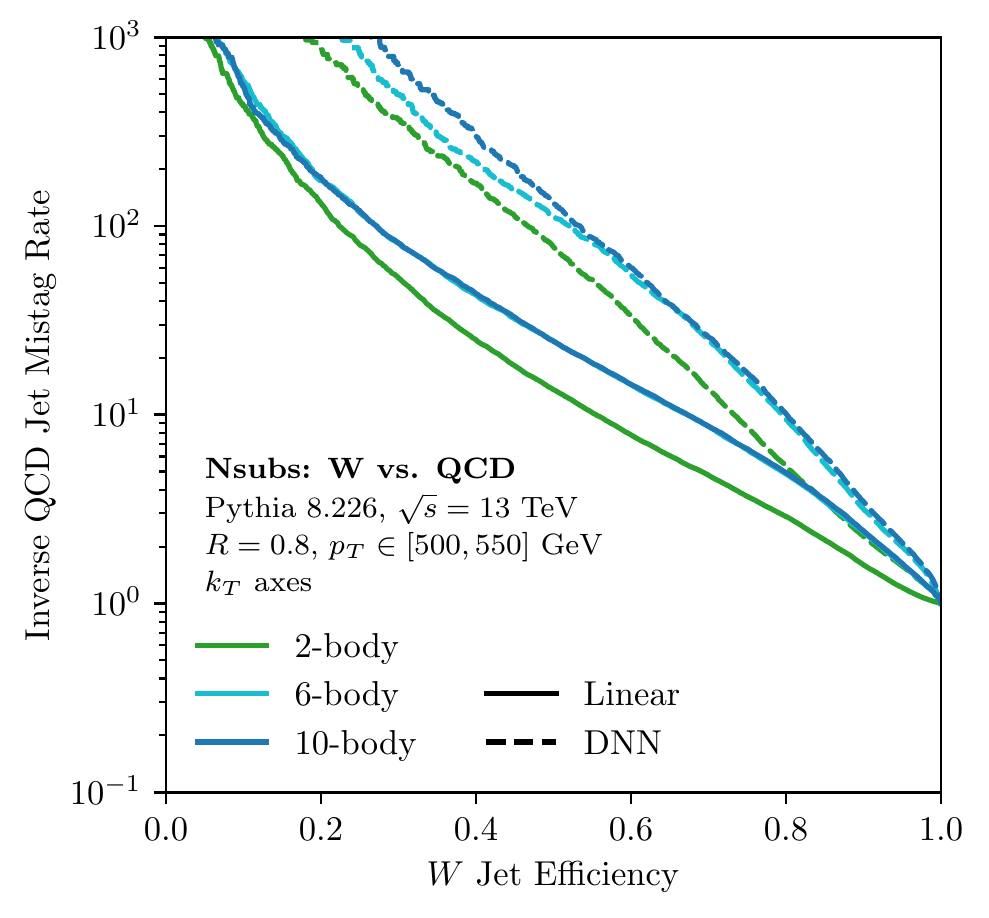}}
\caption{Inverse ROC curves for $W$ tagging with (a) the energy flow basis including degrees up to $d=7$ and (b) the $N$-subjettiness basis up to 10-body phase space information.  In both cases, we show the observables combined linearly (solid) and with a DNN (dashed).  The linear combinations of \Bs approach the nonlinear combinations, particularly for higher signal efficiencies, while the linear combinations of the $N$-subjettiness basis saturate well below the nonlinear combinations as the number of observables is increased. See \Fig{fig:appefpnsubsweep} for the corresponding quark/gluon classification and top tagging results.}
\label{fig:Wefpnsubsweep}
\end{figure}

Next, in \Fig{fig:Wefpsweep}, we test the linear spanning nature of the \Bs by comparing the ROC curves of the linear and nonlinear models trained on \Bs up to different $d$.
With linear regression, there is a large jump in performance in going from $d\le3$ (13 \Bs) to $d\le6$ (314 \Bs), and a slight increase in performance from $d\le6$ to $d\le7$ (1000 \Bs), indicating good convergence to the optimal IRC-safe observable for $W$ jet classification.
To avoid cluttering the plot, $d\le4$ and $d\le5$ are not shown in \Fig{fig:Wefpsweep}, but their ROC curves fall between those of $d\le3$ and $d\le6$, highlighting that the higher $d$ \Bs carry essential information for linear classification. 
By contrast, using nonlinear classification with a DNN, the \Bs performance with $d \le 3$ is already very good, since functions of the low $d$ \Bs can be combined in a nonlinear fashion to construct information contained in higher $d$ composite \Bs.
The linear and nonlinear performance is similar with the $d\le7$ \Bs for operating points of $\varepsilon_s\gtrsim0.5$, though the nonlinear DNN outperforms the linear classifier in the low signal efficiency region.
It should be noted that the linear classifier is not trained specifically for the low signal efficiency region and it may be possible that choosing a different hyperplane could boost performance there.
We leave to future work a more detailed investigation of optimizing the choice of linear classifier.

The performance of the $N$-subjettiness basis with both linear and nonlinear classifiers is shown in \Fig{fig:Wnsubsweep}.
For both linear classification and the DNN, performance appears to saturate with the 6-body (14 $\tau_N$s) phase space, with not much gained in going to 10-body (26 $\tau_N$s) phase space, except for a small increase in the low signal efficiency region for the DNN; we confirmed up to 30-body (86 $\tau_N$s) phase space that no change in ROC curves was observed compared to 10-body phase space.
As expected, there is relativity poor performance with linear classification even as the dimension of phase space is increased.
Classification with a DNN, though, shows an immense increase in performance over linear classification, as expected since the $N$-subjettiness basis is expected to nonlinearly capture all of the relevant IRC-safe kinematic information~\cite{Datta:2017rhs}.
This illustrates that nonlinear combinations of the $N$-subjettiness observables are crucial for extracting the full physics information.

The corresponding quark/gluon and top tagging plots in \Fig{fig:appefpnsubsweep} effectively tell the same story as \Fig{fig:Wefpnsubsweep}, robustly demonstrating the linear spanning nature of the \Bs used for classification across a wide variety of kinematic configurations.
As a side note, in \App{app:moretagging} there are sometimes cases where a linear combination of \Bs yields \emph{improved} performance compared to a DNN on the same inputs, particularly at medium to high signal efficiencies.
Since even a one-node DNN should theoretically be able to learn the linear combination of \Bs learned by the linear classifier, regimes where the linear classifier outperforms the DNN demonstrate the inherent difficulty of training complex multivariate models.

\begin{figure}[t]
\centering
\includegraphics[scale=.6]{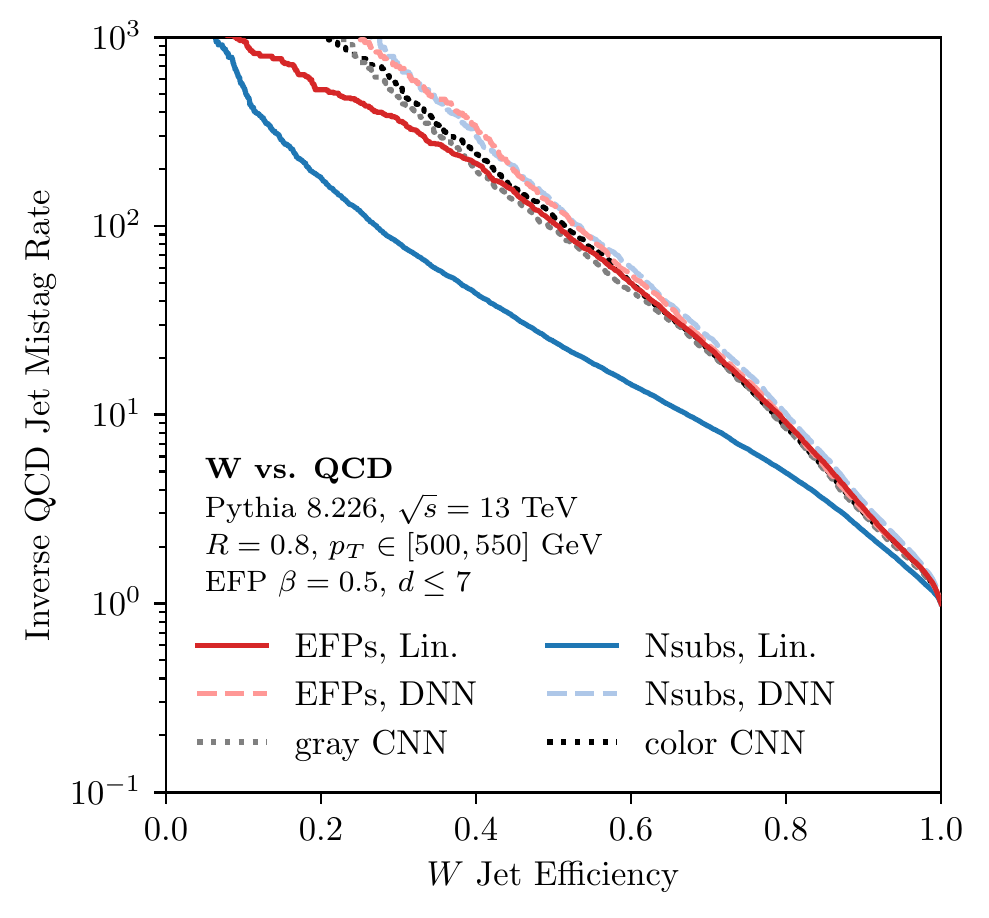}
\caption{Inverse ROC curves for $W$ tagging comparing six different methods:  linear and DNN classification with the energy flow basis up to $d\le 7$, linear and DNN classification with the $N$-subjettiness basis up to 10-body phase space, and grayscale and color jet images with CNNs. The most evident gap is between the linearly-combined $N$-subjettiness basis and the remaining curves, which achieve similar classification performance for medium and high signal efficiencies. See \Fig{fig:apptagcomp} for the corresponding quark/gluon classification and top tagging results.}
\label{fig:Wtagcomp}
\end{figure}

In \Fig{fig:Wtagcomp} we directly compare the \B classification power against the $N$-subjettiness basis and the 1-channel (``grayscale'') and 2-channel (``color'') CNNs. 
For operating points with $\varepsilon_s\gtrsim0.5$, all methods except the linear $N$-subjettiness classifier show essentially the same performance.
The worse performance of the linear \B classifier at low signal efficiencies is expected, since the Fisher linear discriminant is not optimized for that regime.
Overall, it is remarkable that similar classification performance can be achieved with these three very different learning paradigms, especially considering that the DNNs and grayscale CNN implicitly, and the color CNN explicitly, have access to non-IRC-safe information (including Sudakov-safe combinations of the IRC-safe inputs).
This agreement gives evidence that the tagging techniques have approached a global bound on the maximum possible classification power achievable, at least in the context of parton shower simulations.

Once again, the analogous quark/gluon and top tagging plots, shown in \Fig{fig:apptagcomp}, show very similar behavior to the $W$ tagging case in \Fig{fig:Wtagcomp}.
Linear classification with the \Bs performs similarly to the DNNs and CNNs, tending to slightly outperform at high signal efficiencies and underperform at low signal efficiencies.
Ultimately, the choice of tagging method comes down to a trade off between the simplicity of the inputs and the simplicity of the training method, with the \Bs presently requiring more inputs than the $N$-subjettiness basis but with the benefit of using a linear model.
In the future, we plan to study ways of reducing the size of the \B basis by exploiting linear redundancies among the \Bs and using powerful linear methods to automatically select the most important observables for a given task.

\subsection{Opening the energy flow box}
\label{sec:openbox}

\begin{figure}
\centering
\subfloat[]{\label{fig:Wefpnsweep:a}\includegraphics[scale=.76]{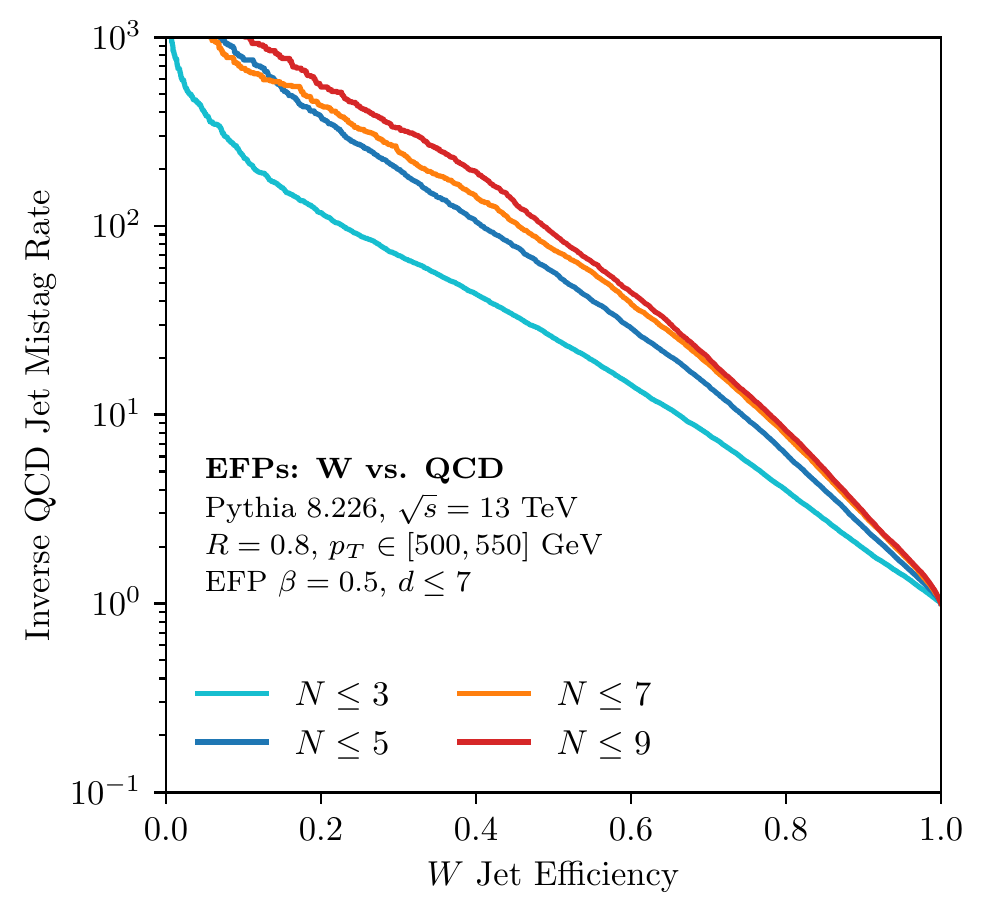}}
\subfloat[]{\label{fig:Wefpnsweep:b}\includegraphics[scale=.76]{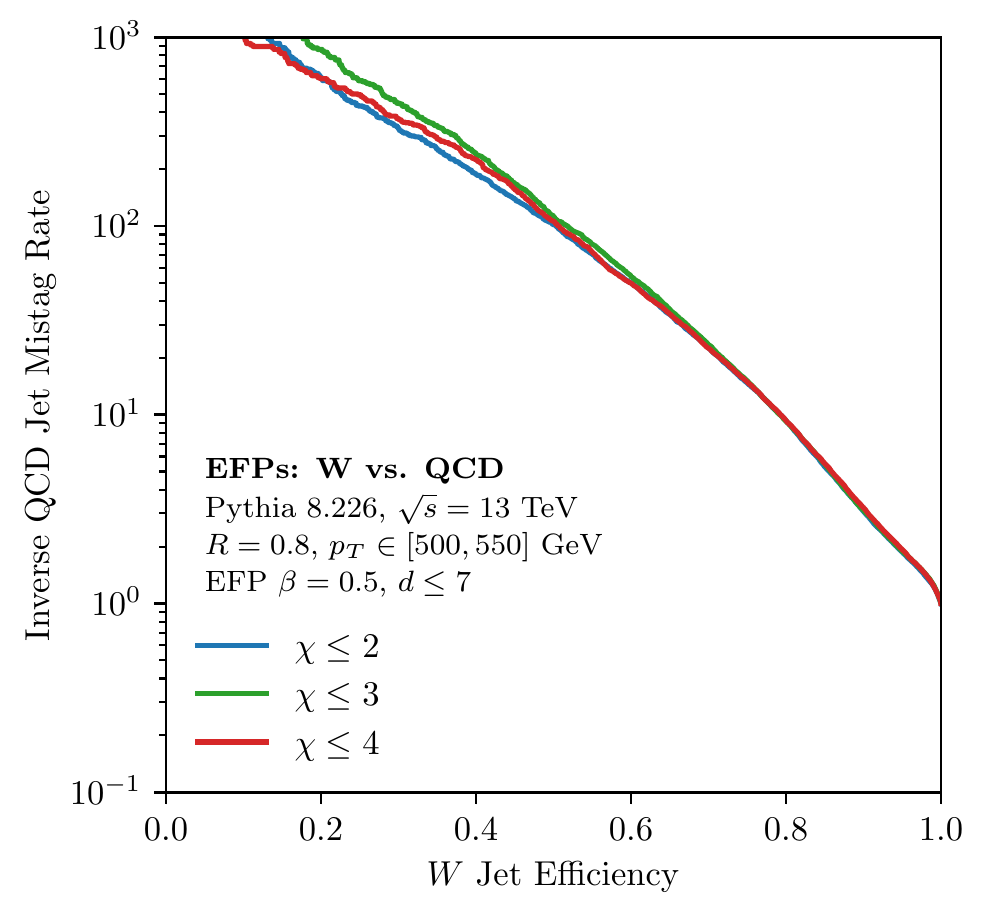}}
\caption{Inverse ROC curves for linear $W$ tagging with the energy flow basis with $d\le 7$, sweeping over (a) which $N$-point correlators and (b) observables of which VE computational complexity $\mathcal O(M^\chi)$ are included in the linear fit.  It is clear that important information is contained in the higher $N$-particle correlators, which can be included because the algorithm in \Sec{sec:complexity} evades the naive $\mathcal O(M^N)$ scaling.
See \Fig{fig:appnchisweep} for the corresponding quark/gluon classification and top tagging results.}
\label{fig:efpnsweep}
\end{figure}

As argued in \Eq{eq:linmodefp}, one of the main advantages of linear methods with the energy flow basis is that one can attempt to ``open the box'' and directly explore what features have been learned.
We leave to future work a full exploration of this possibility, but here we attempt to probe which topological structures within the \B basis carry the classification power for the different tagging problems.
Since we have shown that the \Bs with $d\le 7$ have sufficient classification power to qualitatively match the performance of alternative tagging methods, we will restrict to this set of observables.

In \Fig{fig:Wefpnsweep:a}, we vary the maximum number of vertices in the \B graphs, where the maximum $N$ is $14$ for $d \le 7$, finding that the performance roughly saturates at $N=9$, highlighting the importance of higher $N$ \Bs.
The algorithmic advances described in \Sec{sec:complexity} allow for the efficient computation of these higher $N$ \Bs, which have complexities as intractable as $\mathcal O(M^9)$ with the naive algorithm.
Additionally, note that nearly every \B (all except those corresponding to complete graphs) has a non-vanishing angular weighting function, which is a new feature compared to the ECFs and ECFGs (see \Sec{subsec:goingbeyond}).
In \Fig{fig:Wefpnsweep:b}, we vary the maximum computational complexity $\chi$ of the \B graphs, where the maximum $\chi$ is $4$ for $d \le 7$.
Remarkably, the full performance of linear classification with the $d\le7$ \Bs can be obtained with merely those observables calculable in $\mathcal O(M^2)$ with VE.
Thus, fortuitously for the purposes of jet tagging, it seems that restricting to the most efficiently computable \Bs (in the VE paradigm) is sufficient for extracting the near-optimal IRC-safe observable for jet classification.
Similar results hold for quark/gluon classification and top tagging, shown in \Fig{fig:appnchisweep}.

\section{Conclusions}
\label{sec:conclusion}

In this chapter, we have introduced the \Bs, which linearly span the space of IRC-safe observables.
The core argument, presented in \Sec{sec:basis}, is that one can systematically expand an arbitrary IRC-safe observable in terms of energies and angles and read off the unique resulting analytic structures.
This expansion yields a new way to understand the importance of $C$-correlators \cite{Tkachov:1995kk,Sveshnikov:1995vi,Cherzor:1997ak,Tkachov:1999py} for IRC safety, and it enables a powerful graph-theoretic representation of the various angular structures.
The multigraph correspondence makes manifest a more efficient algorithm than the naive $\mathcal O(M^N)$ one for computing \Bs, overcoming a primary obstacle to exploring higher-$N$ multiparticle correlators for jet substructure.

To demonstrate the power of the energy flow basis, we performed a variety of representative regression and classification tasks for jet substructure.
Crucially, linear methods were sufficient to achieve good performance with the \Bs.
As a not-quite apples-to-apples comparison in three representative jet tagging applications, linear classification with 1000 \Bs achieved comparable performance to a CNN acting on a jet image with $33 \times 33 = 1089$ pixels.
Because of the wide variety of linear learning methods available~\cite{bishop2006pattern}, we expect that the \Bs will be a useful starting point to explore more  applications in jet substructure and potentially elsewhere in collider physics.

There are many possible refinements and extensions to the energy flow basis.
In this chapter, we truncated the \Bs at a fixed maximum degree $d$; alternatively, one could truncate the prime \Bs at a fixed $d$ and compute all composite \Bs up to a specified cutoff.
Since the EFPs yield an overcomplete basis, it could be valuable to cull the list of required multigraphs.
A similar problem of overcompleteness was solved for kinematic polynomial rings in \Ref{Henning:2017fpj}, and that strategy may be relevant for EFPs with a suitable choice of measure.
In the other direction, it may be valuable to make the energy flow basis even more redundant by including EFPs with multiple measures.
With a vastly overcomplete basis, one could use techniques like lasso regression~\cite{tibshirani1996regression} to zero out unnecessary terms.
While we have restricted our attention to IRC-safe observables, it would be straightforward to relax the restriction to just infrared safety.
In particular, the set of IR-safe (but C-unsafe) functions in \Eq{eq:Cexample} can be expanded into multigraphs that have an extra integer decoration on each vertex to indicate the energy scaling.
Finally, the EFPs are based on an expansion in pairwise angles, but one could explore alternative angular expansions in terms of single particle directions or multiparticle factors.

To gain some perspective, we find it useful to discuss the \Bs in the broader context of machine learning for jet substructure.
Over the past few years, there has been a surge of interest in using powerful tools from machine learning to learn useful observables from low-level or high-level representations of a jet~\cite{Cogan:2014oua,deOliveira:2015xxd,Komiske:2016rsd,Almeida:2015jua,Baldi:2016fql,Kasieczka:2017nvn,Pearkes:2017hku,Butter:2017cot,Aguilar-Saavedra:2017rzt,Guest:2016iqz,Louppe:2017ipp,Datta:2017rhs, Baldi:2014kfa,Baldi:2014pta,Gallicchio:2010dq,Gallicchio:2012ez,Gallicchio:2011xq}.
The power of these machine learning methods is formidable, and techniques like neural networks and boosted decision trees have shifted the focus away from single- or few-variable jet substructure taggers to multivariate methods.
On the other hand, multivariate methods can sometimes obscure the specific physics information that the model learns, leading to recent efforts to ``open the box'' of machine learning tools~\cite{deOliveira:2015xxd,Butter:2017cot,Chang:2017kvc,Komiske:2017ubm,Metodiev:2018ftz}.
Even with an open box, though, theoretical calculations of multivariate distributions are impractical (if not impossible).
Furthermore, training multivariate models is often difficult, requiring large datasets, hyperparameter tuning, and preprocessing of the data.

The EFPs represent both a continuation of and a break from these machine learning trends.
The EFPs continue the trend from multivariate to hypervariate representations for jet information, with $\mathcal{O}(100)$ elements needed for effective regression and classification.  
On the other hand, the linear-spanning nature of the EFPs make it feasible to move away from ``black box'' nonlinear algorithms and return to simpler linear methods (explored previously for jet substructure in e.g.\ \cite{Thaler:2011gf,Cogan:2014oua}) without loss of generality.
Armed with the energy flow basis, there is a suite of powerful tools and ideas from linear regression and classification which can now be fully utilized for jet substructure applications, with simpler training processes compared to DNNs and stronger guarantees of optimal training convergence. 
Multivariate methods would ideally be trained directly on data to avoid relying on imperfect simulations, as discussed in \Ref{Komiske:2018oaa}.
The energy flow basis may be compelling for recent data-driven learning approaches~\cite{Komiske:2018oaa,Dery:2017fap, Metodiev:2017vrx} due to its completeness, the simplicity of linear learning algorithms, and a potentially lessened requirement on the size of training samples.

As with any jet observable, the impact of non-perturbative effects on the \Bs is important to understand.
Even with IRC safety, hadronization modifies the distributions predicted by pQCD and therefore complicates first-principles calculations. 
It would be interesting to see if the shape function formalism~\cite{Korchemsky:1999kt,Korchemsky:2000kp} could be used to predict the impact of non-perturbative contributions to \B distributions.
Alternatively, one standard tool that is used to mitigate non-perturbative effects is jet grooming~\cite{Butterworth:2008iy,Ellis:2009su,Ellis:2009me,Krohn:2009th,Dasgupta:2013ihk,Larkoski:2014wba}, which also simplifies first-principles calculations and allows for ``quark'' and ``gluon'' jets to be theoretically well-defined~\cite{Frye:2016aiz}.
We leave a detailed study of the effects of non-perturbative contributions and jet grooming on \Bs to future work.

Eventually, one hopes that the EFPs will be amenable to precision theoretical calculations of jet substructure (see e.g.\ \Refs{Feige:2012vc, Dasgupta:2013ihk,Dasgupta:2013via, Larkoski:2014tva, Procura:2014cba,Larkoski:2015kga,Dasgupta:2015lxh,Frye:2016aiz,Dasgupta:2016ktv,Marzani:2017mva,Larkoski:2017cqq}).
This is by no means obvious, since generic EFPs have different power-counting structures from the ECFs~\cite{Larkoski:2013eya} or ECFGs~\cite{Moult:2016cvt}.
That said, phrasing jet substructure entirely in the language of energy flow observables and energy correlations may provide interesting new theoretical avenues to probe QCD, realizing the $C$-correlator vision of \Refs{Tkachov:1995kk,Sveshnikov:1995vi,Cherzor:1997ak,Tkachov:1999py}.
Most IRC-safe jet observables rely on particle-level definitions and calculations, but there has been theoretical interest in directly analyzing the correlations of energy flow in specific angular directions~\cite{basham1978energy,basham1979energy,Belitsky:2001ij}, particularly in the context of conformal field theory~\cite{Hofman:2008ar,Engelund:2012re,Zhiboedov:2013opa,Belitsky:2013xxa,Belitsky:2013bja}.
The energy flow basis is a step towards connecting the particle-level and energy-correlation pictures, and one could even imagine that the energy flow logic could be applied directly at the path integral level.
Ultimately, the structure of the EFPs is a direct consequence of IRC safety, resulting in a practical tool for jet substructure at colliders as well as a new way of thinking about the space of observables more generally.

\chapter{A Definition of Particle Flavor from Observables}

\section{Introduction}
\label{sec:intro}

Quarks and gluons are fundamental, color-charged particles that are copiously produced at colliders like the Large Hadron Collider (LHC).
Despite their ubiquity, these high-energy quarks and gluons are never observed directly.
Instead, they fragment and hadronize into sprays of color-neutral hadrons, known as \emph{jets}, via quantum chromodynamics (QCD).
As the majority of jets originate from light (up, down, strange) quarks or gluons, a firm understanding of quark and gluon jets is important to many analyses at the LHC.
There has been tremendous recent theoretical and experimental progress in analyzing jets and jet substructure~\cite{Seymour:1991cb,Seymour:1993mx,Butterworth:2002tt,Butterworth:2007ke,Butterworth:2008iy,Abdesselam:2010pt,Altheimer:2012mn,Altheimer:2013yza,Adams:2015hiv,Larkoski:2017jix,Asquith:2018igt}, with a variety of observables~\cite{Berger:2003iw,Almeida:2008yp,Ellis:2010rwa,Thaler:2010tr,Thaler:2011gf,Krohn:2012fg,Larkoski:2013eya,Larkoski:2014uqa,Larkoski:2014pca,Moult:2016cvt,Komiske:2017aww} and algorithms~\cite{Krohn:2009th,Ellis:2009me,Ellis:2009su,Dasgupta:2013ihk,Larkoski:2014wba} developed to expose and probe the underlying physics.
Despite decades of using the notions of ``quark'' and ``gluon'' jets \cite{Nilles:1980ys,Jones:1988ay,Fodor:1989ir,Jones:1990rz,Lonnblad:1990qp,Pumplin:1991kc,Gallicchio:2011xq,Gallicchio:2012ez, Bhattacherjee:2015psa,FerreiradeLima:2016gcz,Bhattacherjee:2016bpy,Komiske:2016rsd,Davighi:2017hok,Cheng:2017rdo,Sakaki:2018opq}, a precise and practical hadron-level definition of jet flavor has not been formulated.

Even setting aside the issue of jet flavor, ambiguity is already present whenever one wants to identify jets in an event~\cite{Salam:2009jx}.
Nonetheless, jets can be made perfectly well-defined: any hadron-level algorithm for finding jets that is infrared and collinear (IRC) safe provides an operational jet definition that can be compared to perturbative predictions.
While different algorithms result in different jets, specifying a jet algorithm allows one to make headway into comparing theoretical calculations and experimental measurements.
Meanwhile, in the case of jet flavor, the lack of a precise, hadron-level definition of ``quark'' and ``gluon'' jets has artificially hindered progress by precluding separate comparisons of quark and gluon jets between theory and experiment.

Typical applications involving ``quark'' and ``gluon'' jets in practice often rely on ill-defined or unphysical parton-level information, such as from the event record of a parton shower event generator.
Progress has been made in providing sharp definitions at the parton-level~\cite{Banfi:2006hf,Buckley:2015gua}, in the context of factorization theorems~\cite{Gallicchio:2011xc,Frye:2016okc,Frye:2016aiz}, and at the conceptual level~\cite{Badger:2016bpw}, but an operational definition, to our knowledge, has never been developed (see \Ref{Gras:2017jty} for a review).
A quark/gluon jet definition%
\footnote{While in some contexts ``jet definition'' means a procedure for finding jets in an event, in this chapter we use ``quark/gluon jet definition'' to mean a definition of jet flavor.}
should ideally work at the hadron level, regardless of whether a rigorous factorization theorem exists, and be practically implementable in both theoretical and experimental settings.

In this chapter, we develop an operational definition of quark and gluon jets that is formulated solely in terms of experimentally-accessible quantities, does not rely on specific theoretical constructs such as factorization theorems, and can be readily implemented in a realistic context.
Intuitively, we define quark and gluon jets as the ``pure'' categories that emerge from two different jet samples.
Our definition operates at the aggregate level, avoiding altogether the troublesome and potentially impossible notion of a per-jet flavor label in favor of quantifying quark and gluon jets by their distributions.

Specifically, given two jet samples $M_1$ and $M_2$ (e.g.\ $Z+$jet and dijet) in a narrow transverse momentum $(p_T)$ bin, with $M_1$ taken to be more ``quark''-like, and a jet substructure feature space $\mathcal O$, we define quark ($q$) and gluon ($g$) jet distributions in the following way:
\begin{align}
\label{eq:opdefintro}
&p_q(\O)\equiv\frac{p_{M_1}(\O)-\kappa_{12}\,p_{M_2}(\O)}{1-\kappa_{12}},&&p_g(\O)\equiv\frac{p_{M_2}(\O)-\kappa_{21}\,p_{M_1}(\O)}{1-\kappa_{21}},
\end{align}
where $\kappa_{12}$ and $\kappa_{21}$ are known as \emph{reducibility factors} and are directly obtainable from the probability distributions $p_{M_1}(\O)$ and $p_{M_2}(\O)$.
The reducibility factors are defined as:
\begin{align}\label{eq:opkappaintro}
&\kappa_{12} \equiv \min_\mathcal O \frac{p_{M_1}(\mathcal O)}{p_{M_2}(\mathcal O)}, &&\kappa_{21} \equiv \min_\mathcal O \frac{p_{M_2}(\mathcal O)}{p_{M_1}(\mathcal O)}.
\end{align}
The reducibility factors in \Eq{eq:opkappaintro} identify the most $M_1$-like and $M_2$-like regions of the substructure phase space by extremizing the sample likelihood ratio.
We take these phase space regions to \emph{define} what it means to be quark-like and gluon-like.
The subtractions in \Eq{eq:opdefintro} then proceed to ``demix'' the two sample distributions as if they were statistical mixtures.
The quark and gluon distributions are defined solely in terms of hadronic fiducial cross section measurements of the two samples, ensuring that our definition is manifestly fully data-driven and non-circular.
This definition relies on a jet algorithm to define the jets in the jet samples, which also allows for further hadron-level processing, such as jet grooming techniques~\cite{Krohn:2009th,Ellis:2009me,Ellis:2009su,Dasgupta:2013ihk,Larkoski:2014wba}, to be folded directly into the quark/gluon jet definition.

One main goal of this chapter is to argue that our operational definition, combined with existing tools, provides a way to obtain information about the likelihood, quark fractions, and quark and gluon distributions in a fully data-driven way, without reference to unphysical notions such as generator labels.
The concepts appearing in our definition are directly related to methods already in use in experimental quark/gluon jet analysis efforts~\cite{CMS-PAS-JME-13-002,Aad:2014gea,Aad:2016oit,CMS-DP-2016-070,ATL-PHYS-PUB-2017-009,Sirunyan:2018asm}.
Quark-gluon likelihood ratios, obtained from parton shower generators, have been implemented by both ATLAS and CMS as optimal discriminants in low-dimensional feature spaces.
Quark fractions, obtained from event generators, for several jet samples have successfully allowed for separate determination of quark and gluon jet properties by solving linear equations.
These analyses already use a statistical-mixture picture of quark and gluon jets, which is a direct consequence of our definition.


Many physics analyses at the LHC would benefit from a clear definition of quark and gluon jets that allows for unambiguous extraction of separate quark and gluon jet distributions and fractions.
Fully data-driven quark/gluon jet taggers have the potential to increase the sensitivity of a variety of new physics searches~\cite{FerreiradeLima:2016gcz,Bhattacherjee:2016bpy}, and related ideas have been developed for model-independent searches for new physics~\cite{Collins:2018epr}.
Experimentally measuring separate quark and gluon distributions of jet observables would significantly improve attempts to extract the strong coupling constant from jet substructure~\cite{Bendavid:2018nar} and to constrain parton shower event generators~\cite{Reichelt:2017hts,Gras:2017jty}.
Extracting data-driven fractions of quark and gluon jets could improve the determination of parton distribution functions and allow for separate measurement of quark and gluon cross sections.
These ideas may also be relevant in the context of heavy ion collisions, where quarks and gluons are expected to be modified differently by the medium and probing the separate modifications to quark and gluon jets would be of significant interest.


We now give a brief summary of the rest of this chapter.
In \Sec{sec:def}, we provide a self-contained overview, motivation, and exploration of our quark/gluon jet definition.
We discuss recent work in \Ref{Gras:2017jty} that developed a ``conceptual'' definition of quark/gluon jets, falling short of providing a full definition that can be reliably used in practice, but highlighting the key elements required of a sensible quark/gluon jet definition.
We then develop the intuition and mathematical tools necessary to construct our operational definition, which satisfies the core conceptual principles while being precise and practically implementable.
After stating our operational definition, we examine its physical and statistical properties in detail.
An exploration of the definition in the context of simple jet substructure observables at leading-logarithmic accuracy is left to \App{sec:explore}.

In \Sec{sec:topicwola}, we discuss how our quark/gluon jet definition benefits from, and provides a foundation for, recent work on data-driven machine learning for jet physics.
The classification without labels (CWoLa) paradigm~\cite{Metodiev:2017vrx} for training classifiers on mixed samples can be used to approximate the mixed-sample likelihood ratio, a key part of implementing our definition.
The jet topics framework~\cite{Metodiev:2018ftz} extracts underlying mutually irreducible distributions from mixture histograms, yielding a practical method to obtain the reducibility factors in \Eq{eq:opkappaintro}.
Using jet topics with the approximated mixed-sample likelihood ratio, obtained from the data via CWoLa, allows for more robust fraction and distribution extraction.
With quark fractions, obtained from the data via jet topics, CWoLa classifiers can be (self-)calibrated in a fully data-driven way.
More broadly, the assumptions required for CWoLa and jet topics---that QCD jet samples are statistical mixtures of mutually irreducible quark and gluon jets---are satisfied by construction with our definition.

In \Sec{sec:qgex}, we showcase a practical implementation of our definition using jet samples from two different processes: $Z$+jet and dijets.
Using six trained models detailed in \App{sec:train}, we apply the procedure outlined in \Sec{sec:topicwola} to extract quark fractions by combining the CWoLa and jet topics methods, finding more robust performance than when using single jet substructure observables.
With the reducibility factors and quark fractions in hand, we extract separate quark and gluon distributions for a variety of jet substructure observables, even those that do not exhibit mutual irreducibility.
We compare the results of using our data-driven definition of quark and gluon jets with a per-jet \pythia-parton definition, finding qualitative and quantitative agreement between the two.
The potential to self-calibrate CWoLa classifiers is also shown with an explicit example.
While our studies are based on parton-shower samples, all of these analyses can be performed in data with the experimental tools already developed for quark and gluon jet physics at the LHC.

We present our conclusions in \Sec{sec:conc}, discussing potential new applications made feasible by this work.
Possible future developments and extensions are highlighted.
%
\clearpage

\section{Defining quark and gluon jets}
\label{sec:def}

\subsection{Review of a conceptual quark/gluon jet definition}
\label{sec:conceptdef}

Due to the complicated radiative showering and fundamentally non-perturbative hadronization that occurs in the course of jets emerging from partons, there is no unambiguous definition of ``quark'' or ``gluon'' jets at the hadron-level.
Despite this challenge, the importance of a clear, well-defined, and practical definition of quark and gluon jets at modern colliders cannot be overstated.
In \Ref{Gras:2017jty}, a significant effort was made to summarize and comment on the concepts of ``quark jet'' and ``gluon jet''.
The authors of \Ref{Gras:2017jty} settled on the following statement as the best way to conceptually define quark jets (and, analogously, gluon jets):

\begin{conceptdef}
A phase space region (as defined by an unambiguous hadronic fiducial cross section measurement) that yields an enriched sample of quarks (as interpreted by some suitable, though fundamentally ambiguous criterion).
\end{conceptdef}

This definition is attractive for numerous reasons.
First, it is explicitly tied to hadronic final states, avoiding dependence, for example, on the unphysical event record of a parton shower generator.
Further, it is specific to the context of a particular measurement and is thus defined regardless of whether the observable and processes in question have rigorous factorization theorems.
Finally, its goal is to tag a region of phase space as quark- or gluon-like rather than to specify a per-jet truth definition of quark and gluon jets.
The main difficulty with this conceptual definition, as noted in \Ref{Gras:2017jty}, is determining the criterion that corresponds to successful quark or gluon jet enrichment.

Despite its attractive qualities, without a practical proposal for implementing this conceptual definition on data, the case studies in \Ref{Gras:2017jty} operationally fell back on less well-defined definitions, such as using initiating parton information from a parton shower generator to tag a quark/gluon jet.
Further, the definition only tags specific regions of phase space as ``quark'' or ``gluon'', such as low or high values of some substructure observable, and provides no framework for discussing jet flavor outside of these regions.
To remedy this issue, we seek to upgrade the conceptual definition to an operational one by giving a concrete, data-driven method for optimally identifying quark- or gluon-enriched regions of phase space and obtaining full quark and gluon jet distributions.

\subsection{Motivating the operational definition}
\label{sec:motivation}

To motivate our definition, suppose that we have two QCD jet samples $M_1$ and $M_2$ in a narrow $p_T$ bin.
One of the mixed samples ($M_1$ without loss of generality) should be ``quark-enriched'' and the other ``gluon-enriched'' relative to each other according to some qualitative criterion.
\Ref{Gras:2017jty} took $M_1$ and $M_2$ to be, respectively, $Z$+jet and dijet samples, a case that we further investigate in \Sec{sec:qgex}.

Assume for now that $M_1$ and $M_2$ are statistical mixtures of quark and gluon jets---an assumption that will \emph{not} be made in our final definition.
Letting the quark fractions of the two mixtures be $f_1$ and $f_2$, the relationship between the distribution of substructure observables in  mixture $M_i$ in terms of the quark and gluon jet distributions is:
\begin{equation}
\label{eq:mixps}
p_{M_i}(\O)=f_i \, p_q(\O)+(1-f_i)\,p_g(\O),
\end{equation}
where the feature space $\O$ is, for our purposes, a set of jet substructure observables taken to be sufficiently rich to encode all relevant information about jet flavor.

\begin{figure}[t]
\centering
\includegraphics[scale=0.9]{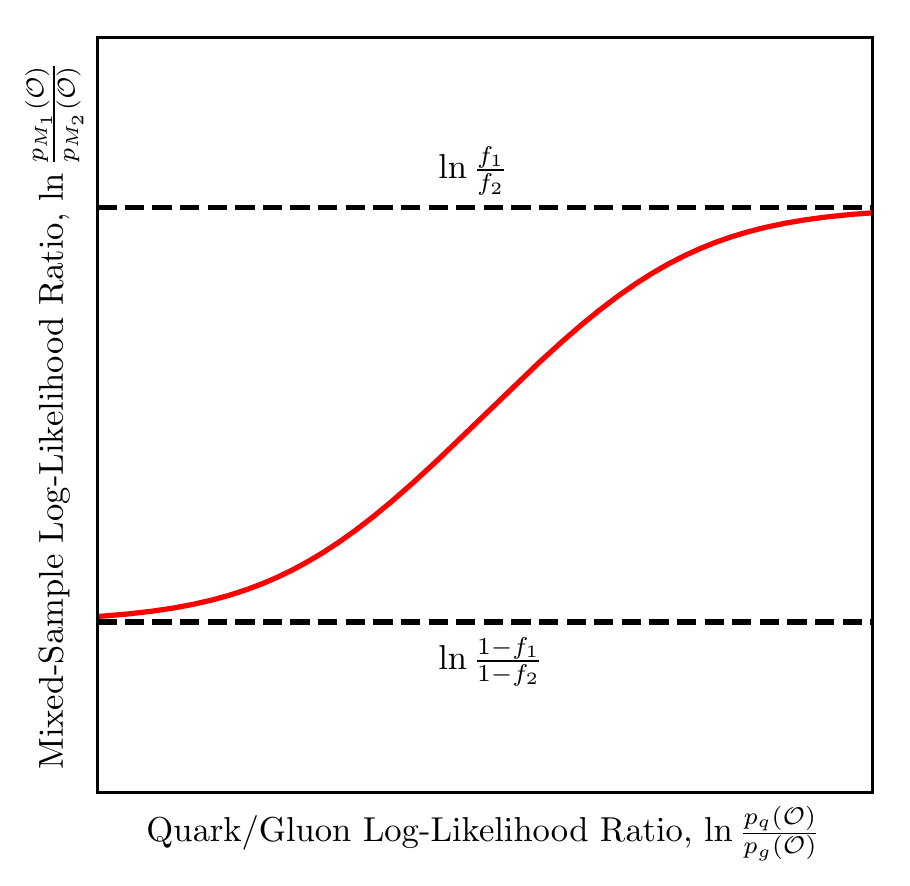}
\caption{
The monotonic relationship between the mixed-sample log-likelihood ratio and the quark-gluon log-likelihood ratio from \Eq{eq:mixnp} for  illustrative fraction values.
The relationship between the maximum and minimum values of the mixed-sample and quark/gluon log-likelihoods from \Eq{eq:kfs} is visually evident in that the red curve horizontally asymptotes to the two black dashed curves.
The plots are shown in terms of the logarithms of the likelihood ratios so that exchanging $M_1\leftrightarrow M_2$ or $q\leftrightarrow g$ simply corresponds to a reflection of the curve.} 
\label{fig:mixloglikes}
\end{figure}

Following the outline of the Conceptual Definition, we consider classification of quark and gluon jets and examine the relationship of this task with classification of one mixture from the other.
By the Neyman-Pearson lemma~\cite{NPlemma}, an optimal classifier for discriminating two classes is their likelihood ratio (or any monotonically-related quantity).
In the case of quark and gluon jets, the likelihood ratio is:
\begin{equation}
\label{eq:np}
L_{q/g}(\O)\equiv\frac{p_q(\O)}{p_g(\O)},
\end{equation}
and, similarly, the optimal classifier for discriminating between $M_1$ and $M_2$ is:
\begin{equation}
\label{eq:mixnp}
L_{M_1/M_2}(\O)\equiv\frac{p_{M_1}(\O)}{p_{M_2}(\O)}=\frac{f_1 \, L_{q/g}(\O)+(1-f_1)}{f_2 \, L_{q/g}(\O)+(1-f_2)}.
\end{equation}
It is easily verified that the mixed-sample likelihood ratio in \Eq{eq:mixnp} is a monotonic function of the quark-gluon likelihood ratio in \Eq{eq:np} as long as $f_1\neq f_2$ (see \Refs{blanchard2016classification,Metodiev:2017vrx}).
The relationship between the mixed-sample likelihood ratio and the quark-gluon likelihood ratio of \Eq{eq:mixnp} is depicted in \Fig{fig:mixloglikes}.
This cleanly demonstrates that the optimal mixed-sample classifier is also the optimal quark-gluon classifier.

Supposing that we can approximate the mixture likelihood ratio sufficiently well, we have distilled the (potentially huge) substructure feature space to a single number which is provably optimal for identifying quark- and gluon-enriched phase space regions.
However, we still lack a procedure for actually identifying the enriched regions; we solely know that they are given by some cut on $L_{q/g}(\O)$, or equivalently a cut on $L_{M_1/M_2}(\O)$.
The key insight for moving closer toward an operational definition is that $L_{q/g}(\O)$, being the optimal discriminant of quark and gluon jets, can be immediately used to identify the most quark-enriched (gluon-enriched) regions as those where $L_{q/g}(\O)$ is at its maximum (minimum). 
In the case that we can find regions of phase space $\mathcal O_q$ and $\mathcal O_g$ where quark and gluon jets respectively are pure, we have that $L_{q/g}(\O_g)=0$ and $L_{g/q}(\O_q)=0$ and we say that the quark and gluon categories are \emph{mutually irreducible} (see \Ref{blanchard2016classification,Metodiev:2018ftz}).

The extrema of the quark/gluon likelihood ratio $L_{q/g}$, corresponding to the enriched regions of phase space, are naturally related to the extrema of the mixture likelihood ratio $L_{M_1/M_2}$.
To this end, it is helpful to define the \emph{reducibility factor} between distributions $A$ and $B$, $\kappa_{AB}$, as:
\begin{equation}
\label{eq:kappa}
\kappa_{AB}\equiv\min_{\O}\frac{p_A(\O)}{p_B(\O)},
\end{equation}
which is the minimum (or more precisely, the infimum) of the likelihood ratio of $A$ and $B$.
Supposing that quarks and gluons are mutually irreducible in the feature space $\O$, the reducibility factors of quark jets to gluon jets (and vice versa) vanish:
\begin{equation}
\label{eq:qgks}
\textbf{Quark and Gluon Jet Mutual Irreducibility}:\quad\quad\kappa_{qg}=0,\quad\quad\kappa_{gq}=0.
\end{equation}

We now show how, assuming quark/gluon mutual irreducibility, the mixture reducibility factors can be related to mixture fractions.
The reducibility factors of the mixed samples can be written down by treating them as mixtures of quarks and gluons as in \Eq{eq:mixps}:
\begin{equation}
\label{eq:mixkappas}
\kappa_{M_iM_j}=\min_{\O} L_{M_i/M_j}(\O)=\min_{\O}\frac{f_i \, L_{q/g}(\O)+(1-f_i)}{f_j \, L_{q/g}(\O)+(1-f_j)}.
\end{equation}
Using our assumptions that $M_1$ is quark-enriched relative to $M_2$, we can write \Eq{eq:mixkappas} as a relation between the mixed-sample reducibility factors and the quark/gluon reducibility factors:
\begin{equation}
\label{eq:mixqgkappas}
\kappa_{M_1M_2}=\frac{f_1\,\kappa_{qg}+(1-f_1)}{f_2\,\kappa_{qg}+(1-f_2)},\quad\quad \kappa_{M_2M_1}=\frac{f_2+(1-f_2)\, \kappa_{gq}}{f_1+(1-f_1)\, \kappa_{gq}},
\end{equation}
where the monotonicity of $L_{M_i/M_j}(\O)$ with $L_{q/g}(\O)$ has been used to push the minimum operation onto the quark-gluon likelihood ratio in \Eq{eq:mixkappas}.
If quarks and gluons are mutually irreducible, we can plug \Eq{eq:qgks} into \Eq{eq:mixqgkappas} to find the reducibility factors of the mixtures:\footnote{An analogous analysis carries through even if non-zero reducibility factors $\kappa_{qg}$ and $\kappa_{gq}$ are specified.}
\begin{equation}
\label{eq:kfs}
\kappa_{12}\equiv\kappa_{M_1M_2}=\frac{1-f_1}{1-f_2},\quad\quad\kappa_{21}\equiv\kappa_{M_2M_1}=\frac{f_2}{f_1}.
\end{equation}
\Fig{fig:mixloglikes} demonstrates that \Eq{eq:mixkappas} defines the asymptotic behavior of the mixed-sample log-likelihood ratio.

Combining the reducibility factors of \Eq{eq:kfs} with the mixture relationship of \Eq{eq:mixps}, we can solve for the underlying quark and gluon jet distributions solely in terms of the well-defined mixture distributions $p_{M_i}(\mathcal O)$ and mixture reducibility factors $\kappa_{ij}$:
\begin{equation}
\label{eq:pqgks}
p_q(\O)=\frac{p_{M_1}(\O)-\kappa_{12}\,p_{M_2}(\O)}{1-\kappa_{12}},\quad\quad p_g(\O)=\frac{p_{M_2}(\O)-\kappa_{21}\,p_{M_1}(\O)}{1-\kappa_{21}}.
\end{equation}
Remarkably, \Eq{eq:pqgks} exposes the underlying quark and gluon jet distributions in terms of experimentally well-defined quantities such as the distribution of jets in mixed samples and their reducibility factors.
Notice also that the quark and gluon distributions each depend on only one of the two mixed-sample reducibility factors. 
Thus, even if only one reducibility factor can be reliably extracted, the corresponding quark or gluon jet distribution can nevertheless be obtained.

Here, we have made several simplifying assumptions, namely that quark and gluon jets can be made well-defined, that $M_1$ and $M_2$ are statistical mixtures of quark and gluon jets, and that quark and gluon jets are mutually irreducible in the feature space $\O$.
\Eq{eq:pqgks} then followed as a consequence, demonstrating that, under these assumptions, it is possible to get access to pure quark and gluon distributions.
What if, on the contrary, we do not make these assumptions, while also requiring that our definition of quark and gluon jets not be circular?
We now proceed to thoroughly explore this idea.

\subsection{An operational definition of quark and gluon jets}
\label{sec:opdef}

We now provide our \emph{operational definition} of quark and gluon jets that builds upon the Conceptual Definition in \Sec{sec:conceptdef} but can be used for practical applications at the LHC and future colliders.
We begin by stating the definition in terms of the notation developed in \Sec{sec:motivation}, and then we proceed to a detailed discussion of its features.

In the absence of any certainty about the underlying structure of samples $M_1$ and $M_2$, we choose to start at the end of \Sec{sec:motivation}, letting \Eq{eq:pqgks} provide a fully-operational definition of quark and gluon jets in terms of experimentally well-defined quantities:
\begin{framed}
\begin{opdef}
Given two samples $M_1$ and $M_2$ of QCD jets at a fixed $p_T$ obtained by a suitable jet-finding procedure, taking $M_1$ to be ``quark-enriched'' compared to $M_2$, and a jet substructure feature space $\O$, the quark and gluon jet distributions are defined to be:
\begin{align}
\label{eq:opdef}
p_q(\O)\equiv\frac{p_{M_1}(\O)-\kappa_{12}\,p_{M_2}(\O)}{1-\kappa_{12}},&&&p_g(\O)\equiv\frac{p_{M_2}(\O)-\kappa_{21}\,p_{M_1}(\O)}{1-\kappa_{21}},
\end{align}
where $\kappa_{12}$, $\kappa_{21}$, $p_{M_1}(\O)$, and $p_{M_2}(\O)$ are directly obtainable from $M_1$ and $M_2$.
\end{opdef}
\end{framed}

There are two immediate points to note about the Operational Definition.
First, it does not attempt to define quark and gluon jets at the level of individual jets, but rather it defines them in aggregate as two well-defined probability distributions.
This is in keeping with the spirit of the Conceptual Definition in \Sec{sec:conceptdef}, which sought to identify enriched regions of phase space rather than to determine a per-jet truth label.
It is also in concert with the basic construction of quantum field theory, which only provides theoretical access to distributional quantities such as cross sections rather than making predictions for individual events.\footnote{Note that (non-deterministic) per-jet labels can be obtained from this definition if needed. For a jet with observable value $O$, one can assign it a ``quark'' label with probability $f\, p_q(O)/(f\, p_q(O) + (1-f)\, p_g(O))$ by using the extracted quark and gluon distributions, $p_q$ and $p_g$, and extracted quark fraction $f$ of the sample. These labels are universal if the observable is monotonically related to the likelihood ratio.}

Second, the Operational Definition does not rely on assumptions of mutual irreducibility of quarks and gluons or the factorization of jet samples as mixtures, instead turning them into derived properties of the definition, as we show below.
In the limit where factorization holds and quarks and gluons are mutually irreducible in the feature space $\mathcal O$, the Operational Definition returns precisely the quark and gluon jets which make sense in that context.
Outside of these potentially-restrictive limits, the definition nonetheless returns two well-defined categories which can be fairly called quark and gluon jets.
The Operational Definition essentially takes the vague notion of ``quark-like'' from the Conceptual Definition and injects mathematical substance by specifying how to extract the quark and gluon distributions.

With the Operational Definition in hand, we now turn the reasoning of \Sec{sec:motivation} on its head to \emph{derive} the mutual irreducibility of quarks and gluons and the mixture nature of the two jet samples $M_1$ and $M_2$.
Using the quark/gluon jet definition in \Eq{eq:opdef}, we can write down the quark/gluon reducibility factors as:
\begin{equation}
\label{eq:derivekqgs}
\kappa_{qg}=\min_{\O}L_{q/g}(\O)=\min_{\O}\frac{(1-\kappa_{21})(L_{M_1/M_2}(\O)-\kappa_{12})}{(1-\kappa_{12})(1-\kappa_{21}L_{M_1/M_2}(\O))}=0,
\end{equation}
where we have used the monotonicity of $L_{q/g}(\mathcal O)$ in $L_{M_1/M_2}(\O)$ and the definition of $\kappa_{12}$ to see that the numerator vanishes while the denominator is non-zero.
An analogous calculation shows that $\kappa_{gq}=0$, and therefore that the distributions of quark and gluon jets as defined by the Operational Definition are always mutually irreducible.

Next, we demonstrate that $M_1$ and $M_2$ are mixtures of the defined quark and gluon jet distributions.
Solving \Eq{eq:opdef} for the distributions of $M_1$ and $M_2$ in terms of the quark/gluon distributions yields:
\begin{align}
\label{eq:solve4m1}
&p_{M_1}(\O)=f_1\, p_q(\O)+(1-f_1)\, p_g(\O),& f_1&\equiv\frac{1-\kappa_{12}}{1-\kappa_{12}\kappa_{21}},\\
&p_{M_2}(\O)=f_2 \, p_q(\O)+(1-f_2)\, p_g(\O),& f_2&\equiv\frac{\kappa_{21}(1-\kappa_{12})}{1-\kappa_{12}\kappa_{21}},
\label{eq:solve4m2}
\end{align}
where we have introduced two numbers $f_1$ and $f_2$ such that $f_1,f_2\in[0,1]$.
We see from \Eqs{eq:solve4m1}{eq:solve4m2} that under the Operational Definition, $M_1$ and $M_2$ have the interpretation of being statistical mixtures of quark and gluon jets where the quark fractions of each sample are $f_1$ and $f_2$, respectively.
Note that while this was entirely anticipated, given the motivation provided in \Sec{sec:motivation}, the Operational Definition manages to avoid the circular reasoning of that section, where a well-defined notion of quark and gluon jets and the statistical-mixture nature of $M_1$ and $M_2$ were assumed to exist before we were able to specify a rigorous procedure to determine them.

There are several additional properties of the Operational Definition worth noting.
First, any additional preprocessing of the jets in $M_1$ and $M_2$ which is operationally defined at the hadron level, such as jet grooming, can be folded into the jet-finding procedure and thus incorporated directly into our definition.
Second, which of $M_1$ or $M_2$ is more ``quark-enriched'' only serves to label which of the resulting distributions is ``quark'' and which is ``gluon'' and does not change the distributions which are produced by this definition.
Finally, while \Eq{eq:opdef} implies the vanishing of the quark/gluon reducibility factors, if a different, non-zero quark/gluon reducibility factor is desired a priori, then the definition may be suitably modified to accommodate those non-zero values.
Thus, the assertion of quark-gluon mutual irreducibility, which is supported by evidence from case studies, can be relaxed to any specified quark/gluon reducibility factors which may then be thought of as inputs to the definition.

In \Sec{sec:topicwola}, we connect the Operational Definition to machinery that has already been developed in the jet substructure and statistical literature, finding that the tools needed to implement the Operational Definition, true to the name, are readily available.
In \App{sec:explore}, we gain some additional insight into the Operational Definition by theoretically exploring it with simple jet substructure observables in a tractable limit of perturbative QCD.

\section{Data-driven jet taggers and topics}
\label{sec:topicwola}

In this section, we connect our Operational Definition of quark and gluon jets to recent developments at the intersection of jet physics and statistical methods, particularly the data-driven paradigms of CWoLa~\cite{Metodiev:2017vrx} and jet topics~\cite{Metodiev:2018ftz}.
CWoLa provides a method to approximate the quark/gluon likelihood ratio by distilling the available information in a huge feature space of jet substructure observables~\cite{Metodiev:2017vrx, Cohen:2017exh, Komiske:2018oaa}.
The jet topics method was introduced and developed in \Ref{Metodiev:2018ftz}, where it was shown that statistical methods could be used to ``disentangle'' quark and gluon jets from mixtures.
We will show how these methods can be combined to form a concrete implementation of the Operational Definition.

\subsection{Classification without labels: Training classifiers on collider data}
\label{sec:cwola}

Recently, there has been an effort to train physics classifiers directly on data despite the lack of labeled truth information, going under the broad term of \emph{weak supervision}.
\Ref{Dery:2017fap} was the first to apply weak supervision methods in a particle physics context, showing that given mixed samples with known signal fractions, a quark/gluon classifier on a few high-level inputs could be trained without access to per-jet truth labels, a paradigm termed learning from label proportions (LLP).
\Ref{Metodiev:2017vrx} developed CWoLa as a method to train a jet classifier via weak supervision on a few generalized angularities~\cite{Berger:2003iw,Almeida:2008yp,Ellis:2010rwa,Larkoski:2014uqa,Larkoski:2014pca}, where signal fractions do not need to be known in order to train the classifier.
\Ref{Komiske:2018oaa} investigated both CWoLa and LLP in the context of high-dimensional, modern machine learning methods, finding that while both methods were performant, CWoLa generalized better and more simply to complex models.
CWoLa has since given rise to new techniques to search for signals of new physics in model-independent ways~\cite{Collins:2018epr}.
These methods are an important step towards making classification at colliders fully data-driven.
Here, we review the CWoLa paradigm in preparation for incorporating it as part of the implementation of our Operational Definition.

Conceptually, CWoLa is extremely simple: given two mixtures $M_1$ and $M_2$ of signal (quark) and background (gluon) jets, train a classifier to distinguish jets in $M_1$ from jets in $M_2$.
This procedure has the attractive property of being able to immediately use any model which can be trained with full supervision.
Furthermore, in the limit that $M_1$ and $M_2$ become pure signal and background, CWoLa smoothly approaches full supervision.
With enough statistics, a feature space that captures all relevant information, and a suitable training procedure, a CWoLa classifier should approach the optimal discriminant between the two mixed samples.\footnote{The generalization to learning from multiple mixtures of signal and background is possible as long as each mixture is assigned a label that is (on average) monotonically related to its signal fraction.}
By the Neyman-Pearson lemma~\cite{NPlemma}, the optimal discriminant between two binary classes is the likelihood ratio.
As discussed in \Sec{sec:motivation}, the mixed-sample likelihood ratio is monotonically related to the quark/gluon jet likelihood ratio.
Thus, CWoLa provides a way of approximating the optimal discriminant between quark and gluon jets given access only to mixed samples.

There are potential concerns, though, that one might have regarding CWoLa in particular and weak supervision in general.
Are enough statistics and a rich-enough feature space available?
Do we have a suitable training procedure?
\Refs{Metodiev:2017vrx, Cohen:2017exh, Komiske:2018oaa} address these concerns and demonstrate that CWoLa indeed works in realistic cases.
For example, CWoLa was used in \Ref{Komiske:2018oaa} to obtain a performant quark/gluon jet classifier by discriminating $Z$+jet and dijet samples using jet images and convolutional neural networks.
As described in \App{sec:train}, there are many other jet representations and machine learning models that are suitable to be trained with CWoLa.
Additionally, previous uses of CWoLa have made assumptions about the samples $M_1$ and $M_2$ being mixtures of well-defined quark and gluon jets, without specifying which definition is being used or attempting to quantify what happens if quark and gluon jets are not the same in the two samples (i.e.\ sample dependence).
From the perspective of this work, those concerns are removed by using the Operational Definition, which turns the problem on its head and lets the samples $M_1$ and $M_2$ \emph{define} quark and gluon jets.
The notion of sample dependence manifests in a new way with our Operational Definition, which we discuss more in our conclusions in \Sec{sec:conc}.

\subsection{Jet topics: Extracting categories from collider data}
\label{sec:topics}

Building on a rich analogy between mixed jet samples and textual documents, \Ref{Metodiev:2018ftz} introduced jet topics and demonstrated how topic modeling could be used to obtain quantitative information about the signal and background distributions from the mixed sample distributions.
The present work extends and elaborates on this approach in order to formulate a practical implementation the Operational Definition of quark and gluon jets in \Sec{sec:opdef}.

Given two samples of quark and gluon jets $M_1$ and $M_2$, the jet topics technique seeks to extract two mutually irreducible categories such that the samples are mixtures of these categories.
To the extent that quark and gluon jets are themselves mutually irreducible, they will correspond to the extracted topics.
There are various procedures for extracting the topics from mixed samples.
\Ref{Metodiev:2018ftz} used a method known as ``demixing'' that was developed in \Ref{katz2017decontamination} in order to obtain the topics.
Other procedures (e.g.\ non-negative matrix factorization~\cite{Arora:2012:LTM:2417500.2417847}) that are popular for textual topic modeling could in principle also be used.
Demixing works by searching for ``anchor bins'' in the mixed sample distributions over a feature space $\O$, which are histogram bins for which the likelihood of $M_1$ to $M_2$ is maximized or minimized.

In the language of \Sec{sec:motivation}, demixing returns reducibility factors $\kappa_{12}$ and $\kappa_{21}$.
With the reducibility factors in hand, the fractions of topic $T_1$ in each mixed sample, $f_{T_1}^{(1)}$ and $f_{T_1}^{(2)}$, can be obtained by solving equations analogous to \Eq{eq:kfs}, and the topic distributions $p_{T_1}(\O)$ and $p_{T_2}(\O)$ are given by \Eq{eq:pqgks} where $q$ is replaced by $T_1$ and $g$ by $T_2$:
\begin{align}
\label{eq:topic1}
&p_{T_1}(\O)=\frac{p_{M_1}(\O)-\kappa_{12}\, p_{M_2}(\O)}{1-\kappa_{12}},&f_{T_1}^{(1)}&=\frac{1-\kappa_{12}}{1-\kappa_{12}\kappa_{21}},\\
&p_{T_2}(\O)=\frac{p_{M_2}(\O)-\kappa_{21}\, p_{M_1}(\O)}{1-\kappa_{21}},&f_{T_1}^{(2)}&=\frac{\kappa_{21}(1-\kappa_{12})}{1-\kappa_{12}\kappa_{21}},
\label{eq:topic2}
\end{align}
where we have assumed without loss of generality that $f_{T_1}^{(1)}>f_{T_1}^{(2)}$.

The jet topics method provides a simple example of the fascinating mileage one is able to achieve from the picture of jets as statistical mixtures.
If the signal (quark) and background (gluon) distributions are mutually irreducible, the topic fractions are the signal fractions, $f_S^{(1)}=f_{T_1}^{(1)}$ and $f_S^{(2)}=f_{T_1}^{(2)}$, from which a number of other useful quantities may be computed.
First, consider some observable $O$ that we wish to cut on to make a signal/background classifier. 
For a given threshold $t$, let the fraction of jets in $M_i$ for which $O$ is greater than $t$ be $f_{M_i}(O>t)$.
Let $\varepsilon_{s}(t)$ be the rate that the signal is correctly identified (the true positive rate) and $\varepsilon_{b}(t)$ be the rate that the background is identified as signal (the false positive rate) by the classifier $(O, t)$.
We then have the equations:
\begin{align}
\label{eq:calibrate}
&f_{M_1}(O>t)=f_S^{(1)}\varepsilon_s(t) + (1-f_S^{(1)})\, \varepsilon_b(t),\\
&f_{M_2}(O>t)=f_S^{(2)}\varepsilon_s(t) + (1-f_S^{(2)})\, \varepsilon_b(t),
\end{align}
which can be solved to give signal and background efficiencies at the given threshold:
\begin{align}
\label{eq:eps}
\varepsilon_s(t)&=\frac{f_{M_1}(O>t)(1-f_S^{(2)}) - f_{M_2}(O>t)(1-f_S^{(1)})}{f_S^{(1)} - f_S^{(2)}},\\
\label{eq:epb}
\varepsilon_b(t)&=\frac{f_{M_2}(O>t)f_S^{(1)} - f_{M_1}(O>t)f_S^{(2)}}{f_S^{(1)} - f_S^{(2)}}.
\end{align}
In this way, the extracted fractions can be used to calibrate the classifier.
Additionally, the pure signal and background distributions of any observable can be obtained from the reducibility factors (or equivalently the extracted fractions): simply change the feature space $\O$ in \Eqs{eq:topic1}{eq:topic2} to whatever observable is desired.

There are several issues to address in attempting to use topic modeling for quark and gluon jets.
How do we know that quark and gluon jets are mutually irreducible in our feature space?
In \App{sec:explore}, we show that quark and gluon jets are {\it not} mutually irreducible in the leading-logarithmic limit of individual Casimir-scaling or Poisson-scaling observables, though this calculation strongly suggests that mutual irreducibility could be achieved in a larger feature space.
\Ref{Metodiev:2018ftz} showed that quark and gluon jets appear to be mutually irreducible in practice for the constituent multiplicity observable, but did not offer a way to fold in additional information.
If we attempt to use multiple observables in the topic modeling procedure, how do we deal with the curse of dimensionality that results from attempting to fill multi-dimensional histograms?
As we now discuss, CWoLa can be combined with jet topics to efficiently use arbitrarily large feature spaces to determine the optimal quark and gluon jet topics.

\subsection{Optimal taggers for optimal topics}
\label{sec:optimal}

To summarize, the CWoLa framework allows trained models to approximate a function monotonic to the quark/gluon likelihood ratio, which is the optimal quark/gluon jet classifier.
Further, the jet topics technique allows for signal and background distributions to be extracted from a given low-dimensional feature space.
Here, we demonstrate how CWoLa and jet topics can be combined into a direct implementation of the Operational Definition of quark and gluon jets from \Sec{sec:opdef}.

When viewed as a likelihood-ratio approximator, a CWoLa-trained model can do more than per-jet classification: it is an efficient method for compressing information in a (potentially) huge but sparsely-populated feature space down to the provably optimal single observable for quark/gluon jet separation.
This approach of taking a CWoLa-trained model output as an interesting observable in its own right solves the curse of dimensionality mentioned at the end of \Sec{sec:topics}.
Furthermore, the guarantee of optimality for the likelihood ratio by the Neyman-Pearson lemma carries over to the jet topics context in that the mutual irreducibility of quark and gluon jets is maximized when the optimal discriminant is used. 
In this sense, optimal taggers give rise to optimal topics.

The marriage of CWoLa and jet topics yields more fruit: since the signal fractions extracted by the topics procedure can be used to calibrate a classifier, the requirement that a CWoLa-trained model be calibrated using known signal fractions is removed.
A CWoLa model now has the potential to be \emph{self-calibrating} in the sense that the model is used to extract the signal fractions, and then the fractions are used to calibrate that same model (other models can also be calibrated).
Furthermore, the optimal topic fractions can be used to extract the pure distribution of any desired observable in a straightforward manner.

This combined paradigm provides a new way to use fully data-driven classifiers in high-energy particle physics, namely as optimal observables for topic fraction extraction. 
The fully data-driven aspect of the entire procedure cannot be emphasized enough as application of these methods to data is the ultimate goal.
The black-box nature of complex classifiers becomes less disturbing in this context, since we can think of the role of the classifier as simply to regress onto the likelihood ratio, without much concern as to how this is done.
As with \Ref{Andreassen:2018apy}, understanding of both the inputs and outputs of a machine learning model allows us to be agnostic with respect to the internal details.

Where does the Operational Definition in \Sec{sec:opdef} fit into this picture?
If we adopt the Operational Definition and define quark and gluon jets to be the categories returned by the topic-finding procedure, this addresses the first issue with jet topics referenced at the end of \Sec{sec:topics}, that we do not know the relation between the extracted topics and quark and gluon jets.
Also, since under this definition the samples $M_1$ and $M_2$ are mixtures of exactly the same quark and gluon jets, the sample dependence concerns mentioned at the end of \Sec{sec:cwola} are alleviated.
The optimality guarantee resulting from the Neyman-Pearson lemma and the good practical performance lend support to the Operational Definition being useful both in theory and practice.
It is no coincidence that the Operational Definition, CWoLa, and jet topics share the same property: they work well when notions of sample independence and mutual irreducibility exist, but still return something sensible as the situation is detuned away from this nice limit.

\section{Quark and gluon jets from dijets and $Z$+jet}
\label{sec:qgex}

In this section, we apply the combined paradigm of CWoLa and jet topics to the realistic context of $Z$+jet and dijet samples, obtaining the distributions of quark and gluon jets via the Operational Definition.%
\footnote{We also investigated applying the Operational Definition to CMS jet mass measurements on similar samples~\cite{Chatrchyan:2013vbb}.  In the dijet sample, though, only average jet mass (instead of individual jet mass) is reported.} 

\subsection{Event generation}
\label{sec:eventgen}

We generated events using \pythia 8.230~\cite{Sjostrand:2014zea} with the default tunings and shower parameters at $\sqrt{s}=\SI{14}{TeV}$.
Hadronization and multiple parton interactions (i.e.\ underlying event) were included and a parton-level $p_T$ cut of $\SI{400}{GeV}$ was applied. 
The $Z$+jet sample was obtained using the \texttt{WeakBosonAndParton:qg2gmZq} and \texttt{qqbar2gmZg} processes, ignoring the photon contribution and requiring the $Z$ to decay invisibly.
The dijet sample was obtained using the \texttt{HardQCD:all} process, excluding bottom quarks.

Final state, non-neutrino particles were clustered with \textsc{FastJet} 3.3.0~\cite{Cacciari:2011ma} using the anti-$k_T$ algorithm~\cite{Cacciari:2008gp} with a jet radius of $R=0.4$.
All jets were required to have $p_T\in[500,550]$ GeV and rapidity $|y|<2.5$.
The hardest jet for $Z$+jet and the hardest two jets for dijets were considered and kept if they passed the above specified cuts.
The unphysical parton-shower-labeled jet flavor was determined by matching the clustered jet to the \pythia parton(s) by requiring that the jet lie within $2R$ of the parton direction from the hard process.
Events in which none of the jets passed this criteria were not considered.
One million jets passing all cuts were retained for both the dijet and $Z$+jet samples. 
The \pythia-labeled quark fraction was 86.3\% for the $Z$+jet sample and 49.8\% for the dijet sample.

\begin{table}[t]
\centering
\begin{tabular}{|c|l|l|}
\hline
 Symbol & Name &Short Description  \\ \hline \hline
$n_\text{const}$ & Constituent Multiplicity & Number of particles in the jet \\
$n_\text{SD}$ & Soft Drop Multiplicity & Probes number of perturbative emissions \\
$N_{95}$ & Image Activity & Number of pixels containing $95\%$ of jet $p_T$\\
$\tau_2^{(\beta=1)}$ & 2-subjettiness & Probes the two-prong nature of the jet\\
$w$ &  Width & Angularity measuring the girth of the jet\\
$m$ & Jet Mass & Mass of the jet\\
\hline \hline
PFN-ID & Particle Flow Net with ID & Particle three-momentum + ID inputs\\
PFN & Particle Flow Net & Particle three-momentum inputs \\
EFN & Energy Flow Net & Using only IRC-safe information\\
EFPs & Energy Flow Polynomials & Linear classification with EFPs\\
CNN & Convolutional Neural Net & Trained on $33\times 33$ 2-channel jet images\\
DNN & Deep Neural Net & Trained on an $N$-subjettiness basis \\
\hline
\end{tabular}
\caption{The individual jet substructure observables (top) and machine learning models (bottom) considered in this study, along with their corresponding symbols and short descriptions.
A full discussion of the observables and models is given in \App{sec:train}.}
\label{tab:obsandmodels}
\end{table}

\subsection{Extracting reducibility factors and fractions}
\label{sec:extract}

For the jet substructure feature space $\mathcal O$, we consider a variety of individual jet substructure observables and trained models.
In \Tab{tab:obsandmodels}, we summarize the observables and models used for our study.
Details of the observable computation, model training, and model architectures are given in \App{sec:train}.

For each of the observables and trained models, we proceed to extract the topic fractions from the $Z$+jet and dijet samples.
We implement a version of the demixing procedure used in \Ref{Metodiev:2018ftz} and described in \Ref{katz2017decontamination}.
Below, we describe the practical procedure used for the studies in this section, including the determination of uncertainties.
Here, we let $O$ indicate either a single observable or the output of a trained model.

\begin{enumerate}
\item {\bf Histograms}: The histograms for $p_\text{$Z+$jet}(O)$ and $p_\text{dijets}(O)$ are computed for a specified binning. Statistical uncertainties are taken to be $\sqrt{N_\text{$Z+$jet}}$ and $\sqrt{N_\text{dijets}}$ coming from one-sigma count uncertainties within each bin.\footnote{These uncertainties, and those derived from them, should only be used to give a sense of scale on the plots. Implementing the Operational Definition in LHC data will require careful consideration of other sources of statistical and systematic uncertainties. For instance, using unfolded distributions may mitigate artificial differences in the samples due to detector effects.}
\item {\bf Likelihood Ratios}: The mixed-sample log-likelihood ratio $\ln p_\text{dijets}(O)/p_\text{$Z+$jet}(O)$ is calculated. The statistical uncertainty is estimated from uncertainty propagation per bin to be:
\begin{equation}
\sigma_{\ln p_\text{dijets}/p_\text{$Z+$jet}} = \sqrt{\frac{1}{N_\text{dijets}} + \frac{1}{N_\text{$Z+$jet}}}.
\end{equation}
\item {\bf Anchor Bins}: Noisy, low-statistics bins are neglected by only considering bins with more than 50 events in each sample. The upper (lower) anchor bin is obtained by finding the maximum (minimum) bin for the log-likelihood ratio minus (plus) its uncertainty.
\item {\bf Reducibility Factors}: The lower (upper) reducibility factor $\kappa_{21}$ ($\kappa_{12}$) is obtained by exponentiating (minus) the log-likelihood ratio evaluated at the lower (upper) anchor bin. Uncertainties on the reducibility factors are obtained by standard uncertainty propagation.
\item {\bf Topics}: The jet topics are obtained from the reducibility factors $\kappa_{12}$ and $\kappa_{21}$ according to the definition in \Eq{eq:opdef}, with uncertainties propagated from the reducibility factors.
\item {\bf Fractions}: Topic fractions are obtained from the reducibility factors $\kappa_{12}$ and $\kappa_{21}$ according to \Eqs{eq:solve4m1}{eq:solve4m2}, with uncertainties propagated from the reducibility factors.
In this study, the topic fraction always corresponds to the quark fraction.
\end{enumerate}

While we use the concrete method above to showcase the viability of our method, there may of course be alternative ways to obtain the anchor bins and reducibility factors.
For instance, it may be interesting to a pursue a binning-free method, where a cumulative density function is used instead of a binned histogram.
Similarly, there may be more suitable ways to ignore low-statistics phase space regions and determine anchor bins.
We leave detailed optimizations of the method for future developments.

In \Fig{fig:histratios}, we show the mixed-sample log-likelihood ratios $\ln p_\text{dijets}(O)/p_\text{$Z+$jet}(O)$ for various jet substructure observables and model outputs.
Overall, we see excellent confirmation that the mixed-sample log likelihood is bounded between the predicted extrema according to the \pythia fractions.
To extract these fractions in a data-driven way, we must of course obtain these extrema from the measured log-likelihood ratios.
Using the procedure outlined above, the resulting anchor bins are shown in the right-most portion of \Fig{fig:histratios}.
Interestingly and satisfyingly, many of the individual observables and essentially all of the models extract extrema consistent with the \pythia fractions.
It is important to note, though, that the \pythia fractions are not fully well-defined hadron-level concepts and are shown solely to provide a conceptual and semi-quantitative guideline for the performance of the method.

\begin{figure}[t]
\centering
\subfloat[]{\includegraphics[scale=.735]{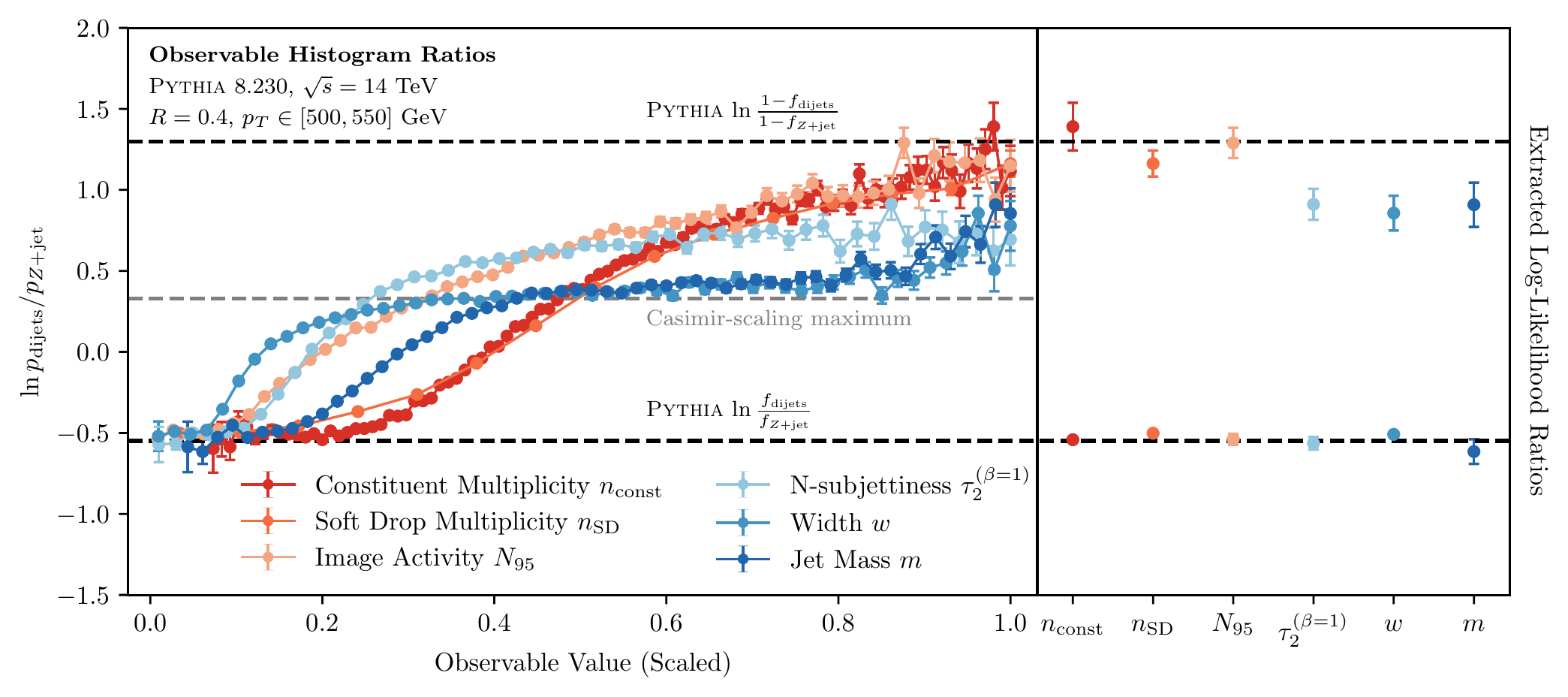}\label{fig:histratios-obs}}

\subfloat[]{\includegraphics[scale=.735]{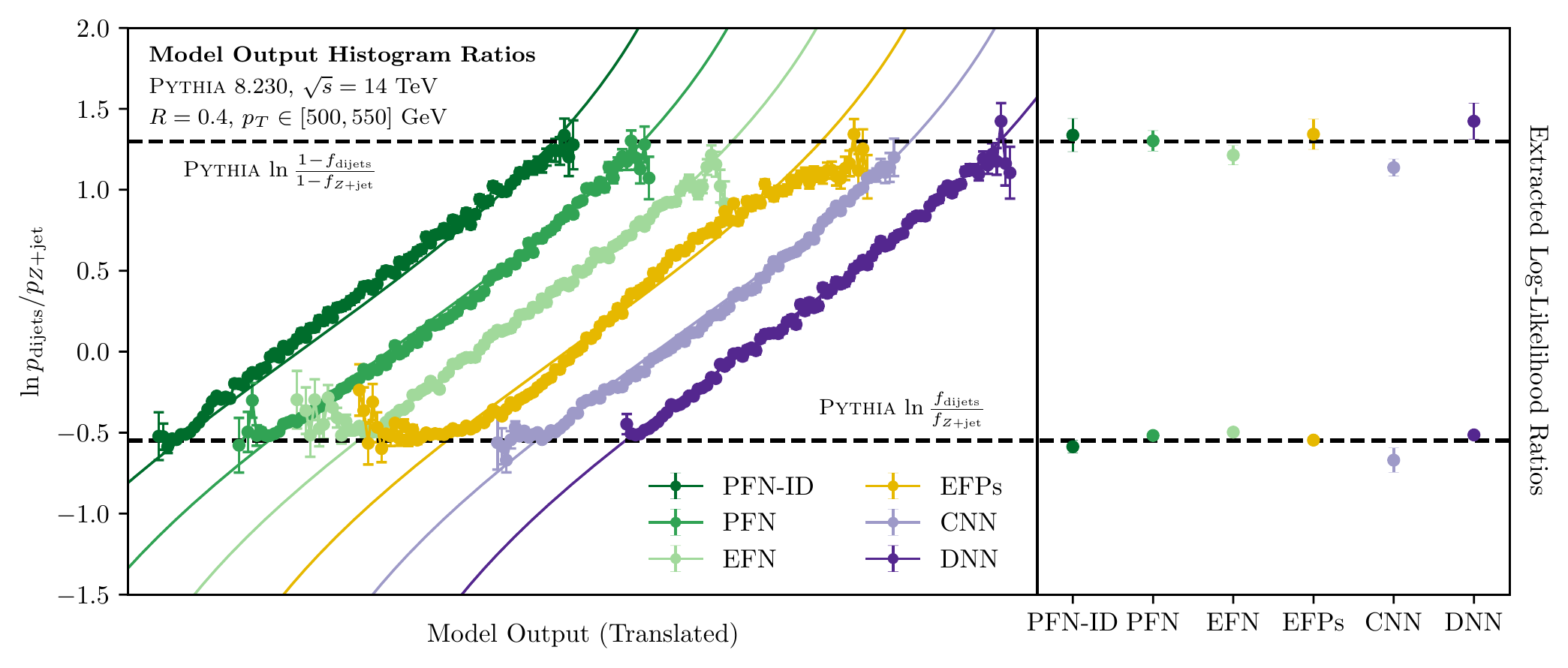}\label{fig:histratios-model}}
\caption{
Mixture log-likelihood ratios and their extrema for (a) individual jet substructure observables and (b) trained models, the latter of which have been translated along the horizontal axis for clarity. 
The black dashed lines indicate the maximum and minimum of the mixture likelihood ratio determined using the \pythia fractions.
The gray dashed line in the observable plot indicates the upper bound obtained for a Casimir-scaling observable from \App{sec:explore}; as expected, jet mass and width approach and remain near the gray line for much of their domain.
While all individual observables asymptote well to the lower black line, only the count observables ($n_{\rm const}$, $n_{\rm SD}$, $N_{95}$) come close to the upper black line, indicating that gluons are more irreducible than quarks.
By contrast, the minimum and maximum for each trained model appear to achieve extremal values close to the black limits.
The solid colored lines in the lower plot indicate the behavior of the optimal classifier, closely related to \Fig{fig:mixloglikes}.
}
\label{fig:histratios}
\end{figure}
\afterpage{\clearpage}

For the substructure observables in \Fig{fig:histratios-obs}, it is evident that the count observables of constituent multiplicity, soft drop multiplicity, and image activity come closest to saturating both the upper and lower bounds.
For mass and width, a clear plateau is exhibited close to the leading logarithmic expectation for Casimir-scaling observables (see \App{sec:explore}).
This difference is reflected in the fact that the count observables extract extrema of the log-likelihood ratio consistent with the \pythia fractions, while the remaining observables systematically underestimate the upper bound.
One feature worth noting is that the lower bound is accurately extracted by every observable; it is the upper bound that is more difficult to saturate with a generic observable.
This indicates that gluon jets are evidently more irreducible than quark jets, and therefore that gluon jet distributions are easier to extract.

For the trained model outputs in \Fig{fig:histratios-model}, we see that the mixed-sample log-likelihood ratios are clearly bounded as expected and agree with the prediction for a well-trained classifier.
The slight deviations from the solid curve in the case of the EFPs arise from the fact that they are trained using Fisher's Linear Discriminant, which optimizes a different objective function, but nonetheless the EFPs exhibit qualitatively similar behavior to the other classifiers.
Compared to the individual substructure observables, the models more robustly saturate the upper and lower bounds of the log-likelihood ratio and demonstrate less sensitivity to changes in the binning of the histograms.
The extracted extrema of the log-likelihood ratio based on the trained models (with the exception of the CNN) are all consistent with one another as well as with the \pythia fractions.
This agreement, present in the variety of different models which process information in very different ways, indicates that there is indeed a robust sense in which ``quark'' and ``gluon'', as qualitatively described by the parton-matched labels, are latent within the mixed samples.

\begin{figure}[t]
\centering
\subfloat[]{\includegraphics[scale=.755]{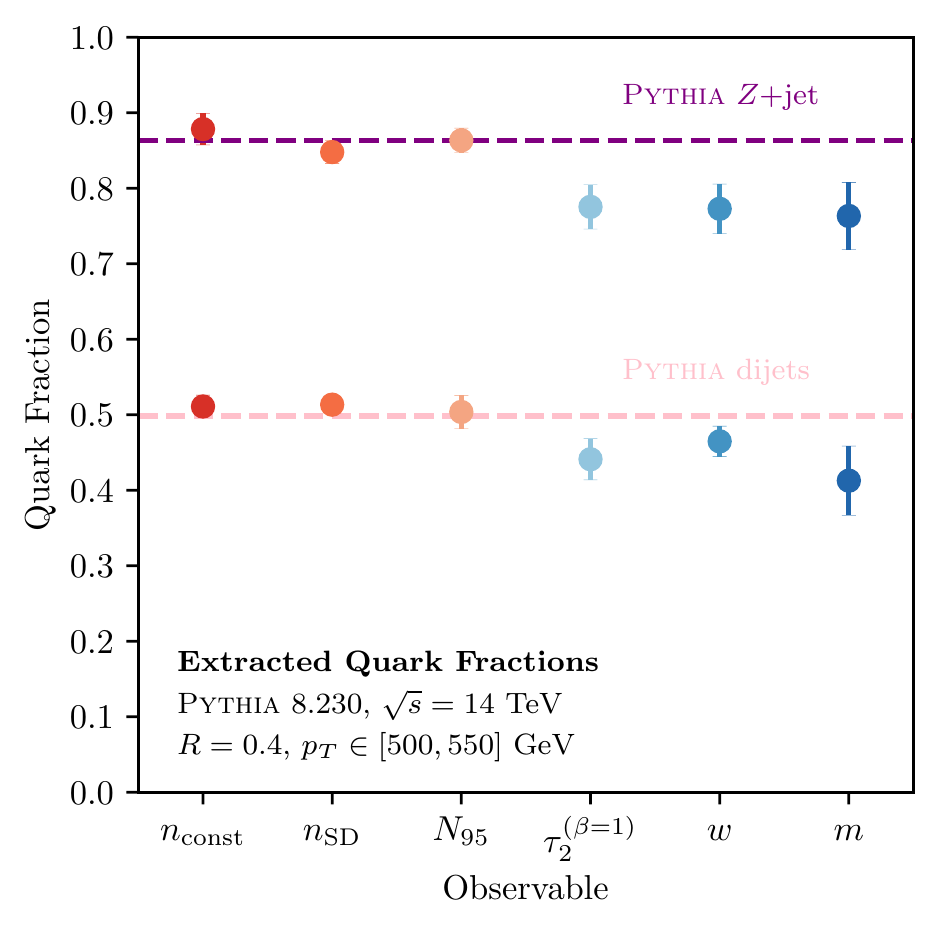}\label{fig:fracs-obs}}
\subfloat[]{\includegraphics[scale=.755]{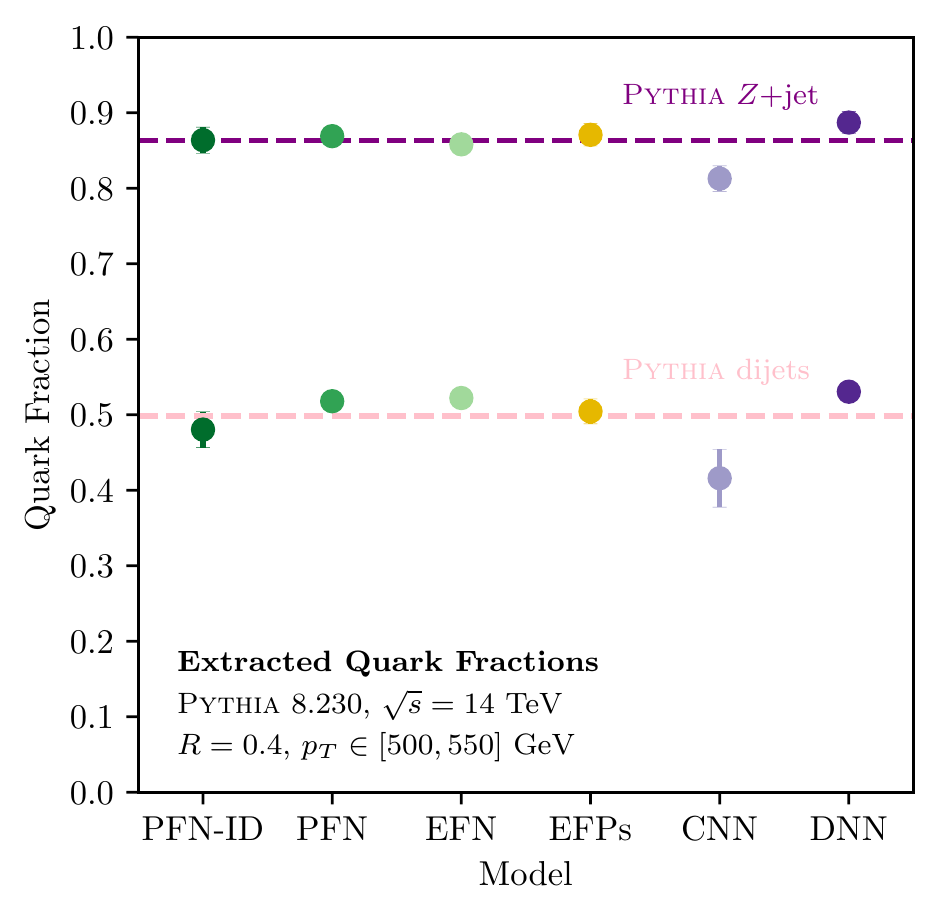}\label{fig:fracs-model}}
\caption{
Extracted quark fraction values for the (a) individual observables and (b) trained models as calculated using the log-likelihood extrema of \Fig{fig:histratios} inserted into in \Eqs{eq:solve4m1}{eq:solve4m2} to obtain the fractions.
}
\label{fig:fracs}
\end{figure}

\begin{figure}[t]
\centering
\subfloat[]{\includegraphics[scale=.755]{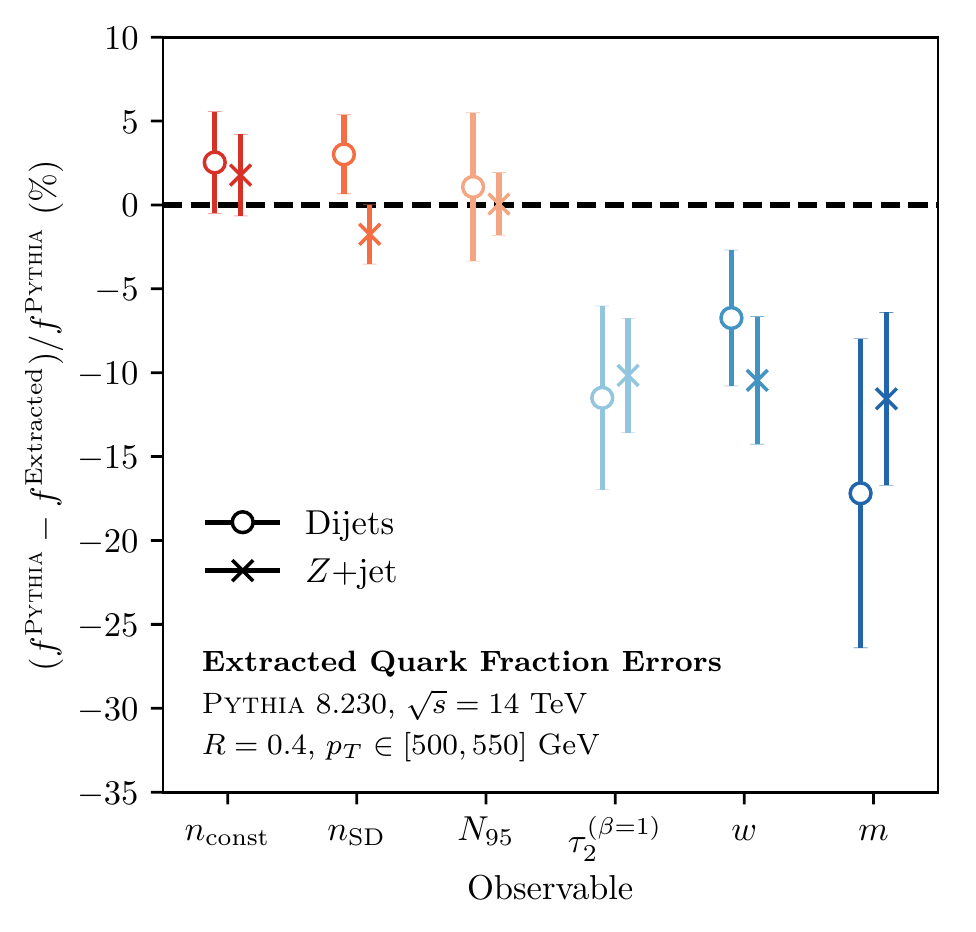}\label{fig:fracerrors-obs}}
\subfloat[]{\includegraphics[scale=.755]{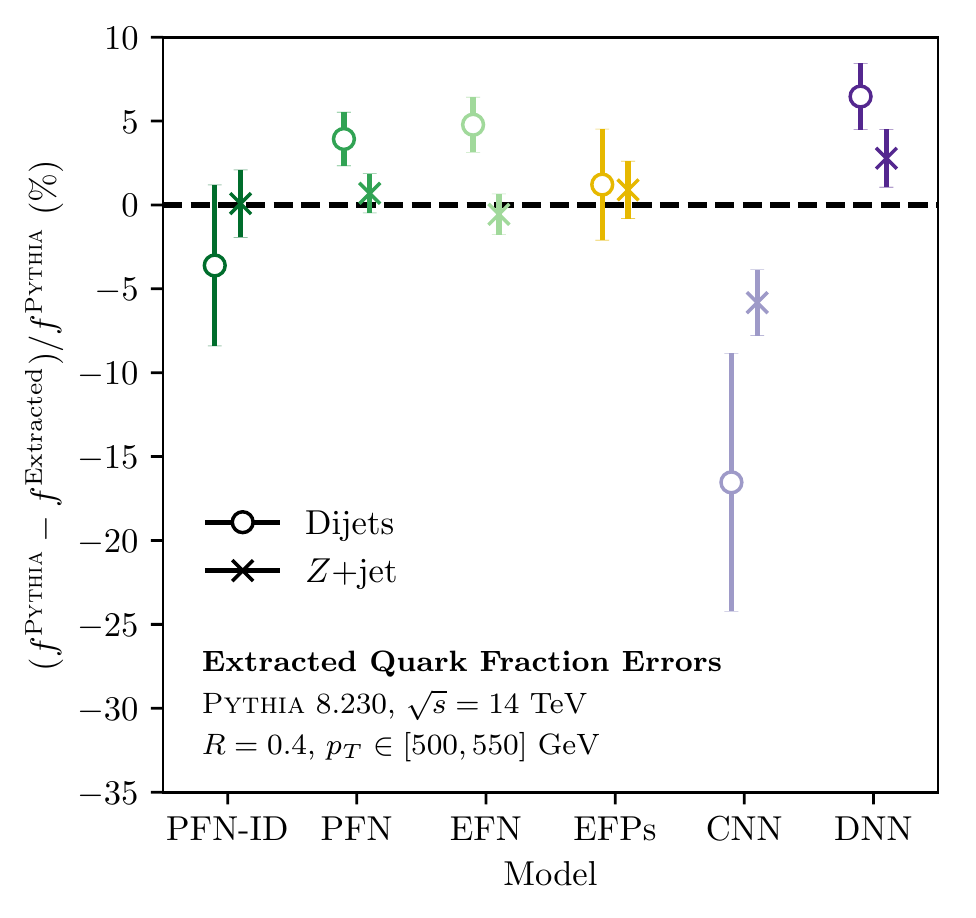}\label{fig:fracerrors-model}}
\caption{
The percent error of the extracted quark fractions (see \Fig{fig:fracs}) relative to the \pythia fractions, obtained using the (a) individual observables and (b) trained models.
By this measure, the best jet observable appears to be $N_{95}$ and the best model is the linear EFP model.
}
\label{fig:fracerrors}
\end{figure}

Using the extracted extrema of the mixed-sample log-likelihood ratio, the reducibility factors can be obtained by appropriate exponentiation.
The quark fractions can then be calculated according to \Eqs{eq:solve4m1}{eq:solve4m2}.
These are shown in \Fig{fig:fracs-obs} for the individual observables and \Fig{fig:fracs-model} for the trained models.
We see that the trained models all extract fractions largely consistent with one another and with the \pythia fractions.
The count substructure observables also extract consistent fractions, while the shape observables exhibit Casimir-scaling behavior, making them unsuitable for identifying mutually-irreducible quark and gluon jets. 
The fractions obtained from the trained models were consistently more robust to different choices of topic extraction procedures, such as the histogram binning.
Despite having little to no handle on the details of the trained models, we are able to obtain important constraints on their behavior and use them to obtain quark/gluon fractions, which are evidently insensitive to these details.

As a more quantitative measure of the quality of the extracted quark fractions, the percent error of the extracted fractions relative to the (unphysical) \pythia fractions is shown in \Figs{fig:fracerrors-obs}{fig:fracerrors-model}.
The count observables and trained models agree within several statistical uncertainties of one another and the \pythia fractions, in many cases achieving $\mathcal O(1\%)$ fidelity.
Again, we caution that the \pythia fractions solely provide a heuristic to demonstrate the performance of the method and should not be taken as fundamental to quark and gluon jets.

\subsection{Self-calibrating classifiers}
\label{sec:selfcal}

\begin{figure}[t]
\centering
\includegraphics[scale=0.85]{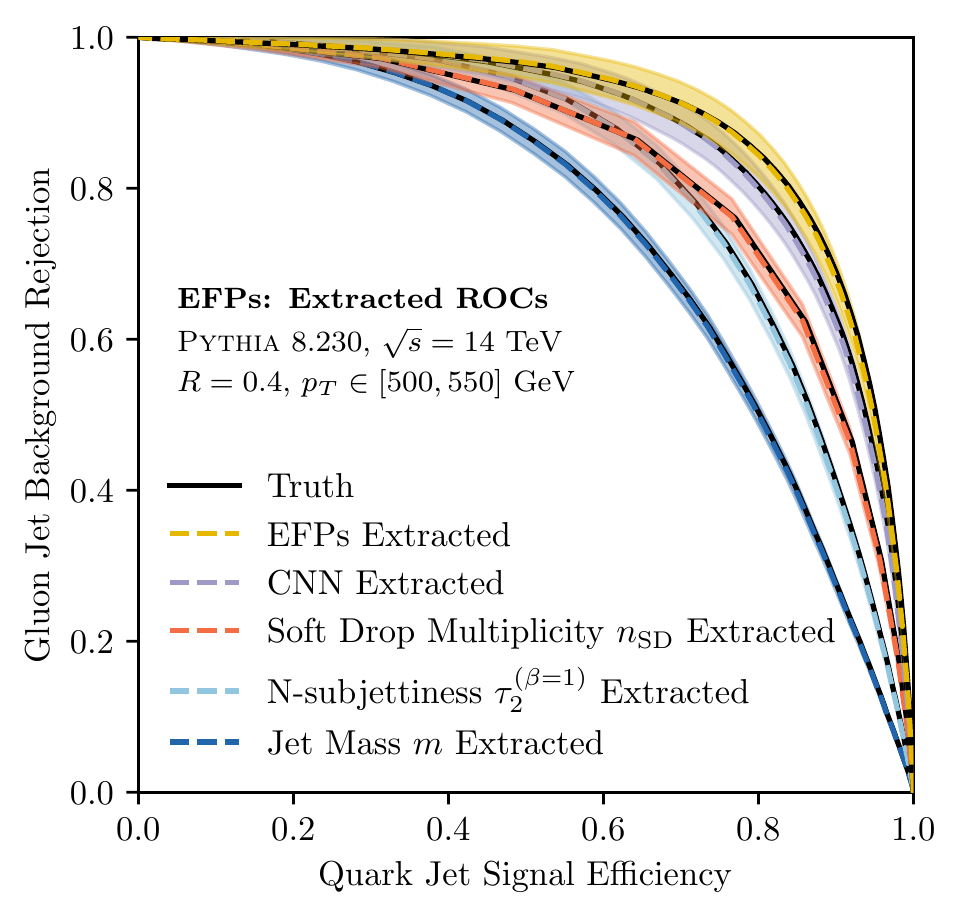}
\caption{
The ROC curves for several substructure observables and trained models using the quark fractions estimated from the EFPs.
The ``Truth'' corresponds to using the \pythia fractions to obtain the ROC curve.
We see good agreement between the data-driven ROC curves and the \pythia-labeled ROC curves.
Further, we see that the CWoLa-trained EFP classifier has effectively self-calibrated itself in this way.}
\label{fig:extractedrocs}
\end{figure}

With the quark fractions of the mixtures in hand, one immediate application is to use them to calibrate the quark/gluon classifiers, as discussed in \Sec{sec:optimal}.
Since uncalibrated classifiers can be used to obtain these fractions, this allows for self-calibrating classifiers in the CWoLa framework.
This liberates the CWoLa framework from necessarily requiring a small test set with known fractions (c.f.~\Ref{Metodiev:2017vrx}).
In the present picture, this ability to self-calibrate is conceptually clear since a sample with ``known'' fractions is equivalent to providing a definition of the underlying categories.

Beyond solely self-calibration of classifiers, the extracted fractions can be used to obtain the receiver operating characteristic (ROC) curves for other trained models or substructure observables, even those that do not themselves exhibit quark/gluon mutual irreducibility.
The extracted ROC curves of a variety of trained model and substructure observables using the EFP-extracted quark fractions are shown in \Fig{fig:extractedrocs}, with estimated uncertainty bands coming from uncertainties in the extracted fractions.
They are compared to the calibrated ROC curve using the \pythia-labeled fractions, achieving very good agreement between the two.
Note that the uncertainties are smaller for worse classifiers, which is intuitive given the limit that a perfectly-random classifier can be identified as such without any fraction information.
Overall, this concretely demonstrates that the self-calibration of CWoLa-trained classifiers can be achieved in a purely data-driven way.

\subsection{Obtaining observable distributions from extracted fractions}
\label{sec:obs}

With the reducibility factors of the mixtures, the distributions of substructure observables can be extracted for quark and gluon jets separately.
This corresponds to a direct application of the Operational Definition of quarks and gluons in \Eq{eq:opdef}.
This is similar in spirit to the procedure implemented in \Refs{Aad:2014gea,ATL-PHYS-PUB-2017-009} of using quark/gluon fractions estimated by convolving matrix elements and parton distribution functions and then solving systems of linear equations.
The key distinction is that, in our case, the fractions (and reducibility factors) themselves are data-driven.

In \Fig{fig:extractedhists}, we use the reducibility factors defined by the EFP classifier to extract quark and gluon distributions for the six individual substructure observables. 
We see excellent agreement between the data-driven, operationally-defined quark and gluon distributions and the ones specified by the \pythia fractions.
Importantly, this procedure works for any substructure observable, even ones such as jet mass and width which do not manifest quark/gluon mutual irreducibility.

\begin{figure}[t]
\centering
\includegraphics[scale=.66]{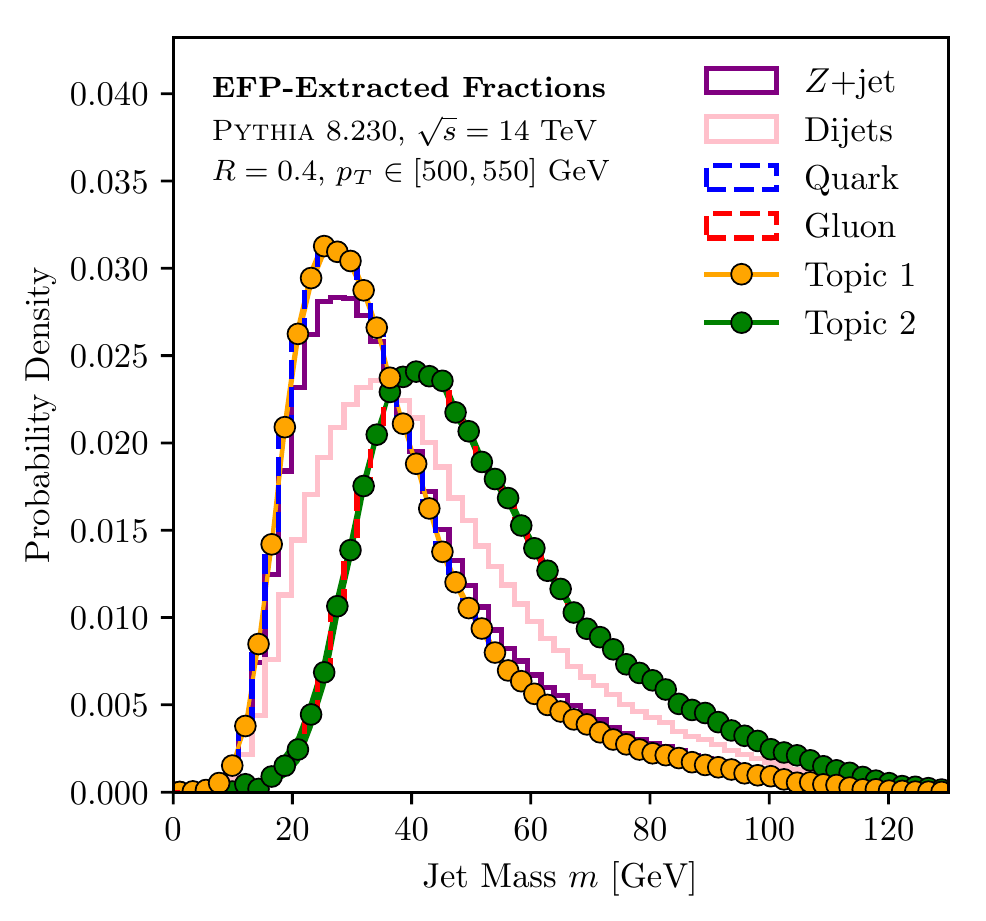}
\includegraphics[scale=.66]{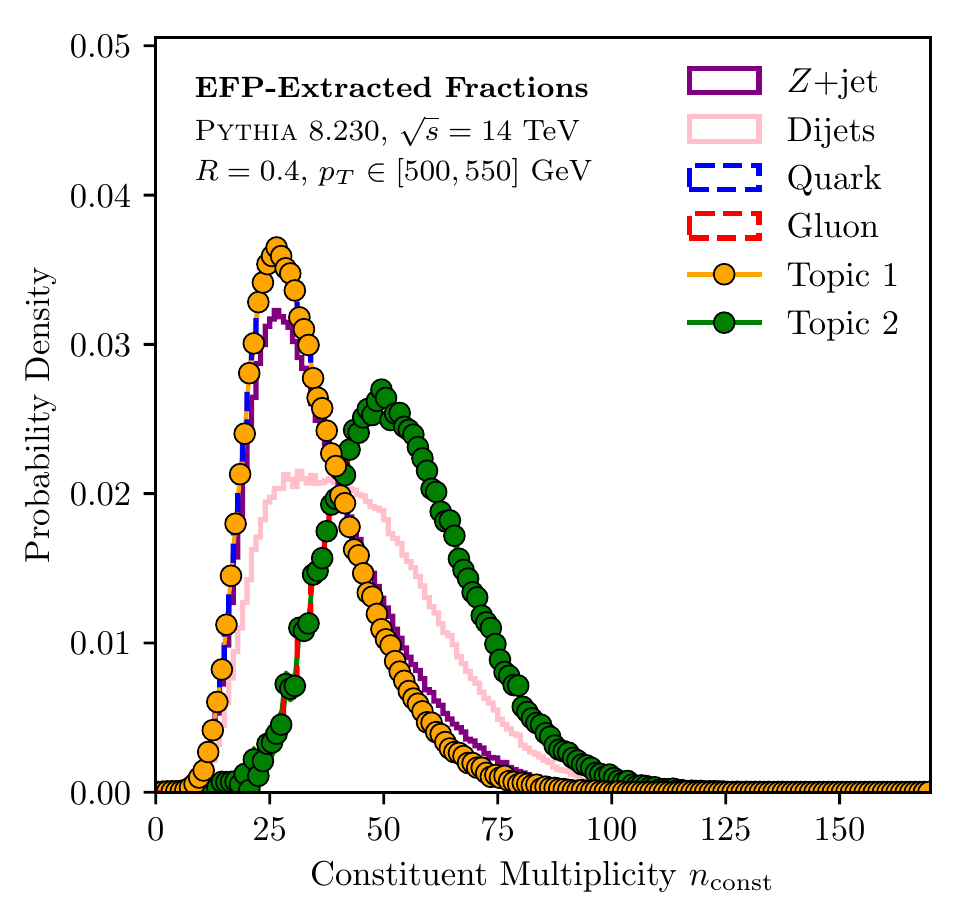}

\includegraphics[scale=.66]{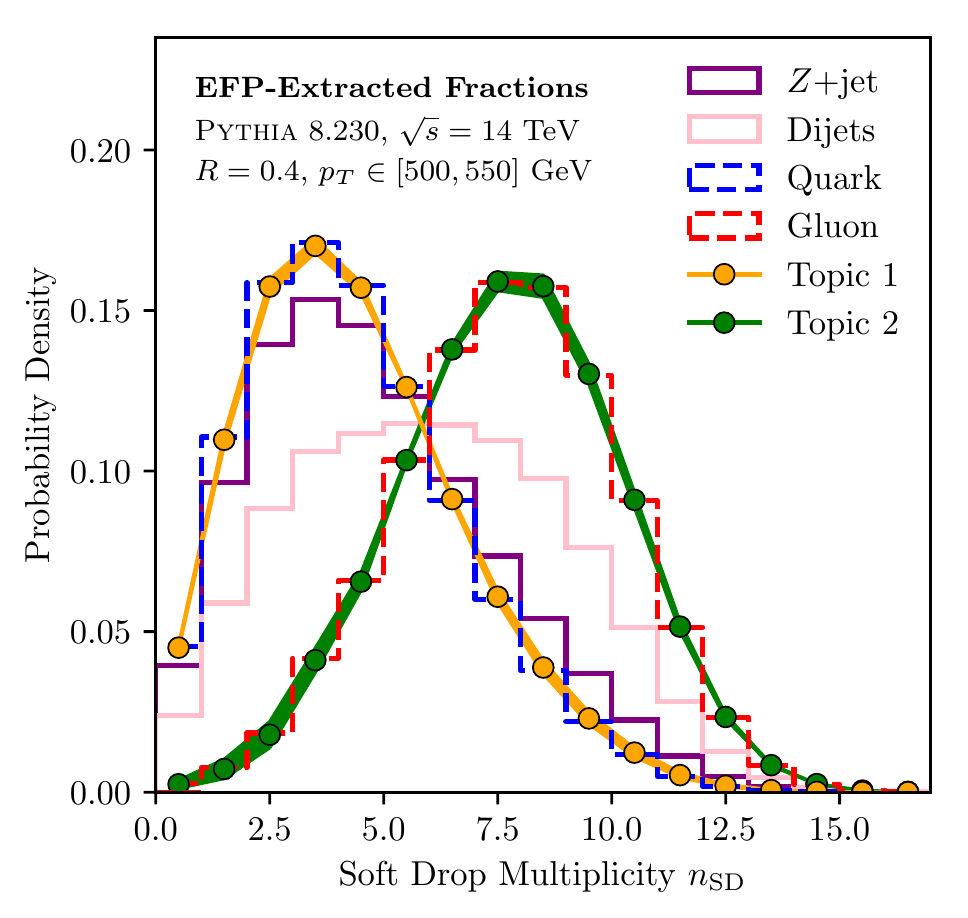}
\includegraphics[scale=.66]{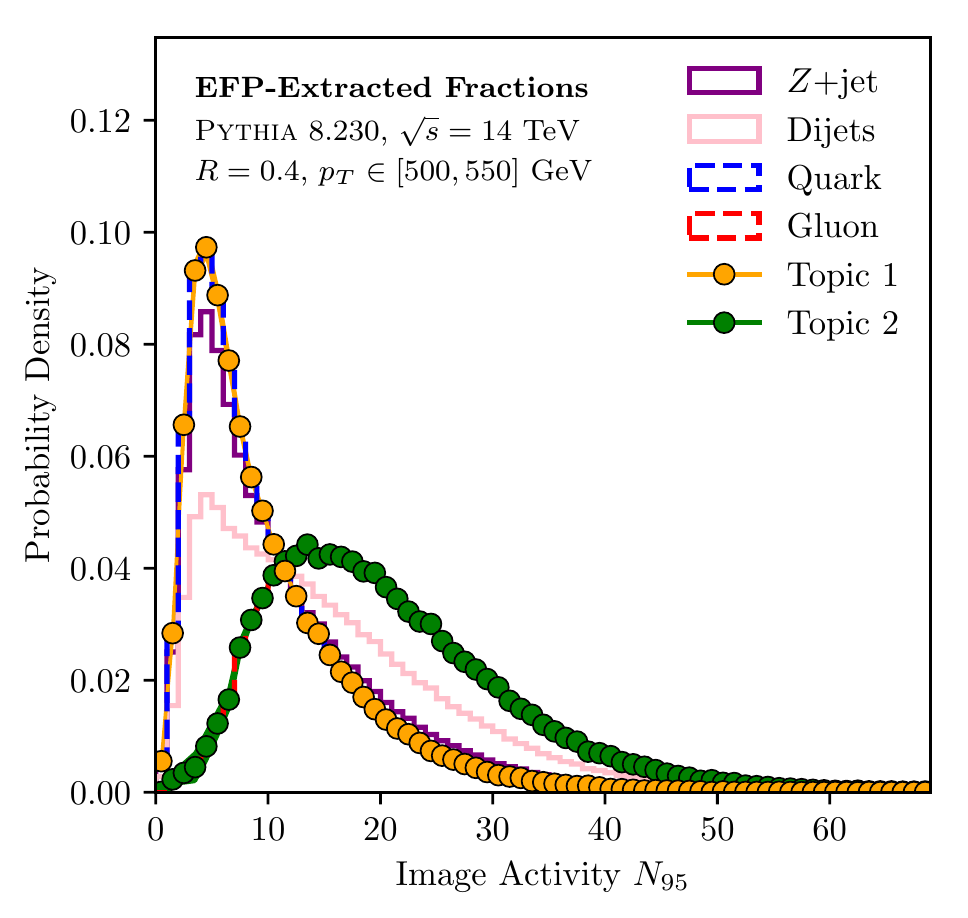}

\includegraphics[scale=.66]{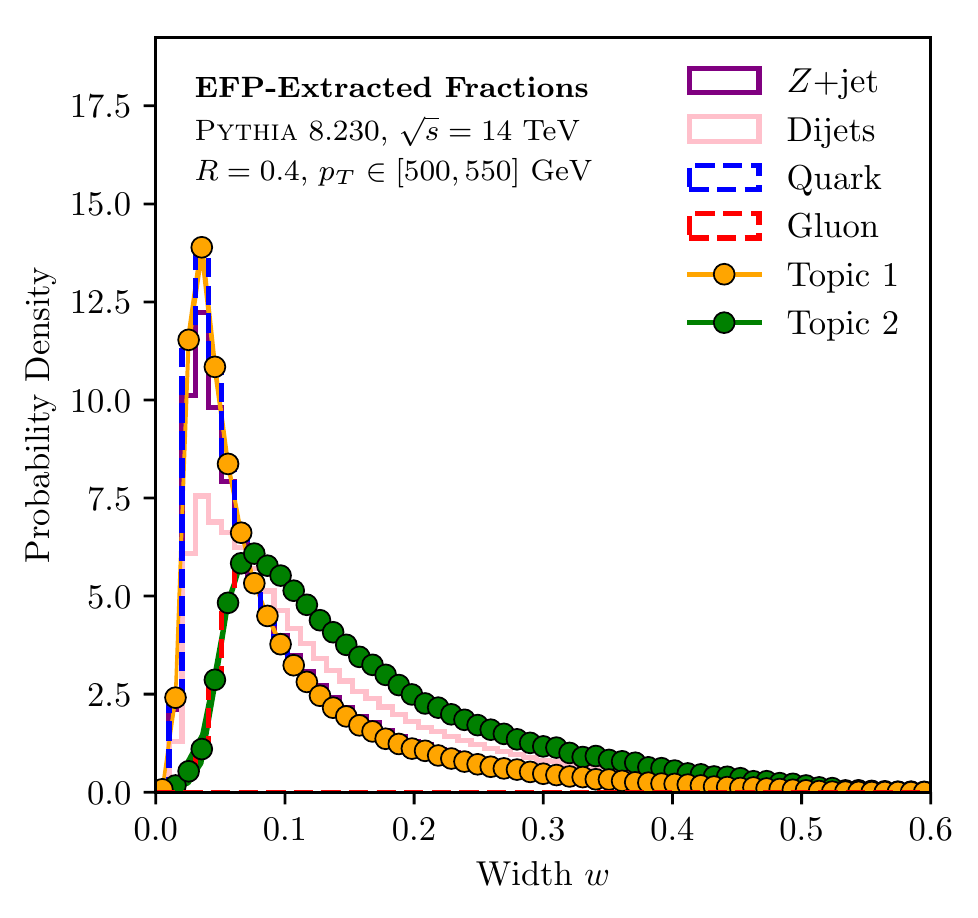}
\includegraphics[scale=.66]{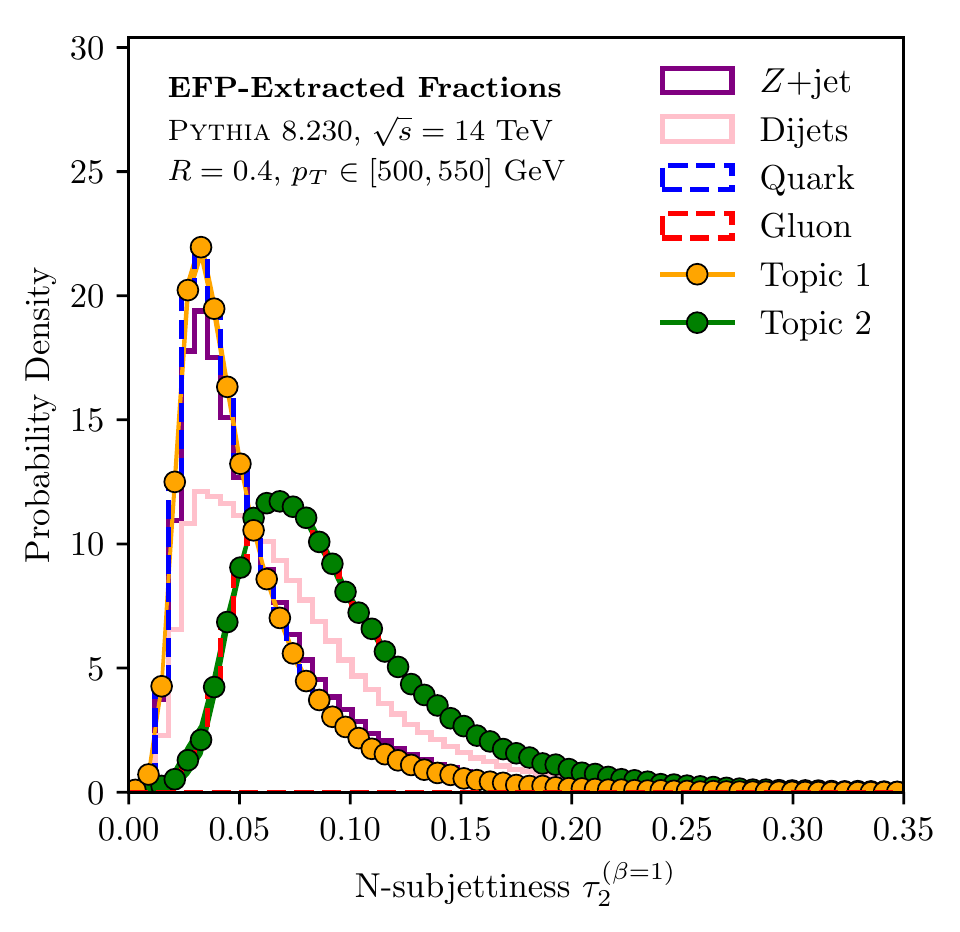}
\caption{
The distributions of the six substructure observables in the $Z+$jet sample (purple) and dijet sample (pink), with the quark and gluon distributions determined from the \pythia fractions (blue and red, respectively) and the jet topics (orange and green) using EFP-extracted reducibility factors.
We see excellent agreement between the jet topics and the \pythia-determined distributions of quark and gluon jets.
}
\label{fig:extractedhists}
\end{figure}
\afterpage{\clearpage}

\section{Conclusions}
\label{sec:conc}

In this chapter, we provided an Operational Definition of quark and gluon jets, based solely on physical cross section measurements. 
We connected our definition to the existing CWoLa and jet topics paradigms, showing how they each fit naturally into the implementation of the definition.
Taking two mixed samples, for which there is a qualitative notion that one is more ``quark-like'' than the other, the Operational Definition returns a quantitive understanding through mutually-irreducible quark and gluons distributions.
Practically, we implemented this definition by approximating the mixed-sample likelihood ratio, relating it to the pure quark/gluon likelihood ratio, and finding its extrema to determine mixed-sample reducibility factors.
With the reducibility factors in hand, the quark fractions for the mixed samples can be readily obtained.
In a broad sense, our Operational Definition harmonizes with the statistical picture of jet samples at colliders, where individual jets do not carry intrinsic flavor labels and one only ever has access to mixed samples in data.

To illustrate the power of the Operational Definition, we tested it in the realistic context of $Z$+jet and dijet processes.
We applied our quark/gluon jet definition to twelve different observables:  six individual substructure observables, and six trained machine learning models which distilled a huge feature space down to a single optimal observable.
The six individual observables naturally fall into two categories, count and shape observables, and we confirmed that the count observables yield much more accurate quark fractions (relative to a \pythia baseline).
With the minor exception of the CNN, the machine learning models all did well at extracting the fractions.
While the performance of the best individual observable ($N_{95}$) and the best machine learning model (linear EFPs) were comparable, the machine learning models were overall more robust to changes in histogram binning and to the technique used for determining the reducibility factors; this in turn contributes to the robustness of the Operational Definition.
Having determined the quark fractions, we extracted pure quark and gluon distributions for various jet substructure observables.
Crucially, this worked even for observables that do not exhibit quark/gluon mutually irreducibility, as long as the observable used to extract the fractions does.
Additionally, we demonstrated that CWoLa classifiers could be self calibrated using fractions obtained from an uncalibrated classifier, thereby removing a potential hurdle in using CWoLa in practice.

The techniques in this chapter represent a novel use of classification in particle physics.
Instead of tagging quark and gluon jets, we used a CWoLa-trained deep learning classifier to approximate the full mixed-sample likelihood ratio.
This is in the same spirit of recent work on deep learning~\cite{Andreassen:2018apy,Komiske:2017ubm,Komiske:2017aww,Chang:2017kvc,Roxlo:2018adx,deOliveira:2017pjk,Paganini:2017hrr,Paganini:2017dwg,DAgnolo:2018cun}, where the ``black box'' nature of the trained model is not of central importance to the success or understanding of the method.
No longer is the output of a neural network viewed as an arbitrary quantity used only for classification, but rather as a robust approximation to the likelihood ratio, which turns out also to be optimal for extracting categories from the data.
Surprisingly, while individual quark and gluon jets cannot be tagged perfectly, we were able to use a data-driven classifier to extract the full quark and gluon distributions of an observable to percent-level accuracy.
This approach paves the way for fully data-driven collider physics, making use of machine learning techniques trained directly on data while producing results insensitive to the details of the ``black box''.

We conclude by discussing potential extensions of the methods used in this chapter.
As mentioned in \Sec{sec:topicwola}, a key concern in jet tagging is sample dependence, i.e.\ whether a ``quark jet'' in one sample is the same as a ``quark jet'' in another.
While the Operational Definition sidesteps the issue of sample dependence in the case of two mixed samples, it is natural to ask what happens with three or more mixed samples.
Concretely, once the Operational Definition is applied to two mixed jet samples, one can ask the degree to which a third sample $M$ is explained by the existing quark and gluon distributions.
It turns out that there is a unique and well-defined generalization of the reducibility factor, discussed in \Ref{katz2017decontamination}, that precisely captures this notion and yields a quantifiable notion of sample dependence:
\begin{equation}\label{eq:other}
\kappa \equiv \max_{f_q,\,f_g}\{f_q + f_g \,|\,\exists \text{ dist. } p_o(\O) \text{ s.t. } p_{M}(\O) = f_q p_q(\O) + f_g p_g(\O) + (1-f_q-f_g)p_o(\O)\},
\end{equation}
where $0\le f_q,f_g\le 1$ and $f_q+f_g\le1$.
In \Eq{eq:other}, $\kappa$ is the maximum amount of $M$ explainable by the quark and gluon distributions, requiring minimal addition of an ``other'' distribution $p_o(\mathcal O)$.
Understanding sample dependence is a general challenge, even with parton-shower-extracted templates, so it is gratifying that our framework naturally suggests a tool to address this problem.
Sample dependence can also be studied by directly comparing the quark and gluon jet definitions provided by different pairs of jet samples ($Z$+jet, dijets, $\gamma$+jet, etc.)\ at different transverse momenta and jet radii.
We leave explorations of these important ideas, as well as more detailed optimizations of the method, to future work.

Extending this thinking, one might attempt to provide a concrete jet flavor definition beyond the two-category case of quarks and gluons.
For instance, while the difference in radiation patterns between different-flavor light-quark jets is much smaller than between quark and gluon jets, it may be possible to use the techniques described in this chapter to define differently-flavored quark jets.
The subtle difference in radiation patterns between different light-quark has been studied in the context of jet charge observables in \Ref{Krohn:2012fg} and in the context of machine learning in \Ref{Fraser:2018ieu}.
To use our methods in this case would require advances in multiple-category CWoLa and jet topics, though the conceptual underpinnings would be the same as for the two-category case studied here.
Further, one could extend such a definition to provide well-defined jet flavor definitions for a variety of other boosted hadronic objects, potentially including subtle distinctions like longitudinal versus transverse polarization of boosted $W/Z$ bosons.
More broadly, the concept of mutual irreducibility as a means of defining categories may find additional applications in high-energy physics due to its utility in disentangling overlapping distributions using pure phase space signatures.
\chapter{Conclusions}

In this thesis, I approached quantum field theory and collider physics from a new perspective using only the observable energy flow information of an event.

In Chapter 2, I established a metric space for events using the energy mover's distance by combining the energy flow with ideas from optimal transport.
I unified many classic observables and jet definitions, as well as infrared and collinear safety, as geometric objects in this new space.
I also lifted this reasoning to construct a space of theories by treating a theory as a distribution over energy flows.
This new event geometry formalism not only allows for a clarified understanding of existing concepts, but also enables meaningful new geometrically-inspired developments for understanding particle interactions.
A fascinating avenue for further exploration is to circumvent scattering amplitudes and perform theoretical calculations directly in this new space of events.

In Chapter 3, I developed the energy flow polynomials as a basis of infrared- and collinear-safe observables by systematically expanding in particle energies and angles.
I demonstrated that this basis encompasses a variety of existing observables and includes entirely new observable structures.
I applied these observables to a variety of classification tasks for jets and jet substructure, such as quark versus gluon jet classification and boosted object identification.
Infrared and collinear safety emerged as a consistency condition due to the redunancy of describing an event using particles rather than its energy flow.
This is similar to how gauge symmetry emerges from the redundancy of describing spin-one massless particles using four-vectors, highlighting the power of using manifestly observable information with the energy flow formalism.

In Chapter 4, I defined particle flavor using observable information by connecting factorization with concepts from topic modeling, refining the notion of flavor from a per-event label to a statistical category.
I applied this definition to the specific case of quark and gluon jets at colliders, resolving a long-standing ambiguity in the literature and showcasing how this data-driven flavor definition can be used to statistically disentangle different types of particles.
More broadly, the statistical structure of factorized observables that I have highlighted has also been used to develop new methods for training data-driven collider classifiers and for searching for new physics in a model-independent way.

Throughout the thesis, many of the developments were enabled by connecting concepts from particle physics with ideas and techniques from statistics and computer science.
While there have been many instances of problems in the natural sciences being cast as machine learning problems, it has been much rarer for these connections to provide purely theoretical or conceptual advances in our understanding of quantum field theory or collider physics.
I hope that the perspective in this thesis is the first of many profound interdisplinary insights which will allow us to discover new particles and to better understand the fundamental interactions of nature.

\appendix
\chapter{Energy Moving with Massive Particles}
\label{sec:mass}

In this appendix, I explore an alternative definition of energy flow appropriate for massive particles, with a corresponding change in the measures used to define the EMD.
The energy flow in \Eq{eq:energyflow} treats events as sets of particles that have energy-like weights $\{E_i\}$ and geometric directions $\{\hat{n}_i\}$.
The EMD in \Eq{eq:emd} is based on pairwise distances $\{\theta_{ij}\}$ that are only functions of the $\hat{n}_i$ and $\hat{n}_j$ directions.
The exact definitions of $E_i$, $\hat{n}_i$, and $\theta_{ij}$ may vary depending on the collider context and other choices.
For massless final-state particles in $e^+e^-$ collisions, it is typical to take the energy $E$ to be equal to the total momentum $|\vec p\,|$, and the geometric direction $\hat{n}$ to be equal to the unit vector $\vec{p} / E$.
For massless particles in hadronic collisions, it is natural to use transverse momentum $p_T$ and a geometric direction based on azimuth $\phi$ and pseudorapidity $\eta$.

It is straightforward to adapt the energy flow to massive particles (see related discussion in \Ref{Mateu:2012nk}).
For the energy measure, the natural choices are energy in the $e^+e^-$ case and transverse energy in the hadronic case:
\begin{equation}
E_i = \sqrt{|\vec p_i|^2+m_i^2}, \qquad
E_{Ti}=\sqrt{p_{Ti}^2+m_i^2}.
\end{equation}
Both of these reduce nicely to the expected expressions in the $m_i\to0$ limit.
For the geometric direction, the natural choices are velocity and transverse velocity, written in four-vector notation:
\begin{equation}
n^\mu_i = \frac{p_i^\mu}{E_i} = (1,\vec{v})^\mu, \qquad
n^\mu_{Ti} = \frac{p_i^\mu}{E_{Ti}} = (\cosh y_i, \vec{v}_{Ti}, \sinh y_i)^\mu,
\end{equation}
where $\vec{v} = \vec{p}_i / E_i$ is the particle three-velocity, $\vec{v}_{Ti} = \vec{p}_{Ti} / E_{Ti}$ is the particle transverse two-velocity, and $y_i$ is the particle rapidity.
Again, these have the expected behavior in the $m_i\to0$ limit, and for finite mass, the velocities are bounded as $|\vec{v}| \in [0,1]$ and $|\vec{v}_T| \in [0,1]$.

To define the EMD, we choose the following pairwise angular distance:
\begin{equation}
\label{eq:theta_massive}
\theta_{ij}=\sqrt{ - (n_i^\mu-n_j^\mu)^2},
\end{equation}
where one replaces $n^\mu$ with $n_T^\mu$ in the hadronic case.
The first minus sign is needed because the difference between two time-like vectors with $n^2 \in [0,1]$ is space-like.
This expression reduces to the usual expression $\theta_{ij}=\sqrt{2n_i^\mu n_{j\mu}}$ in the massless limit.

\begin{figure}[t]
\centering
\includegraphics[scale=0.7]{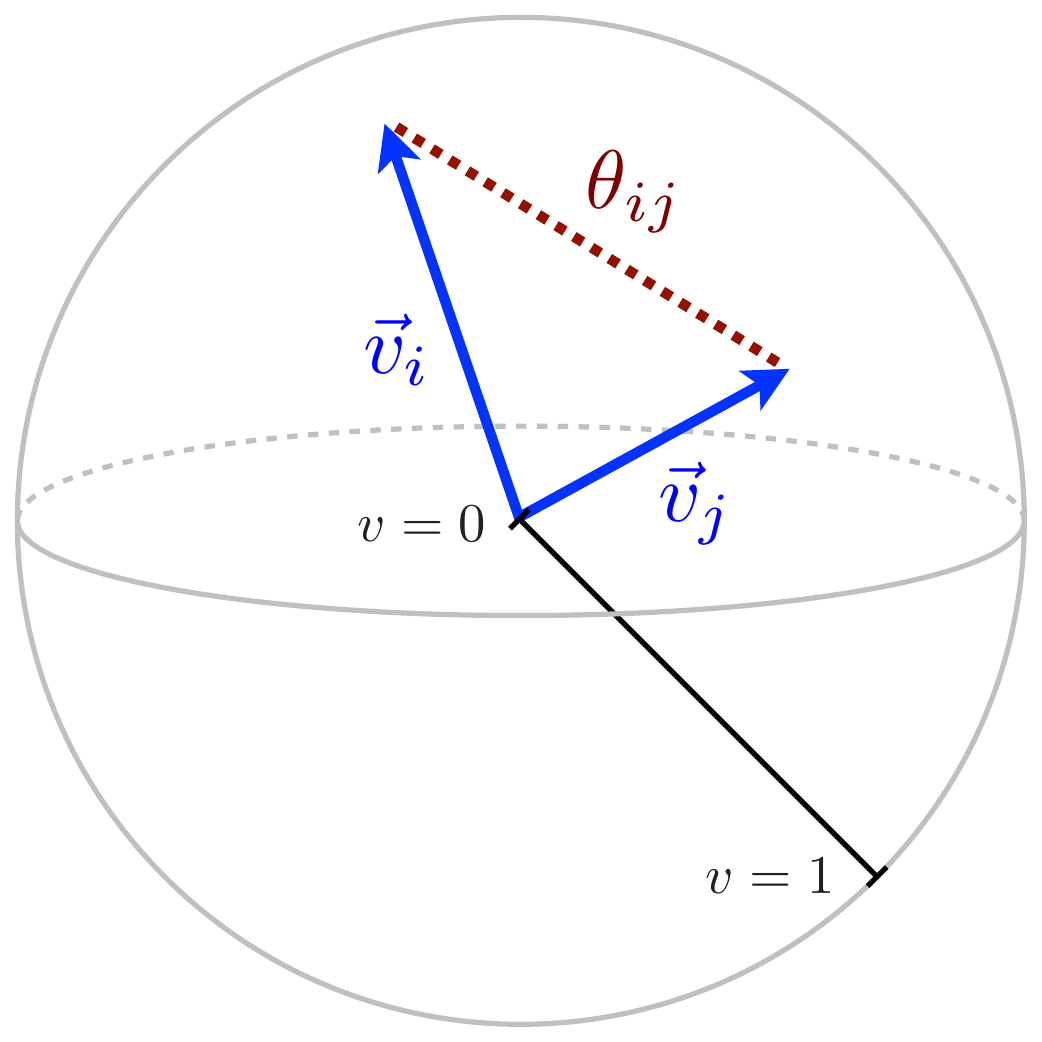}
\caption{\label{fig:sphere} The space of (massive) particle kinematics, with pairwise distances corresponding to Euclidean distances in this space.
Massless particles with $v=1$ live on the boundary and particles at rest with $v=0$ are at the origin.
One can interpret this figure as a snapshot of the event taken at a time $t$ after the collision, when the particles have traveled a distance $v t$.
}
\end{figure}

To gain intuition for this geometric distance between massive particles, it is instructive to expand out \Eq{eq:theta_massive} in the $e^+e^-$ case:
\begin{equation}
\theta_{ij}=\sqrt{\left(\vec v_i-\vec v_j\right)^2}=\sqrt{v_i^2+v_j^2 - 2 \, v_i \, v_j \cos \Omega_{ij}},
\end{equation}
where $\Omega_{ij}=\arccos \hat n_i\cdot \hat n_j$ is the purely geometric angle between particles $i$ and $j$.
We see that the velocity magnitude $v = |\vec{v}|$ acts as a radial coordinate on the sphere, and the pairwise distances $\theta_{ij}$ are just the Euclidean distances between two points in the unit ball, with distances $v_i$ and $v_j$ from the origin and angle $\Omega_{ij}$ between them.
Massless particles live entirely on the boundary with $v=1$ and massive particles live inside the ball with $0\le v<1$.
An illustration of this massive particle phase space is shown in \Fig{fig:sphere}.

The use of this massive distance measure has an interesting interplay with some of the studies in the body of the paper.
For example, the analysis of thrust in \Sec{subsubsec:thrust} involved finding the EMD to the manifold of back-to-back massless particle configurations of potentially unequal energy.
Using the massive particle distance, one could consider finding the EMD to the manifold of all possible two-particle configurations, including massive particles.
For $\beta = 2$, this is equivalent to partitioning the event into two halves with masses $M_A$ and $M_B$ and corresponding energies $E_A$ and $E_B$, and minimizing the quantity $M_A^2 / E_A + M_B^2 / E_B$.
A nice feature of this approach is that the optimal two particle configuration has the same energies and velocities as one would get from clustering the particles in each half.
Note that this approach is closely related to (but not identical to) the original definition of heavy jet mass in \Ref{Clavelli:1981yh} based on minimizing $M_1^2 + M_2^2$.

The idea of optimizing jet regions based on $M^2 / E$ also appears in the jet maximization approach~\cite{Georgi:2014zwa}.
In fact, using the massive distance measure in \Eq{eq:XConeasEMD} with $\beta = 2$ and $N=1$, and repeating the logic in \Ref{Thaler:2015uja}, we recover precisely the algorithm in \Ref{Georgi:2014zwa}, where the parameter $R$ controls the size of the resulting jet region.
The EMD approach yields a natural way to extend the jet maximization algorithm to $N > 1$, and also allows for an alternative definition of $N$-jettiness from \Eq{eq:Njettiness_with_beam_asEMD} based on time-like axes.

As another example, the analysis of recombination schemes in \Sec{sec:seqrec} involved minimizing the transportation cost to merge two particles into one.
For $\beta = 2$, this was equivalent to the $E$-scheme, namely $\kappa = 1$ in \Eq{eq:recomb}, up to the subtlety noted in footnote~\ref{footnote:Escheme}.
Using the massive particle distance, the merged particle in the $E$-scheme has the energy and direction:
\begin{equation}
E_c = E_i + E_j, \qquad n^\mu_c = \frac{E_i n_i^\mu + E_j n_j^\mu}{E_i + E_j},
\end{equation}
where appropriate $T$ subscripts should be included in the hadronic case.
The combined four-vector is
\begin{equation}
p^\mu_c = E_c n^\mu_c = p_i^\mu + p_j^\mu,
\end{equation}
which is a valid expression in both the $e^+e^-$ and hadronic cases.
Thus, the combined four-vector is just the sum of the two particles, which is indeed the desired $E$-scheme behavior.
Note, however, that the interpretation of the jet radius is very different if one uses the massive particle distance, since clustering happens in velocity space.
We leave further studies of the massive particle distance to future work.

\chapter{Energy Flow and the Energy-Momentum Tensor}
\label{sec:stressenergy}

\begin{figure}[t]
\centering
\includegraphics[scale=0.8]{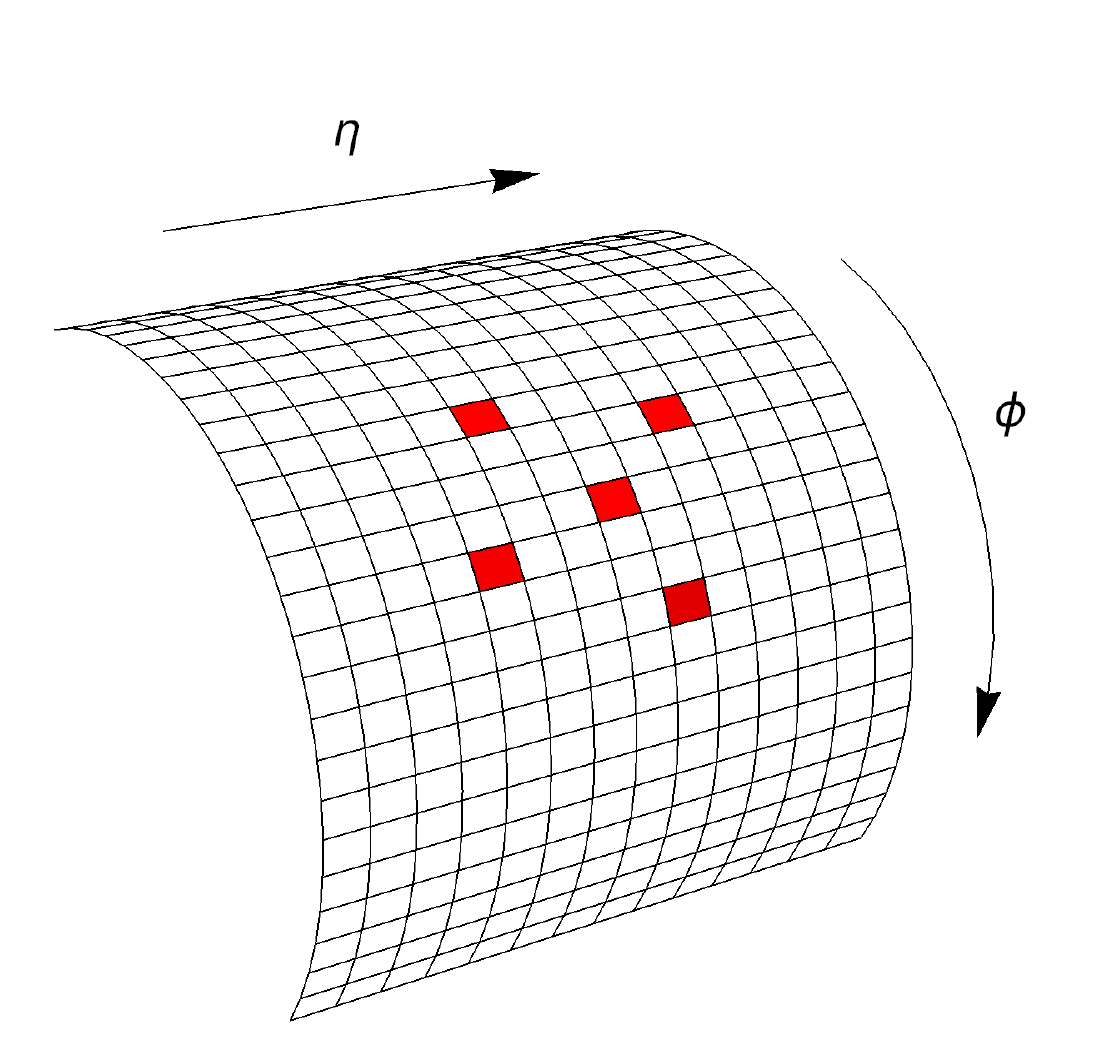}
\caption{An example calorimeter cell configuration to measure a 5-point energy correlator. The red regions indicate the five calorimeter cells chosen to measure the energy infinitely far from the interaction. For each event, the values of the five energy deposits are multiplied together to obtain the value of the observable in \Eq{eq:npointcorr}.}
\label{fig:5correxample}
\end{figure}

In this appendix, I review the connection between the energy flow of an event, as described by the energy-momentum tensor, to multiparticle energy correlators~\cite{Tkachov:1995kk,Sveshnikov:1995vi,Cherzor:1997ak,Tkachov:1999py}.

Consider an idealized hadronic calorimeter cell at position $\hat n$ in pseudorapidity-azimuth $(\eta,\phi)$-space, spanning a patch of size $d\eta\,d\phi$. The \emph{energy flow operator} $\mathcal E_T(\hat n)$ corresponding to the total transverse momentum density flowing into the calorimeter cell can be written in terms of the energy-momentum tensor $T_{\mu\nu}$~\cite{Sveshnikov:1995vi,Korchemsky:1997sy,Lee:2006nr,Bauer:2008dt,Mateu:2012nk} as:
\begin{equation}\label{eq:energyflow}
\mathcal E_T(\hat n) = \frac{1}{\cosh^3  \eta} \lim_{R\to\infty} R^2\int_0^\infty dt\, \hat n_i\, T^{0i}(t, R\hat n),
\end{equation}
with its action on a state $\ket X$ of $M$ massless particles given by:
\begin{equation}\label{eq:transcal}
\mathcal E_T(\hat n) \ket X = \sum_{i\in X} p_{T,i} \delta(\eta -  \eta_i)\delta(\phi - \phi_i) \ket X.
\end{equation}

Next, consider $N$ calorimeter cells at positions $(\hat n_1, \cdots, \hat n_N)$. 
An illustration of an example calorimeter cell configuration is shown in \Fig{fig:5correxample}. 
For an event $X$, multiply together the measured energy deposits in each of these $N$ cells. 
The corresponding observable is then the energy $N$-point correlator as defined in \Refs{basham1978energy,basham1979energy}:
\begin{equation}\label{eq:npointcorr}
\mathcal E_T(\hat n_1) \cdots \mathcal E_T(\hat n_N)\ket X = \sum_{i_1 \in X} \cdots \sum_{i_N \in X} \left[\prod_{j=1}^Np_{T,i_j}\delta^2(\hat n_j-\hat p_{i_j})\right]\ket X,
\end{equation}
where $\hat n_a=( \eta_a,\phi_a)$ for the calorimeters cells and $\hat p_a=( \eta_a,\phi_a)$ for the particles in the event.

We can define a new set of observables in terms of the $N$-point correlators in \Eq{eq:npointcorr}. 
Consider averaging \Eq{eq:npointcorr} over all calorimeter cells with an arbitrary angular weighting function $f_N(\hat n_1,\ldots,\hat n_N)$. 
The resulting observables are then of the form:
\begin{align}\label{eq:corravg}
\mathcal C_N^{f_N} \ket X &= \int d^2\hat n_1 \cdots d^2 \hat n_N \, f_N(\hat n_1,\ldots, \hat n_N)\mathcal E_T(\hat n_1) \cdots \mathcal E_T(\hat n_N) \ket X
\\& = \sum_{i_1\in X} \cdots \sum_{i_N \in X} p_{T,i_1} \cdots p_{T,i_N} f_N(\hat p_{i_1},\ldots, \hat p_{i_N})\ket X,\label{eq:ccorr}
\end{align}
namely, these observables $\mathcal C_N^{f_N}$ written in the form of \Eq{eq:ccorr} are exactly the $C$-correlators defined in \Eq{eq:genccorr}. 
Thus the averaging procedure in \Eq{eq:corravg} relates the particle-level $C$-correlators of \Eq{eq:genccorr} to the energy flow of the energy-momentum tensor $T_{\mu\nu}$.

\chapter{Quark versus Gluon Classification and Top Tagging Results}
\label{app:moretagging}

\begin{figure}[t]
\centering
\subfloat[]{\includegraphics[scale=.76]{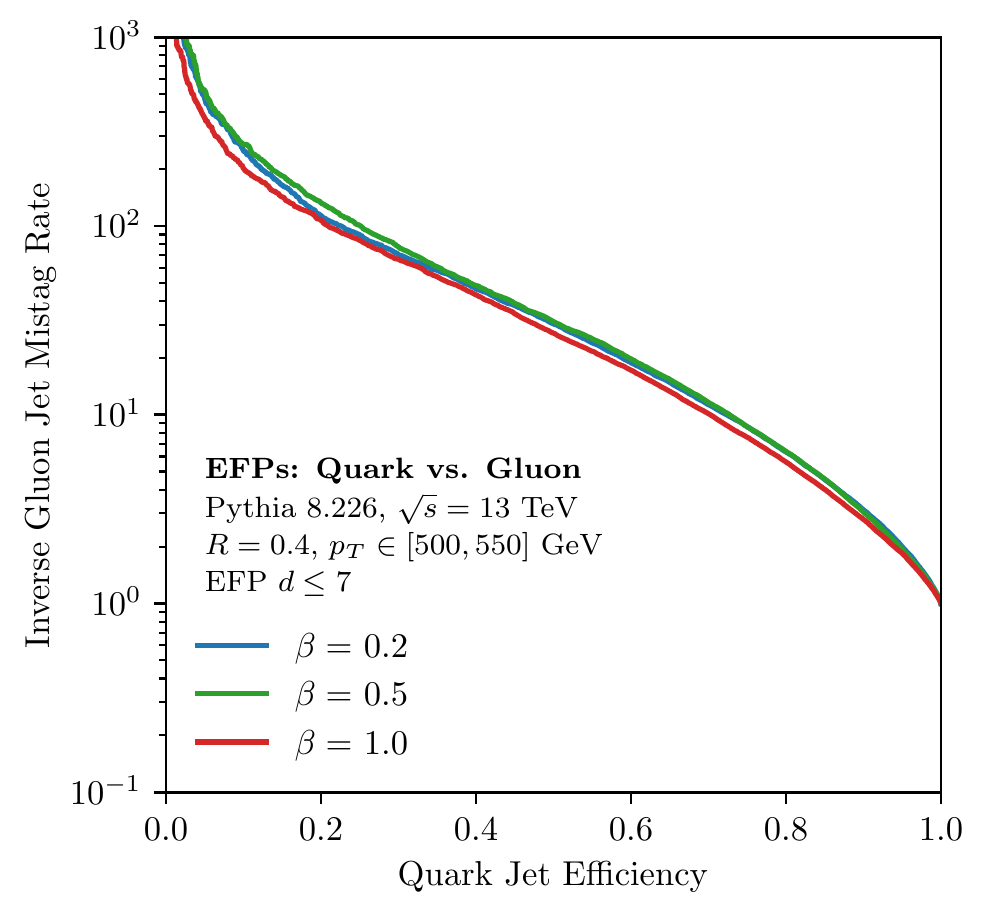}}
\subfloat[]{\includegraphics[scale=.76]{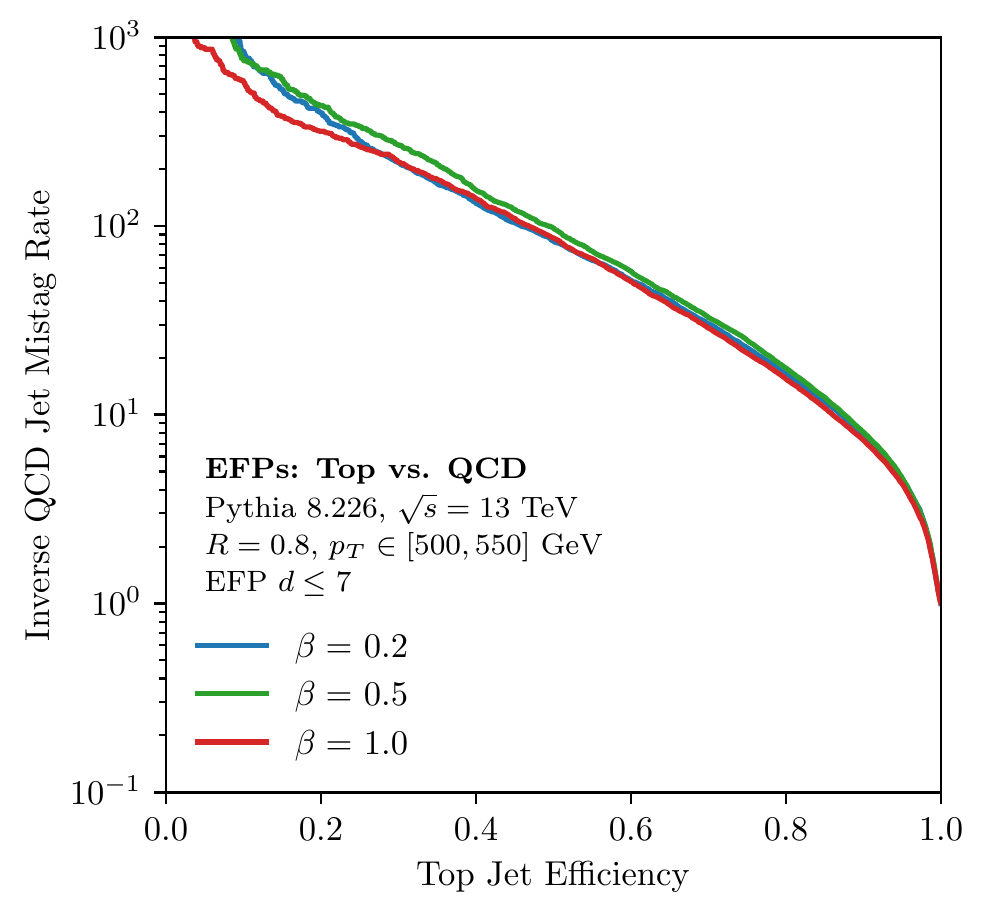}}
\caption{Same as \Fig{fig:Wbetasweep}, but for (a) quark/gluon classification and (b) top tagging.  Similar to the $W$ tagging case, the $\beta = 0.5$ choice has the best performance (marginally) for both tagging problems.}
\label{fig:appbetasweep}
\end{figure}

In this appendix, I supplement the $W$ tagging results of \Sec{sec:linclass} with the corresponding results for quark/gluon classification and top tagging.
The details of the event generation are given in \Sec{sec:eventgen}.

We compare the \B linear classification performance with $\beta = 0.2$, $\beta = 0.5$, and $\beta = 1$ in \Fig{fig:appbetasweep}.  
Consistent with the $W$ tagging case in \Fig{fig:Wbetasweep}, we find that the optimal performance is achieved with $\beta = 0.5$.  We therefore use $\beta = 0.5$ for the remainder of this study.

\begin{figure}[t]
\centering
\subfloat[]{\includegraphics[scale=.76]{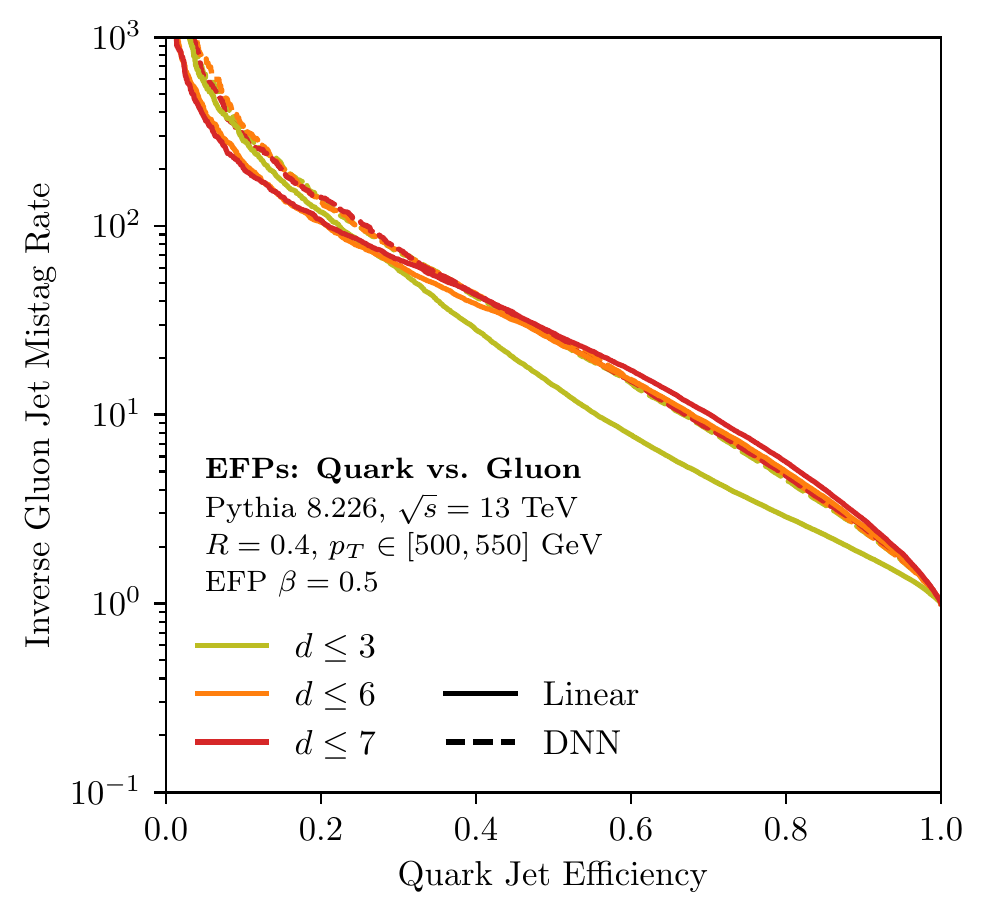}}
\subfloat[]{\includegraphics[scale=.76]{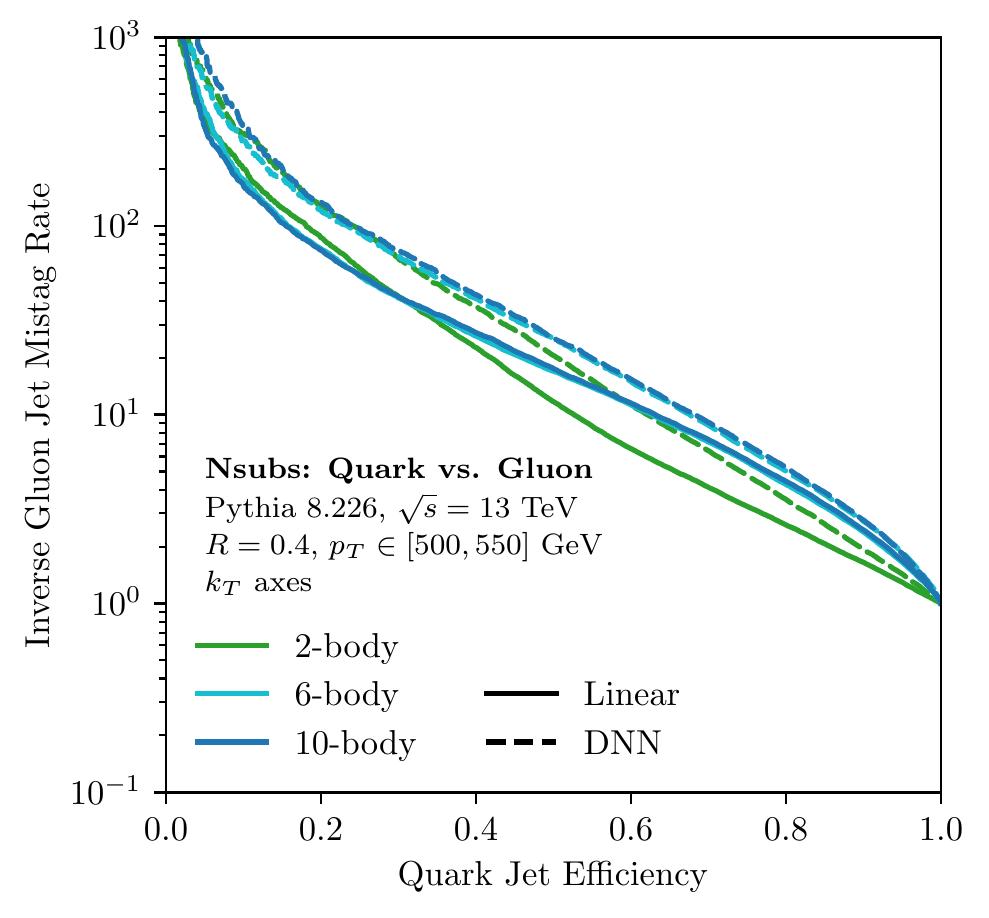}}

\subfloat[]{\includegraphics[scale=.76]{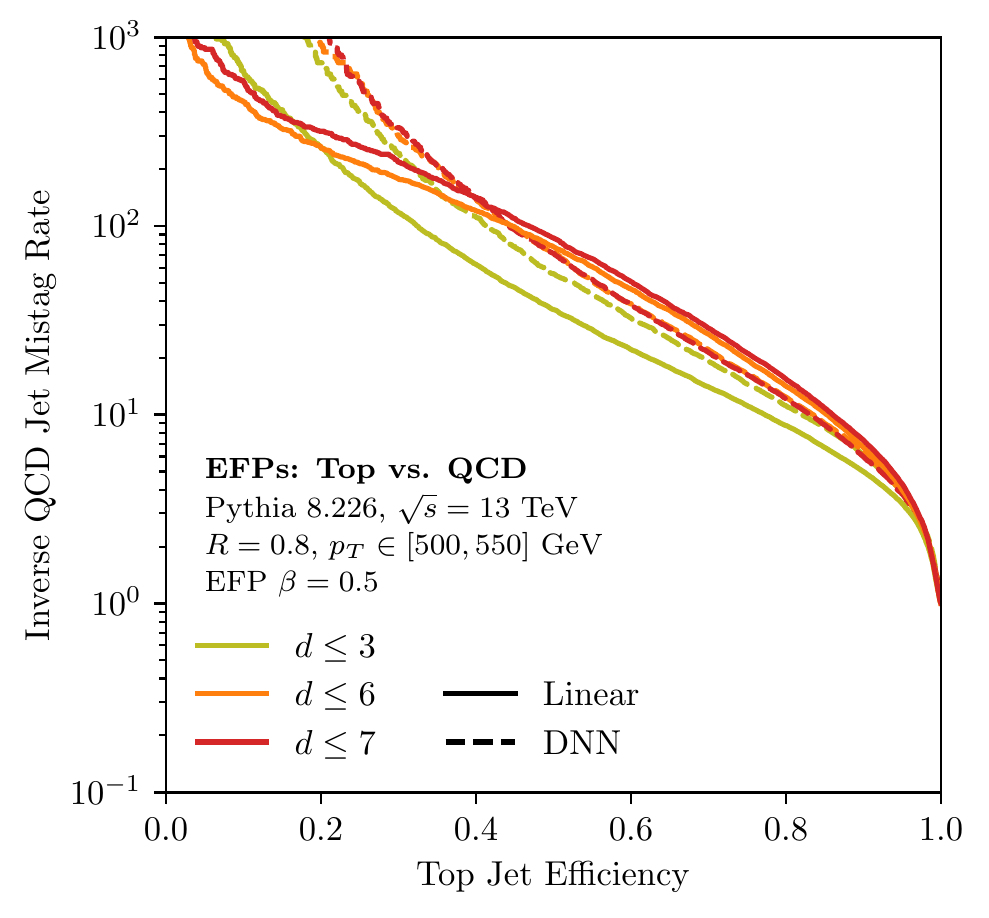}}
\subfloat[]{\includegraphics[scale=.76]{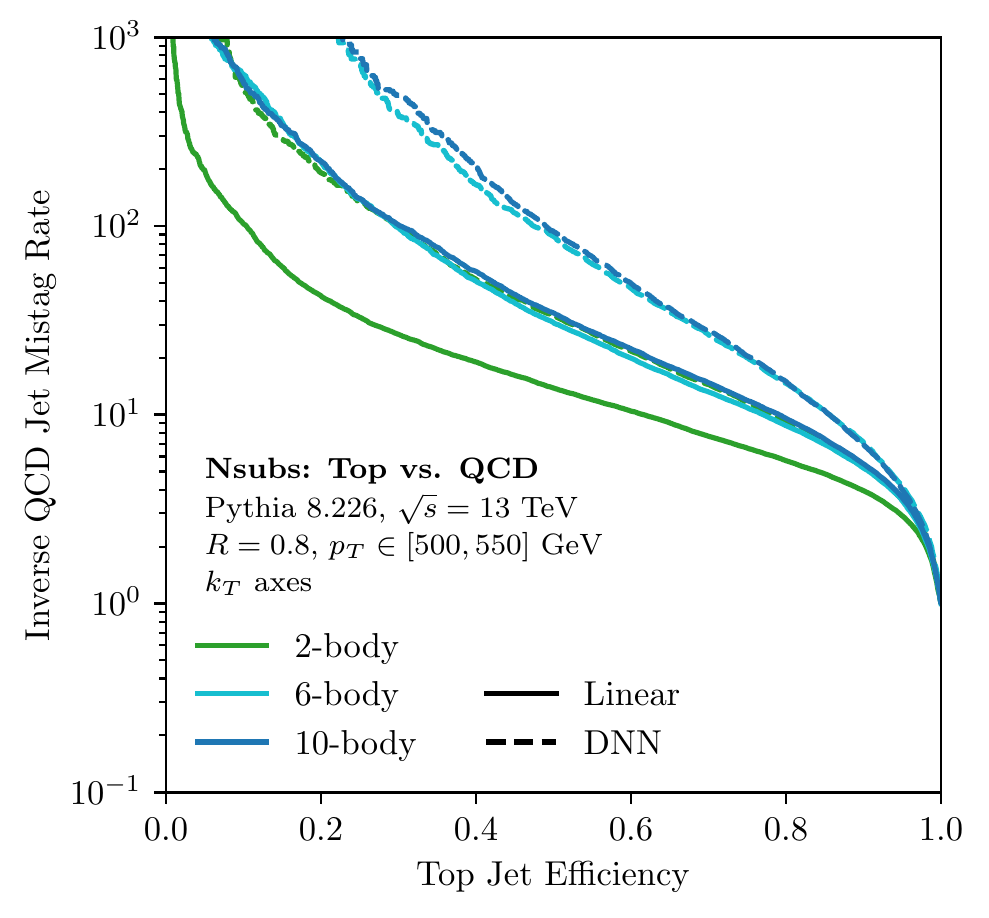}}
\caption{Same as \Fig{fig:Wefpnsubsweep}, but for quark/gluon classification (top) and top tagging (bottom).  As in the $W$ tagging case, the linear combinations of \Bs can be seen to approach (or even exceed) the nonlinear combinations, particularly for higher signal efficiencies.
\label{fig:appefpnsubsweep}}
\end{figure}

In \Fig{fig:appefpnsubsweep} we compare the linear and nonlinear performances of the energy flow basis and the $N$-subjettiness basis.
There is a clear gap between the linear and nonlinear $N$-subjettiness classifiers, whereas no such gap exists for the \Bs.  Interestingly, the linear classifier of \Bs tends to outperform the DNN at medium and high signal efficiencies, indicating the difficulty of training high-dimensional neural networks.
This behavior was not seen in \Fig{fig:Wefpnsubsweep}, most likely because the achievable efficiency is overall higher in the $W$ tagging case.

\begin{figure}[t]
\centering
\subfloat[]{\includegraphics[scale=.76]{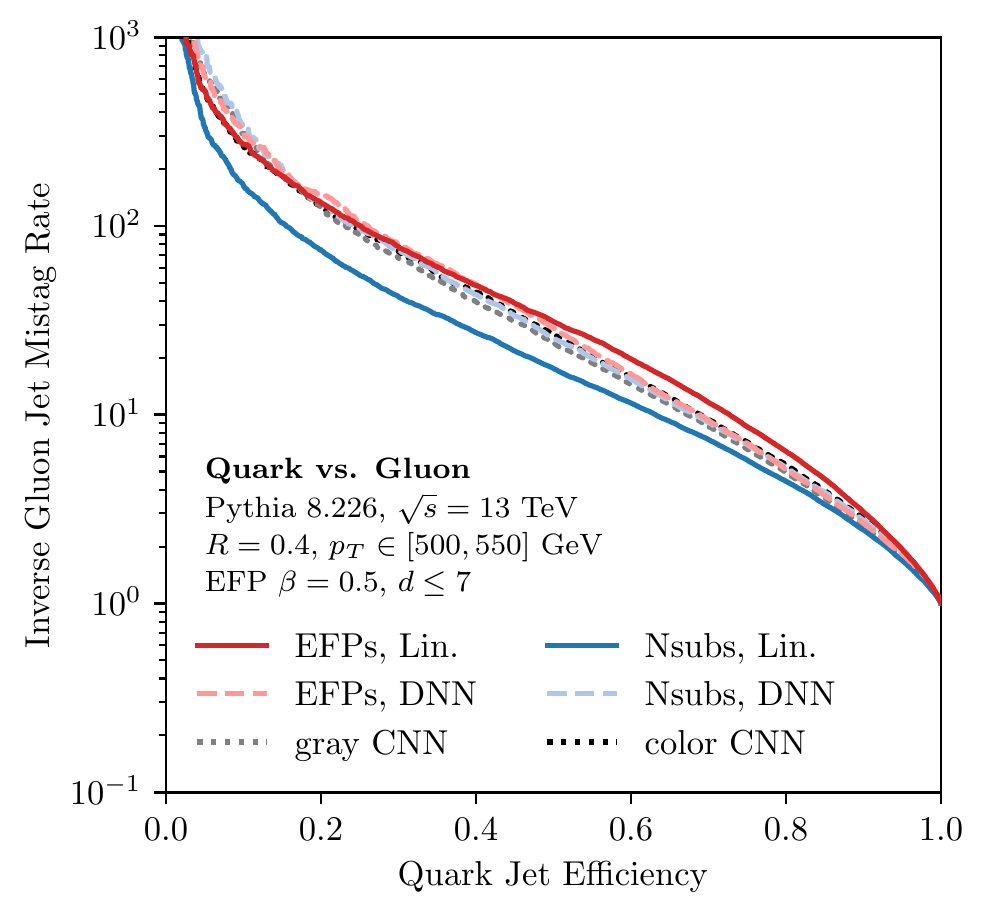}}
\subfloat[]{\includegraphics[scale=.76]{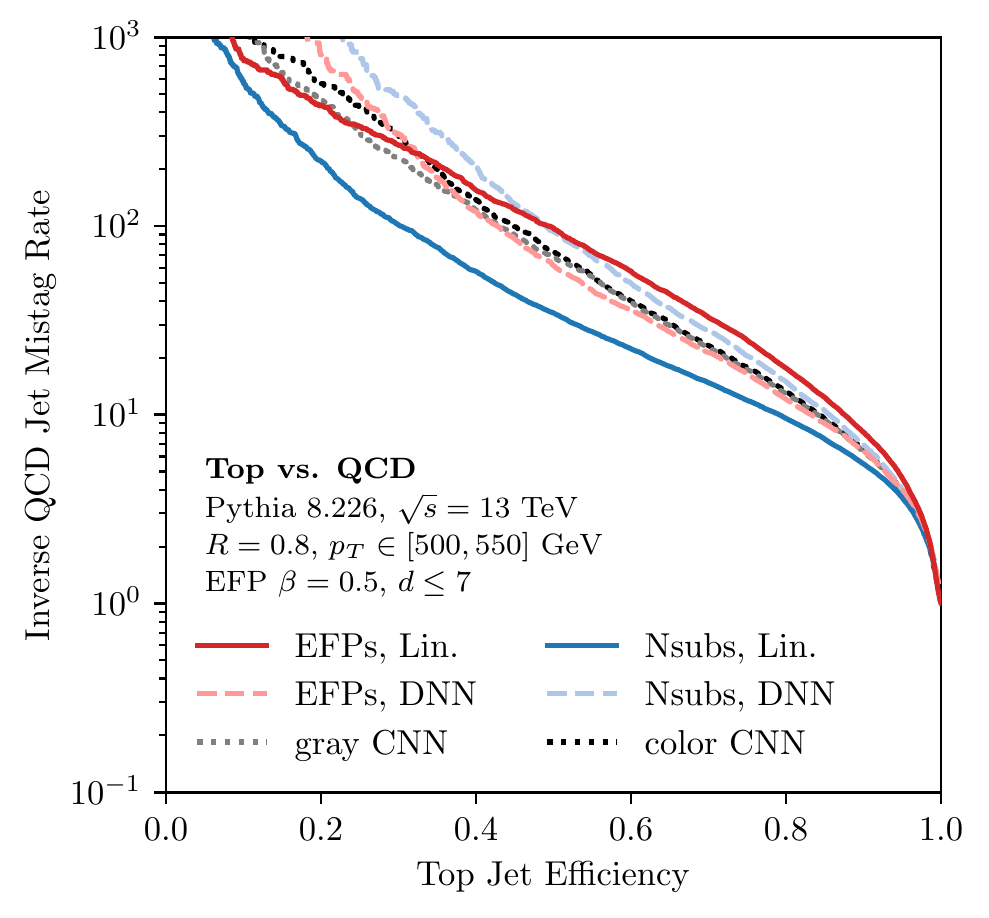}}
\caption{Same as \Fig{fig:Wtagcomp} for (a) quark/gluon classification and (b) top tagging.
As in the $W$ tagging case, the linear classification with \Bs can match (or even outperform) the other methods at high signal efficiencies.}
\label{fig:apptagcomp}
\end{figure}

A summary of the six tagging methods is shown in \Fig{fig:apptagcomp}, comparing linear and nonlinear combinations of the energy flow basis and $N$-subjettiness basis to grayscale and color jet images. 
As in \Fig{fig:Wtagcomp}, linear combinations of \B tend to match or outperform the other methods, especially at high signal efficiencies.

\begin{figure}[t]
\centering
\subfloat[]{\includegraphics[scale=.76]{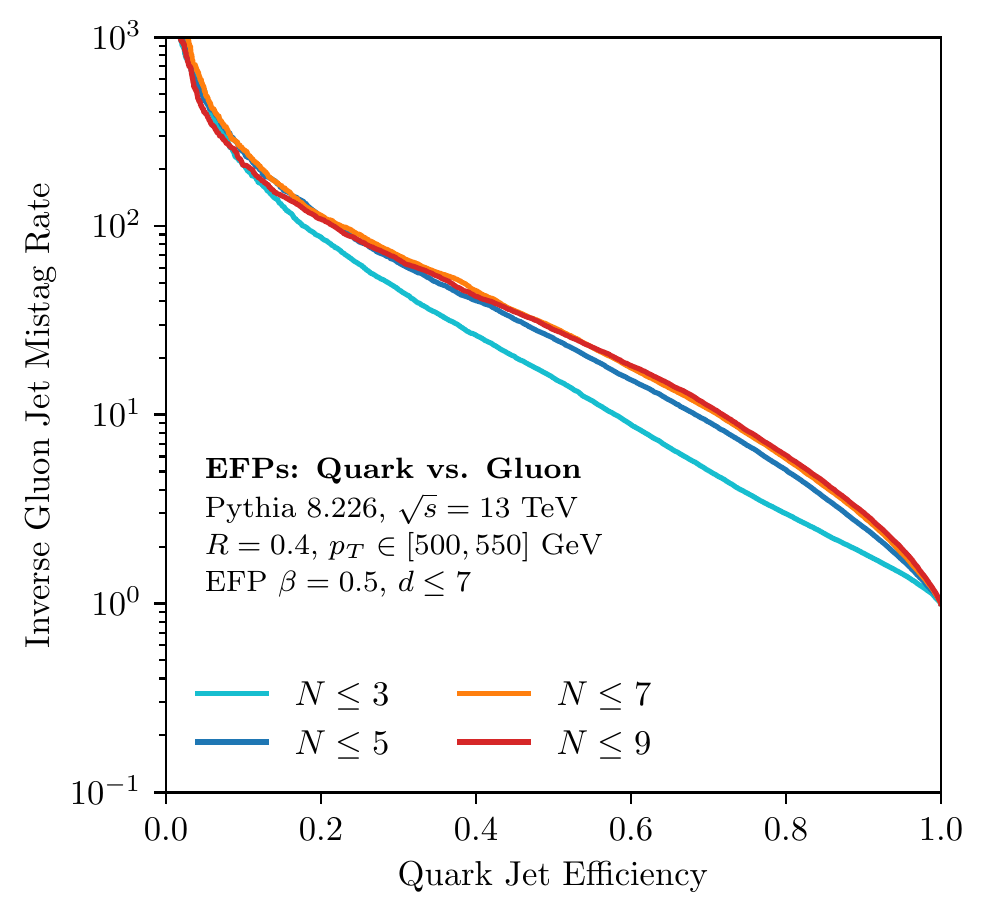}}
\subfloat[]{\includegraphics[scale=.76]{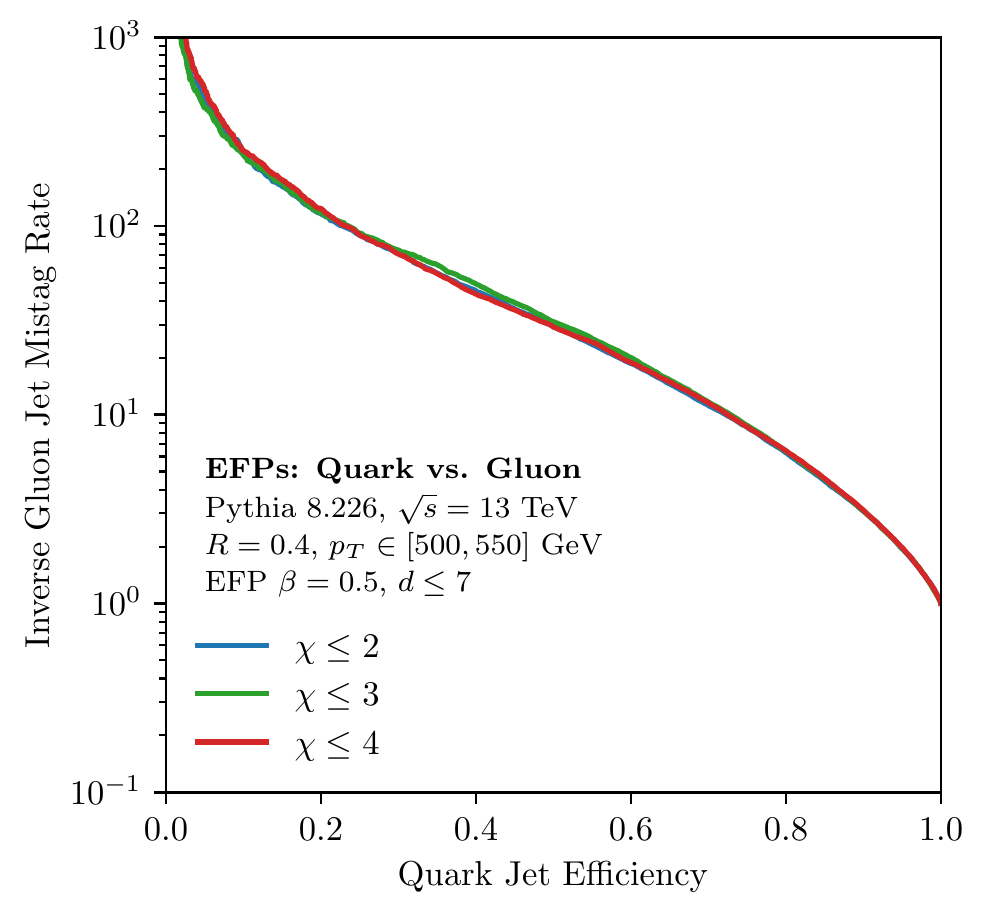}}

\subfloat[]{\includegraphics[scale=.76]{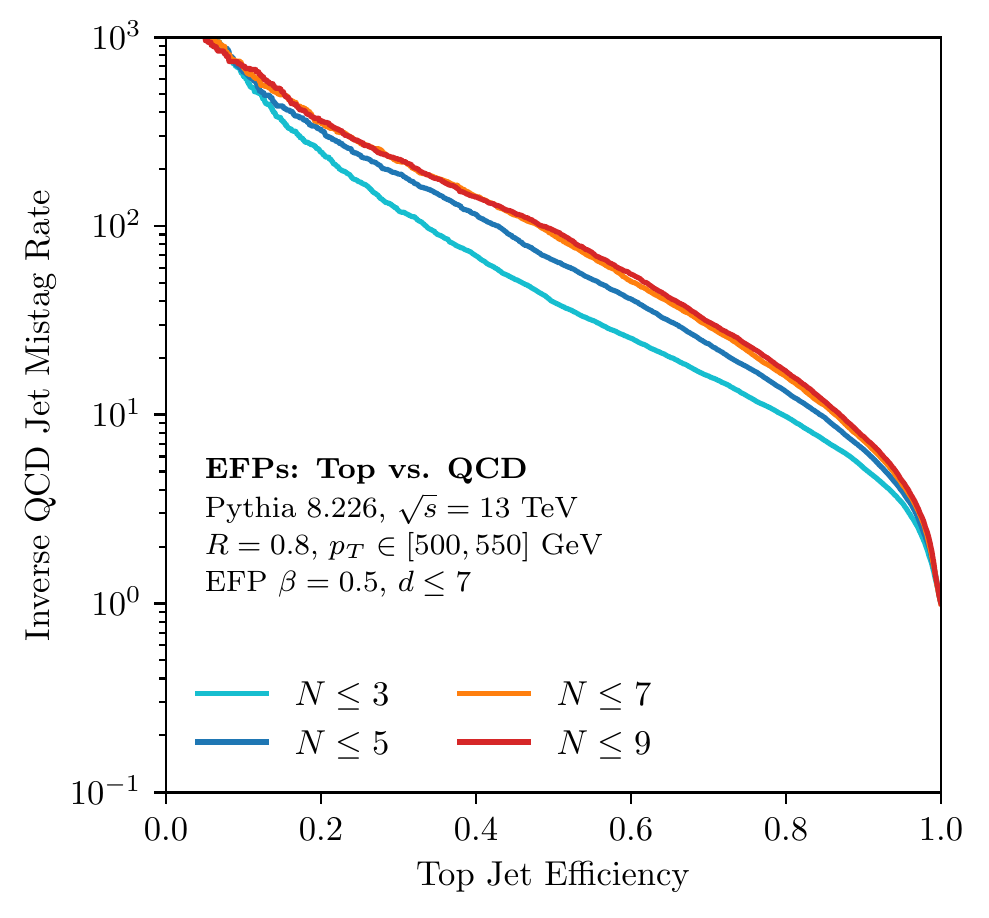}}
\subfloat[]{\includegraphics[scale=.76]{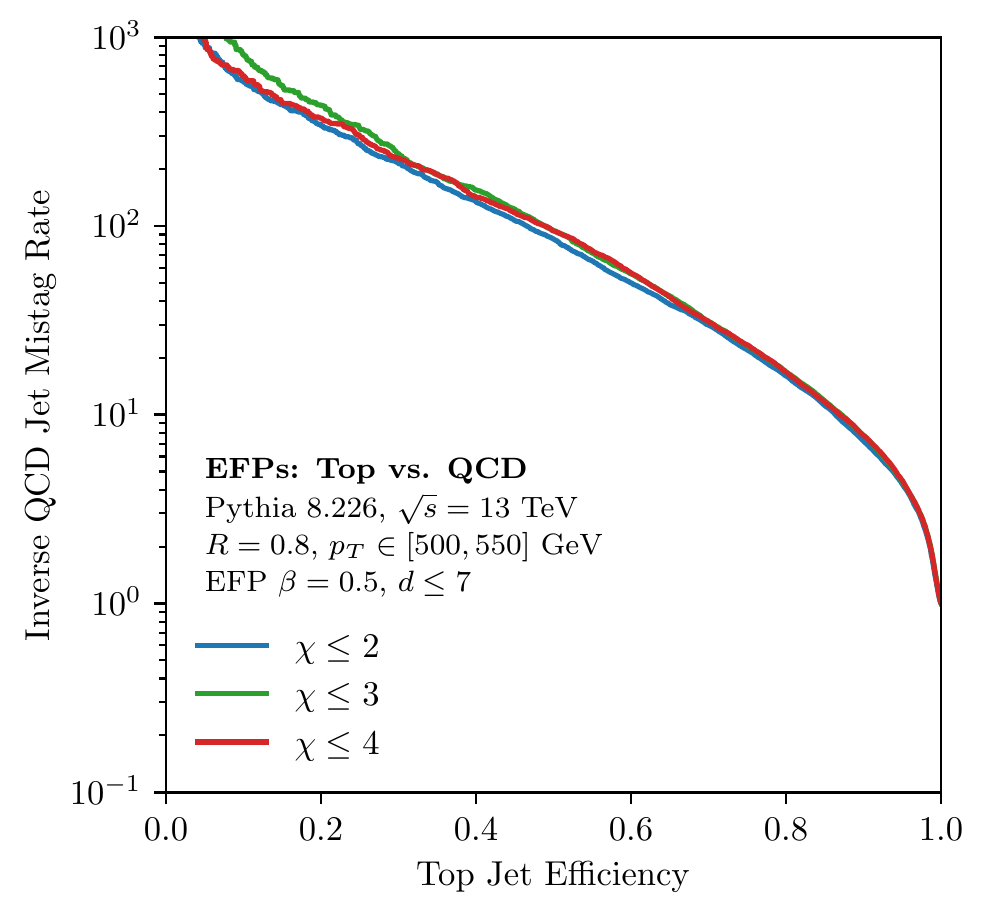}}
\caption{Same as \Fig{fig:efpnsweep} but for (top) quark/gluon classification and (bottom) top tagging.}
\label{fig:appnchisweep}
\end{figure}
\afterpage{\clearpage}

Finally, we truncate the set of \Bs with $d\le 7$ by the number of vertices $N$ and by the VE computational complexity $\chi$ in \Fig{fig:appnchisweep}.  
As in \Fig{fig:efpnsweep}, the higher $N$-particle correlators contribute to the classification performance up to at least $N=7$, whereas the higher-complexity \Bs beyond $\chi = 2$ do not significantly contribute to the classification performance.

\chapter{Exploration of Casimir- and Poisson-scaling Observables}
\label{sec:explore}

In this appendix, I explore the Operational Definition of quark and gluon jets in the leading-logarithmic (LL) limit, focusing on two theoretically-tractable classes of jet observables: Casimir-scaling and Poisson-scaling observables.
Though I only work to lowest non-trivial order, these calculations demonstrate that our framework for defining quark and gluon jets is suitable to theoretical exploration in addition to practical experimental implementation.
In the LL limit of perturbative QCD, quarks and gluons differ in their emission profiles only by their color charges: $C_F=4/3$ for quarks and $C_A=3$ for gluons.
Thus, in the LL limit, quarks and gluons are well-defined (at least at the parton level), providing a simplified context to explore the Operational Definition.
We find different non-zero quark/gluon reducibility factors for Casimir-scaling and Poisson-scaling observables, substantiating the need to use a richer space of jet substructure observables to approximate the full likelihood ratio.

Casimir-scaling observables include common jet substructure observables, such as the jet mass $m$ or IRC-safe angularities~\cite{Berger:2003iw,Almeida:2008yp,Ellis:2010rwa,Larkoski:2014uqa,Larkoski:2014pca}, that are dominated at LL accuracy by a single hard emission.
Their cumulative distributions satisfy $\Sigma_g (m)= \Sigma_q(m)^{C_A/C_F}$, where $p_i(m) = d\Sigma_i/dm$.
Solely using this scaling property, the quark/gluon reducibility factors of Casimir-scaling observables are:
\begin{align}\label{eq:kqgcas}
&\kappa_{qg}^\text{Cas.} = \min_m\frac{p_q(m)}{p_g(m)} = \min_m\frac{\frac{d\Sigma_q}{dm}}{\frac{C_A}{C_F} \Sigma_q^{C_A/C_F - 1}\frac{d\Sigma_q}{dm}}  = \frac{C_F}{C_A} \min_m\Sigma_q^{1-C_A/C_F} = \frac{C_F}{C_A}, \\
&\kappa_{gq}^\text{Cas.} = \min_m\frac{p_g(m)}{p_q(m)} = \min_m\frac{\frac{C_A}{C_F} \Sigma_q^{C_A/C_F - 1}\frac{d\Sigma_q}{dm}}{\frac{d\Sigma_q}{dm}}  = \frac{C_A}{C_F} \min_m\Sigma_q^{C_A/C_F-1} = 0, \label{eq:kgqcas}
\end{align}
where $C_A/C_F > 1$ and $\min_m\Sigma_i(m)=0$ have been used to obtain the last equality.
These results are universal to all Casimir-scaling observables and are independent of the remaining details of the observables at LL accuracy.

\begin{figure}[t]
\centering
\subfloat[]{\includegraphics[scale=0.75]{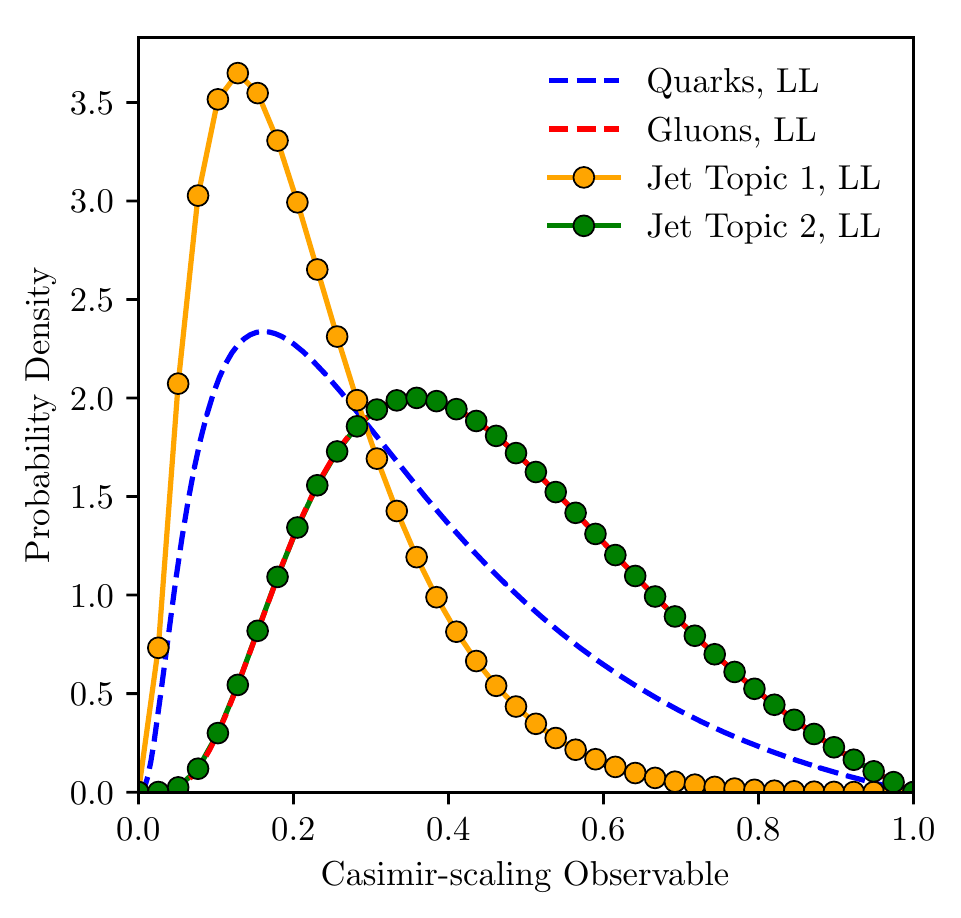}\label{fig:CStopics}}
\subfloat[]{\includegraphics[scale=0.75]{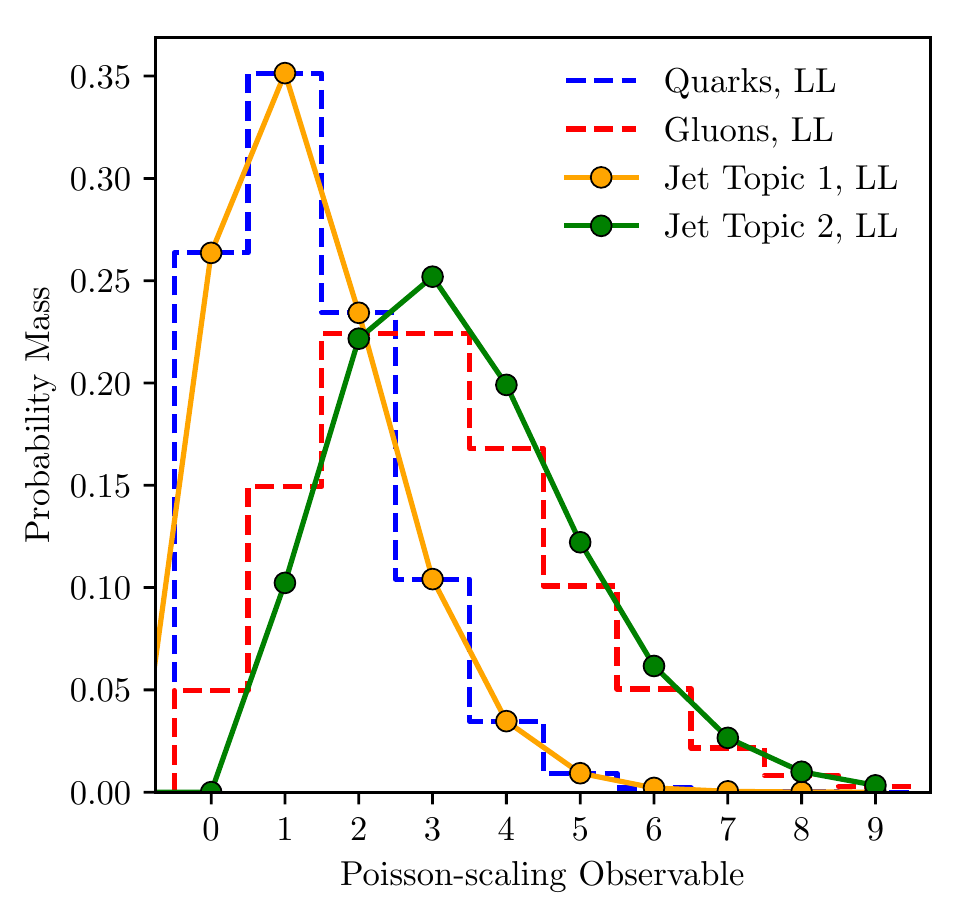}\label{fig:PStopics}}
\caption{
Quark and gluon distributions at LL accuracy for (a) Casimir-scaling and (b) Poisson-scaling observables, together with the corresponding jet topics.
The reducibility of the quark Casimir-scaling distribution and the gluon Poisson-scaling distribution are evident.
While neither of these observables individually results in mutually irreducible quarks and gluons, considering them jointly does.
}
\label{fig:LLtopics}
\end{figure}

The non-zero reducibility factor in \Eq{eq:kqgcas} indicates that quark and gluon jets are not mutually irreducible in the space of Casimir-scaling observables.
In particular, the quark distribution of any Casimir-scaling observable is a mixture of the (irreducible) gluon distribution and some other distribution, as shown in \Fig{fig:CStopics}.
Note that this {\it does not} imply that quark jets are fundamentally reducible, since this is just a property derived from Casimir-scaling observables in the LL limit.
That said, as noted at the end of \Sec{sec:opdef}, if \Eq{eq:kqgcas} were fundamental to quark and gluon jets, one could simply include this reducibility factor in the Operational Definition.

We next consider Poisson-scaling observables, which count the number of perturbative emissions and have qualitatively different quark-gluon reducibility factors.
One example is the soft drop multiplicity $n_{\rm SD}$~\cite{Frye:2017yrw}, which counts the number of emissions restricted to a certain phase space region.
At LL, Poisson-scaling observables are distributed according to Poissonian distributions with means $C_F \lambda$ for quarks and $C_A \lambda$ for gluons, where $\lambda$ is a constant proportional to the area of the emission plane that is counted.
The quark-gluon reducibility factors corresponding to these distributions are then:
\begin{align}\label{eq:kqgpois}
&\kappa_{qg}^\text{Pois.} = \min_n\frac{p_q(n)}{p_g(n)} = \min_n\frac{(C_F \lambda)^n e^{- C_F \lambda}}{(C_A \lambda)^n e^{- C_A \lambda}} =e^{-(C_F - C_A)\lambda}  \min_n \left(\frac{C_F}{C_A}\right)^n = 0,\\
&\kappa_{gq}^\text{Pois.} = \min_n\frac{p_g(n)}{p_q(n)} = \min_n\frac{(C_A \lambda)^n e^{- C_A \lambda}}{(C_F \lambda)^n e^{- C_F \lambda}} = e^{-(C_A-C_F) \lambda} \min_n \left(\frac{C_A}{C_F}\right)^n = e^{-(C_A-C_F)\lambda},\label{eq:kgqpois}\end{align}
since $C_A/C_F > 1$ and $n$ can take any non-negative integer value.

Evidently, Poisson-scaling observables display the {\it opposite} behavior of Casimir-scaling observables: the gluon distribution is a mixture of the (irreducible) quark distribution and some other distribution, as shown in \Fig{fig:PStopics}.
Further, the reducibility factor is not universal to all Poisson-scaling observables but rather depends exponentially on the parameter $\lambda$.
Though $\lambda\sim\mathcal O(1)$ was considered in \Ref{Frye:2017yrw}, perturbative QCD allows for arbitrarily large $\lambda$ by counting emissions in larger and larger regions.
As $\lambda$ increases, the reducibility factor falls to zero much more quickly than the overlap in the distributions decreases, and thus quark and gluon jets rapidly approach mutual irreducibility.
While perturbative control is lost for large $\lambda$ due to non-perturbative effects, considering this limit suggests that there is no fundamental impediment to the mutual irreducibility of quarks and gluons from the perspective of perturbative QCD, at least at LL accuracy.

From these two classes of observables, we see that enriching the feature space beyond individual Casimir-scaling and Poisson observables to $\mathcal O = \{m,n_{\rm SD}\}$ yields $\kappa_{qg} = \kappa_{gq} = 0$ for the combined feature space in the LL limit.
This benefit of using a rich feature space motivates our approach of training data-driven classifiers on complete substructure information to probe the full quark/gluon jet likelihood ratio, rather than relying on individual specially-crafted substructure observables.

\chapter{Details of Observables and Machine Learning Models}
\label{sec:train}

In this appendix, I give details for the jet substructure study in \Sec{sec:qgex}, describing the observables, machine learning models, and model training used.

For the individual substructure observables, three of them use custom implementations: constituent multiplicity $n_\text{const}$, image activity $N_{95}$~\cite{Pumplin:1991kc} (number of pixels in a $33\times33$ jet image containing 95\% of the $p_T$), and jet mass $m$.
The remaining three observables are computed using \textsc{FastJet contrib} 1.033~\cite{fjcontrib}.
The \textsc{RecursiveTools} 2.0.0-beta1 module is used to calculate soft drop multiplicity $n_{\rm SD}$~\cite{Frye:2017yrw} with parameters $\beta=-1$, $z_\text{cut}=0.005$, and $\theta_\text{cut}=0$.
The \textsc{Nsubjettiness} 2.2.4 module is used to calculate the $N$-subjettiness~\cite{Thaler:2010tr,Thaler:2011gf} observables $\tau_N^{(\beta)}$ with $k_T$ axes as recommended in \Ref{Datta:2017rhs}, in particular $\tau_2^{(\beta=1)}$ and jet width $w$ (implemented as $\tau_1^{(\beta=1)}$).

For our trained models, we use several different jet representations and machine learning architectures.
In reverse order compared to \Tab{tab:obsandmodels}, they are:
\begin{itemize}

\item \textbf{DNN}: The $N$-subjettiness basis~\cite{Datta:2017rhs} is a phase space basis in the sense that $3K-4$ independent $N$-subjettiness observables map non-linearly onto $K$-body phase space.
We use 20-body phase space consisting of the following set of $N$-subjettiness basis elements:
\begin{equation}
\left\{\tau_1^{(1/2)},\,\tau_1^{(1)},\,\tau_1^{(2)},\tau_2^{(1/2)},\,\tau_2^{(1)},\,\tau_2^{(2)},\,\ldots,\,\tau_{K-2}^{(1/2)},\,\tau_{K-2}^{(1)},\,\tau_{K-2}^{(2)},\tau_{K-1}^{(1/2)},\,\tau_{K-1}^{(1)}\right\},
\end{equation}
i.e.\ $\tau_N^{(\beta)}$ with $N\in\{1,\ldots,19\}$ and $\beta\in\{1/2,1,2\}$, except $\tau_{19}^{(2)}$ is absent, all computed using the \textsc{Nsubjettiness} 2.2.4 module of \textsc{FastJet contrib} 1.033.
A DNN consisting of three 100-unit fully-connected layers and a 2-unit softmaxed output was trained on the $N$-subjettiness basis inputs.

\item \textbf{CNN}: The jet images approach~\cite{Cogan:2014oua} treats calorimeter deposits as pixel intensities and represents the jet as an image. 
Convolutional neural networks (CNNs) are the typical model of choice when learning from such a representation, and have been successfully implemented for quark/gluon classification~\cite{Komiske:2016rsd}, $W$ tagging~\cite{deOliveira:2015xxd}, and top tagging~\cite{Baldi:2016fql,Guest:2016iqz}.
We calculate $33\times33$ jet images spanning $2R\times2R$ in the rapidity-azimuth plane.
In the language of \Ref{Komiske:2016rsd}, we formulate ``color'' jet images with two channels: the $p_T$ per pixel and the multiplicity per pixel.
Images were standardized by subtracting the mean and dividing by the per-pixel standard deviation of the training set.

A CNN architecture similar to that used in \Ref{Komiske:2016rsd} was employed: three convolutional layers with 48, 32, and 32 filters and filter sizes of $8\times 8$, $4\times 4$, and $4\times 4$, respectively, followed by a 128-unit dense layer. 
Maxpooling of size $2\times2$ was performed after each convolutional layer with a stride length of 2.
The dropout rate was taken to be 0.1 for all convolutional layers and was not used for the dense layer.

\item \textbf{EFPs}: The Energy Flow basis~\cite{Komiske:2017aww} is a linear basis for IRC-safe observables in the sense that any IRC-safe observable is arbitrarily well approximated by a linear combination of Energy Flow Polynomials (EFPs).
As a result of this remarkable property, linear methods can be used for classification and regression and are highly competitive with modern machine learning methods.
The \href{https://energyflow.network}{\texttt{EnergyFlow}} 0.8.2 package~\cite{energyflow} was used to compute EFPs up to $d\le7,\,\chi\le3$ with $\beta=0.5$ using the normalized default hadronic measure.
This yields 996 EFPs in total, including the trivial constant EFP.
This set was used to train a Fisher's Linear Discriminant model with scikit-learn~\cite{scikit-learn}.

\item \textbf{EFN, PFN, PFN-ID}:
Various particle-level network architectures have been proposed to take advantage of the structure of events or jets as sequences of vectors~\cite{Louppe:2017ipp,Andreassen:2018apy,Butter:2017cot,Cheng:2017rdo,Egan:2017ojy,Komiske:2018cqr}.
We choose to focus on the Energy Flow Networks (EFNs) recently introduced in \Ref{Komiske:2018cqr} and shown to be competitive with other particle-level models.
The EFN architecture is designed to have the properties desirable of a model that takes jet constituents as inputs: it is able to handle variable length lists but, critically, is manifestly symmetric under permutations of the elements in the input.
The inputs to an EFN are lists of particles, where a particle is described by its energy fraction, rapidity, and azimuthal angle (the latter two translated to the origin according to the $E$-scheme jet axis).
EFNs construct an internal latent representation of the jet using the particle-level inputs, weighting each particle's contribution by its energy fraction in order to ensure the IRC safety of the internal observables, and then combine the internal jet observables using a DNN backend.
The \href{https://energyflow.network}{\texttt{EnergyFlow}} package contains an implementation of EFNs.

The EFN architecture can be generalized to learn potentially IRC-unsafe internal observables.
This variant is termed a Particle Flow Network (PFN), which can easily incorporate additional particle features such as flavor information; see \Ref{Komiske:2018cqr} for a more thorough discussion. 
In addition to the IRC-safe EFN, our study uses a PFN with only kinematic inputs, and a PFN-ID with both kinematic and particle flavor (or ID) information.
For each network, the per-particle frontend subnetwork has three fully-connected 100-unit layers corresponding to an internal latent representation of 100 jet observables, and the per-jet backend has three fully-connected 100-unit layers that combines the internal latent observables.
The EFN, PFN, and PFN-ID networks differ only in their inputs and whether the energy fractions are used as weights for the internal sum over particles (for the EFN) or passed to the frontend subnetwork (for the PFN and PFN-ID).

\end{itemize}

All of the above models (excepting the linear EFPs) were implemented and trained using Keras~\cite{keras} with the TensorFlow~\cite{tensorflow} backend.
Training/validation and test datasets were each constructed using 500,000 events for each jet sample being considered.
The training/validation dataset is further divided with 90\% used for training and the remaining 10\% used for validation. 
Properties common to all networks were the use of ReLU activations~\cite{nair2010rectified} for each non-output layer, a 2-unit softmaxed output layer, He-uniform initialization~\cite{heuniform} of the model weights, the categorical crossentropy loss function, the Adam optimization algorithm~\cite{adam}, a learning rate of 0.001, and a patience parameter of 10 epochs monitoring the validation loss.
Models are trained 25 times, making use of different random weight initializations, and the best one is selected according to the maximum Area Under the (mixed sample ROC) Curve.
The hyperparameters of each model were not optimized for either classification performance or accuracy of the ultimately extracted fractions but rather are demonstrative of typical performance that can be achieved.
Practical users of the Operational Definition should tune the hyperparameters for their own purpose.

Finally, it should be noted that other data-driven criteria can be used to select optimal trained models, though we do not explore this further here.
One idea is that since the regions of the ROC curve that are relevant for topic extraction are those with very low and very high signal efficiency, in practice it may be beneficial to optimize training for these regions directly.
A method for optimizing loss-function based training by operating point is described in \Ref{rocopt}, and it would be fascinating to explore this for training better models for topic extraction.

\begin{singlespace}
\bibliography{main}
\bibliographystyle{JHEP}
\end{singlespace}

\end{document}